\documentclass[a4paper,10pt,numbers=noenddot,headsepline=true,DIV=8,BCOR=10mm]{scrbook}

\usepackage[utf8]{inputenc}
\usepackage[T1]{fontenc}
\usepackage{textcomp}
\usepackage[english,ngerman]{babel}
\usepackage{color}
\usepackage{amsmath}
\usepackage{amsopn}
\usepackage{amsbsy}
\usepackage[cal=boondoxo,bb=boondox]{mathalfa}
\usepackage{enumerate}
\usepackage[plainpages=false,pdfpagelabels,breaklinks=true,pdfborderstyle={/S/U/W 1.2}]{hyperref}
\usepackage{graphics}
\usepackage{units}
\usepackage{tabularx}

\numberwithin{equation}{section} 
\numberwithin{table}{section} 
\numberwithin{figure}{section} 

\usepackage[amsmath,thmmarks,thref,hyperref]{ntheorem}

\theoremstyle{plain}
\newtheorem{thm}{Theorem}[section]
\newtheorem{defn}[thm]{Definition}

\newtheorem{lem}[thm]{Lemma}
\newtheorem{cor}[thm]{Corrolary}
\newtheorem{prop}[thm]{Proposition}
\newtheorem{assumption}[thm]{Assumption}

\theorembodyfont{\upshape}
\newtheorem{remark}[thm]{Remark}

\theoremstyle{nonumberplain}
\newtheorem{example}{Example}

\theoremsymbol{\ensuremath{_\Box}}
\newtheorem{proof}{Proof}
\usepackage[style=alphabetic,
	citestyle=alphabetic,
	firstinits=true,
	backend=biber,
	hyperref=auto,
	sorting=nyt,
	maxcitenames=3,
	maxbibnames=99,
	minnames=3,
	arxiv=pdf,
	url=false,
	isbn=false,
	dateabbrev=true,
	datezeros=true,
	bibencoding=utf8]{biblatex}
\newbibmacro{string+doi}[1]{%
	\iffieldundef{doi}{#1}{\href{http://dx.doi.org/\thefield{doi}}{#1}}}
\DeclareFieldFormat
	{title}{\usebibmacro{string+doi}{\mkbibemph{#1}}\isdot}
\DeclareFieldFormat
	[article,inbook,incollection,inproceedings,patent,thesis,unpublished]
	{title}{\usebibmacro{string+doi}{\mkbibemph{#1}}\isdot}
\DeclareFieldFormat
	[article,inbook,incollection,inproceedings,patent,thesis,unpublished]
	{pages}{#1}
\DeclareFieldFormat
	[article]
	{journaltitle}{#1\isdot}
\DeclareFieldFormat
	[article]
	{volume}{\textbf{#1}}
\DefineBibliographyStrings{english}{pages={}}
\DeclareBibliographyDriver{article}{%
	\usebibmacro{bibindex}%
	\usebibmacro{begentry}%
	\usebibmacro{author/translator+others}%
	\setunit{\labelnamepunct}\newblock
	\usebibmacro{title}%
	\usebibmacro{journal+issuetitle}%
	\setunit*{\addcomma\space}
	\printfield{pages}
	\setunit*{\addcomma\space}
	\printfield{year}
	\usebibmacro{finentry}%
}

\newbibmacro*{journal+issuetitle}{%
	\usebibmacro{journal}%
	\setunit*{\addspace}%
	\iffieldundef{series}
	{}
	{\newunit
		\printfield{series}%
		\setunit{\addspace}}%
	\iffieldundef{volume}
		{}
		{\printfield{volume}%
		\setunit{\addspace}}%
	\setunit*{\addcolon\space}%
	\usebibmacro{issue}%
	\newunit}

\bibliography{/Users/max/Library/texmf/tex/latex/max/bibliography}

\usepackage[scaled]{beramono}
\usepackage{helvet}
\usepackage[charter]{mathdesign} 
\SetMathAlphabet{\mathcal}{normal}{OMS}{cmsy}{m}{n} 
\SetMathAlphabet{\mathcal}{bold}{OMS}{cmsy}{m}{n} 
\usepackage{makeidx}
\usepackage{./supplemental/diagxy}
\usepackage{xcolor}

\hypersetup{
	bookmarks=true,
	pdftitle={Semiclassical Dynamics and Magnetic Weyl Calculus},
	pdfauthor={Max Lein},
	pdfproducer={http://www.maxlein.com}
}

\colorlet{chapter}{black!75}
\addtokomafont{chapter}{\color{chapter}}

\makeatletter
\renewcommand*{\chapterformat}{%
\begingroup
\setlength{\unitlength}{1mm}%
\begin{picture}(20,40)(0,5)%
\setlength{\fboxsep}{0pt}%
\put(20,15){\line(1,0){\dimexpr
\textwidth-20\unitlength\relax\@gobble}}%
\put(0,0){\makebox(20,20)[r]{%
\fontsize{28\unitlength}{28\unitlength}\selectfont\thechapter
\kern-.04em
}}%
\put(20,15){\makebox(\dimexpr
\textwidth-20\unitlength\relax\@gobble,\ht\strutbox\@gobble)[l]{%
\ \normalsize\color{black}\chapapp~\thechapter\autodot
}}%
\end{picture} 
\endgroup
}

\newcolumntype{Y}{>{\small\raggedright}X} 

\providecommand{\FourierTrafo}{\mathfrak{F}}
\providecommand{\sympForm}{\sigma}
\providecommand{\Pspace}{\Xi}
\providecommand{\Cont}{\mathcal{C}}
\providecommand{\BCont}{\mathcal{BC}}
\providecommand{\Cpol}{\Cont_{\mathrm{pol}}}
\providecommand{\Fourier}{\FourierTrafo}
\providecommand{\spec}{\mathrm{spec}}
\providecommand{\hooklongrightarrow}{\lhook\joinrel\longrightarrow}
\providecommand{\R}{\mathbb{R}}
\providecommand{\C}{\mathbb{C}}
\providecommand{\T}{\mathbb{T}}
\renewcommand{\R}{\mathbb{R}}
\renewcommand{\C}{\mathbb{C}}
\renewcommand{\T}{\mathbb{T}}
\providecommand{\ie}{i.~e.~}
\providecommand{\eg}{e.~g.~}

\providecommand{\ie}{i.~e.~}
\providecommand{\eg}{e.~g.~}

\providecommand{\R}{\mathbb{R}}

\providecommand{\C}{\mathbb{C}}
\renewcommand{\C}{\mathbb{C}}

\providecommand{\T}{\mathbb{T}}
\renewcommand{\T}{\mathbb{T}}
\providecommand{\N}{\mathbb{N}}
\providecommand{\Z}{\mathbb{Z}}

\renewcommand{\Re}{\mathrm{Re} \,}
\renewcommand{\Im}{\mathrm{Im} \,}
\providecommand{\Hil}{\mathcal{H}}

\providecommand{\eps}{\varepsilon}
\providecommand{\Cont}{\mathcal{C}}

\providecommand{\ker}{\mathrm{ker} \, }

\providecommand{\supp}{\mathrm{supp} \,}
\providecommand{\ker}{\mathrm{ker} \,}

\providecommand{\trace}{\mathrm{Tr} \,}

\providecommand{\dd}{\mathrm{d}}
\providecommand{\id}{\mathrm{id}}

\providecommand{\order}{\mathcal{O}}
\providecommand{\Fourier}{\mathcal{F}}
\providecommand{\noverk}[2]{{#1 \choose #2}}
\providecommand{\trace}{\mathrm{Tr}}

\providecommand{\abs}[1]{\left \lvert #1 \right \rvert}
\providecommand{\sabs}[1]{\lvert #1 \vert}
\providecommand{\babs}[1]{\bigl \lvert #1 \bigr \rvert}
\providecommand{\Babs}[1]{\Bigl \lvert #1 \Bigr \rvert}
\providecommand{\norm}[1]{\left \lVert #1 \right \rVert}
\providecommand{\snorm}[1]{\lVert #1 \rVert}
\providecommand{\bnorm}[1]{\bigl \lVert #1 \bigr \rVert}
\providecommand{\Bnorm}[1]{\Bigl \lVert #1 \Bigr \rVert}
\providecommand{\scpro}[2]{\left \langle #1 , #2 \right \rangle}
\providecommand{\sscpro}[2]{\langle #1 , #2 \rangle}
\providecommand{\bscpro}[2]{\bigl \langle #1 , #2 \bigr \rangle}
\providecommand{\Bscpro}[2]{\Bigl \langle #1 , #2 \Bigr \rangle}

\providecommand{\sopro}[2]{\vert #1 \rangle \langle #2 \vert}

\providecommand{\expval}[1]{\left \langle #1 \right \rangle}
\providecommand{\sexpval}[1]{\langle #1 \rangle}
\providecommand{\bexpval}[1]{\bigl \langle #1 \bigr \rangle}

\providecommand{\Op}{\mathfrak{Op}}
\providecommand{\OpA}{\Op^A}

\providecommand{\Fs}{\FourierTrafo_{\sympForm}}
\providecommand{\PA}{\mathsf{P}^A}
\providecommand{\Weyl}{\sharp}
\providecommand{\magW}{\Weyl^B}
\providecommand{\magWel}{\magW_{\eps,\lambda}}
\providecommand{\Hoer}[1]{S^{#1}}
\providecommand{\Hoerr}[1]{S^{#1}_{\rho,0}}
\providecommand{\Hoerrd}[1]{S^{#1}_{\rho,\delta}}
\providecommand{\Hoermrd}[3]{S^{#1}_{#2,#3}}
\providecommand{\Hoerran}[1]{S^{#1}_{\rho,0}(\syspace)}
\providecommand{\Hoerrdan}[1]{S^{#1}_{\rho,\delta}(\syspace)}
\providecommand{\Hoermrdan}[3]{S^{#1}_{#2,#3}(\syspace)}
\providecommand{\Hoeran}[1]{S^{#1}(\syspace)}
\providecommand{\WeylSys}{W}
\providecommand{\WignerTrafo}{\mathcal{W}}
\providecommand{\Pe}{{\mathsf{P}}}
\providecommand{\PA}{{\P^A}}
\providecommand{\Qu}{\mathsf{Q}}

\providecommand{\GAe}{{\Gamma^A_{\eps}}}

\providecommand{\GAe}{{\Gamma^A_{\eps}}}
\providecommand{\gBe}{{\gamma^B_{\eps}}}
 
\providecommand{\oBel}{\omega^B_{\eps,\lambda}} 

\providecommand{\spec}{\mathrm{spec}}

\providecommand{\Qe}{\mathsf{Q}}

\providecommand{\PSpace}{T^{\ast} \R^d}

\providecommand{\ordere}[1]{\mathcal{O}(\eps^{#1})}
\providecommand{\orderl}[1]{\mathcal{O}(\lambda^{#1})}

\providecommand{\orderone}{\mathcal{O}(1)}
\providecommand{\eprec}{\epsilon}
\providecommand{\ordereprec}{\mathcal{O}(\eprec +)}

\providecommand{\order}{\mathcal{O}}

\providecommand{\Piref}{\Pi_{\mathrm{ref}}}
\providecommand{\piref}{\pi_{\mathrm{ref}}}

\providecommand{\Hfast}{\mathcal{H}_{\mathrm{fast}}}
\providecommand{\Hslow}{\mathcal{H}_{\mathrm{slow}}}

\providecommand{\Fs}{\mathcal{F}_{\sigma}}

\providecommand{\Weyl}{\star}
\providecommand{\Weyle}{\Weyl_{\eps}}
\providecommand{\magW}{\sharp^B}

\providecommand{\magWel}{\Weyl^B_{\eps,\lambda}}

\providecommand{\Ael}{A^{\eps,\lambda}}

\providecommand{\Bte}{\tilde{B}^{\eps}}

\providecommand{\Op}{\mathrm{Op}}

\providecommand{\OpAel}{\mathrm{Op}^A_{\eps,\lambda}}
\providecommand{\Ope}{\mathrm{Op}_{\eps}}

\providecommand{\WeylSysAel}{{W^A_{\eps,\lambda}}}

\providecommand{\WignerTrafoel}{{\mathcal{W}^A_{\eps,\lambda}}}

\providecommand{\minsAl}{\vartheta^A_{\lambda}}

\providecommand{\Hoer}[1]{\mathcal{S}^{#1}}
\providecommand{\SemiHoer}[1]{\mathrm{A}\Hoer{#1}}
\providecommand{\SemiHoerr}[1]{\mathrm{A}\Hoerr{#1}}

\providecommand{\Schwartz}{\mathcal{S}}

\providecommand{\sfnorm}[3]{\Vert #1 \Vert_{#2 , #3}}

\providecommand{\Int}{\mathfrak{Int}}
\providecommand{\ad}{\mathrm{ad}}
\providecommand{\adfrak}{\mathfrak{ad}}

\providecommand{\Adfrak}{\mathfrak{Ad}}

\providecommand{\Hil}{\mathcal{H}}
\providecommand{\Cont}{\mathcal{C}}
\providecommand{\BCont}{\mathcal{C}_{\mathrm{b}}}
\providecommand{\Weyl}{\sharp}

\providecommand{\piref}{\pi_{\mathrm{ref}}}

\providecommand{\Hper}{H_{\mathrm{per}}}

\providecommand{\Zak}{\mathcal{Z}}
\providecommand{\BZ}{M^*}
\providecommand{\WS}{M}
\providecommand{\heff}{{h_{\mathrm{eff}}}}
\providecommand{\El}{{\mathsf{E}}}

\providecommand{\Hoer}[1]{\mathcal{S}^{#1}}
\providecommand{\WeylProd}{\star}
\providecommand{\Weyle}{\WeylProd_{\eps}}
\providecommand{\PSpace}{\Xi}
\providecommand{\Op}{\mathrm{Op}}

\providecommand{\minsAl}{\vartheta^A_{\lambda}}
\providecommand{\Bte}{\tilde{B}^{\eps}}
\providecommand{\eprec}{\epsilon}
\providecommand{\eprec}{\epsilon_{\mathrm{prec}}}
\providecommand{\ordereprec}{\mathcal{O}(\eprec +)}

\providecommand{\ordere}[1]{\mathcal{O}(\eps^{#1})}
\providecommand{\orderl}[1]{\mathcal{O}(\lambda^{#1})}
\providecommand{\orderone}{\mathcal{O}(1)}

\providecommand{\heff}{h_{\mathrm{eff}}}
\providecommand{\Cpol}{\mathcal{C}^{\infty}_{\mathrm{pol}}}

\providecommand{\heff}{{h_{\mathrm{eff}}}}
\providecommand{\spec}{\sigma}

\providecommand{\PA}{\mathsf{P}^A}
\providecommand{\Qe}{\mathsf{Q}}
\providecommand{\WeylSysAel}{{W^A_{\eps,\lambda}}}

\providecommand{\WeylSys}{W}

\providecommand{\Htau}{\Hil_{\tau}}
\providecommand{\Schwartz}{\mathcal{S}}
\providecommand{\magW}{\Weyl^B}
\providecommand{\magWel}{\magW_{\eps,\lambda}}
\providecommand{\WS}{M}
\providecommand{\Hper}{H_{\mathrm{per}}}

\providecommand{\Weyl}{\mathrm{W}^M_{\varepsilon}}

\providecommand{\noverk}[2]{\left ( 
\begin{matrix}
	#1 \\
	#2 \\
\end{matrix}
\right )}

\providecommand{\Schwartz}{\mathcal{S}}
\providecommand{\Piref}{\Pi_{\mathrm{ref}}}

\providecommand{\Zak}{\mathcal{Z}}
\providecommand{\Hfast}{\mathcal{H}_{\mathrm{fast}}}
\providecommand{\Hslow}{\mathcal{H}_{\mathrm{slow}}}
\providecommand{\Hper}{H_{\mathrm{per}}} 

\providecommand{\piref}{\pi_{\mathrm{ref}}}
\providecommand{\Index}{\mathcal{I}}
\providecommand{\Htau}{L^2_{\tau}(M^{\ast},L^2(M))}

\providecommand{\OpA}{\mathrm{Op}^A}
\providecommand{\Fs}{\mathcal{F}_{\sigma}}

\providecommand{\GAe}{{\Gamma^A_{\eps}}}

\providecommand{\Ael}{A^{\eps,\lambda}}

\providecommand{\Hslow}{\mathcal{H}_{\mathrm{slow}}}
\providecommand{\Hfast}{\mathcal{H}_{\mathrm{fast}}}

\providecommand{\Qe}{{Q_{\eps}}}

\providecommand{\OpAel}{\mathrm{Op}^A_{\eps,\lambda}}

\providecommand{\Ope}{\mathrm{Op}_{\eps}}

\providecommand{\WignerTrafoel}{{\mathcal{W}^A_{\eps,\lambda}}}
\providecommand{\magBel}{\WeylProd^B_{\eps,\lambda}}

\providecommand{\WeylSysAel}{{W^A_{\eps,\lambda}}}
\providecommand{\GAe}{{\Gamma^A_{\eps}}}
\providecommand{\gBe}{{\gamma^B_{\eps}}}

\providecommand{\oBel}{\omega^B_{\eps,\lambda}}

\providecommand{\Hoer}[1]{S^{#1}}
\providecommand{\Hoerr}[2]{S^{#1}_{#2}}
\providecommand{\Hoermr}[2]{S^{#1}_{#2}}
\providecommand{\Hoermrd}[3]{S^{#1}_{#2,#3}}
\providecommand{\SemiHoer}[1]{\mathrm{A}\Hoer{#1}}

\providecommand{\SemiHoermr}[2]{\mathrm{A}\mathcal{S}^{#1}_{#2}}

\providecommand{\heff}{h_{\mathrm{eff}}}
\providecommand{\reff}{{r_{\mathrm{eff}}}}
\providecommand{\keff}{{k_{\mathrm{eff}}}}
\providecommand{\BerryC}{\mathcal{A}}

\providecommand{\Eb}{E_{\ast}}

\providecommand{\KA}{{\mathsf{K}^A}}
\providecommand{\Reps}{{\mathsf{R}}}

\providecommand{\order}{\mathcal{O}}
\providecommand{\ordern}{\mathcal{O}_{\norm{\cdot}}}
\providecommand{\BZak}{Zak~}
\providecommand{\Opx}{\mathrm{Op}}
\providecommand{\Opk}{\mathfrak{Op}}
\providecommand{\SemiTau}{A \Cont^{\infty}_{\tau}}

\providecommand{\Rep}{\mathfrak{Rep}} 
\providecommand{\repcom}{\star}
\providecommand{\reptau}{\repcom_{\theta,\tau}}
\providecommand{\repom}{\repcom^{\omega}_{\theta,\tau}}

\providecommand{\Xgroup}{\mathcal{X}}
\providecommand{\AB}{\mathfrak{A}^B}
\providecommand{\Alg}{\mathcal{A}}
\providecommand{\Hil}{\mathcal{H}}

\providecommand{\CA}{\mathfrak{C}_{\Alg}}
\providecommand{\CBA}{\CA^B}

\providecommand{\Ralg}{{\R^d};\Alg}
\providecommand{\syspace}{{\R^d}^*;\Alg^{\infty}}

\providecommand{\CBA}{\mathfrak{C}^B_{\Alg}}
\providecommand{\MBA}{\mathfrak{M}^B_{\Alg}}
\providecommand{\Salg}{\mathcal{S}_{\Alg}}
\providecommand{\SalgComp}{F_{\Alg}}
\providecommand{\Quasi}{\mathsf{Q}}

\providecommand{\WeylSyslt}{W^{\lambda}_{\tau}}

\providecommand{\Aut}{\mathrm{Aut}}
\providecommand{\dualX}{\hat{\Xgroup}}
\providecommand{\Xprod}{\Alg \rtimes_{\theta,\tau} \Xgroup}
\providecommand{\sXprod}{\mathfrak{C}_{\Alg}}
\providecommand{\twistedXprod}{\Alg \rtimes_{\theta,\tau}^{\omega} \Xgroup}
\providecommand{\stwistedXprod}{\sXprod^{\omega}}
\providecommand{\Ideal}{\mathfrak{I}}
\providecommand{\sFXprod}{\Fourier^{-1} \mathfrak{C}^{\omega}_{\Alg}}
\providecommand{\Quasi}{\mathcal{Q}}
\providecommand{\Orbit}{\mathcal{O}}
\providecommand{\Quasialg}{\mathbf{Q}_{\Alg}}
\providecommand{\Gelf}{\mathcal{G}_{\Alg}}

\providecommand{\crossB}{\repcom^B_{\theta}}
\providecommand{\mult}{\mathcal{M}}
\providecommand{\fxb}{\mathfrak{B}^B_{\Alg}}
\providecommand{\BopL}{\mathcal{B} \bigl ( L^2(\R^d) \bigr )}
\providecommand{\Calg}{\mathfrak{C}}
\providecommand{\sXprodR}{\mathfrak{C}^B_{\Alg}}
\providecommand{\sFXprodR}{\Fourier \sXprodR}
\providecommand{\Index}{\mathcal{I}}
\providecommand{\Quasicover}{\textbf{Q}}

\makeindex

\begin{document}

\frontmatter

%
%


\begin{titlepage}
	\vspace{-64mm}
	\begin{center}
	\large TECHNISCHE UNIVERSITÄT MÜNCHEN\\
	Zentrum Mathematik
	\end{center}
	\vspace{19mm}
	\begin{center}
	\bf\huge  Semiclassical Dynamics and \linebreak Magnetic Weyl Calculus\\
	\end{center}
	\vspace{8mm}
	\begin{center}
	{\Large Maximilian Stefan Lein\\}
	\end{center}
	\vspace{19mm}
	\noindent 
	Vollständiger Abdruck der von der Fakultät für Mathematik der Technischen Universität München zur Erlangung des akademischen Grades eines
	\begin{center}
	Doktors der Naturwissenschaften (Dr.~rer.~nat.)
	\end{center}
	genehmigten Dissertation.\vspace{19mm}\\
	\begin{tabular}{lrl}
	$\qquad$ Vorsitzende: && Univ.-Prof. Dr. Simone Warzel\\
	\vspace{-4mm}\\
	$\qquad$ Prüfer der Dissertation:
	& 1. & Univ.-Prof. Dr. Herbert Spohn\\ \vspace{-4mm}\\
	& 2. & Prof. Dr. Radu Purice, IMAR, Bukarest, Rumänien \\ \vspace{-4mm}\\
	& 3. & Univ.-Prof. Dr. Stefan Teufel, Eberhard-Karls-Universität Tübingen \\
	&    & (schriftliches Gutachten)\\ \vspace{-4mm}\\
	\end{tabular}
	\vspace{12mm}\\
	Die Dissertation wurde am 18.10.2010 bei der Technischen Universität München eingereicht und durch die Fakultät für Mathematik am 19.01.2011 angenommen.
\end{titlepage}
%
%

\chapter*{Acknowledgements}
I would like to express my gratitude towards a few people who have directly or indirectly contributed to this thesis. 

First and foremost, I would like to thank Herbert Spohn for his support and faith in my abilities. He has encouraged and supported frequent visits to conferences and other countries. But the thing I am by far most grateful for is his instrumental role in the inception of my lecture `Quantization and Semiclassics' in winter 2009/2010. This was a very, very exhilarating experience for me that I will not forget. 

Secondly, I would like to thank my coworkers Fabian Belmonte, Giuseppe De Nittis, Martin Fürst, Marius Măntoiu, Serge Richard David Sattlegger, and Marcello Seri. Without their contributions, this thesis would be at most a third of what it is now. I thank Giuseppe for his friendship and hospitality not just during my visit at SISSA. Marius' invitation to Santiago enabled me to get to know a new continent. I continue to draw a lot of pleasure from discussions and interaction with you and other colleagues in the field. 

There are other people who deserve special mention: I owe Radu Purice a debt of gratitude for many stimulating discussions and the possibility of organizing ROGER 2009 in Sibiu as well as his invitation to Bukarest in 2007. I appreciate he has agreed to be one of the referees to this thesis. Similarly, I would like to mention Stefan Teufel whose lecture on Quantum Dynamics in 2002 gave me the last push I needed towards the mathematical aspects of physics. Furthermore, the graduate seminar he gave together with Herbert Spohn and Detlef Dürr was the single most enjoyable class I have sat in during my studies. He has also kindly agreed to be the third referee to this thesis. 

I benefitted from the work environment at M5 which all its members helped create. One person deserves mention here and that is Wilma Ghamam who is the good soul of the department. She always knows where things are, who to ask and where people are. Her unpretentious attitude is really unique. 

I also recognize the kind support of DAAD for a $4$-month short term scholarship that allowed me to visit Prof.~Littlejohn at UC Berkeley and the financial support of DFG and Fondecyt under the Grant 1090008 which have made my $4$-week visit to Santiago possible. In that context, I also thank Georgi Raikov and Rafael de Tiedra Aldecoa for preparing my visit to Chile. 
\medskip

\noindent
The other, non-scientific half of my life is just as important: my family and friends have kept my life interesting and fun. The memories from the last four years I share with Julius Bahr, Ann Etienne, Stephanie Farnleitner, Sabrina Oberacher, my brother Fabian and my sister Alexandra are invaluable.

\chapter*{Abstracts}

\section*{Kurzzusammenfassung} 
\label{abstract:german}
\selectlanguage{ngerman}
Weyl Quantisierung und semiklassische Techniken können benutzt werden, um Leitungseigenschaften von kristallinen Festkörpern zu verstehen, die externen, langsam variierenden elektromagnetischen Feldern ausgesetzt werden. Der Fall, in dem das Magnetfeld schwach, aber konstant ist, wird von bisherigen mathematischen Ergebnissen nicht abgedeckt. Genau das ist das Regime des Quanten-Hall-Effekts und es gilt zu verstehen, wieso die transversale Leitfähigkeit quantisiert ist. Möchte man für diesen Fall semiklassische Bewegungsgleichungen rigoros herleiten, muss man den konventionellen Weyl-Kalkül durch einen magnetischen ersetzen, der einen semiklassischen Parameter enthält. 

Mathematisch gesehen hat man es mit magnetischen Pseudodifferentialoperatoren zu tun, die auch für sich gesehen von Interesse sind. Daher widmen wir diesen zwei weitere Kapitel, die sich mit deren Eigenschaften befassen. 


\section*{Abstract} 
\label{abstract:english}
\selectlanguage{english}
Weyl quantization and related semiclassical techniques can be used to study conduction properties of crystalline solids subjected to slowly-varying, external electromagnetic fields. The case where the external magnetic field is constant, is not covered by existing theory as proofs involving usual Weyl calculus break down. This is the regime of the so-called quantum Hall effect where quantization of transverse conductance is observed. To rigorously derive semiclassical equations of motion, one needs to systematically develop a \emph{magnetic} Weyl calculus which contains a semiclassical parameter. 

Mathematically, the operators involved in the analysis are magnetic pseudodifferential operators, a topic which by itself is of interest for the mathematics and mathematical physics community alike. Hence, we will devote two additional chapters to further understanding of properties of those operators. 


\tableofcontents

\mainmatter

\chapter{Introduction} 
\label{intro}
Initially, the word quantization referred to the fact that in some physical systems, atoms could only absorb or emit light of certain frequencies, for instance, \ie that energy could only be exchanged in certain chunks of fixed size called \emph{quanta}. Nowadays, it refers to the task of associating quantum analogs to all parts of a classical system, \ie to find a set of procedures that associates quantum observables to classical observables, (quasi-)classical states to quantum states and an evolution equation that governs the dynamics of quantum states and quantum observables. Historically, this happened rather quickly after the inception of modern quantum mechanics around 1926 when the forefathers such as Bohr, Dirac, Heisenberg and Schrödinger have tried -- and eventually succeeded -- to conceptually `derive' quantum mechanics for a particle moving in $\R^d$ by analogy from classical mechanics. 

In contrast to Chapter~\ref{magWQ} where I will give a modern introduction to the subject of quantization, let me sketch the task from the perspective of early quantum mechanics: to the co-founders of modern quantum mechanics, the keys to understanding quantum mechanics were the commutation relations\index{commutation relations!non-magnetic} \cite[pp.~100]{Dirac:foundations_qm_ger:1930} \cite[Chapter~IV.1]{Heisenberg:prinzipien_quantentheorie:1930}
\begin{align*}
	q_k q_j - q_j q_k &= 0 \\
	p_k p_j - p_j p_k &= 0 \\
	p_k q_j - q_j p_k &= i \hbar 
\end{align*}
of position $q$ and momentum $p$ as well as a vague intuition that the classical Poisson bracket $\{ \cdot , \cdot \}$\index{Poisson bracket!non-magnetic} needed to be replaced by the quantum commutator $\tfrac{i}{\hbar} [\cdot , \cdot]$.\index{commutator} Both are bilinear and antisymmetric in their arguments, act as derivations and satisfy the Jacobi identity \cite[p.~99, eq.~(7)--(9)]{Dirac:foundations_qm_ger:1930}. These similarities suggested them to propose 
\begin{align*}
	\partial_t F(t) &= \tfrac{i}{\hbar} [ H , F(t) ] 
\end{align*}
as the equations of motion\index{Heisenberg equations of motion} for a quantum observable $F$ in lieu of 
\begin{align*}
	\partial_t f(t) &= \{ h , f(t) \} 
\end{align*}
for a classical observable $f$. 

In the first edition of `The Principles of Quantum Mechanics,' Dirac somewhat haphazardly explains the link between the `non-commutative observables' $q$ and $p$ and the operators multiplication by $q$ and $- i \hbar \nabla_q$.\footnote{Compared to the first edition of Dirac's book, the presentation of Chapter~IV in the third edition has been much improved. } Furthermore, Dirac \cite[p.~109]{Dirac:foundations_qm_ger:1930}, Schrödinger \cite{Schroedinger:quantisierung_ew_problem_1:1926,Schroedinger:quantisierung_ew_problem_2:1926,Schroedinger:quantisierung_ew_problem_3:1926,Schroedinger:quantisierung_ew_problem_4:1926} and Heisenberg \cite[p.~86]{Heisenberg:prinzipien_quantentheorie:1930} proposed $f(q,-i \hbar \nabla_q)$ as the quantization of the classical observable $f(q,p)$. Dirac noticed an inherent ambiguity in this procedure: one needs to make a choice of operator ordering \cite[p.~103]{Dirac:foundations_qm_ger:1930}\index{operator ordering}: 
\begin{quote}
	Es darf nicht übersehen werden, daß die Reihenfolge der Faktoren in Produkten, die im Ausdruck für $H$ vorkommen, von Bedeutung sein kann, da unsere Variablen nicht alle vertauschbar sind. 
\end{quote}
Hence, there may be more than one candidate as the quantization of a classical observable and the prescription `in position representation, replace $q$ by multiplication with $q$ and $p$ by $-i \hbar \nabla_q$' is incomplete. 

The first step in the right direction was taken by Weyl in 1927 \cite{Weyl:qm_gruppentheorie:1927} who wrote down a consistent quantization formula: if $f : T^* \R^d \longrightarrow \R$ is a suitable physical observable on phase space $T^* \R^d \cong \R^d \times {\R^d}^*$, then he defined its quantization to be\footnote{For consistency with the remainder of this thesis, we choose a different sign convention compared to \cite{Weyl:qm_gruppentheorie:1927}. } 
\begin{align*}
	\Op(f) := f(\Qe,\Pe) := \frac{1}{(2\pi)^d} \int_{\R^d} \dd x \int_{{\R^d}^*} \dd \xi \, (\Fs f)(x,\xi) \, e^{- i (\xi \cdot \Qe - x \cdot \Pe)} 
\end{align*}
where $q$ and $p$ are elevated to operators $\Qe = \hat{q}$ and $\Pe = - i \hbar \nabla_q$ on $L^2(\R^d)$ and $\Fs$ denotes a symmetrized Fourier transform\index{Fourier transform!symplectic} on $\R^d \times {\R^d}^*$ defined as 
\begin{align}
	(\Fs f)(x,\xi) := \frac{1}{(2\pi)^d} \int_{\R^d} \dd y \int_{{\R^d}^*} \dd \eta \, e^{i (\eta \cdot x - y \cdot \xi)} \, f(y,\eta) 
	. 
\end{align}
Thus it seems appropriate that the calculus which emerged from this point of view bears his name. 

With seemingly no connection to Weyl's work, Wigner \cite{Wigner:Wigner_transform:1932} showed how quantum states $\psi \in L^2(\R^d)$ can be written as pseudo probability measures on phase space. It is truly remarkable and an indication of Wigner's genius that he had found the correct formula just by `staring at the problem.' 

It took until 1947 to put all pieces of the puzzle together: in a seminal work by Moyal \cite{Moyal:Weyl_calculus:1949}, the relation between Weyl quantization, Wigner transform and the Moyal product (also known as Weyl product) has been worded out systematically and it is justified to say that this marked the birth of Weyl calculus as we know it today.

\section{Physical aspects: quantization of magnetic systems} 
\label{intro:physical_aspects}
Let us consider a single classical particle without spin moving in $\R^d$. In the framework of hamiltonian mechanics, the state of the particle is represented by a point in phase space $\Pspace := (T^* \R^d , \omega)$ where $T^* \R^d \cong \R^d \times {\R^d}^*$ is the cotangent bundle of configuration space $\R^d$ and $\omega$ the so-called symplectic form. Points in phase space will be denoted by capital letters $X = (x,\xi) , Y = (y,\eta) , Z = (z , \zeta) \in \Pspace$ with space components $x , y, z \in \R^d$ and momentum components $\xi , \eta , \zeta \in {\R^d}^*$. The symplectic form $\omega$ is a two-form, \ie a skew-symmetric bilinear form on the space of vector fields on $\R^d \times {\R^d}^*$ whose representation matrix $\bigl ( \omega_{kj} \bigr )_{1 \leq k,j \leq d}$ in terms of coordinates is invertible. Its purpose is to associate vector fields $X_h$ to functions $h : \Pspace \longrightarrow \R$ on phase space via 
\begin{align}
	\omega (X_h , \cdot) := \dd h 
	. 
	\label{intro:physical_aspects:eqn:vector_field}
\end{align}
In the absence of a magnetic field, the \emph{canonical symplectic form} $\omega^0 = \sum_{j = 1}^d \dd x_j \wedge \dd \xi_j$ with matrix representation 
\begin{align*}
	\bigl ( \omega^0_{kj} \bigr )_{1 \leq k,j \leq d} &= \left (
	\begin{matrix}
		0 & - \id \\
		+ \id & 0 \\
	\end{matrix}
	\right ) 
\end{align*}
relates the gradient of the energy function to the hamiltonian vector field\index{hamiltonian vector field} associated to the energy function $h$ via 
\begin{align*}
	X_h = \left (
	\begin{matrix}
		0 & - \id \\
		+ \id & 0 \\
	\end{matrix}
	\right )^{-1} \left (
	\begin{matrix}
		\nabla_x h \\
		\nabla_{\xi} h \\
	\end{matrix}
	\right ) 
	. 
\end{align*}
This vector field determines the hamiltonian flow $\phi_t$\index{hamiltonian flow} through Hamilton's equations of motion
\begin{align*}
	\left (
	\begin{matrix}
		\dot{x} \\
		\dot{\xi} \\
	\end{matrix}
	\right ) &= X_h 
	= \left (
	\begin{matrix}
		0 & - \id \\
		+ \id & 0 \\
	\end{matrix}
	\right )^{-1} \left (
	\begin{matrix}
		\nabla_x h \\
		\nabla_{\xi} h \\
	\end{matrix}
	\right ) 
	. 
\end{align*}
The flow $\phi_t$ tracks the trajectory $\bigl ( x(t) , \xi(t) \bigr ) = \phi_t(x_0,\xi_0)$ for the initial conditions $(x_0 , \xi_0) \in \Pspace$. Instead of looking at the special observables position and momentum, we can write down equations of motion for an arbitrary observable $f \in \Cont^{\infty}(\Pspace,\R)$: the symplectic form $\omega^0$ induces a Poisson bracket via 
\begin{align*}
	\{ f , g \} := - \omega^0 \bigl ( X_f , X_g \bigr ) = \sum_{j = 1}^d \bigl ( \partial_{\xi_j} f \, \partial_{x_j} g - \partial_{x_j} f \, \partial_{\xi_j} g \bigr ) 
\end{align*}
where $f , g \in \Cont^{\infty}(\Pspace,\R)$ are observables and $X_f , X_g$ are the associated vector fields (see equation~\eqref{intro:physical_aspects:eqn:vector_field}). Now one can show that the equations of motion for a time-evolved observable $f(t) := f \circ \phi_t$ are given by 
\begin{align*}
	\partial_t f(t) &= \{ h , f(t) \}
	. 
\end{align*}
If the particle is subjected to a magnetic field, $B$, we have two options to integrate it into the classical formalism: (i) we use minimal substitution or (ii) we incorporate $B$ into the symplectic form. First of all, the magnetic field can be seen as a closed two-form $B \in \bigwedge^2 (\R^d)$, $\dd B = 0$. Since $\R^d$ is star-shaped, a $k$-form $\omega \in \bigwedge^k (\R^d)$ is closed if and only if it is exact, \ie there exists a $k-1$-form $\alpha$ such that $\omega = \dd \alpha$. Hence, to each $B$ there exist vector potentials $A \in \bigwedge^1(\R^d)$ such that $B = \dd A$. The components of $B$ (with respect to an orthonormal basis of $\R^d$) are related to those of the vector potential via $B_{kj} = \partial_{x_k} A_j - \partial_{x_j} A_k$. Vector potentials are highly non-unique and worse-behaved than the magnetic field they represent. 

If $A$ is a vector potential to $B$, then only \emph{kinetic} momentum $\xi_{\mathrm{kin}} := \xi - A(x)$ is a physically relevant observable that does not depend on the choice of gauge. This follows from the Lagrangian approach to magnetic systems \cite[Chapter~7.6]{Marsden_Ratiu:intro_mechanics_symmetry:1999}. Minimal coupling\index{minimal coupling} is the recipe to replace $f(x,\xi)$ by $f \bigl ( x , \xi - A(x) \bigr ) =: f_A(x,\xi)$ as observables and to consider the equations of motion given by 
\begin{align}
	\partial_t f_A(t) = \{ h_A , f_A(t) \} 
	, 
	&&
	f_A(0) = f_A 
	, 
	\label{intro:physical_aspects:eqn:non_mag_eom_Poisson}
\end{align}
where $f_A$ and $h_A$ are the minimally substituted observables and $\{ \cdot , \cdot \}$ the usual, non-magnetic Poisson bracket. If one considers only minimally substituted observables, then the corresponding equations of motion essentially do not depend on the choice of vector potential $A$, but only on the magnetic field $B$. 
\medskip

\noindent
Alternatively, the geometry of phase space $\Pspace$ can be changed: we equip $T^* \R^d$ with the \emph{magnetic symplectic form}\index{symplectic form!magnetic} 
\begin{align}
	\omega^B &= \sum_{j = 1}^d \dd x_j \wedge \dd \xi_j + \frac{1}{2} \sum_{k , j = 1}^d B_{kj} \dd x_k \wedge \dd x_j 
\end{align}
which induces the \emph{magnetic Poisson bracket} 
\begin{align}
	\{ f , g \}_B = \sum_{j = 1}^d \bigl ( \partial_{\xi_j} f \, \partial_{x_j} g - \partial_{x_j} f \, \partial_{\xi_j} g \bigr ) - \sum_{k , j = 1}^d B_{kj} \, \partial_{x_k} f \, \partial_{x_j} g 
	. 
\end{align}
The hamiltonian vector field\index{hamiltonian vector field} now depends on $B$, 
\begin{align*}
	X_h^B &= \left (
	\begin{matrix}
		B & - \id \\
		+ \id & 0 \\
	\end{matrix}
	\right )^{-1} \left (
	\begin{matrix}
		\nabla_x h \\
		\nabla_{\xi} h \\
	\end{matrix}
	\right )
	. 
\end{align*}
By a simple calculation, we get $\{ f_A , g_A \}(x,\xi) = \{ f , g \}_B \bigl ( x , \xi - A(x) \bigr )$ and thus the solution $f(t)$ of 
\begin{align*}
	\partial_t f(t) &= \{ h , f(t) \}_B 
	, 
	&&
	f(0) = f 
	, 
\end{align*}
and $f_A(t)$ which solves equation~\eqref{intro:physical_aspects:eqn:non_mag_eom_Poisson} are the same, written down in different coordinates. In this more geometric formulation, we are simply working with kinetic momentum all along. Both descriptions of a \textbf{classical} magnetic system are equivalent.\medskip

\noindent
As we will see, this is \textbf{not} the case for quantum systems. Let us start with the usual recipe: since both \emph{classical} descriptions of magnetic systems are equivalent (and minimal substitution being the more popular choice), it is quite sensible to define the magnetic quantization of an observable $f$ as the usual Weyl quantization of the minimally substituted observable $f_A$, 
\begin{align}
	\Op_A(f) := \Op(f_A) = \frac{1}{(2\pi)^d} \int_{\Pspace} \dd x \, \dd \xi \, (\Fs f_A)(x,\xi) \, e^{- i (\xi \cdot Q - x \cdot P)} 
	. 
	\label{intro:physical_aspects:eqn:OpA_wrong}
\end{align}
Although we will elaborate on the drawbacks of this prescription in detail in Chapter~\ref{magWQ:magnetic_weyl_calculus:naive_ansatz}, let us sketch the origin of the flaws: in magnetic quantum systems on $\R^d$, the building block observables\index{building block observables!magnetic} position and momentum are position and kinetic momentum, 
\begin{align*}
	\Qe &= \hat{x} 
	, \\
	\Pe^A &= - i \eps \nabla_x - A(\hat{x})
	, 
\end{align*}
where $A$ is a vector potential representing $B = \dd A$ and $\eps \leq 1$ a dimensionless parameter that sits in the same place as $\hbar$. Formally, these operators satisfy the commutation relations\index{commutation relations!magnetic} 
\begin{align}
	i [\Qe_l , \Qe_j] = 0 
	&&
	i [\Pe^A_l , \Qe_j] = \eps \delta_{lj} 
	&&
	i [\Pe^A_l , \Pe^A_j] = - \eps B_{lj}(\Qe) 
\end{align}
which \emph{should} be encoded into the composition law 
\begin{align*}
	\WeylSys(X) \WeylSys(Y) &= e^{i \frac{\eps}{2} \sigma(X,Y)} \WeylSys(X+Y) \in \mathcal{B} \bigl ( L^2(\R^d) \bigr ) 
\end{align*}
of the so-called Weyl system\index{Weyl system!non-magnetic} $\WeylSys(X) := e^{- i (\xi \cdot \Qe - x \cdot \Pe)}$. However, since the magnetic field does not appear in the definition of $\WeylSys(X)$, the presence of the magnetic field is not properly taken into account. This is the reason why $\Op_A$ is in general not gauge-covariant:\index{gauge-covariance} if $A' = A + \dd \chi$ is an equivalent gauge, then $\Op_{A'}(f)$ and $\Op_A(f)$ are in general not unitarily equivalent. On the other hand, physical properties such as the spectrum of the system must not depend on the choice of gauge! 
\medskip

\noindent
Even though bits and pieces of the correct solution were used as early as 1951 \cite{Luttinger:magnetic_field_periodic_potential:1951}, it was not until 1999 that a gauge-covariant Weyl calculus was first written down in its entirety by Müller \cite{Mueller:product_rule_gauge_invariant_Weyl_symbols:1999}. The idea is to replace translations by magnetic translations and to put the magnetic vector potential into the Weyl system $\WeylSys^A(X) := e^{- i \sigma(X,(\Qe,\Pe^A))}$. A simple trotterization shows that $\WeylSys^A(Y)$ acts on $u \in L^2(\R^d)$ by 
\begin{align*}
	\bigl ( \WeylSys^A(Y) u \bigr )(x) = e^{- \frac{i}{\eps} \Gamma^A([x,x+\eps y])} e^{- i \eta \cdot (x + \frac{\eps}{2} y)} \, u(x + \eps y) 
\end{align*}
where 
\begin{align*}
	\Gamma^A([x,x+\eps y]) := \int_{[x,x + \eps y]} A 
\end{align*}
is the circulation along the line segment connecting $x$ and $x + \eps y$. If $A' = A + \dd \chi$ is an equivalent gauge, then 
\begin{align*}
	\int_{[x,x + \eps y]} A' = \int_{[x,x + \eps y]} A + \int_{\partial [x , x + \eps y]} \chi 
	= \int_{[x,x + \eps y]} A + \chi(x + \eps y) - \chi(x)
\end{align*}
holds by Stokes theorem and $\WeylSys^{A'}(X) = e^{+ \frac{i}{\eps} \chi(\Qe)} \WeylSys^A(X) e^{- \frac{i}{\eps} \chi(\Qe)}$ is unitarily equivalent to $\WeylSys^A(X)$. Hence, magnetic Weyl quantization\index{Weyl quantization!magnetic} 
\begin{align}
	\Op^A(f) := \frac{1}{(2\pi)^d} \int_{\Pspace} \dd X \, (\Fs f)(X) \, \WeylSys^A(X) 
	. 
	\label{intro:physical_aspects:eqn:OpA_right}
\end{align}
inherits the gauge-covariance of the Weyl system and the quantizations with respect to equivalent gauges define \emph{unitarily equivalent} operators. The composition law of $\WeylSys^A$\index{Weyl system!composition law} now contains an additional magnetic contribution, 
\begin{align}
	\WeylSys^A(X) \WeylSys^A(Y) &= e^{i \frac{\eps}{2} \sigma(X,Y)} \, e^{- i \Gamma^B_{\eps}(\sexpval{\Qe,\Qe + \eps x, \Qe + \eps x + \eps y})} \, \WeylSys^A(X+Y) 
	, 
	&& 
	\forall X , Y \in \Pspace , 
\end{align}
which is the exponential of the scaled magnetic flux $\Gamma^B_{\eps}$ through the triangle with corners $\Qe$, $\Qe + \eps x$ and $\Qe + \eps x + \eps y$ (see Theorem~\ref{asymptotics:thm:equivalenceProduct}). Associated to $\Op^A$, there is a magnetic Wigner transform $\WignerTrafo^A$\index{Wigner transform!magnetic} which connects quantum expectation values to phase space averages, 
\begin{align*}
	\bscpro{u}{\Op^A(f) v} = \frac{1}{(2\pi)^{\nicefrac{d}{2}}} \int_{\Pspace} \dd X \, f(X) \, \bigl ( \WignerTrafo^A(v,u) \bigr )(X)
	, 
\end{align*}
as well as a non-commutative product $\magW$\index{Weyl product!magnetic} on the level of functions on phase space which emulates the operator product, \ie 
\begin{align*}
	\Op^A(f \magW g) := \Op^A(f) \, \Op^A(g) 
	. 
\end{align*}
By covariance of $\Op^A$, the product only depends on the magnetic field and after some effort, one arrives at an explicit integral expression for $\magW$: 
\begin{align}
	(f \magW g)(X) &= \frac{1}{(2 \pi)^{2d}} \int_{\Pspace} \dd Y \int_{\Pspace} \dd Z \, e^{+i \sigma(X,Y+Z)} \, e^{i \frac{\eps}{2} \, \sigma(Y,Z)} \, 
	\cdot \notag \\
	&\qquad \qquad \quad \cdot 
	e^{- \frac{i}{\eps} \Gamma^B(\sexpval{x-\frac{\eps}{2}(y+z),x + \frac{\eps}{2}(y-z) , x+\frac{\eps}{2}(y+z))}} \, 
	\bigl ( \Fs f \bigr )(Y) \, \bigl ( \Fs g \bigr )(Z) 
	\label{intro:physical_aspects:eqn:magW}
\end{align}
If $\eps \ll 1$, the product $\magW$ can be expanded asymptotically in $\eps$ and we get an expansion of the operator product, 
\begin{align*}
	\Op^A(f \magW g) \asymp \Op^A \Bigl ( \mbox{$\sum_{n = 1}^{\infty}$} \eps^n \, (f \magW g)_{(n)} \Bigr ) 
	= \sum_{n = 0}^{\infty} \eps^n \, \Op^A \bigl ( (f \magW g)_{(n)} \bigr ) 
	. 
\end{align*}
This idea has been put to good use when studying perturbation expansions and semiclassical limits (see \eg \cite{Littlejohn_Weigert:diagonalization_multi_wave:1993,PST:sapt:2002} for ordinary Weyl calculus and Chapter~\ref{magsapt} and \cite{Mueller:product_rule_gauge_invariant_Weyl_symbols:1999,DeNittis_Lein:Bloch_electron:2009,Fuerst_Lein:scaling_limits_Dirac:2008} for applications of magnetic Weyl calculus). 


\section{Mathematical aspects: magnetic $\Psi$DOs} 
\label{intro:mathematical_aspects}
%
Motivated by problems in mathematical physics, mathematicians and mathematical physicists sought to apply and generalize pseudodifferential techniques to magnetic problems. We start by example: let $H^A$ be a covariant selfadjoint magnetic operator, \eg 
\begin{align*}
	H^A = \tfrac{1}{2} {\Pe^A}^2 + V(\Qe) 
\end{align*}
on $L^2(\R^d)$ where $A$ is a vector potential to the magnetic field $B$. Functions of $H^A$ inherit its gauge-covariance, \ie for any equivalent gauge $A' = A + \dd \chi$, 
\begin{align*}
	f(H^{A'}) &= e^{+ i \chi(\Qe)} \, f(H^A) \, e^{- i \chi(\Qe)}
\end{align*}
holds true. Particular examples are resolvents $(H^A - \zeta)^{-1}$, $\zeta \not \in \sigma(H^A)$, and the semigroup $e^{- t H^A}$ (if $H^A$ is bounded from below). The Green function $G^A(x,y;\zeta)$, \ie the operator kernel of the resolvent $(H^A - \zeta)^{-1}$, contains many of the operator's properties and often occurs in the analysis of magnetic systems. Since $G^A(\cdot,\cdot;\zeta)$ depends explicitly on the choice of vector potential and thus a convenient choice of vector potential may be necessary (\eg the symmetric gauge in case of a constant magnetic field for $d = 2,3$). Physical properties such as the spectrum and the conductivity tensor to name just two, however, should only depend on the magnetic field $B$ and not on the vector potential $A$. The components of the conductivity tensor and other expectation values can be written as trace (per volume) of covariant operators, 
\begin{align*}
	\sigma^B := \mathrm{Tr} \, f(H^A) 
	, 
\end{align*}
which are independent of the choice of gauge since the trace is invariant under conjugation with unitary operators. Hence, the idea was to somehow get rid of the dependence on the vector potential $A$ and to start the analysis with an expression that only depends on the magnetic field $B$. Eventually, it was noticed that the dependence of the kernel $K^A_H$ of a covariant operator $T^A$ on the vector potential is relatively simple: the function 
\begin{align*}
	\tilde{K}^B_T := e^{+ i \Gamma^A([x,y])} \, K^A_T(x,y)
\end{align*}
no longer depends on the choice of vector potential (as can be checked explicitly by noting that the extra phase factor cancels a phase factor that stems from replacing translations with magnetic translations). This was noticed early on in articles by Peierls \cite{Peierls:Diamagnetismus:1933}, Luttinger \cite{Luttinger:magnetic_field_periodic_potential:1951} and Schwinger  \cite{Schwinger:vac_polarization:1951}, and used extensively later on in rigorous works as well, see \eg \cite{Cornean_Nenciu_Pederson:Faraday_effect:2005}. Magnetic fields $B$ are always better-behaved than vector potentials $A$ representing them and thus it is easier to analyze $\tilde{K}^B_T$ than the original operator kernel. For instance, if one wants to show exponential decay of the Green function $G^A(\cdot,\cdot;\zeta)$ for $\zeta \not \in \sigma(H^A)$ in $x-y$, then equivalently, one can prove exponential decay of $\tilde{G}^B(x,y;\zeta) := e^{+i \Gamma^A([x,y])} \, G^A(x,y;\zeta)$. 

The dependence of products of two gauge-covariant operators $T^A$ and $S^A$ also contain characteristic magnetic phase factors: let $F^A := T^A \, S^A$ be the product of two bounded gauge-covariant operators with kernels $K^A_T$ and $K^A_S$. Then 
\begin{align*}
	\tilde{K}^B_F(x,y) &= e^{+ i \Gamma^A([x,y])} \, \int_{\R^d} \dd z \, K^A_T(x,z) \, K^A_S(z,y) 
	\\
	&
	= \int_{\R^d} \dd z \, e^{+ i \Gamma^A([x,y])} e^{- i \Gamma^A([x,z])} e^{- i \Gamma^A([z,y])} \, \tilde{K}^B_T(x,z) \, \tilde{K}^B_S(z,y) 
	\\
	&= \int_{\R^d} \dd z \, e^{- i \Gamma^B(\sexpval{x,y,z})} \, \tilde{K}^B_T(x,z) \, \tilde{K}^B_S(z,y) 
\end{align*}
contains the exponential of the magnetic flux through the triangle with corners $x$, $y$ and $z$. In view of equations~\eqref{intro:physical_aspects:eqn:magW} and \eqref{intro:physical_aspects:eqn:OpA_right}, these additional magnetic phase factors are hardly surprising. 

Their universality suggests a more systematic approach to magnetic pseudodifferential operators which gives access to a rich toolbox of results that can be re-used and exploit the structure of magnetic problems. Măntoiu and Purice have laid the foundation in \cite{Mantoiu_Purice:magnetic_Weyl_calculus:2004} and transcribed the most fundamental results of pseudodifferential theory to the magnetic context (\eg \cite{Iftimie_Mantiou_Purice:magnetic_psido:2006,Iftimie_Mantoiu_Purice:commutator_criteria:2008}). In addition, an algebraic approach in the spirit of \cite{Mantoiu:Cstar_algebras_dynamical_systems_at_infinity:2002} and \cite{Amrein_Boutet_Georgescu:C0_groups_commutator_methods:1996} was proposed in \cite{Mantoiu_Purice_Richard:twisted_X_products:2004} so that pseudodifferential and algebraic techniques may be combined to one's advantage \cite{Lein_Mantoiu_Richard:anisotropic_mag_pseudo:2009,Mantoiu_Purice:continuity_spectra:2009}. 


\section{Structure and main results} 
\label{intro:scope}
This thesis is based on four publications by the author \cite{Lein:two_parameter_asymptotics:2008}, Lein, Măntoiu and Richard \cite{Lein_Mantoiu_Richard:anisotropic_mag_pseudo:2009}, De~Nittis and Lein \cite{DeNittis_Lein:Bloch_electron:2009} and Belmonte, Lein and Măntoiu \cite{Belmonte_Lein_Mantoiu:mag_twisted_actions:2010}. A fifth article with Fürst is in preparation \cite{Fuerst_Lein:scaling_limits_Dirac:2008}. 
\medskip

\noindent
Chapter~\ref{magWQ} outlines the basic formalism of magnetic Weyl calculus. It starts with a pedagogical introduction to usual Weyl calculus that sets the stage for developing a magnetic version. This is done in Chapters~\ref{magWQ:magnetic_weyl_calculus} and \ref{magWQ:extension}. A list of important results that are needed in the remainder of this thesis is given in Chapter~\ref{magWQ:important_results}. Among them are $L^2$-boundedness of operators in $\Op^A(\Hoerrd{0})$, basic facts on selfadjointness, Beals and Bony commutator criteria and results on inversion. The material is mostly taken from publications by Măntoiu and Purice \cite{Mantoiu_Purice:magnetic_Weyl_calculus:2004} and Iftimie, Măntoiu and Purice \cite{Iftimie_Mantiou_Purice:magnetic_psido:2006,Iftimie_Mantoiu_Purice:commutator_criteria:2008}. 
\medskip

\noindent
Chapter~\ref{asymptotics} which is based on \cite{Lein:two_parameter_asymptotics:2008} is devoted to the development of a functional calculus for observables $\Qe$ and $\Pe^A$ that satisfy the following commutation relations:\index{commutation relations!magnetic} 
\begin{align}
	i [ \Qe_l , \Qe_j ] = 0 
	&&
	i [ \Pe^A_l , \Qe_j ] = \eps \delta_{lj} 
	&&
	i [ \Pe^A_l , \Pe^A_j ] = - \eps \lambda B_{lj}(\Qe) 
\end{align}
Here, $\eps \ll 1$ is a semiclassical parameter and $\lambda \leq 1$ quantifies the coupling to the magnetic field. We start with a brief discussion concerning the realization of the above commutation relations as operators on Hilbert spaces $\Hil \cong L^2(\R^d)$: the results in the remainder of Chapter~\ref{asymptotics} -- in particular the form of the product $\magW$ and its asymptotic expansions as well as the semiclassical limit -- hold true as long as $\Qe$ and $\Pe^A$ are unitarily equivalent to 
\begin{align}
	\Qe' &:= \hat{x} \\
	{\Pe'}^A &:= - i \eps \nabla_x - \lambda A(\hat{x}) \notag 
\end{align}
equipped with the usual domains as operators on $L^2(\R^d)$. This will be of importance in Chapter~\ref{magsapt}. We reiterate formulas for magnetic Weyl quantization and Wigner transform in Chapters~\ref{asymptotics:weyl_quantization}--\ref{asymptotics:mag_Wigner_transform} where the two parameter $\eps$ and $\lambda$ are put in the right places. 

The first main result of this thesis is contained in Chapter~\ref{asymptotics:expansions}. We prove asymptotic expansions for the magnetic Weyl product $f \magW g$ of two Hörmander class symbols $f \in \Hoermr{m_1}{\rho}$ and $g \in \Hoermr{m_2}{\rho}$ using oscillatory integral techniques: a two-parameter expansion in $\eps$ \emph{and} $\lambda$ (Theorem~\ref{asymptotics:thm:asymptotic_expansion}) is shown first, a one-parameter expansion in $\eps$ is an immediate corollary (Corollary~\ref{asymptotics:cor:asymptotic_expansion_eps}) and finally, an expansion in the coupling constant $\lambda$ (Theorem~\ref{asymptotics:thm:lambda_expansion}) is proven as well. This makes the formal derivation by Müller rigorous \cite{Mueller:product_rule_gauge_invariant_Weyl_symbols:1999}. The expansion of $\magW$ in terms of the semiclassical parameter is immediately put to good use in the proof of an Egorov-type theorem, Theorem~\ref{asymptotics:semiclassical_limit:thm:semiclassical_limit_observables}, which connects quantum and classical time evolution. 

The last section establishes that to first order, perturbation expansions in $\eps$ derived with the help of usual Weyl calculus and magnetic Weyl calculus must agree up to errors of order $\order(\eps^2)$. This explains why usual Weyl calculus reproduces the correct results in applications, although typically, stronger assumptions need to be placed on the magnetic field. 
\medskip

\noindent
Chapter~\ref{magsapt} deals with an application which stimulated the author's interest in magnetic pseudodifferential operators in the first place: the derivation for effective dynamics for the magnetic Bloch electron. Here, a single particle is subjected to a periodic potential and a slowly varying electromagnetic field. The results which have been obtained in collaboration with G.~De~Nittis in \cite{DeNittis_Lein:Bloch_electron:2009} generalize the work by Panati, Spohn and Teufel \cite{PST:effective_dynamics_Bloch:2003}. 

After introducing the model in Chapter~\ref{magsapt:intro} and rewriting it in a suitable form in Chapter~\ref{magsapt:rewriting:Zak}, an equivariant version of magnetic Weyl calculus is introduced in Chapter~\ref{magsapt:rewriting:magnetic_weyl_calculus}. Since magnetic Weyl calculus incorporates the magnetic field in a natural manner, some points in the original publication can be simplified. For instance, it is not necessary to work with weighted symbol classes $S^w$, traditional Hörmander symbols are here used instead. 

The next section explains the physical content behind the philosophy of space-adiabatic perturbation theory \cite{PST:sapt:2002}, the main tool employed in the derivation of effective quantum and semiclassical dynamics in Chapter~\ref{magsapt:modulation_field}. Only the necessary modifications are mentioned since the proofs carry over from \cite{PST:effective_dynamics_Bloch:2003} \emph{mutadis mutandis.}\medskip

\noindent
Chapter~\ref{algebraicPOV} introduces magnetic quantization and magnetic pseudodifferential theory from an algebraic point of view and serves as preparation for Chapter~\ref{psiDO_reloaded}. The first two sections draw heavily from \cite{Mantoiu_Purice_Richard:twisted_X_products:2004,Mantoiu_Purice_Richard:Cstar_algebraic_framework:2007} while the last relies on \cite{Amrein_Boutet_Georgescu:C0_groups_commutator_methods:1996} and \cite{Mantoiu:Cstar_algebras_dynamical_systems_at_infinity:2002}. Properties of resolvents and spectra of magnetic $\Psi$DOs can be linked to special $C^*$-subalgebras of $\mathcal{B} \bigl ( L^2(\R^d) \bigr )$. They are representations of so-called twisted crossed products $\twistedXprod$ which are the topic of Chapter~\ref{algebraicPOV:twisted_crossed_products}. After we recall Gelfand theory, crossed products are introduced as completions of $L^1(\Xgroup ; \Alg)$ where $\Xgroup$ is an abelian Polish group acting on an abelian $C^*$-algebra $\Alg$ via $\theta : \Xgroup \longrightarrow \mathrm{Aut}(\Alg)$. Typically, $\Alg$ consists of bounded, uniformly continuous functions on $\Xgroup = \R^d , \T^d , \Z^d$ and this `anisotropy algebra' characterizes the behavior of the magnetic field $B$ as well as that of functions on phase space $\Pspace = \Xgroup \times \hat{\Xgroup}$ in the position variable, \ie $x \mapsto h(x,\xi) \in \Alg$. Next, a magnetic twist $\omega(q ; x , y) := e^{- i \Gamma^B(\sexpval{q,q+x,q+x+y})}$ is introduced in Chapter~\ref{algebraicPOV:twisted_crossed_products:twisted_crossed_products}. It enters in the twisted convolution $\repom$  which serves as a product on $L^1(\Xgroup ; \Alg)$ and $\twistedXprod$ and is related to the magnetic Weyl product $\magW$ by partial Fourier transform. The associativity of $\repom$ is ensured by the so-called $2$-cocycle property of $\omega$. In case $\Alg$ is a $C^*$-subalgebra of $\BCont_u(\Xgroup)$, more can be said about the structure of the twisted crossed products and their representations: there exists a natural representation on $L^2(\Xgroup)$ called Schrödinger representation. The twisted crossed product $\Cont_{\infty} \rtimes^{\omega}_{\theta,\tau} \Xgroup$ plays a special role since it is mapped onto the compact operators $\mathcal{K} \bigl ( L^2(\Xgroup) \bigr )$ by the Schrödinger representation. This characterization of the compact operators enters in the analysis of essential spectra in Chapter~\ref{psiDO_reloaded:spectral:ess_spec}. Gauge-covariance of representations is explained in terms of cohomology. The connex to Weyl calculus is made in Chapter~\ref{algebraicPOV:generalized_weyl_calculus}. 

The last section introduces the concept of affiliation which is the abstract analog of a functional calculus for selfadjoint operators, \ie it is a morphism $\Phi : \Cont_{\infty}(\R) \longrightarrow \Calg$ mapping to a $C^*$-algebra $\Calg$. However, $\Calg$ need not be the algebra of bounded operators on a Hilbert space $\mathcal{B}(\Hil)$. Spectra and essential spectra as sets can be recovered in this formalism as well: if $\pi : \Calg \longrightarrow \Calg'$ is a morphism between $C^*$-algebras and $\Phi : \Cont_{\infty}(\R) \longrightarrow \Calg$ an observable affiliated to $\Calg$, then $\pi \circ \Phi : \Cont_{\infty}(\R) \longrightarrow \Calg'$ is an observable affiliated to $\Calg'$. The spectrum of $\pi \circ \Phi$ tends to be smaller as that of $\Phi$ since morphisms are norm-decreasing. Two particular examples of morphisms $\pi$ are representations and projections onto $\Calg / \Ideal$ where $\Ideal \subseteq \Calg$ is a two-sided ideal. The latter is used in the characterization of essential spectra of pseudodifferential operators. The chapter finishes with a short discussion on tensor products of $C^*$-algebras which can be used to treat observables which are `direct integrals' or sequences of observables within algebraic framework. 
\medskip

\noindent
Chapter~\ref{psiDO_reloaded} combines algebraic and pseudodifferential methods to analyze properties of Moyal resolvents and essential spectra of magnetic pseudodifferential operators with certain behavior in the $x$ variable. It is based on a joint work with Măntoiu and Richard \cite{Lein_Mantoiu_Richard:anisotropic_mag_pseudo:2009}. This part of the thesis is dedicated to the study of magnetic pseudodifferential operators whose behavior in the position variable is characterized by some algebra $\Alg$ composed of bounded, uniformly continuous functions on $\R^d$. In other words, these operators are magnetic quantizations of functions $f : \Pspace \longrightarrow \C$ for which $x \mapsto f(x,\xi) \in \Alg$ holds for all $\xi \in {\R^d}^*$, and the components of the magnetic field $B$ are also elements of $\Alg$. This algebra is called \emph{anisotropy}. After introducing the smooth elements of $\Alg$, anisotropic Hörmander classes are defined in Chapter~\ref{psiDO_reloaded:magB_anisotropic_symbols:anisotropic_symbol_spaces}. The task of showing that the anisotropy is preserved under the Moyal product $\magW$ is taken up in Chapter~\ref{psiDO_reloaded:magB_anisotropic_symbols:symbol_composition}. It is also shown that the asymptotic expansions obtained in Chapter~\ref{asymptotics:expansions} are compatible with the anisotropy. 

Chapter~\ref{psiDO_reloaded:relevant_cStar_algebras} introduces a few $C^*$-algebras that are relevant for the results on inversion and affiliation in Chapter~\ref{psiDO_reloaded:inversion_and_affiliation}. The first main result is that the anisotropy is preserved under inversion: if $f$ is a real-valued, elliptic anisotropic Hörmander symbol $\Hoermr{m}{\rho} ({\R^d}^* ; \Alg^{\infty})$ of positive order $m$, then Moyal resolvents $(f - z)^{(-1)_B} \in \Hoermr{-m}{\rho} ({\R^d}^* ; \Alg^{\infty})$ are -- if they exist -- also \emph{anisotropic} Hörmander symbols of order $-m$ (Theorem~\ref{psiDO_reloaded:inversion_affiliation:thm:inversion}). The elegant proof is not based on a parametrix construction, but rather on a combination of an analytic result, Proposition~6.31 in \cite{Iftimie_Mantoiu_Purice:commutator_criteria:2008}, and a fact from the intersection of analysis and algebra \cite[Corollary~2.5]{Lauter:operator_theoretical_approach_melrose_algebras:1998}. The existence of the family of Moyal resolvents yields a principle of affiliation of $f$ to the twisted crossed product $\Fourier \bigl ( \Alg \rtimes^{\omega^B}_{\theta} \R^d \bigr )$ (Theorem~\ref{psiDO_reloaded:inversion_affiliation:thm:affiliation}). 

Chapter~\ref{psiDO_reloaded:spectral} is dedicated to the spectral analysis of magnetic pseudodifferential operators. Affiliating suitable functions on phase space to twisted crossed products allows the treatment of potentially unbounded pseudodifferential operators. For these observables, it is shown how the spectrum and the essential spectrum as sets can be recovered. Assume the anisotropy algebra is unital and contains $\Cont_{\infty}(\R^d)$, the functions on $\R^d$ vanishing at infinity. Then, the intuitive notion that the behavior of the potentials and magnetic fields at infinity is responsible for the essential spectrum is made rigorous in Theorem~\ref{psiDO_reloaded:spectral:thm:essential_spectrum}: the essential spectrum of a magnetic pseudodifferential operator is written as the union of spectra of magnetic $\Psi$DOs that `live on orbits at infinity.' This notion is made precise by borrowing tools from Gelfand theory and $C^*$-dynamical systems. Although the link between the Calkin algebra, the quotient of bounded operators by the ideal of the compact operators, and the essential spectrum is well-known [citation], we have obtained much more detailed information: depending on the anisotropy algebra $\Alg$, we may even be able to calculate the essential spectrum from the spectrum of a few simpler pseudodifferential operators. The fact that there is no Hilbert space analog of this decomposition highlights the usefulness of the combination of abstract algebraic and pseudodifferential methods. 
\medskip

\noindent
A brief outlook is given in the last chapter. An Appendix contains some additional information and auxiliary results needed in some of the proofs. 


\chapter{Magnetic Weyl Calculus} 
\label{magWQ}
%
The problem of `consistently' assigning operators on a Hilbert space to classical functions on phase space has seen quite a few attempts. As one of the basic questions, a coherent answer first written up in its entirety by Moyal in 1949~\cite{Moyal:Weyl_calculus:1949} who proposed to use 
\begin{align}
	\bigl ( \Op(f) u \bigr )(x) = \frac{1}{(2\pi)^d} \int_{\R^d} \dd y \int_{{\R^d}^*} \dd \eta \, e^{- i (y - x) \cdot \eta} \, f \bigl ( \tfrac{1}{2} (x + y) , \eta \bigr ) \, u(y) 
	\label{magWQ:eqn:non_mag_WQ}
\end{align}
for suitable functions $f : T^* \R^d \longrightarrow \C$ and $u \in \Hil$ where the Hilbert space $\Hil$ is typically $L^2(\R^d)$ or some Sobolev space, for instance. The problem is how to extend this definition to the case of a particle subjected to a magnetic field. Up until very late in the game, the standard recipe has been to utilize minimal coupling, \ie apply equation~\eqref{magWQ:eqn:non_mag_WQ} to $f \circ \vartheta^A(x,\xi) := f \bigl ( x , \xi - A(x) \bigr )$. Before we explain why this is not the correct solution, let us review standard Weyl calculus first.

\section{Standard Weyl calculus} 
\label{magWQ:standard_weyl_calculus}
There are many texts on standard Weyl calculus, \eg \cite{Folland:harmonic_analysis_hase_space:1989,Hoermander:Weyl_calculus:1979,Hoermander:analysis_PDO1:1983,Stein:harmonic_analysis:1993}, and although we will not stick to any of them in particular, we do not make any claims of originality. Our presentation emphasizes the structural aspects and introduces the paradigms which make the generalization to magnetic Weyl calculus logical and intuitive.

\subsection{Comparison of classical and quantum mechanical frameworks} 
\label{magWQ:standard_weyl_calculus:structural_frameworks}
Understanding of quantization requires knowledge of classical \emph{and} quantum mechanics. A quantization procedure is not merely a method to `consistently assign operators on $L^2(\R^d)$ to functions on phase space,' but rather a collection of procedures. A nice overview of the two frameworks can be found in the first few sections of chapter~5 in \cite{Waldmann:deformation_quantization:2008} and we will give a condensed account here: physical theories consist roughly of three parts: 
\begin{enumerate}[(i)]
	\item \emph{State space: }
	states describe the current configuration of the system and need to be encoded in a mathematical structure. 
	\item \emph{Observables: }
	they represent quantities physicists would like to measure. Related to this is the idea of spectrum as the set of possible outcomes of measurements as well as expectation values (if ones deals with distributions of states). 
	\item \emph{Evolution equation: }
	usually, one is interested in the time evolution of states as well as observables. As energy is the observable conjugate to time, energy functions generate time evolution. 
\end{enumerate}

\subsubsection{Hamiltonian framework of classical mechanics} 
\label{magWQ:standard_weyl_calculus:structural_frameworks:framework_classical_mechanics}
Pure states in classical mechanics are simply points on phase space $\Pspace := (T^* \R^d , \omega)$, \ie the cotangent bundle $T^* \R^d \cong \R^d \times {\R^d}^*$ endowed with a symplectic form $\omega$ that is usually taken to be $\omega_0 = \dd x \wedge \dd \xi$. This symplectic form determines the form of the evolution equation associated to the energy function $h$ called the hamiltonian. Mixed states are merely probability measures $\mu$ on $\Pspace$, \ie positive Borel measures normalized to $1$. An observable $f$ is a smooth function on $\Pspace$ with values in $\R$. Then the expectation value\index{expectation value!classical} of $f$ with respect to the state $\mu$ is given by the phase space average 
\begin{align*}
	\mathbb{E}_{\mu}(f) := \int_{T^* \R^d} \dd \mu(X) \, f(X) 
	. 
\end{align*}
The symplectic structure on $T^* \R^d$ induces a Poisson structure on $\Cont^{\infty}(T^* \R^d)$: with pointwise addition, multiplication and complex conjugation as involution, $\Cont^{\infty}(T^* \R^d)$ forms a Poisson-$\ast$-algebra. The Poisson bracket\index{Poisson bracket} defined via the symplectic form as 
\begin{align}
	\{ f , g \}_{\omega} := - \omega(X_f , X_g) 
\end{align}
where $X_f$ and $X_g$ satisfy $\omega(X_f , \cdot) = \dd f$ and $\omega(X_g , \cdot) = \dd g$, respectively. The evolution of observables is generated by 
\begin{align}
	\frac{\dd }{\dd t} f(t) = \bigl \{ h , f(t) \bigr \}_{\omega} 
	\label{magWQ:eqn:classical_evolution}
\end{align}
with $f(t) := f \circ \phi_t$ where $\phi_t$ is the hamiltonian flow\index{hamiltonian flow} generated by $h$. Equivalently, pure or mixed states can be time-evolved instead of the observables. Put in quantum mechanical terms, evolving observables corresponds to the Heisenberg picture\index{Heisenberg picture!classical}, evolving states corresponds to the Schrödinger picture\index{Schrödinger picture!classical}. 


\subsubsection{Quantum mechanics} 
\label{magWQ:standard_weyl_calculus:structural_frameworks:quantum_mechanics}
Here, pure states are rays in a Hilbert space $\Hil$, or, equivalently, orthogonal projections onto a state $\psi \in \Hil$. Mixed states are density operators $\rho$ that are positive trace-class operators normalized to $1$. Physical observables are selfadjoint, densely defined operators on $\Hil$. With the adjoint as involution and addition and multiplication defined as usual, they form a $\ast$-algebra. The role of the Poisson bracket is played by the commutator $[A , B] := A \, B - B \, A$. Expectation values\index{expectation value!quantum} are computed via the trace 
\begin{align}
	\mathbb{E}_{\rho}(A) := \mathrm{tr} \bigl ( \rho \, A \bigr ) 
	. 
\end{align}
The dynamics of the observables in the Heisenberg picture\index{Heisenberg picture!quantum} are generated by 
\begin{align}
	\frac{\dd }{\dd t} A(t) = \frac{i}{\hbar} \bigl [ H , A(t) \bigr ] 
\end{align}
which is structurally equivalent to equation~\eqref{magWQ:eqn:classical_evolution}. The unitary time-evolution group $U(t) = e^{- i \frac{t}{\hbar} H}$ satisfies the \emph{Schrödinger equation}\index{Schrödinger equation} 
\begin{align}
	i \hbar \frac{\dd }{\dd t} U(t) = H \, U(t) 
	. 
\end{align}
Then the time-evolved observable is given by 
\begin{align}
	A(t) = \mathrm{ad} \bigl (U(t) \bigr ) \bigl ( A \bigr ) := U(t)^* \, A \, U(t)  = e^{+ i \frac{t}{\hbar} H} \, A \, e^{- i \frac{i}{\hbar} H} 
	. 
\end{align}
The spectrum of an observable $\spec(A)$ defined in the usual functional analytic sense gives the possible outcomes of measurements in experiments while the projection-valued measure contains the statistics. 


\subsubsection{Comparison of the two frameworks} 
\label{magWQ:standard_weyl_calculus:structural_frameworks:comparison}
Now that we have an understanding of the mathematical structures which we have juxtaposed in Table~\ref{magWQ:table:overview_frameworks}, what can we deduce from this? First of all, what is usually considered a quantization only gives a third of the total answer: a map $\Op$ from suitable functions on phase space to operators on the Hilbert space $L^2(\R^d)$. However, it is clear that $\Op$ cannot simply map `functions onto operators:' the quantum algebra of observables $\Alg_{\mathrm{qm}}$ is significantly different from the classical algebra $\Alg_{\mathrm{cl}}$ -- it is \emph{noncommutative}. If we `map back' from $\Alg_{\mathrm{qm}}$ to $\Alg_{\mathrm{cl}}$, assuming that is possible, then we get a \emph{modified, noncommutative product $\Weyl$} of functions (that is the point of view of deformation quantization, see \eg \cite{Waldmann:deformation_quantization:2008}) which satisfies 
\begin{align*}
	\Op(f) \, \Op(g) = \Op(f \Weyl g) 
	. 
\end{align*}
The associated dequantization map $\Op^{-1}$ should also connect quantum states (written as projections or, more generally, density operators) with measures on phase space. So let $\rho_u := \sopro{u}{u}$ be a pure state, $\snorm{u}_{L^2} = 1$, and consider the expectation value of $\Op(f)$ with respect to $\rho_u$, $\mathbb{E}_{\rho_u} \bigl ( \Op(f) \bigr ) = \trace \bigl ( \rho_u \, \Op(f) \bigr ) = \sscpro{u}{\Op(f) u}$. Does there exist a measure $\mu_u$ on $T^* \R^d$ such that 
\begin{align*}
	\mathbb{E}_{\rho_u} \bigl ( \Op(f) \bigr ) = \sscpro{u}{\Op(f) u} \overset{?}{=} \int_{T^* \R^d} \dd \mu_u(X) \, f(X) = \mathbb{E}_{\mu_u}(f) 
	? 
\end{align*}
If so, what properties does $\mu_u = (2\pi)^{-\nicefrac{d}{2}} \, \WignerTrafo(\rho_u)$ have? It is clear that the properties of $\WignerTrafo$ should follow from the properties of $\Op$. Quite naturally, we demand the following from a `good' quantization procedure: 

\paragraph{Linearity} 
\label{magWQ:standard_weyl_calculus:structural_frameworks:comparison:linearity}
The map $\Op$ should be linear, \ie for two classical observables $f , g \in \mathfrak{A}_{\mathrm{cl}}$ taken from the algebra of classical observables and $\alpha , \beta \in \C$, we should have 
\begin{align*}
	\Op \bigl ( \alpha f + \beta g \bigr ) = \alpha \, \Op(f) + \beta \, \Op(g) \in \mathfrak{A}_{\mathrm{qm}} 
	. 
\end{align*}
Here, $\mathfrak{A}_{\mathrm{qm}}$ is an algebra of quantum observables. 


\paragraph{Compatibility with involution} 
\label{magWQ:standard_weyl_calculus:structural_frameworks:comparison:involution}
$\Op$ should intertwine complex conjugation and taking adjoints, \ie for all $f \in \mathfrak{A}_{\mathrm{cl}}$ 
\begin{align*}
	\Op (f^*) = \Op(f)^* \in \mathfrak{A}_{\mathrm{qm}} 
	. 
\end{align*}
%

\paragraph{Products} 
\label{magWQ:standard_weyl_calculus:structural_frameworks:comparison:products}
The two products \emph{cannot} be equivalent: the operator product is noncommutative and hence for general $f , g \in \mathfrak{A}_{\mathrm{cl}}$ 
\begin{align*}
	\Op(f \cdot g) \neq \Op(f) \cdot \Op(g) 
	. 
\end{align*}
Instead, for \emph{suitable} functions $f,g$, we can define a \emph{non-commutative} product $\Weyl$ on the level of functions on phase space such that 
\begin{align*}
	\Op(f \Weyl g) := \Op(f) \cdot \Op(g) \in \mathfrak{A}_{\mathrm{qm}} 
	. 
\end{align*}
\emph{A priori} it is not at all clear whether $f \Weyl g \in \mathfrak{A}_{\mathrm{cl}}$. 


\paragraph{Poisson bracket and commutator} 
\label{magWQ:standard_weyl_calculus:structural_frameworks:comparison:brackets}
The classical Poisson bracket and the quantum commutator play similar roles: they are both derivations, satisfy the Jacobi identity, and, in some sense, should be analogs of one another, 
\begin{align*}
	\bigl \{ f , g \bigr \} \leftrightsquigarrow \frac{i}{\hbar} \bigl [ \Op(f) , \Op(g) \bigr ] 
	. 
\end{align*}
Just like with the product, the quantization of the Poisson bracket usually does not coincide with with $\nicefrac{i}{\hbar}$ times the commutator. 

%
\begin{table}
	\begin{tabularx}{\textwidth}{>{\small\raggedright\hsize=3cm}X | >{\small\raggedright}X >{\small\raggedright\arraybackslash}X}
		\makebox[1cm]{} & \textit{Classical} & \textit{Quantum} \\ \hline 
		\textit{States} & positive normalized Borel measures $\mu$ on phase space $\Pspace$ & density operators on $L^2(\R^d)$ \\ [0.75ex] 
		\textit{Observables} & commutative Poisson-$\ast$-algebra $\Alg_{\mathrm{cl}}$ of functions on $\Pspace$ & noncommutative $\ast$-algebra $\Alg_{\mathrm{qm}}$ of operators acting on the Hilbert space $L^2(\R^d)$ \\ [0.75ex] 
		\textit{Building block observables} & position $x$ and momentum $p$ & position $\Qe$ and momentum $\Pe$ \\ [0.75ex] 
		\textit{Possible results of measurements} & $\mathrm{im}(f)$ & $\spec(A)$ \\ [0.75ex] 
		\textit{Generator of evolution} & hamiltonian function $h : \Pspace \longrightarrow \R$ & hamiltonian operator $H : \mathcal{D}(H) \longrightarrow L^2(\R^d)$ \\ [0.75ex] 
		\textit{Infinitesimal time evolution equation} & $\frac{\dd }{\dd t} f(t) = \bigl \{ h , f(t) \bigr \}$ & $\frac{\dd }{\dd t} A(t) = \frac{i}{\hbar} \bigl [ H , A(t) \bigr ]$ \\ [0.75ex] 
		\textit{Integrated time evolution} & hamiltonian flow $\phi_t$ as one-parameter group of automorphisms & $\mathrm{ad} \bigl ( e^{- i \frac{t}{\hbar} H} \bigr ) \bigl ( \cdot \bigr ) = e^{+ i \frac{t}{\hbar} H} \, \cdot \, e^{- i \frac{t}{\hbar} H}$ as one-parameter group of automorphisms \\ 
	\end{tabularx}
	\caption{Comparison of classical and quantum framework}
	\label{magWQ:table:overview_frameworks}
\end{table}
%



\subsection{The Weyl system} 
\label{magWQ:standard_weyl_calculus:weyl_system}
In position representation, the `building block operators' of quantum mechanics on $\R^d$, position $\Qe = \hat{x}$ which acts as multiplication by $x$, $(\Qe u)(x) = x \, u(x)$, and momentum $\Pe = - i \nabla_x$, $(\Pe u)(x) = - i (\nabla_x u)(x)$, are characterized by their commutation relations\index{commutation relations!non-magnetic} 
\begin{align}
	i [\Qe_l , \Qe_j] = 0 
	&&
	i [\Pe_l , \Pe_j] = 0 
	&&
	i [\Pe_l , \Qe_j] = \delta_{lj} 
	. 
\end{align}
Since commutators of unbounded operators are problematic \cite[Chapter~VIII.5]{Reed_Simon:M_cap_Phi_1:1972}, it is technically more convenient to encode the commutation relations into the so-called Weyl system which is a collection of unitary operators $\{ \WeylSys(X) \}_{X \in \Pspace}$ defined via the non-magnetic symplectic form $\sigma(X,Y) := \xi \cdot y - x \cdot \eta$ as\index{Weyl system!non-magnetic} 
\begin{align}
	\WeylSys(X) := e^{- i \sigma(X,(\Qe,\Pe))} = e^{- i (\xi \cdot \Qe - x \cdot \Pe)} 
	. 
\end{align}
Here $X \equiv (x,\xi)$, $Y \equiv (y,\eta)$ and $Z \equiv (z,\zeta)$ are points on phase space $\Pspace := T^* \R^d \cong \R^d \times {\R^d}^*$. Greek letters $\xi$, $\eta$ and $\zeta$ denote the momenta associated to $x$, $y$ and $z$. A simple Trotter argument shows that $\WeylSys(Y)$ acts on $u \in L^2(\R^d)$ as 
\begin{align}
	\bigl ( \WeylSys(Y) u \bigr )(x) = e^{- i (x + \frac{y}{2}) \cdot \eta} u(x + y) 
	\label{magWQ:eqn:usual_Weyl_System_applied}
	.
\end{align}
$\WeylSys : \Pspace \longrightarrow \mathcal{U} \bigl ( L^2(\R^d) \bigr )$, $X \mapsto \WeylSys(X)$, forms a strongly continuous projective representation of the group $\Pspace \cong \R^d \times {\R^d}^*$: for any $X , Y \in \Pspace$, the product of two Weyl operators gives another Weyl operator times a phase, 
\begin{align}
	\WeylSys(X) \, \WeylSys(Y) = e^{\frac{i}{2} \sigma(X,Y)} \WeylSys(X + Y) 
	. 
	\label{magWQ:eqn:usual_Weyl_System_composition_law}
\end{align}
This will be the key ingredient when determining the product formula. 


\subsection{Weyl quantization} 
\label{magWQ:standard_weyl_calculus:weyl_quantization}
We can define a convenient variant of the Fourier transform on $\Pspace$ via the symplectic form $\sigma$: for $f \in \Schwartz(\Pspace)$, we define\index{symplectic Fourier transform}\index{symplectic Fourier transform} 
\begin{align}
	\bigl ( \Fs f \bigr )(X) := \frac{1}{(2\pi)^d} \int_{\Pspace} \dd X' \, e^{i \sigma(X,X')} \, f(X') 
	. 
\end{align}
One easily checks that $\Fs$ is an involution, $\Fs^2 = \id_{\Schwartz}$, and thus $\Fs^{-1} = \Fs$. For Schwartz functions, we can now replace one exponential factor in 
\begin{align*}
	f(Y) = \frac{1}{(2\pi)^{2d}} \int_{\Pspace} \dd X \int_{\Pspace} \dd X' \, e^{i \sigma(Y , X)} e^{i \sigma(X,X')} \, f(X') 
\end{align*}
by $\WeylSys(X)$ and get the \emph{Weyl quantization of $f$},\index{Weyl quantization!non-magnetic} 
\begin{align}
	\Op(f) := \frac{1}{(2\pi)^d} \int_{\Pspace} \dd X \, (\Fs f)(X) \, \WeylSys(X) 
	. 
	\label{magWQ:eqn:usual_WQ}
\end{align}
The fact that $\WeylSys(X)$ is a projective group representation and the definition of the symplectic Fourier transform imply $\bigl ( (\Fs f)(X) \, \WeylSys(X) \bigr )^* = (\Fs f^*)(-X) \, \WeylSys(-X)$ and thus Weyl quantization corresponds to symmetric operator ordering,\index{operator ordering} 
\begin{align*}
	\Op(f)^* &= \frac{1}{(2\pi)^d} \int_{\Pspace} \dd X \, (\Fs f)^*(X) \, \WeylSys(X)^* 
	= \frac{1}{(2\pi)^d} \int_{\Pspace} \dd X \, (\Fs f^*)(-X) \, \WeylSys(-X) 
	\\
	&
	= \Op(f^*) 
	. 
\end{align*}
Morally, the right-hand side of the above definition reduces to $f(\Qe,\Pe)$: 
\begin{align*}
	\frac{1}{(2\pi)^d} \int_{\Pspace} \dd X \, &e^{- i \sigma(X,(\Qe,\Pe))} \, \frac{1}{(2\pi)^d} \int_{\Pspace} \dd X' e^{i \sigma(X,X')} \, f(X') = 
	\\
	&= \frac{1}{(2\pi)^{2d}} \int_{\Pspace} \dd X' \left ( \int_{\Pspace} \dd X e^{i \sigma(X,X' - (\Qe,\Pe))} \right ) \, f(X') 
	\\
	&= \int_{\Pspace} \dd X' \, \delta \bigl ( X' - (\Qe,\Pe) \bigr ) \, f(X') = f(\Qe,\Pe) 
\end{align*}
Obviously, the reader should add quotation marks to the above. We can write down the action of a Weyl quantized operator on a wave function explicitly: 
\begin{lem}\label{magWQ:standard_weyl_calculus:weyl_quantization:lem:action_Op}
	The Weyl quantization of $h \in \Schwartz(\Pspace)$ defines a bounded operator on $L^2(\R^d)$ whose operator norm is bounded by 
	\begin{align*}
		\bnorm{\Op(h)}_{\mathcal{B}(L^2(\R^d))} \leq (2\pi)^{-d} \bnorm{\Fs h}_{L^1(\Pspace)} < \infty 
	\end{align*}
	and for all $u \in L^2(\R^d)$, we have 
	\begin{align}
		\bigl ( \Op(h) u \bigr )(x) &= \frac{1}{(2\pi)^d} \int_{\R^d} \dd y \int_{{\R^d}^*} \dd \eta \, e^{- i(y - x) \cdot \eta} \, h \bigl ( \tfrac{1}{2}(x + y) , \eta \bigr ) \, u(y) 
		\\
		&=: \frac{1}{(2\pi)^{\nicefrac{d}{2}}} \int_{\R^d} \dd y \, K_h(x,y) \, u(y) 
		\label{magWQ:standard_weyl_calculus:eqn:def_integral_kernel}
		. 
	\end{align}
\end{lem}
\begin{proof}
	We can interpret $\Op(h)$ as a Bochner integral with respect to the operator norm on $L^2(\R^d)$ which immediately leads to the desired bound: 
	\begin{align*}
		\bnorm{\Op^A(h)}_{\mathcal{B}(L^2(\R^d))} &\leq \frac{1}{(2\pi)^d} \int_{\Pspace} \dd X \, \bnorm{(\Fs)(X) \, \WeylSys(X)}_{\mathcal{B}(L^2(\R^d))} 
		\\
		&
		= \frac{1}{(2\pi)^d} \int_{\Pspace} \dd X \, \babs{(\Fs)(X)} 
		= (2\pi)^{-d} \, \bnorm{\Fs f}_{L^1(\Pspace)}
	\end{align*}
	Elementary manipulations using equation~\eqref{magWQ:eqn:usual_Weyl_System_applied} yield equation~\eqref{magWQ:standard_weyl_calculus:eqn:def_integral_kernel}. 
\end{proof}
%
%
\begin{remark}\label{magWQ:standard_weyl_calculus:rem:kernel_map}
	The kernel map\index{kernel map!non-magnetic} $h \mapsto K_h$ is an isomorphism between $\Schwartz(\Pspace)$ and $\Schwartz(\R^d \times \R^d)$ and thus extends to tempered distributions. This is the starting point for defining Weyl calculus on distributions. 
\end{remark}
To be able to treat the prefactor in a coherent manner, we will add one more defintion: $\Int$ is just the regular integral modulo a factor of $(2\pi)^{-\nicefrac{d}{2}}$. 
\begin{defn}[Integral map]
	Let $h \in \Schwartz(\Pspace)$ be a function and $K_h$ be the integral kernel defined via equation~\eqref{magWQ:standard_weyl_calculus:eqn:def_integral_kernel}. Then we define 
	\begin{align*}
		\Op(h) =: \Int(K_h) 
		. 
	\end{align*}
\end{defn}
%


\subsection{The Wigner transform} 
\label{magWQ:standard_weyl_calculus:wigner_trafo}
If we look at the definition of $\Op(f)$, equation~\eqref{magWQ:eqn:usual_WQ}, then it is quite natural to see how the expectation values of the form $\bscpro{v}{\Op(f) u}$ can be rewritten as phase space averages of $f$: rewriting the expectation value as 
\begin{align}
	\sscpro{v}{\Op(f) u} &= \frac{1}{(2\pi)^{\nicefrac{d}{2}}} \int_{\Pspace} \dd X \, (\Fs f)(X) \, \bscpro{v}{\WeylSys(X) u} 
	\label{magWQ:standard_weyl_calculus:eqn:phase_space_exp_value} 
	\\
	&= \frac{1}{(2\pi)^d} \int_{\Pspace} \dd X \, f(X) \, \Fs \bigl ( \bscpro{v}{\WeylSys(\cdot) u} \bigr )(-X) 
	, 
	\notag 
\end{align}
suggests to look at the symplectic Fourier transform of the expectation value of the Weyl system. Let us start with the first building block: 
\begin{defn}[Fourier-Wigner transform\index{Fourier-Wigner transform!non-magnetic}]\label{magWQ:standard_weyl_calculus:fourier_wigner_trafo}
	Let $u,v \in \Schwartz(\R^d)$. Then we define the magnetic Fourier-Wigner transform $\rho(u,v)$ to be 
	\begin{align}
		\bigl ( \rho(u,v) \bigr )(X) := (2\pi)^{-\nicefrac{d}{2}} \, \bscpro{v}{\WeylSys(X) u}
	\end{align}
\end{defn}
\begin{remark}
	Our choice not to include the factor of $(2\pi)^{-\nicefrac{d}{2}}$ in the definition of the Fourier-Wigner transform will lead to a `missing' factor of $(2\pi)^{-\nicefrac{d}{2}}$ in the Wigner transform. This way, the Wigner transform is unitary and the inverse of the kernel map as defined in Remark~\ref{magWQ:standard_weyl_calculus:rem:kernel_map}. 
\end{remark}
\begin{lem}\label{magWQ:standard_weyl_calculus:lem:Fourier-Wigner_trafo}
	Let $u,v \in \Schwartz(\R^d)$. Then it holds 
	\begin{align*}
		\bigl ( \rho(u,v) \bigr )(X) = (2\pi)^{-\nicefrac{d}{2}} \, \bscpro{v}{\WeylSys(X) u} 
		&= \frac{1}{(2\pi)^{\nicefrac{d}{2}}} \int_{\R^d} \dd y \, e^{-i y \cdot \xi} \, {v}^{\ast} \bigl ( y - \tfrac{x}{2} \bigr ) \, u \bigl ( y + \tfrac{x}{2} \bigr ) 
	\end{align*}
	and $\rho(u,v) \in \Schwartz(\Pspace)$
\end{lem}
\begin{proof}
	Plugging equation~\eqref{magWQ:eqn:usual_Weyl_System_applied} into the scalar product, we get 
	\begin{align*}
		\bigl ( \rho(u,v) \bigr )(X) &= (2\pi)^{-\nicefrac{d}{2}} \, \bscpro{v}{\WeylSys(X) u} 
		= \frac{1}{(2\pi)^{\nicefrac{d}{2}}} \int \dd y \, {v}^{\ast}(y) \, \bigl ( \WeylSys(X) u \bigr )(y) \\
		&= \frac{1}{(2\pi)^{\nicefrac{d}{2}}} \int_{\R^d} \dd y \, {v}^{\ast}(y) \, e^{-i (y + \frac{x}{2}) \cdot \xi} \, u(y+x) \\ 
		&= \frac{1}{(2\pi)^{\nicefrac{d}{2}}} \int_{\R^d} \dd y \, e^{-i y \cdot \xi} \, {v}^{\ast} \bigl ( y - \tfrac{x}{2} \bigr ) \, u \bigl ( y + \tfrac{x}{2} \bigr ) 
		. 
	\end{align*}
	Since $\rho(u,v)$ is the partial Fourier transform of a Schwartz function, $\rho(u,v)$ exists in $\Schwartz(\Pspace)$. 
\end{proof}
To write the quantum expectation value as a phase space averate, we still have to push over the Fourier transform. 
\begin{defn}[Wigner transform\index{Wigner transform!non-magnetic}]
	Let $u,v \in \Schwartz(\R^d)$. The Wigner transform $\WeylSys(u,v)$ is defined as the symplectic Fourier transform of $\rho(u,v)$, 
	\begin{align*}
		\bigl ( \WignerTrafo(u,v) \bigr )(X) := \bigl ( \Fs \rho(u,v) \bigr )(-X) 
		. 
	\end{align*}
\end{defn}
\begin{remark}
	There is a reason why we need to use $\WeylSys(-X)$ and not $\WeylSys(+X)$: the symplectic Fourier transform is unitary on $L^2(\Pspace)$ and 
	\begin{align*}
		\bscpro{f}{g}_{L^2(\Pspace)} = \bscpro{\Fs f}{\Fs g}_{L^2(\Pspace)} = \bigl ( (\Fs f)^* , \Fs g \bigr ) = \bigl ( (\Fs f^*)(- \, \cdot) , \Fs g \bigr )
	\end{align*}
	holds. The extra sign stems from the fact that we are missing complex conjugation in integral~\eqref{magWQ:standard_weyl_calculus:eqn:phase_space_exp_value}. 
\end{remark}
\begin{lem}\label{magWQ:lem:WignerTransform}
	The Wigner transform $\WignerTrafo(u,v)$\index{Wigner transform!non-magnetic} with respect to $u,v \in \Schwartz(\R^d)$ is an element of $\Schwartz(\Pspace)$ and given by 
	\begin{align*}
		\bigl ( \WignerTrafo(u,v) \bigr )(X) = \frac{1}{(2\pi)^{\nicefrac{d}{2}}} \int_{\R^d} \dd y \, e^{- i y \cdot \xi} \, {v}^{\ast} \bigl ( x - \tfrac{y}{2} \bigr ) \, u \bigl ( x + \tfrac{y}{2} \bigr ) 
		. 
	\end{align*}
\end{lem}
\begin{proof}
	As a symplectic Fourier transform of the Schwartz function $\bigl ( \rho(u,v) \bigr )(- \cdot)$ (see Lemma~\ref{magWQ:standard_weyl_calculus:lem:Fourier-Wigner_trafo}), the Wigner transform $\WignerTrafo(u,v)$ is again in $\Schwartz(\Pspace)$. Hence, the following integrals exist and we compute 
	\begin{align*}
		\bigl ( \WignerTrafo(u,v) \bigr )(X) &= \frac{1}{(2\pi)^d} \int_{\Pspace} \dd Y \, e^{i \sigma(-X,Y)} \rho(u,v) (Y) \\ 
		&= \frac{1}{(2\pi)^{\nicefrac{3d}{2}}} \int_{\Pspace} \dd Y \, \int_{\R^d} \dd z \, e^{- i (\xi \cdot y - x \cdot \eta)} e^{- i z \cdot \eta} \, v^{\ast} \bigl ( z - \tfrac{y}{2} \bigr ) \, u \bigl ( z + \tfrac{y}{2} \bigr ) \\ 
		&= \frac{1}{(2\pi)^{\nicefrac{3d}{2}}} \int_{\R^d} \dd y \, \int_{\R^d} \dd z \, \int_{{\R^d}^*} \dd \eta \, e^{i (x-z) \cdot \eta} e^{- i y \cdot \xi} \, v^{\ast} \bigl ( z - \tfrac{y}{2} \bigr ) \, u \bigl ( z + \tfrac{y}{2} \bigr ) \\
		&= \frac{1}{(2\pi)^{\nicefrac{d}{2}}} \int_{\R^d} \dd y \, e^{- i y \cdot \xi} \, v^{\ast} \bigl ( x - \tfrac{y}{2} \bigr ) \, u \bigl ( x + \tfrac{y}{2} \bigr ) 
		. 
	\end{align*}
	This concludes the proof. 
\end{proof}
\begin{remark}
	We can easily extend the Wigner transform of operator kernels: if $K_T$ is the kernel of the operator $T = \Int(K_T)$, then we define 
	\begin{align}
		\WignerTrafo K_T (x,\xi) := \frac{1}{(2\pi)^{\nicefrac{d}{2}}} \int_{\R^d} \dd y \, e^{- i y \cdot \xi} \, K_T \bigl ( x + \tfrac{y}{2} , x - \tfrac{y}{2} \bigr )
		\label{magWQ:standard_weyl_calculus:eqn:Wigner_trafo_kernel}
	\end{align}
	and $\WignerTrafo$ is a bijection between $\Schwartz(\R^d \times \R^d)$ and $\Schwartz(\Pspace)$. 
\end{remark}
%
The Wigner transform of even a pure state in general does not define a true probability measure: 
\begin{example}
	Take $d = 1$ and consider $u(x) = x \, e^{- \frac{x^2}{4}}$, for instance, the first excited state of the harmonic oscillator. Then we calculate the Wigner transform to be 
	\begin{align*}
		\bigl ( \WignerTrafo(u,u) \bigr )(x,\xi) &= \frac{1}{\sqrt{2\pi}} \int_{\R} \dd y \, e^{- i y \cdot \xi} \, u^* \bigl ( x - \tfrac{y}{2} \bigr ) \, u \bigl ( x + \tfrac{y}{2} \bigr ) 
		\\
		&= \frac{1}{\sqrt{2\pi}} \int_{\R^d} \dd y \, e^{- i y \cdot \xi} \, \bigl ( x - \tfrac{y}{2} \bigr ) \, \bigl ( x + \tfrac{y}{2} \bigr ) \, e^{- \frac{1}{4} [(x - \frac{y}{2})^2 + (x + \frac{y}{2})^2]} 
		\\
		&= 2 e^{- \frac{x^2}{2}} \frac{1}{\sqrt{2\pi}} \int_{\R} \dd y \, e^{- i y \cdot 2 \xi} \, (x^2 - y^2) \, e^{- \frac{y^2}{2}} 
		\\ 
		&= 2 e^{- \frac{x^2}{2}} \, \Bigl ( x^2 e^{- 2 \xi^2} + \tfrac{1}{2} \partial_{\xi}^2 \bigl ( e^{- 2 \xi^2} \bigr ) \Bigr )
		\\
		&
		= 2 \Bigl ( x^2  + \tfrac{1}{2} (4 \xi)^2 - 2 \Bigr ) \, e^{- 2 x^2} e^{- 2 \xi^2}
		. 
	\end{align*}
	Hence $\WignerTrafo(u,u) \not\geq 0$ and $\WignerTrafo(u,u) \, \dd X$ is not a probability measure. 
\end{example}
%
For convenience of the reader, we list some properties of the Wigner transform which are easy to prove: 
\begin{thm}[Properties of the Wigner transform]
	Let $u , v \in \Schwartz(\R^d)$, $x \in \R^d$ and $\xi \in {\R^d}^*$. 
	\begin{enumerate}[(i)]
		\item $\WignerTrafo(v,v)$ is a real-valued function, but not necessarily positive. 
		\item The marginals of the Wigner transform of $u , v \in \Schwartz(\R^d)$ with respect to $x$ and $\xi$ are 
		\begin{align*}
			\frac{1}{(2\pi)^{\nicefrac{d}{2}}} \int_{\R^d} \dd x \, \WignerTrafo(v,u) (x,\xi) &= (\Fourier u)^*(\xi) \, (\Fourier v)(\xi) 
			, 
			\\
			\frac{1}{(2\pi)^{\nicefrac{d}{2}}} \int_{{\R^d}^*} \dd \xi \, \WignerTrafo(v,u) (x,\xi) &= u^*(x) \, v(x) 
			. 
		\end{align*}
		\item $\displaystyle \int_{\Pspace} \dd X \, \WignerTrafo(v,u)(X) = (2\pi)^{\nicefrac{d}{2}} \, \scpro{u}{v}$ 
		\item $\bnorm{\WignerTrafo(u,v)}_{L^2(\Pspace)} = \norm{u}_{L^2(\R^d)} \norm{v}_{L^2(\R^d)}$
		\item Let $R$ be the reflection operator defined by $(R u)(x) := u(-x)$. Then $\WignerTrafo(R v , R u)(X) = \WignerTrafo(v , u)(-X)$ holds. 
		\item $\WignerTrafo(v^*,u^*)(x,\xi) = \WignerTrafo(u,v)(x,-\xi)$ 
		\item $\WignerTrafo \bigl ( U(y) v , U(y) u \bigr )(x,\xi) = \WignerTrafo(v , u)(x - y,\xi)$ for all $y \in \R_x^d$
	\end{enumerate}
\end{thm}
\begin{proof}
	\begin{enumerate}[(i)]
		\item We have to show $\WignerTrafo(u,u)^* = \WignerTrafo(u,u)$: plugging in the complex conjugate of $\WignerTrafo(u,u)$, we get 
		\begin{align*}
			\WignerTrafo(u,u)^*(x,\xi) &= \biggl ( \frac{1}{(2\pi)^{\nicefrac{d}{2}}} \int_{\R^d} \dd y \, e^{- i y \cdot \xi} \, u^{\ast} \bigl ( x - \tfrac{y}{2} \bigr ) \, u \bigl ( x + \tfrac{y}{2} \bigr )  \biggr )^* 
			\\
			&= \frac{1}{(2\pi)^{\nicefrac{d}{2}}} \int_{\R^d} \dd y \, e^{+ i y \cdot \xi} \, u \bigl ( x - \tfrac{y}{2} \bigr ) \, u^* \bigl ( x + \tfrac{y}{2} \bigr ) 
			\\
			&= \frac{1}{(2\pi)^{\nicefrac{d}{2}}} \int_{\R^d} \dd y \, e^{- i y \cdot \xi} \, u \bigl ( x + \tfrac{y}{2} \bigr ) \, u^* \bigl ( x - \tfrac{y}{2} \bigr ) 
			= \WignerTrafo(u,u)(x,\xi) 
			. 
		\end{align*}
		We have already given an example where $\WignerTrafo(u,u)$ is not positive. 
		\item If we take the marginals with respect to $x$, then up to a factor of $(2\pi)^{\nicefrac{d}{2}}$ that is due to the choice of convention in Definition~\ref{magWQ:standard_weyl_calculus:fourier_wigner_trafo}, we get 
		\begin{align*}
			\int_{\R^d} \dd x \, \WignerTrafo(v,u)(x,\xi) &= \frac{1}{(2\pi)^{\nicefrac{d}{2}}} \int_{\R^d} \dd x \int_{\R^d} \dd y \, e^{-i y \cdot \xi} \, u \bigl ( x - \tfrac{y}{2} \bigr )^* \, v \bigl ( x + \tfrac{y}{2} \bigr ) 
			\\
			&= \frac{1}{(2\pi)^{\nicefrac{d}{2}}} \int_{\R^d} \dd x' \int_{\R^d} \dd y \, e^{-i y \cdot \xi} \, u (x')^* \, v (x' + y) 
			\\
			&= \frac{1}{(2\pi)^{\nicefrac{d}{2}}} \int_{\R^d} \dd x' \int_{\R^d} \dd y' \, e^{-i (y'-x') \cdot \xi} \, u (x')^* \, v (y') 
			\\
			&= (2\pi)^{\nicefrac{d}{2}} \, (\Fourier u)^*(\xi) \, (\Fourier v)(\xi)
			. 
		\end{align*}
		The other marginal can be obtained analogously. 
		\item This follows immediately from (ii), the fact that $u$ and $v$ are square integrable and the Plancherel theorem. 
		\item We plug in the definition of the Wigner transform and compute 
		\begin{align*}
			\bnorm{\WignerTrafo(v,u)}_{L^2(\Pspace)}^2 &= \int_{\R^d} \dd x \int_{{\R^d}^*} \dd \xi \, \babs{\WignerTrafo(v,u)(x,\xi)}^2 
			\\
			&= \frac{1}{(2\pi)^d} \int_{\R^d} \dd x \int_{{\R^d}^*} \dd \xi \int_{\R^d} \dd y \, \int_{\R^d} \dd y' \, e^{+i y \cdot \xi} e^{- i y' \cdot \xi}
			\cdot \\
			&\qquad \qquad \qquad \cdot 
			u^* \bigl ( x - \tfrac{y}{2} \bigr ) \, v \bigl ( x + \tfrac{y}{2} \bigr ) \, u \bigl ( x - \tfrac{y'}{2} \bigr ) \, v^* \bigl ( x + \tfrac{y'}{2} \bigr ) 
			\\
			&= \frac{1}{(2\pi)^d} \int_{\R^d} \dd x \int_{\R^d} \dd y \, \int_{\R^d} \dd y' \, \biggl ( \int_{{\R^d}^*} \dd \xi \, e^{+i (y - y') \cdot \xi} \biggr ) 
			\cdot \\
			&\qquad \qquad \qquad \cdot 
			u^* \bigl ( x - \tfrac{y}{2} \bigr ) \, v \bigl ( x + \tfrac{y}{2} \bigr ) \, u \bigl ( x - \tfrac{y'}{2} \bigr ) \, v^* \bigl ( x + \tfrac{y'}{2} \bigr ) 
			\\
			&= \int_{\R^d} \dd x \int_{\R^d} \dd y \, u^* \bigl ( x - \tfrac{y}{2} \bigr ) \, v \bigl ( x + \tfrac{y}{2} \bigr ) \, u \bigl ( x - \tfrac{y}{2} \bigr ) \, v^* \bigl ( x + \tfrac{y}{2} \bigr ) 
			. 
		\end{align*}
		After two changes of variables, this simplifies to 
		\begin{align*}
			\ldots &= \biggl ( \int_{\R^d} \dd x \, u^* (x) \, u (x) \biggr ) \, \biggl ( \int_{\R^d} \dd y \, v (y) \, v^* (y) \biggr ) 
			\\
			&= \norm{u}_{L^2(\R^d)}^2 \norm{v}_{L^2(\R^d)}^2 
			. 
		\end{align*}
		\item This is a direct consequence of the definition, 
		\begin{align*}
			\WignerTrafo \bigl ( R v , R u \bigr )(x,\xi) &= \frac{1}{(2 \pi)^{\nicefrac{d}{2}}} \int_{\R^d} \dd y \, e^{- i y \cdot \xi} \, (R u)^* \bigl ( x - \tfrac{y}{2} \bigr ) \, (R v) \bigl ( x + \tfrac{y}{2} \bigr ) 
			\\
			&= \frac{1}{(2 \pi)^{\nicefrac{d}{2}}} \int_{\R^d} \dd y \, e^{- i y \cdot \xi} \, u^* \bigl ( - x + \tfrac{y}{2} \bigr ) \, v \bigl ( - x - \tfrac{y}{2} \bigr ) 
			\\
			&= \frac{1}{(2 \pi)^{\nicefrac{d}{2}}} \int_{\R^d} \dd y' \, e^{- i y' \cdot (- \xi)} \, u^* \bigl ( (- x) - \tfrac{y'}{2} \bigr ) \, v \bigl ( (- x) + \tfrac{y'}{2} \bigr ) 
			\\
			&= \WignerTrafo(v,u)(- x, -\xi) 
			. 
		\end{align*}
		\item Follows directly from the definition of the Wigner transform. 
		\item Follows directly from the definition of the Wigner transform. 
	\end{enumerate}
\end{proof}
Some other classical results concerning Wigner functions and Wigner measures can be found in \cite{Lions_Paul:Wigner_measures:1993} and \cite{Mecklenbraeuker_Hlawatsch:Wigner_distribution:1997}. 
\begin{cor}
	For $u,v \in \Schwartz(\R^d)$ and $f \in \Schwartz(\Pspace)$ we have 
	\begin{align*}
		\bscpro{v}{\Op(f) u} = \frac{1}{(2\pi)^{\nicefrac{d}{2}}} \int_{\Pspace} \dd X \, f(X) \, \bigl ( \WignerTrafo(u,v) \bigr )(X) 
		. 
	\end{align*}
\end{cor}
%
The Wigner transform is essentially the inverse of Weyl quantization: 
\begin{prop}\label{weyl_calculus:wigner_transform:prop:inverse_weyl_quantization}
	Let $T \in \Op \bigl ( \Schwartz(\Pspace) \bigr ) \subset \mathcal{B} \bigl ( L^2(\R^d) \bigr )$ be an operator with operator kernel $K_T$. The map $\Op^{-1} : \Op \bigl ( \Schwartz(\Pspace) \bigr ) \longrightarrow \Schwartz(\Pspace)$ defined by 
	\begin{align*}
		\Op^{-1} (T) := \WignerTrafo K_T
	\end{align*}
	is the inverse to Weyl quantization, \ie we have $T = \Op \bigl ( \WignerTrafo K_T \bigr)$ for all $T \in \Op \bigl ( \Schwartz(\Pspace) \bigr )$. Conversely, if we take any $f \in \Schwartz(\Pspace)$ with Weyl kernel $K_f$, then $\Op^{-1} \bigl ( \Op(f) \bigr ) = \WignerTrafo K_f = f$ holds. 
\end{prop}
\begin{proof}
	Let $K_T$ be the operator kernel associated to the operator $T \in \Op \bigl ( \Schwartz(\Pspace) \bigr )$. Then $K_T$ has to be in $\Schwartz(\R^d \times \R^d)$: as $T$ is an operator that has been obtained by Weyl quantization, there is a unique preimage $f_T \in \Schwartz(\Pspace)$ and its Weyl kernel $K_{f_T}$ has to be in $\Schwartz(\R^d \times \R^d)$. 
	
	We have to confirm that $\Op \bigl ( \Op^{-1}(T) \bigr ) = T$ and $\Op^{-1} \bigl ( \Op(f) \bigr ) = K_f$ hold. Plugging in the definition and making a change of variables, we get 
	\begin{align*}
		\bigl ( \Op \bigl ( \Op^{-1}&(T) \bigr ) \varphi \bigr )(x) = \frac{1}{(2\pi)^d} \int_{\R^d_x} \dd y \int_{\R^d_{\xi}} \dd \eta \, e^{- i (y - x) \cdot \eta} \, \bigl ( \WignerTrafo K_T \bigr ) \bigl ( \tfrac{1}{2}(x + y) , \eta \bigr ) \, \varphi(y) 
		\\
		&= \frac{1}{(2\pi)^{\nicefrac{d}{2}}} \int_{\R^d_x} \dd y \, K_T \bigl ( \tfrac{1}{2}(x + y + (x - y)) , \tfrac{1}{2} (x + y - (x - y)) \bigr ) \, \varphi(y) 
		\\
		&= \frac{1}{(2\pi)^{\nicefrac{d}{2}}} \int_{\R^d_x} \dd y \, K_T (x,y) \, \varphi(y) 
		= \bigl ( T \varphi \bigr )(x) 
		. 
	\end{align*}
	On the other hand, let $T = \Op(f)$ be the Weyl quantization of $f \in \Schwartz(\Pspace)$. Then $\Op^{-1} T = f$ follows from direct calculation: using 
	\begin{align*}
		K_f \bigl ( x + \tfrac{y}{2} , x - \tfrac{y}{2} \bigr ) = (\Fourier_2 f)(x , -y) 
		= (\Fourier_2^{-1} f)(x , y) 
		, 
	\end{align*}
	we get 
	\begin{align*}
		\Op^{-1} \bigl ( \Op(f) \bigr ) &= \WignerTrafo K_f(x,\xi) =  \frac{1}{(2\pi)^{\nicefrac{d}{2}}} \int_{\R^d_x} \dd y \, e^{- i y \cdot \xi} \, K_f \bigl ( x - \tfrac{1}{2} y , x + \tfrac{1}{2} y \bigr ) 
		\\
		&= \frac{1}{(2\pi)^{d}} \int_{\R^d_x} \dd y \int_{\R^d_{\xi}} \dd \eta \, e^{- i y \cdot \xi} e^{+ i y \cdot \eta} \, f ( x , \eta ) 
		= f(x , \xi) 
		. 
	\end{align*}
	To show that the dequantization $\Op^{-1}$ maps $\Op(\Schwartz(\Pspace))$ onto $\Schwartz(\Pspace)$, we invoke Remark~\ref{magWQ:standard_weyl_calculus:rem:kernel_map} and Lemma~\ref{magWQ:lem:WignerTransform} which state that the kernel map $K : \Schwartz(\Pspace) \longrightarrow \Schwartz(\R^d \times \R^d)$, $f \mapsto K_f$, and the Wigner transform $\WignerTrafo : \Schwartz(\R^d \times \R^d) \longrightarrow \Schwartz(\Pspace)$ are bijective. Hence the composition of the kernel map $K$ and the Wigner transform $\WignerTrafo$ is a bijection as well. In fact, 
	\begin{align*}
		\WignerTrafo \circ K : \Schwartz(\Pspace) \longrightarrow \Schwartz(\Pspace) 
	\end{align*}
	is the identity map by the above calculation. 
\end{proof}
%


\subsection{The Weyl product} 
\label{magWQ:standard_weyl_calculus:weyl_product}
The Weyl product\index{Weyl product!non-magnetic} emulates the operator product on the level of functions (or later: tempered distributions) on phase space: for two Schwartz functions $f , g \in \Schwartz(\Pspace)$, there is a Schwartz function $f \Weyl g \in \Schwartz(\Pspace)$ such that 
\begin{align*}
	\Op ( f \Weyl g ) = \Op(f) \, \Op(g) 
\end{align*}
holds. Obviously, $\Weyl$ inherits the noncommutativity of the operator product. 
\begin{thm}
	For two Schwartz functions $f , g \in \Schwartz(\Pspace)$, the tempered distribution $f \Weyl g$ which satisfies $\Op \bigl ( f \Weyl g \bigr ) = \Op(f) \, \Op(g)$ is the Schwartz function given by 
	\begin{align}
		( f \Weyl g )(X) &= \frac{1}{(2 \pi)^{2d}} \int_{\Pspace} \dd Y \int_{\Pspace} \dd Z \, e^{i \sigma(X , Y + Z)} \, e^{\frac{i}{2} \sigma(Y,Z)} \, (\Fs f)(Y) \, (\Fs g)(Z) 
		\\
		&= \frac{1}{\pi^{2d}} \int_{\Pspace} \dd Y \int_{\Pspace} \dd Z \, e^{- i 2 \sigma(X - Y', X - Z')} \, f(Y') \, g(Z') 
		\notag 
		. 
	\end{align}
\end{thm}
\begin{proof}[Sketch]
	Using the definition of $\Op$ and the composition law of the Weyl system (equation~\eqref{magWQ:eqn:usual_Weyl_System_composition_law}), we get 
	\begin{align*}
		\Op(f) \, &\Op(g) u = \frac{1}{(2\pi)^{2d}} \int_{\Pspace} \dd Y \int_{\Pspace} \dd Z \, (\Fs f)(Y) \, (\Fs g)(Z) \, \WeylSys(Y) \, \WeylSys(Z) u 
		\\
		&= \frac{1}{(2\pi)^{2d}} \int_{\Pspace} \dd Y \int_{\Pspace} \dd Z \, (\Fs f)(Y) \, (\Fs g)(Z) \, e^{\frac{i}{2} \sigma(Y,Z)} \, \WeylSys(Y+Z) u 
		\\
		&= \frac{1}{(2\pi)^d} \int_{\Pspace} \dd Z \left ( \frac{1}{(2\pi)^d} \int_{\Pspace} \dd Y \, e^{\frac{i}{2} \sigma(Y,Z-Y)} \, (\Fs f)(Y) \, (\Fs g)(Z-Y) \right ) \WeylSys(Z) u 
	\end{align*}
	for any $u \in \Schwartz(\R^d)$. We recognize the inner integral as $\bigl ( \Fs (f \Weyl g) \bigr )(Z)$ and thus we add a Fourier transform to obtain the first of the two equivalent forms of the product formula: 
	\begin{align*}
		(f \Weyl g)(X) &= \Fs \left ( \frac{1}{(2\pi)^d} \int_{\Pspace} \dd Y \, e^{\frac{i}{2} \sigma(Y,\cdot-Y)} \, (\Fs f)(Y) \, (\Fs g)(\cdot-Y) \right )(X) 
		\\
		&= \frac{1}{(2\pi)^{2d}} \int_{\Pspace} \dd Z \int_{\Pspace} \dd Y \, e^{i \sigma(X,Z)} \, e^{\frac{i}{2} \sigma(Y,Z-Y)} \, (\Fs f)(Y) \, (\Fs g)(Z-Y)
		\\
		&= \frac{1}{(2\pi)^{2d}} \int_{\Pspace} \dd Y \int_{\Pspace} \dd Z \, e^{i \sigma(X,Y+Z)} \, e^{\frac{i}{2} \sigma(Y,Z)} \, (\Fs f)(Y) \, (\Fs g)(Z) 
		. 
	\end{align*}
	One can derive the second form of $f \Weyl g$ from elementary manipulations which we will detail only for the magnetic case. We postpone the proof that $f \Weyl g \in \Schwartz(\Pspace)$ to Chapter~\ref{magWQ:magnetic_weyl_calculus:magnetic_product} where we consider the magnetic Weyl product. 
\end{proof}
%


\subsection{Quantization of Hörmander symbols} 
\label{magWQ:standard_weyl_calculus:symbol_quantization}
So we have a nice quantization procedure, but we are still missing something up to now: we cannot quantize $h(x,\xi) = \tfrac{1}{2} \xi^2 + V(x)$ yet, because even for rapidly decaying potentials, $h \not \in \Schwartz(\Pspace)$. We will show how to use duality techniques derived from two papers by Gracia-Bondìa and Várilly \cite{Gracia_Bondia_Varilly:distributions_phasespace_1:1988,Gracia_Bondia_Varilly:distributions_phasespace_2:1988} to extend \emph{magnetic} Weyl calculus from $\Schwartz(\Pspace)$ to $\Schwartz'(\Pspace)$ in Chapter~\ref{magWQ:extension}. Thus, we content ourselves stating facts at this point. 
\begin{defn}[Hörmander class symbol $\Hoerrd{m}$]
	The Hörmander symbols of order $m$ and type $(\rho,\delta)$, $0 \leq \delta \leq \rho \leq 1$, are defined as 
	\begin{align*}
		\Hoerrd{m} := \Bigl \{ f \in \Cont^{\infty}(\Pspace) \; \vert \; \forall a, \alpha \in \N_0^d \exists C_{a \alpha} > 0 : \babs{\partial_x^a \partial_{\xi}^{\alpha} f(x,\xi)} \leq C_{a \alpha} \expval{\xi}^{m - \rho \abs{\alpha} + \delta \abs{a}} \Bigr \} 
		. 
	\end{align*}
	The Fréchet topology is generated by the following family of seminorms: 
	\begin{align}
		\bnorm{f}_{m , a \alpha} := \sup_{(x,\xi) \in \Pspace} \expval{\xi}^{-m + \rho \abs{\alpha} - \delta \abs{a}} \, \babs{\partial_x^a \partial_{\xi}^{\alpha} f(x,\xi)} 
	\end{align}
\end{defn}
%
The following fact is proven in \cite{Iftimie_Mantiou_Purice:magnetic_psido:2006}, for instance: 
\begin{thm}
	Let $f \in \Hoerrd{m}$, $0 \leq \delta < \rho \leq 1$ or $\rho = 0 = \delta$, regarded as a tempered distribution on $\Pspace$. Then $\Op(f)$ is a continuous map from $\Schwartz(\R^d)$ to $\Schwartz'(\R^d)$. 
\end{thm}
Hörmander class symbols have nice composition properties. 
\begin{thm}
	For $0 \leq \delta < \rho \leq 1$ or $\rho = 0 = \delta$, we have $\Hoerrd{m_1} \Weyl \Hoerrd{m_2} \subseteq \Hoerrd{m_1 + m_2}$. 
\end{thm}
The proofs rely on oscillatory integral techniques (see Chapter~\ref{magWQ:extension:important_subclasses} and Appendix~\ref{appendix:oscillatory_integrals}) and we postpone them until we treat the magnetic case. 



\section{Magnetic Weyl calculus} 
\label{magWQ:magnetic_weyl_calculus}
%
Now we would like to quantize a particle in $\R^d$ subjected to a magnetic field $B$. The correct building block operators are position and \emph{kinetic} momentum, 
\begin{align*}
	\Qe &= \hat{x} \\
	\Pe^A &= \Pe - A(\Qe) = - i \nabla_x - A(\hat{x}) 
	. 
\end{align*}
Different components of kinetic momentum no longer commute, 
\begin{align*}
	i [ \Pe^A_l , \Pe^A_j ] = - B_{lj}(\Qe) 
	, 
\end{align*}
and we expect this to complicate things considerably. We would like a quantization procedure that maps momenta onto the kinetic momentum operator. The quantization formula, 
\begin{align*}
	\Op(f) = \frac{1}{(2\pi)^d} \int_{\Pspace} \dd X \, (\Fs f)(X) \, \WeylSys(X) 
	, 
\end{align*}
has only two slots where the magnetic field could enter: (i) We could use minimal substitution and quantize $f \bigl ( x , \xi - A(x) \bigr ) =: f \circ \vartheta^A(x,\xi)$ where $\dd A = B$. The Weyl system is not touched. (ii) We could alter the Weyl system, but leave $f$ unchanged. 

The standard recipe is (i) which comes from the observation that \emph{classically} minimal substitution gives an accurate description of the physics and is \emph{equivalent} to putting the magnetic field into the symplectic form \cite[Chapters~6.7 and 7.6]{Marsden_Ratiu:intro_mechanics_symmetry:1999}. Before we demonstrate the shortcomings of this attempt, we will state two common assumptions on the magnetic fields $B$ and associated vector potentials $A$ for brevity. 
\begin{assumption}[Polynomially bounded fields]
	We assume that the components of the magnetic fields $B$ and associated vector potentials $A$ have components in $\Cont^{\infty}_{\mathrm{pol}}(\R^d)$. 
\end{assumption}
In many instances, we need to work with a more restricted class of magnetic fields: 
\begin{assumption}[Bounded magnetic fields]
	We assume that the components of the magnetic fields $B$ are $\BCont^{\infty}(\R^d)$ functions, \ie smooth, bounded functions with bounded derivatives to any order. Associated vector potentials $A$, \ie $\dd A = B$, are always assumed to have components in $\Cont^{\infty}_{\mathrm{pol}}(\R^d)$. 
\end{assumption}
\emph{Whenever we say bounded or polynomially bounded magnetic field, we actually invoke one of these two assumptions. }
\begin{remark}
	If a magnetic field $B$ is bounded or polynomially bounded, it is always possible to choose a polynomially bounded vector potential, \eg we may use the transversal gauge (equation~\eqref{appendix:eqn:transversal_gauge}). 
\end{remark}

\subsection{Standard ansatz: minimal coupling} 
\label{magWQ:magnetic_weyl_calculus:naive_ansatz}
The standard recipe used throughout most of the literature is to combine minimal coupling with usual Weyl quantization: if $A$ is a vector potential associated to a polynomially bounded $B$, then we define 
\begin{align*}
	\Op_A(h) := \Op(h \circ \vartheta^A) 
	= \frac{1}{(2\pi)^d} \int_{\Pspace} \dd X \, \bigl ( \Fs (h \circ \vartheta^A) \bigr )(X) \, \WeylSys(X) 
\end{align*}
as magnetic Weyl quantization of $h$. For suitable functions, \eg $h \in \Schwartz(\Pspace)$ and $u \in \Schwartz(\R^d)$, we easily convince ourselves that 
\begin{align*}
	\bigl ( \Op_A(h) u \bigr )(x) &= \frac{1}{(2\pi)^d} \int_{\R^d} \dd y \int_{{\R^d}^*} \negmedspace \negmedspace \dd \eta \, e^{- i (y - x) \cdot \eta} \, h \bigl ( \tfrac{1}{2}(x + y) , \eta - A \bigl ( \tfrac{1}{2} (x + y) \bigr ) \bigr ) \, u(y) 
	\\
	&= \frac{1}{(2\pi)^d} \int_{\R^d} \dd y \int_{{\R^d}^*} \negmedspace \negmedspace \dd \eta \, e^{- i (y - x) \cdot (\eta + A(\frac{1}{2}(x + y)))} \, h \bigl ( \tfrac{1}{2}(x + y) , \eta \bigr ) \, u(y) 
	. 
\end{align*}
By hand, one can check that indeed, we get 
\begin{align*}
	\Op_A(\xi) = \Pe^A 
\end{align*}
and 
\begin{align*}
	\Op_A \bigl ( \tfrac{1}{2} \xi^2 + V \bigr ) = \tfrac{1}{2} {\Pe^A}^2 + V(Q) 
	. 
\end{align*}
The operators on the right-hand side are gauge-covariant, \ie if $A' = A + \dd \chi$ is an equivalent gauge, then $\Op_{A + \dd \chi} \bigl ( \tfrac{1}{2} \xi^2 + V \bigr ) = e^{+ i \chi(Q)} \, \Op_A \bigl ( \tfrac{1}{2} \xi^2 + V \bigr ) \, e^{-i \chi(Q)}$, for instance. The unitary operator $U_{\chi} = e^{+ i \chi(Q)}$ relates wave functions and operators in the gauges $A$ and $A' = A + \dd \chi$. The quadratic hamiltonian $\tfrac{1}{2} \xi^2 + V$ is the most frequently studied and thus the lack of gauge-covariance of the $\Op_A$ for generic suitable functions (\eg those in $\Cont^{\infty}_{\mathrm{pol \, u}}(\Pspace)$) was not discovered until 1999 by Müller \cite{Mueller:product_rule_gauge_invariant_Weyl_symbols:1999}: if we choose another, equivalent gauge $A' = A + \dd \chi$, then $\Op_A(h)$ and $\Op_{A + \dd \chi}(h)$ generally fail to be unitarily equivalent, \ie the \emph{physical and mathematical properties would depend on the choice of gauge}. To see this, let us calculate the difference explicitly: 
\begin{align*}
	\Bigl ( \bigl ( &\Op_{A + \dd \chi}(h) - e^{+ i \chi(Q)} \, \Op_A(h) \, e^{-i \chi(Q)} \bigr ) u \Bigr )(x) 
	= \\
	&= \frac{1}{(2\pi)^d} \int_{\R^d} \dd y \biggl ( \int_{{\R^d}^*} \dd \eta \, \Bigl (
	e^{- i (y - x) \cdot (\eta + A(\frac{1}{2}(x + y)) + \nabla_x \chi(\frac{1}{2}(x + y)))} \, h \bigl ( \tfrac{1}{2}(x + y) , \eta \bigr ) 
	+ \Bigr . \biggr . \\
	&\qquad \qquad \qquad \qquad \left . \Bigl . - e^{+ i \chi(x)} e^{- i (y - x) \cdot (\eta + A(\frac{1}{2}(x + y)))} \, h \bigl ( \tfrac{1}{2}(x + y) , \eta \bigr ) e^{- i \chi(y)} \Bigr ) \right ) \, u(y) 
	\\
	&= \frac{1}{(2\pi)^d} \int_{\R^d} \dd y \biggl ( \int_{{\R^d}^*} \dd \eta \, e^{- i (y - x) \cdot (\eta + A(\frac{1}{2}(x + y)))} \, h \bigl ( \tfrac{1}{2}(x + y) , \eta \bigr ) \, 
	\cdot \biggr . \\
	&\qquad \qquad \qquad \qquad \qquad \cdot \biggl .
	\Bigl (
	e^{- i (y - x) \cdot \nabla_x \chi(\frac{1}{2}(x + y))} - e^{- i (\chi(y) - \chi(x))} \Bigr ) \, \biggr ) \, u(y) 
\end{align*}
In case $\Op_A(h)$ transforms covariantly, then the above expression equals $0$. This is the case if and only if 
\begin{align*}
	\bigl ( \Fourier_2 (h \circ \vartheta^A) \bigr ) \bigl ( \tfrac{1}{2}(x + y) , y - x \bigr ) \, \Bigl ( e^{- i (y - x) \cdot \nabla_x \chi(\frac{1}{2}(x + y))} - e^{- i (\chi(y) - \chi(x))} \Bigr ) = 0
\end{align*}
vanishes in the distributional sense. We can check by hand that this is the case if $h$ is a polynomial of degree $\leq 2$ in $\xi$. In all other cases, this expression has no reason to vanish as one can check by plugging in $h = \xi_l \xi_j \xi_k$, for instance. 


\subsection{Covariant quantization formula} 
\label{magWQ:magnetic_weyl_calculus:covariant_quantization}
The second approach, a modification of the Weyl system, does not suffer from this defect. If the fields are polynomially bounded, then we define the 
\begin{defn}[Magnetic Weyl system\index{Weyl system!magnetic}]
	Assume $B$ is polynomially bounded. Then the Weyl system associated to $\Pspace \cong T^* \R^d$ is the strongly continuous map 
	\begin{align}
		\WeylSys^A : \Pspace \longrightarrow \mathcal{U} \bigl ( L^2(\R^d) \bigr ) 
	\end{align}
	which for each $X \in \Pspace$ is defined as 
	\begin{align*}
		\WeylSys^A(X) := e^{- i \sigma(X,(\Qe,\Pe^A))} 
		. 
	\end{align*}
\end{defn}
The map $\WeylSys^A : X \mapsto \WeylSys^A(X)$ is a strongly continuous projective representation of the group $\Pspace \cong \R^d \times {\R^d}^*$ and acts on wave functions as magnetic translations times a phase. 
\begin{lem}\label{magWQ:magnetic_weyl_calculus:lem:action_Weyl_system}
	For $u \in \Schwartz(\R^d)$, we have 
	\begin{align*}
		\bigl ( \WeylSys^A(Y) u \bigr )(x) = e^{- i (x + \frac{y}{2}) \cdot \eta} e^{- i \Gamma^A([x,x+y])} u(x+y) 
		=: e^{- i (x + \frac{y}{2}) \cdot \eta} \, \lambda^A(x;y) \, u(x+y) 
		\label{magWQ:magnetic_weyl_calculus:eqn:action_Weyl_system}
	\end{align*}
	where $\lambda^A(x;y) := e^{- i \int_{[x,x+y]} A}$ is the exponential of the magnetic circulation through the line segment connecting $x$ and $x+y$. In particular, $\WeylSys^A(Y)$ is a covariant operator, \ie 
	\begin{align*}
		\WeylSys^{A + \dd \chi}(Y) = e^{+ i \chi(\Qe)} \, \WeylSys^A(Y) \, e^{- i \chi(\Qe)} 
		. 
	\end{align*}
\end{lem}
\begin{proof}
	Equation~\eqref{magWQ:magnetic_weyl_calculus:eqn:action_Weyl_system} follows from a simple Trotter argument. The covariance is a direct consequence as well if one takes 
	\begin{align*}
		\int_{[x,x+y]} (A + \dd \chi) = \int_{[x,x+y]} A + \bigl ( \chi(x+y) - \chi(x) \bigr ) 
	\end{align*}
	into account. 
\end{proof}
\begin{remark}
	As $\Schwartz(\R^d) \subset L^2(\R^d)$ is dense, all of these statements extend immediately to $u \in L^2(\R^d)$ and $\WeylSys^A(X)$ is a unitary operator for each $X \in \Pspace$. We will show later on in Corollary~\ref{magWQ:extension:cor:irreducibility_Weyl_system} that $\{ \WeylSys^A(X) \}_{X \in \Pspace}$ is irreducible. 
\end{remark}
Just as in the non-magnetic case, the composition properties of the Weyl system encapsulate the commutation relations of the building block operators $\Qe$ and $\Pe^A$. Compared to the non-magnetic Weyl system, we get an extra phase factor, the exponential of a magnetic flux through a triangle. 
\begin{lem}\label{magWQ:magnetic_weyl_calculus:lem:composition_Weyl_system}
	For all $X,Y \in \Pspace$ the following holds: 
	\begin{align}
		\WeylSys^A(X) \WeylSys^A(Y) = e^{\frac{i}{2} \sigma(X,Y)} \omega^B(\Qe;x,y) \WeylSys^A(X+Y) 
		\label{magWQ:magnetic_weyl_calculus:eqn:composition_Weyl_system}
	\end{align}
	where 
	\begin{align}
		\omega^B(q;x,y) := e^{- i \int_{\sexpval{q,q+x,q+x+y}} B} =: e^{- i \Gamma^B(\sexpval{q,q+x,q+x+y})}
	\end{align}
	is the exponential of the magnetic flux through the triangle with corners $q$, $q + x$ and $q + x + y$. 
\end{lem}
\begin{proof}
	Let $u \in \Schwartz(\R^d)$. Then we have 
	\begin{align*}
		\bigl ( \WeylSys^A(X) \WeylSys^A(Y) u \bigr ) (q) &= e^{- i (q + \frac{x}{2}) \cdot \xi} \lambda^A(q;x) \bigl ( \WeylSys^A(Y) u \bigr )(q+x) \\
		&= e^{- i (q + \frac{x}{2}) \cdot \xi} \lambda^A(q;x) \, e^{- i (q + x + \frac{y}{2}) \cdot \eta} \lambda^A(q+x;y) \, u(q + x + y) \\ 
		&= e^{- i (q + \frac{x}{2}) \cdot \xi} e^{- i (q + x + \frac{y}{2}) \cdot \eta}  e^{i (q + \frac{1}{2}(x+y)) \cdot (\xi + \eta)} 
		\cdot \\
		&\qquad \qquad \cdot \lambda^A(q;x) \lambda^A(q+x;y) \lambda^A(q;x+y)^{-1} \cdot \\ 
		&\qquad \qquad \cdot e^{- i (q + \frac{1}{2}(x+y)) \cdot (\xi + \eta)} \lambda^A(q;x+y) u(q + x + y) \\
		&= e^{\frac{i}{2} \sigma(X,Y)} \omega^B(\Qe;x,y) \, \bigl ( \WeylSys^A(X+Y) u \bigr )(q) 
		. 
	\end{align*}
	We have used that $\lambda^A(q;x)$ is the exponential of $\Gamma^A([q,q+x]) = \int_{[q,q+x]} A$ and applied Stoke's Theorem to rewrite the sum of the circulations along the edges of the triangle as magnetic flux through the enclosed area. 
\end{proof}
Since we have proven in Lemma~\ref{magWQ:magnetic_weyl_calculus:lem:action_Weyl_system} that the magnetic Weyl system transforms covariantly under a change of gauge, magnetic Weyl quantization inherits this property: 
\begin{defn}[Magnetic Weyl quantization\index{Weyl quantization!magnetic}]
	For all functions $h \in \Schwartz(\Pspace)$ and polynomially bounded magnetic fields, we define 
	\begin{align}
		\Op^A(h) := \frac{1}{(2\pi)^d} \int_{\Pspace} \dd X \, (\Fs h)(X) \, \WeylSys^A(X) 
		\label{magWQ:magnetic_weyl_calculus:eqn:magnetic_Weyl_quantization}
	\end{align}
	%
	in the weak sense. 
\end{defn}
\begin{lem}
	Assume the magnetic field is polynomially bounded. Then the magnetic Weyl quantization of $h \in \Schwartz(\Pspace)$ defines a bounded operator on $L^2(\R^d)$ whose operator norm is bounded by 
	\begin{align*}
		\bnorm{\Op^A(h)}_{\mathcal{B}(L^2(\R^d))} \leq (2\pi)^{-d} \bnorm{\Fs h}_{L^1(\Pspace)} < \infty 
	\end{align*}
	and acts on $u \in \Schwartz(\R^d)$ as 
	\begin{align}
		\bigl ( \Op^A(h) u \bigr )(x) &= \frac{1}{(2\pi)^d} \int_{\Pspace} \dd Y \, e^{- i (y-x) \cdot \eta} \, e^{-i \Gamma^A([x,y])} \, h \bigl ( \tfrac{1}{2} (x+y) , \eta \bigr ) \, u(y) 
		\label{magWQ:magnetic_weyl_calculus:eqn:Op_A_applied} \\
		&= \frac{1}{(2\pi)^{\nicefrac{d}{2}}} \int_{\R^d} \dd y \, \lambda^A(x;y-x) \, (\Fourier_2 h) \bigl ( \tfrac{1}{2} (x+y) , y-x \bigr ) \, u(y) 
		\notag \\
		&=: \frac{1}{(2\pi)^{\nicefrac{d}{2}}} \int_{\R^d} \dd y \, K^A_h(x,y) \, u(y) 
		\notag 
	\end{align}
	where $\Fourier_2 h$ is the Fourier transform in the second argument of $h$. The magnetic Weyl quantization is $\Op^A$ covariant: for an equivalent gauge $A' = A + \dd \chi$, we have 
	\begin{align}
		\Op^{A + \dd \chi}(h) = e^{+ i \chi(\Qe)} \, \Op^A(h) \, e^{- i \chi(\Qe)} 
		. 
	\end{align}
\end{lem}
\begin{proof}
	The $L^2$ operator norm bound follows just as in the non-magnetic case (see proof of Lemma~\ref{magWQ:standard_weyl_calculus:weyl_quantization:lem:action_Op}). 
	
	We use Lemma \ref{magWQ:magnetic_weyl_calculus:lem:action_Weyl_system} and write out the symplectic Fourier transform: 
	\begin{align*}
		\bigl ( \OpA&(h) u \bigr )(x) = \frac{1}{(2\pi)^d} \int_{\Pspace} \dd Y \, \bigl ( \Fs h \bigr )(Y) \bigl ( \WeylSys^A(Y) u \bigr )(x) \\
		&= \frac{1}{(2\pi)^{2d}} \int_{\Pspace} \dd Y \int_{\Pspace} \dd Z \, e^{i \sigma(Y,Z)} \, h(Z) \, e^{-i (x + \frac{y}{2}) \cdot \eta} \lambda^A(x;y) \, u(x+y) \\
		%
		%
		&= \frac{1}{(2\pi)^{2d}} \int_{\R^d} \dd y \int_{{\R^d}^*} \negmedspace \dd \zeta \int_{\R^d} \dd z \int_{{\R^d}^*} \negmedspace \dd \eta \, e^{i (z - \frac{1}{2} (x+y)) \cdot \eta} \, e^{-i (y-x) \cdot \zeta} 
		\cdot \\
		&\qquad \qquad \qquad \qquad \qquad \qquad \qquad \qquad \; \; \cdot 
		h(z,\zeta) \, \lambda^A(x;y-x) \, u(y) \\
		&= \frac{1}{(2\pi)^d} \int_{\R^d} \dd y \, \int _{{\R^d}^*} \negmedspace \dd \zeta \, e^{-i (y-x) \cdot \zeta} \, \lambda^A(x;y-x) \, h \bigl ( \tfrac{1}{2} (x+y) ,\zeta \bigr ) \, u(y) 
	\end{align*}
	%
	The covariance of $\Op^A(h)$ follows immediately from the covariance of the Weyl system proven in Lemma~\ref{magWQ:magnetic_weyl_calculus:lem:action_Weyl_system}. 
\end{proof}
\begin{remark}\label{magWQ:magnetic_weyl_calculus:rem:unitarity_kernel_map}
	Also here, the kernel map $K^A : h \mapsto K^A_h$\index{kernel map!magnetic} which associates to each $h \in \Schwartz(\Pspace)$ the operator kernel of $\Op^A(h) = \Int(K^A_h)$ is an isomorphism between $\Schwartz(\Pspace)$ and $\Schwartz(\R^d \times \R^d)$. 
\end{remark}
%
%
\begin{remark}
	Just as usual Weyl quantization, the magnetic quantization rule defined via equation~\eqref{magWQ:magnetic_weyl_calculus:eqn:magnetic_Weyl_quantization} orders operators symmetrically and real-valued functions $f \in \Schwartz(\Pspace)$ are mapped onto bounded, selfadjoint operators. If the function takes values in the complex numbers, the operator adjoint of the magnetic Weyl quantization of $f$ is equal to the quantization of the complex conjugated function, 
	\begin{align*}
		\Op^A(f)^* = \Op^A(f^*) 
		. 
	\end{align*}
	We could opt for a different operator ordering\index{operator ordering} by modifying equation~\eqref{magWQ:magnetic_weyl_calculus:eqn:Op_A_applied}: if $\tau \in [0,1]$, then we could equally well set 
	\begin{align*}
		\bigl ( \Op^A_{\tau}(f) u \bigr )(x) &:= \frac{1}{(2\pi)^d} \int_{\Pspace} \dd Y \, e^{- i (y-x) \cdot \eta} \, e^{-i \Gamma^A([x,y])} \, h \bigl ( (1-\tau) x + \tau y , \eta \bigr ) \, u(y) 
	\end{align*}
	as magnetic Weyl quantization of $f$. Then taking adjoints on the level of operators no longer reduces to complex conjugation, it has to be replaced by $f^{\magW}$ defined in equation~\eqref{algebraicPOV:generalized_weyl_calculus:eqn:repom_involution}. 
\end{remark}
%


\subsection{The Magnetic Wigner transform} 
\label{magWQ:magnetic_weyl_calculus:magnetic_wigner_transform}
The magnetic Wigner transform is the `inverse' of $\Op^A$ and can be used to connect quantum states (projections and density operators) with signed probability measures on phase space~$\Pspace$. 
\begin{defn}[Magnetic Wigner transform\index{Wigner transform!magnetic}]
	Let $u,v \in \Schwartz(\R^d)$. Then the magnetic Wigner transform $\WignerTrafo^A(u,v)$ for polynomially bounded $B$ is defined as 
	\begin{align*}
		\WignerTrafo^A(u,v)(X) := (2\pi)^{-\nicefrac{d}{2}} \, \bigl ( \Fs \bscpro{v}{\WeylSys^A(\cdot) u} \bigr )(-X) . 
	\end{align*}
\end{defn}
\begin{lem}\label{magWQ:magnetic_weyl_calculus:lem:Wigner_transform}
	For a polynomially bounded magnetic field $B$ and associated vector potential $A$, the magnetic Wigner transform $\WignerTrafo^A(u,v)$ of $u,v \in \Schwartz(\R^d)$ calculates to be 
	\begin{align}
		\WignerTrafo^A(u,v)(X) &= \frac{1}{(2\pi)^{\nicefrac{d}{2}}} \int_{\R^d} \dd y \, e^{- i y \cdot \xi} e^{-i \Gamma^A([x - \frac{y}{2} ,x + \frac{y}{2} ])} \, {v}^{\ast} \bigl ( x - \tfrac{y}{2} \bigr ) \, u \bigl ( x + \tfrac{y}{2} \bigr ) 
		\label{magWQ:magnetic_weyl_calculus:eqn:Wigner_transform}
	\end{align}
	and maps $\Schwartz(\R^d) \otimes \Schwartz(\R^d) \cong \Schwartz(\R^d \times \R^d)$ bijectively onto $\Schwartz(\Pspace)$. 
\end{lem}
\begin{proof}
	Formally, the result follows from direct calculation. 
	The second claim, 
	\begin{align*}
		\Schwartz(\R^d \times \R^d) \times \Schwartz(\R^d \times \R^d) \ni (u,v) \mapsto \WignerTrafo^A(u,v) \in \Schwartz(\Pspace)
	\end{align*}
	follows from $e^{- i \Gamma^A([x - \frac{y}{2} , x + \frac{y}{2}])} \, v^{\ast} \bigl ( x - \tfrac{y}{2} \bigr ) \, u \bigl ( x + \tfrac{y}{2} \bigr ) \in \Schwartz(\R^d \times \R^d)$ and the fact that the partial Fourier transformation is an isomorphism on $\Schwartz$. 
\end{proof}
\begin{remark}
	The Wigner transform can be easily extended to a map from $L^2(\R^d \times \R^d)$ to $L^2(\Pspace) \cap \Cont_{\infty}(\Pspace)$ where $\Cont_{\infty}(\Pspace)$ is the space of continuous functions on phase space which decay at $\infty$. For more details, see \cite[Proposition~1.92]{Folland:harmonic_analysis_hase_space:1989}, for example. 
\end{remark}
\begin{lem}\label{magWQ:magnetic_weyl_calculus:magnetic_wigner_transform:lem:qm_exp_val_phase_space_average}
	For polynomially bounded $B$, $u,v \in \Schwartz(\R^d)$ and $h \in \Schwartz(\Pspace)$ the quantum expectation value of $\Op^A(h)$ with respect to $u$ and $v$ can be expressed as the phase space average of $f$ with respect to the Wigner transform of $\sopro{u}{v}$, 
	\begin{align*}
		\bscpro{v}{\Op^A(h) u} = \frac{1}{(2\pi)^{\nicefrac{d}{2}}} \int_{\Pspace} \dd X \, h(X) \, \WignerTrafo^A(u,v)(X) 
		. 
	\end{align*}
\end{lem}
\begin{proof}
	The claim follows from direct computation. 
\end{proof}
%
The Wigner transform can also be used to `dequantize' operators: if the operator kernel $K_T$ of an operator $T$ is of class $\Schwartz(\R^d \times \R^d)$, then it is the quantization of an $h \in \Schwartz(\Pspace)$. That is, $\WignerTrafo^A$ is the inverse of $h \mapsto K^A_h$: 
\begin{align*}
	\bigl ( \WignerTrafo^A K^A_h \bigr )(x,\xi) &= \frac{1}{(2\pi)^{\nicefrac{d}{2}}} \int_{\R^d} \dd y \, e^{-i y \cdot \xi} \, e^{-i \Gamma^A([x - \frac{y}{2} , x + \frac{y}{2} ])} \,  K^A_h \bigl ( x + \tfrac{y}{2} , x - \tfrac{y}{2} \bigr ) 
	\\ 
	&= \frac{1}{(2\pi)^d} \int_{\R^d} \dd y \, e^{-i y \cdot \xi} \, e^{-i \Gamma^A([x - \frac{y}{2} , x + \frac{y}{2} ])} \, 
	\int_{{\R^d}^*} \dd \eta \, e^{- i ((x - \frac{y}{2}) - (x + \frac{y}{2})) \cdot \eta} 
	\cdot \\ 
	&\qquad \qquad \qquad \qquad \cdot 
	e^{- i \Gamma^A([x + \frac{y}{2} , x - \frac{y}{2}])} \, h \bigl ( \tfrac{1}{2} \bigl ( x + \tfrac{y}{2} \bigr ) + \tfrac{1}{2} \bigl ( x - \tfrac{y}{2} \bigr ) , \eta \bigr ) 
	\\ 
	&= \frac{1}{(2\pi)^d} \int_{\R^d} \dd y \int_{{\R^d}^*} \dd \eta \, e^{-i y \cdot \xi} \, e^{+ i y \cdot \eta} \, h ( x , \eta ) 
	= h(x,\xi) 
\end{align*}
Hence, we have just proven 
\begin{lem}\label{magWQ:magnetic_weyl_calculus:lem:inverseWeylQ}
	Assume $T = \Int(K_T) \in \BopL$ is an operator whose operator kernel $K_T$ is a Schwartz function. Then the \emph{inverse magnetic quantization} is in $\Schwartz(\Pspace)$ given by 
	\begin{align}
		{\Op^A}^{-1}(T)(x,\xi) &= \WignerTrafo^A K_{T}(x,\xi) 
		\notag \\
		%
		&= \frac{1}{(2\pi)^{\nicefrac{d}{2}}} \int_{\R^d} \dd y \, e^{-i y \cdot \xi} \, e^{-i \Gamma^A([x - \frac{y}{2} , x + \frac{y}{2} ])} \,  K_{T} \bigl ( x + \tfrac{y}{2} , x - \tfrac{y}{2} \bigr ) 
		. 
	\end{align}
\end{lem}
%


\subsection{The magnetic Weyl product} 
\label{magWQ:magnetic_weyl_calculus:magnetic_product}
The derivation of the product formula is slightly more involved than in the non-mag\-netic case and we need to use the magnetic Wigner transform. By covariance, the magnetic Weyl product $\magW$\index{Weyl product!magnetic} only depends on the magnetic field rather than the vector potential, 
\begin{align*}
	\Op^A(f) \, \Op^A(g) = \Op^A (f \magW g) 
	. 
\end{align*}
\begin{thm}[\cite{Mueller:product_rule_gauge_invariant_Weyl_symbols:1999, Iftimie_Mantiou_Purice:magnetic_psido:2006}]\label{magWQ:magnetic_weyl_calculus:thm:equivalence_product}
	Assume the magnetic field $B$ is polynomially bounded. Then for two symbols $f , g \in \Schwartz(\Pspace)$, the magnetic composition law is given by 
	\begin{align}
		(f \magW g)(X) &= \frac{1}{(2 \pi)^{2d}} \int_{\Pspace} \dd Y \, \int_{\Pspace} \dd Z \, e^{+i \sigma(X,Y+Z)} \,  e^{\frac{i}{2} \sigma(Y,Z)} \, 
		\cdot \notag \\
		&\qquad \qquad \qquad \qquad \quad \cdot 
		\omega^B \bigl ( x-\tfrac{1}{2}(y+z),x + \tfrac{1}{2}(y-z) , x+\tfrac{1}{2}(y+z) \bigr ) 
		\cdot \notag \\
		&\qquad \qquad \qquad \qquad \quad \cdot 
		\bigl ( \Fs f \bigr )(Y) \, \bigl ( \Fs g \bigr )(Z) 
		\label{magWQ:magnetic_weyl_calculus:eqn:Fourier_form_magnetic_composition} \\
		&= \frac{1}{\pi^{2d}} \, \int_{\Pspace} \dd \tilde{Y} \, \int_{\Pspace} \dd \tilde{Z} \, e^{- i 2 \sigma(\tilde{Y} - X , \tilde{Z} - X)} 
		\cdot \notag \\
		&\qquad \qquad \qquad \qquad \quad \cdot 
		\omega^B \bigl ( x - \tilde{y} + \tilde{z} , -x + \tilde{y} + \tilde{z} , x + \tilde{y} + \tilde{z} \bigr ) \,  f(\tilde{Y}) \, g(\tilde{Z}) 
		\notag 
	\end{align}
	and the product $f \magW g \in \Schwartz(\R^d)$ is also a Schwartz function. 
\end{thm}
Before we can prove this statement, we need an auxiliary result: take two operators $T$ and $S$ whose operator kernels $K_T$ and $K_S$ are in $\Schwartz(\R^d \times \R^d)$. Then the operator kernel of $T S$ is given by 
\begin{align*}
	(K_T \diamond K_S)(x,y) := \frac{1}{(2\pi)^{\nicefrac{d}{2}}} \int_{\R^d} \dd z \, K_T(x,z) \, K_S(z,y) 
	. 
\end{align*}
\begin{lem}\label{magWQ:magnetic_weyl_calculus:lem:composition_of_operator_kernels}
	For any $K_T , K_S \in \Schwartz(\R^d \times \R^d)$, the product $K_T \diamond K_S$ is also in $\Schwartz(\R^d \times \R^d)$, \ie $\diamond : \Schwartz(\R^d \times \R^d) \times \Schwartz(\R^d \times \R^d) \longrightarrow \Schwartz(\R^d \times \R^d)$. 
\end{lem}
The proof can be found in Appendix~\ref{appendix:oscillatory_integrals}. 
\begin{proof}[Theorem~\ref{magWQ:magnetic_weyl_calculus:thm:equivalence_product}]
	The Weyl product is implicitly defined through 
	\begin{align*}
		\Op^A(f) \, \Op^A(g) &=: \Op^A(f \magW g) 
		. 
	\end{align*}
	The integral kernels of $\Op^A(f)$ and $\Op^A(g)$ are of Schwartz class, $K_f^A , K_g^A \in \Schwartz(\R^d \times \R^d)$, and hence, by Lemma~\ref{magWQ:magnetic_weyl_calculus:lem:composition_of_operator_kernels}, the integral kernel of $\Op^A(f) \, \Op^A(g)$ is also an element of $\Schwartz(\R^d \times \R^d)$. If we combine this with Lemmas~\ref{magWQ:magnetic_weyl_calculus:lem:Wigner_transform} and \ref{magWQ:magnetic_weyl_calculus:lem:inverseWeylQ}, we conclude 
	\begin{align*}
		f \magW g = \WignerTrafo^A K_{\Op^A(f) \, \Op^A(g)} \in \Schwartz(\Pspace)
	\end{align*}
	where $K_{\Op^A(f) \, \Op^A(g)} \in \Schwartz(\R^d \times \R^d)$ is the kernel of $\Op^A(f) \, \Op^A(g)$. 
	\medskip
	
	\noindent
	\textbf{Step 1: Rewrite in terms of Weyl system. }
	Plugging in the definition of $\Op^A$, we get 
	\begin{align*}
		\Op^A(f) \, \Op^A(g) 
		&= \frac{1}{(2\pi)^{2d}} \int_{\Pspace} \dd Y \, \int_{\Pspace} \dd Z \, \bigl ( \Fs f \bigr )(Y) \, \bigl ( \Fs g \bigr )(Z) \, \WeylSys^A(Y) \WeylSys^A(Z) \\
		&= \frac{1}{(2\pi)^{2d}} \int_{\Pspace} \dd Y \, \int_{\Pspace} \dd Z \, \bigl ( \Fs f \bigr )(Y) \, \bigl ( \Fs g \bigr )(Z) \, e^{\frac{i}{2} \sigma(Y,Z)} 
		\cdot \\
		&\qquad \qquad \qquad \qquad \qquad \; \cdot
		 \omega^B(\Qu,\Qu+y,\Qu+y+z) \WeylSys^A(Y+Z) \\
		&= \frac{1}{(2\pi)^{2d}} \int_{\Pspace} \dd Z \biggl ( \int_{\Pspace} \dd Y \, \bigl ( \Fs f \bigr )(Y) \, \bigl ( \Fs g \bigr )(Z-Y) \, e^{\frac{i}{2} \sigma(Y,Z)} 
		\biggr . \cdot \\
		\biggl . &\qquad \qquad \qquad \qquad \qquad \; \cdot 
		 \omega^B(\Qu,\Qu+y,\Qu+z) \biggr ) \WeylSys^A(Z) 
		. 
	\end{align*}
	In order to find the kernel of this operator, we need to find the kernel for $\hat{L}(y,Z) := \omega^B(\Qu,\Qu+y,\Qu+z) \WeylSys^A(Z)$ which parametrically depends on $y$ and $Z = (z,\zeta)$. 
	\medskip

	\noindent
	\textbf{Step 2: Find the operator kernel for $\hat{L}(y,Z)$. } 
	Let $\varphi \in L^2(\R^d)$. Then we have 
	\begin{align*}
		\bigl ( \hat{L}(y,Z) u \bigr )(q) 
		&= \omega^B(q,q+y,q+z) \, e^{-i (q + \frac{z}{2}) \cdot \eta} \, e^{-i \Gamma^A([q,q+z])} \, u(q+z) \\
		&= \int_{\R^d} \dd q' \, e^{-i (q' - \frac{z}{2}) \cdot \eta} \, e^{-i \Gamma^A([q'-z,q'])} 
		\cdot \\
		&\qquad \qquad \; \cdot 
		\omega^B(q'-z,q'+y-z,q') \, \delta \bigl ( q'-(q+z) \bigr ) \, u(q') \\
		%
		&=: (2\pi)^{-\nicefrac{d}{2}} \int_{\R^d} \dd q' \, K_{L}(y,Z;q,q') \, u(q') 
		, 
	\end{align*}
	and we need to find $\WignerTrafo^A K_{L}(y,Z; \cdot , \cdot) (X)$, 
	\begin{align*}
		\WignerTrafo^A K_{L}(y,Z; \cdot , \cdot)(X) 
		&= \int_{\R^d} \dd q \, e^{-i q \cdot \xi} \, e^{-i \Gamma^A([x - \frac{q}{2} , x + \frac{q}{2}])} \, K_{L} \bigl ( y,Z;x + \tfrac{q}{2},x - \tfrac{q}{2} \bigr ) \\
		&= e^{i \sigma(X,Z)} \, \omega^B \bigl ( x - \tfrac{z}{2},x - \tfrac{z}{2} + y , x + \tfrac{z}{2} \bigr ) =: L(y,Z;X) 
		. 
	\end{align*}
	\medskip

	\noindent
	\textbf{Step 3: Magnetic composition law. }%
	Now we plug $L(y,Z;X)$ back into the operator equation and obtain 
	\begin{align}
		(f \magW g)(X) 
		&= \frac{1}{(2\pi)^{2d}} \int_{\Pspace} \dd Z \, \int_{\Pspace} \dd Y \, \bigl ( \Fs f \bigr )(Y) \, \bigl ( \Fs g \bigr )(Z-Y) \;  e^{\frac{i}{2} \sigma(Y,Z)} \, L(y,Z;X) 
		\notag \\ 
		&= \frac{1}{(2\pi)^{2d}} \int_{\Pspace} \dd Y \, \int_{\Pspace} \dd Z \, e^{i \sigma(X,Y+Z)} \,  e^{\frac{i}{2} \sigma(Y,Z)} 
		\cdot \label{magWQ:magnetic_weyl_calculus:eqn:equivalence:eqn1} \\
		&\quad \cdot 
		\omega^B \bigl ( x - \tfrac{1}{2} (y+z),x + \tfrac{1}{2} (y-z) , x + \tfrac{1}{2} (y+z) \bigr ) \bigl ( \Fs f \bigr )(Y) \, \bigl ( \Fs g \bigr )(Z)
		. 
		\notag 
	\end{align}
	This formula is the starting point for Müller's and our derivation of the asymptotic expansion of the product. However, we can show the equivalence to the product formula obtained by two of the authors in \cite{Mantoiu_Purice:magnetic_Weyl_calculus:2004} by writing out the symplectic Fourier transforms, 
	\begin{align*}
		\mbox{RHS of \eqref{magWQ:magnetic_weyl_calculus:eqn:equivalence:eqn1}}
		&= \frac{1}{(2\pi)^{4d}} \int_{\Pspace} \dd Y \, \int_{\Pspace} \dd \tilde{Y} \, \int_{\Pspace} \dd Z \, \int_{\Pspace} \dd \tilde{Z} \, e^{i \sigma(X - \tilde{Y},Y)} \, e^{i \sigma(X - \tilde{Z},Z)} \, e^{\frac{i}{2} \sigma(Y,Z)} \, \cdot \\ 
		&\qquad \qquad \quad \cdot \omega^B \bigl ( x - \tfrac{1}{2} (y+z),x + \tfrac{1}{2} (y-z) , x + \tfrac{1}{2} (y+z) \bigr ) \, f(\tilde{Y}) \, g(\tilde{Z}) 
		. 
	\end{align*}
	If one writes out the exponential prefactors explicitly, sorts all terms containing $\xi$ and $\eta$ and then integrates over those variables, one obtains 
	\begin{align*}
		\frac{1}{\pi^{2d}} \, \int_{\Pspace} \dd \tilde{Y} \, \int_{\Pspace} \dd \tilde{Z} \, &e^{- i 2 \sigma(X - \tilde{Y} , X - \tilde{Z})} 
		\cdot \\
		&\cdot 
		\omega^B \bigl ( \tilde{y} - \tilde{z} + x , \tilde{y} + \tilde{z} - x , - \tilde{y} + \tilde{z} + x \bigr ) \, f(\tilde{Y}) \, g(\tilde{Z}) 
		. 
	\end{align*}
	This concludes the proof. 
\end{proof}
%



\section{Extension to larger classes of functions} 
\label{magWQ:extension}
Up to now, we only know how to quantize and compose Schwartz functions on phase space. This is certainly not satisfactory, not only does that exclude functions depending only on one variable $x$ or $\xi$, but also the most common hamiltonian function, $h(x,\xi) = \tfrac{1}{2} \xi^2 + V(x)$. A rather straightforward, but implicit approach to extend magnetic Weyl calculus is to proceed as in \cite{Gracia_Bondia_Varilly:distributions_phasespace_1:1988,Gracia_Bondia_Varilly:distributions_phasespace_2:1988} and employ duality techniques. However, if one wants to show that a certain class of functions is contained in the magnetic Moyal algebra (an algebra composed of tempered distributions with `nice' composition properties, see~Definition~\ref{magWQ:extension:defn:mag_Moyal_algebra}), one has to employ hands-on oscillatory integral techniques \cite{Hoermander:Weyl_calculus:1979} (see also Appendix~\ref{appendix:oscillatory_integrals}). Our presentation here follows \cite{Mantoiu_Purice:magnetic_Weyl_calculus:2004}. 
\begin{assumption}
	Throughout this section, we will assume the fields to be polynomially bounded, \ie the components of $B$ and $A$ are $\Cont^{\infty}_{\mathrm{pol}}$ functions. 
\end{assumption}

\subsection{Extension via duality} 
\label{magWQ:extension:extension_via_duality}
The first step is to extend $\Op^A$ for polynomially bounded magnetic fields $B$ from Schwartz functions to tempered distributions. It is helpful to think in terms of integral kernels: if $h \in \Schwartz(\Pspace)$, then the integral kernel $K^A_h$ of $\Op^A(h)$ reads 
\begin{align}
	K^A_h(x,y) &= \frac{1}{(2\pi)^{\nicefrac{d}{2}}} e^{-i \Gamma^A([x,y])} \, \int_{{\R^d}^*} \dd \eta \, e^{-i (y - x) \cdot \eta} \, h \bigl ( \tfrac{1}{2}(x + y) , \eta \bigr ) 
	\notag \\
	&
	= \lambda^A(x;y-x) \, \bigl ( \Fourier_2 h \bigr) \bigl ( \tfrac{1}{2}(x + y) , y - x \bigr ) 
	. 
	\label{magWQ:extension:eqn:operator_kernel}
\end{align}
Hence, we can define $K^A : \Schwartz(\Pspace) \longrightarrow \Schwartz(\R^d \times \R^d)$, $h \mapsto K^A_h$, as the map which associates to any $h \in \Schwartz(\Pspace)$ the operator kernel of $\Op^A(h)$. Since $K^A$ decomposes into a partial Fourier transform, a linear coordinate transform and a multiplication with a $\Cont^{\infty}_{\mathrm{pol}}(\R^d)$ function, it defines a linear topological isomorphism between $\Schwartz(\Pspace)$ and $\Schwartz(\R^d \times \R^d)$ which extends to an isomorphism between $\Schwartz'(\Pspace)$ and $\Schwartz'(\R^d \times \R^d)$. If we define $\Int$ in the obvious way, $\Op^A(h) =: \Int(K^A_h)$, and use that $\Int$ induces two isomorphisms \cite[Section~50, Theorem~51.6]{Treves:topological_vector_spaces:1967}, 
\begin{align*}
	\Int &: \Schwartz(\R^d \times \R^d) \longrightarrow \mathcal{L} \bigl ( \Schwartz'(\R^d) , \Schwartz(\R^d) \bigr ) 
	\\
	\Int &: \Schwartz'(\R^d \times \R^d) \longrightarrow \mathcal{L} \bigl ( \Schwartz(\R^d) , \Schwartz'(\R^d) \bigr ) 
	, 
\end{align*}
we conclude we can write \emph{any} continuous map from $\Schwartz(\R^d)$ to $\Schwartz'(\R^d)$ as the Weyl quantization of a tempered distribution $\Schwartz'(\Pspace)$. If we endow $\mathcal{L} \bigl ( \Schwartz'(\R^d) , \Schwartz(\R^d) \bigr )$ and $\mathcal{L} \bigl ( \Schwartz(\R^d) , \Schwartz'(\R^d) \bigr )$ with the topology of uniform convergence on bounded subsets, we have two linear, continuous injections 
\begin{align*}
	\mathcal{L} \bigl ( \Schwartz'(\R^d) , \Schwartz(\R^d) \bigr ) \hooklongrightarrow \BopL \hooklongrightarrow \mathcal{L} \bigl ( \Schwartz(\R^d) , \Schwartz'(\R^d) \bigr ) 
	. 
\end{align*}
In particular, these embeddings imply that \emph{any} bounded operator $T$ on $L^2(\R^d)$ has a distributional operator kernel. Putting all this together, we have proven 
\begin{prop}\label{magWQ:extension:prop:extension_Op_A}
	If the magnetic field $B$ is polynomially bounded with vector potential $A$, then $\Op^A$ defines topological linear isomorphisms 
	\begin{align*}
		\Op^A &: \Schwartz(\Pspace) \longrightarrow \mathcal{L} \bigl ( \Schwartz'(\R^d) , \Schwartz(\R^d) \bigr ) 
		\\
		\Op^A &: \Schwartz'(\Pspace) \longrightarrow \mathcal{L} \bigl ( \Schwartz(\R^d) , \Schwartz'(\R^d) \bigr ) 
		. 
	\end{align*}
	If $A'$ is an equivalent polynomially bounded vector potential, \ie $\dd A' = B = \dd A$, then there exists $\chi \in \Cont^{\infty}_{\mathrm{pol}}(\R^d)$ such that $A' = A + \dd \chi$ and for all $h \in \Schwartz'(\Pspace)$, the operators $\Op^A(h)$ and $\Op^{A + \dd \chi}(h)$ are unitarily equivalent as maps in $\mathcal{L} \bigl ( \Schwartz(\R^d) , \Schwartz'(\R^d) \bigr )$, 
	\begin{align*}
		e^{+ i \chi(Q)} \, \Op^A(h) \, e^{- i \chi(Q)} = \Op^{A + \dd \chi}(h) 
		. 
	\end{align*}
\end{prop}
Similarly, the Wigner transform also admits an extension to $\Schwartz'$: 
\begin{prop}\label{magWQ:extension:prop:extension_Wigner_trafo}
	For polynomially bounded fields $B$, the magnetic Fourier transform extends from $\WignerTrafo^A : \Schwartz(\R^d \times \R^d) \longrightarrow \Schwartz(\Pspace)$ to an isomorphism between distributions, 
	\begin{align*}
		\WignerTrafo^A : \Schwartz'(\R^d \times \R^d) \longrightarrow \Schwartz'(\Pspace)
		. 
	\end{align*}
\end{prop}
\begin{proof}
	From the explicit formula, equation~\eqref{magWQ:magnetic_weyl_calculus:lem:Wigner_transform}, we again see that $\WignerTrafo^A$ is a combination of Fourier transform, linear change of variables and multiplication by a phase (a $\Cont^{\infty}_{\mathrm{pol}}$ function). Thus, it extends to tempered distributions by duality. 
\end{proof}
An important consequence is the irreducibility of the Weyl system: 
\begin{cor}\label{magWQ:extension:cor:irreducibility_Weyl_system}
	The magnetic Weyl system $\WeylSys^A : \Pspace \longrightarrow \mathcal{U} \bigl ( L^2(\R^d) \bigr )$ for polynomially bounded fields is irreducible, \ie there are no nontrivial subspaces of $L^2(\R^d)$ invariant under $\{ \WeylSys^A(X) \}_{X \in \Pspace}$. 
\end{cor}
\begin{proof}
	Assume there exists a nontrivial invariant subspace $\mathcal{K}$. Let $u \in \mathcal{K} \setminus \{ 0 \}$ and $u_{\perp} \in \mathcal{K}^{\perp} \setminus \{ 0 \}$. Then for all $X \in \Pspace$ 
	\begin{align*}
		\bscpro{u_{\perp}}{\WeylSys^A(X) u} = 0 
	\end{align*}
	holds by assumption. This also implies $\Fs \bigl ( \bscpro{u_{\perp}}{\WeylSys^A(\cdot) u} \bigr )(- X) = \WignerTrafo^A(u_{\perp},u)(X) = 0$ and by Proposition~\ref{magWQ:extension:prop:extension_Wigner_trafo} 
	\begin{align*}
		\bnorm{\bscpro{u_{\perp}}{\WeylSys^A(\cdot) u}}_{L^2(\Pspace)} = \bnorm{\Fs \bigl ( \bscpro{u_{\perp}}{\WeylSys^A(\cdot) u} \bigr )}_{L^2(\Pspace)} = 0 
		. 
	\end{align*}
	On the other hand, we can calculate the norm of $\bscpro{u_{\perp}}{\WeylSys^A(X) u}$ explicitly, 
	\begin{align*}
		\bnorm{\bscpro{u_{\perp}}{\WeylSys^A(\cdot) u}}_{L^2(\Pspace)} = \snorm{u_{\perp}}_{L^2(\R^d)} \, \snorm{u}_{L^2(\R^d)} \neq 0 
		, 
	\end{align*}
	and we have arrived at a contradiction. 
\end{proof}
We can also characterize the space of compact and Hilbert-Schmidt operators with functions on phase space. This result should be compared to Propositions~\ref{algebraicPOV:twisted_crossed_products:prop:basic_properties_rep}, \ref{psiDO_reloaded:relevant_cStar_algebras:thm:description_FXprod_via_S} and \ref{psiDO_reloaded:relevant_cStar_algebras:prop:another_description_FXprod} which basically say that the `missing' elements are those which lack smoothness and perhaps even continuity, but can be approximated by smooth functions. 
\begin{prop}
	\begin{enumerate}[(i)]
		\item $\Op^A$ induces a unitary map from $L^2(\Pspace)$ to $\mathcal{B}_2 \bigl ( L^2(\R^d) \bigr )$, the ideal of Hilbert-Schmidt operators.\index{Hilbert-Schmidt operators} 
		\item The family $\Op^A \bigl ( \Fs L^1(\Pspace) \bigr )$ is dense in the closed ideal of compact operators $\mathcal{K} \bigl ( L^2(\R^d) \bigr )$.\index{compact operators} 
	\end{enumerate}
\end{prop}
\begin{proof}
	\begin{enumerate}[(i)]
		\item As $K^A \Schwartz(\Pspace) = \Schwartz(\R^d \times \R^d)$ is dense in $L^2(\R^d \times \R^d)$, we can approximate any $L^2$ function by a sequence of Schwartz functions. Furthermore, $\Int : L^2(\R^d \times \R^d) \longrightarrow \mathcal{B}_2 \bigl ( L^2(\R^d) \bigr )$ is unitary and thus also the composition $\Op^A = \Int \circ K^A$ is a unitary map between $L^2(\R^d \times \R^d)$ and $\mathcal{B}_2 \bigl ( L^2(\R^d) \bigr )$. 
		\item All operators with kernels in $\Fs L^1(\Pspace) \cap L^2(\Pspace)$ are Hilbert-Schmidt and thus also compact and $\Op^A \bigl ( \Fs L^1(\Pspace) \cap L^2(\Pspace) \bigr )$ are dense in $\mathcal{B}_2 \bigl ( L^2(\R^d) \bigr )$ with respect to the Hilbert-Schmidt as well as the operator norm. Hence, $\Op^A \bigl ( \Fs L^1(\Pspace) \cap L^2(\Pspace) \bigr )$ is also dense in $\mathcal{K} \bigl ( L^2(\R^d) \bigr )$. 
	\end{enumerate}
\end{proof}
Now that we have successfully extended $\Op^A$ and $\WignerTrafo^A$ to tempered distributions, we turn attention to the product $\magW$. To be able to invoke duality, we crucially need the next Lemma: 
\begin{lem}\label{magWQ:extension:lem:magnetic_product_scalar_product}
	For any $f , g \in \Schwartz(\Pspace)$, we have 
	\begin{align*}
		\int_{\Pspace} \dd X \, (f \magW g)(X) = \int_{\Pspace} \dd X \, f(X) \, g(X) = \bscpro{f^*}{g}_{L^2(\Pspace)} = \bigl ( f , g \bigr ) 
		. 
	\end{align*}
\end{lem}
\begin{proof}
	Formally, we obtain the result by some easy manipulations which can be made rigorous by regularizing the integral, Fubini's theorem and dominated convergence: 
	\begin{align*}
		\int_{\Pspace} \dd X \, &(f \magW g)(X) = \frac{1}{(2\pi)^{2d}} \int_{\Pspace} \dd X \int_{\Pspace} \dd Y \int_{\Pspace} \dd Z \, e^{i \sigma(X , Y + Z)} \, e^{\frac{i}{2} \sigma(Y,Z)} 
		\cdot \\
		&\qquad \qquad \qquad \qquad \qquad \qquad \cdot 
		\omega^B \bigl ( x - \tfrac{1}{2}(y + z) ; x + \tfrac{1}{2} (y - z) , x + \tfrac{1}{2} (y + z) \bigr ) 
		\cdot \\
		&\qquad \qquad \qquad \qquad \qquad \qquad \cdot 
		\bigl ( \Fs f \bigr )(Y) \, \bigl ( \Fs g \bigr )(Z) 
		\\
		&= \int_{\Pspace} \dd Y \, e^{\frac{i}{2} \sigma(Y,-Y)} 
		\, 
		\omega^B \bigl ( x - \tfrac{1}{2}(y - y) ; x + \tfrac{1}{2} (y + y) , x + \tfrac{1}{2} (y - y) \bigr ) 
		\cdot \\
		&\qquad \qquad \cdot 
		\bigl ( \Fs f \bigr )(Y) \, \bigl ( \Fs g \bigr )(-Y) 
		\\
		&= \int_{\Pspace} \dd Y \, \bigl ( \Fs f \bigr )(Y) \, \bigl ( \Fs g \bigr )(-Y) 
		= \int_{\Pspace} \dd Y \, f(Y) \, g(Y) = \bigl ( f , g \bigr ) = \bscpro{f^*}{g} 
	\end{align*}
	The magnetic flux is $0$ as the area of the collapsed triangle vanishes. 
\end{proof}
The basis for the extension of $\magW$ to tempered distributions is the following Lemma: 
\begin{cor}
	For $f , g , h \in \Schwartz(\Pspace)$, we have 
	\begin{align*}
		\bigl ( f \magW g , h \bigr ) = \bigl ( f , g \magW h \bigr ) = \bigl ( g , h \magW f \bigr ) 
		. 
	\end{align*}
\end{cor}
Starting from the above equality, we can define the Weyl product of a tempered distribution $F$ and a Schwartz function $f$. 
\begin{defn}[Extension of $\magW$\index{Weyl product!extension by duality} via duality]\label{magWQ:extension:defn:duality_extension_product}
	For $F \in \Schwartz'(\Pspace)$ and $g \in \Schwartz(\Pspace)$, we define the Weyl product $F \magW g$ by duality as 
	\begin{align*}
		\bigl ( F \magW g , h \bigr ) &:= \bigl ( F , g \magW h \bigr ) 
		\\
		\bigl ( g \magW F , h \bigr ) &:= \bigl ( F , h \magW g \bigr ) 
	\end{align*}
	for all $h \in \Schwartz(\Pspace)$. 
\end{defn}
\begin{prop}
	Defintion~\ref{magWQ:extension:defn:duality_extension_product} extends the magnetic Weyl product to the case where one factor is a tempered distribution and we get two continuous bilinear maps 
	\begin{align*}
		\magW &: \Schwartz'(\Pspace) \times \Schwartz(\Pspace) \longrightarrow \Schwartz'(\Pspace) \\
		\magW &: \Schwartz(\Pspace) \times \Schwartz'(\Pspace) \longrightarrow \Schwartz'(\Pspace) 
		. 
	\end{align*}
\end{prop}
\begin{proof}
	Associativity and compatibility with the involution can be checked easily by direct computation and using the $2$-cocycle property of the exponential of the magnetic flux. Lemma~\ref{magWQ:extension:lem:magnetic_product_scalar_product} also yields $1 \magW f = f = f \magW 1$ for $1 \in \Schwartz'(\Pspace)$ and $f \in \Schwartz(\Pspace)$. 
\end{proof}
\begin{prop}\label{magWQ:extension:prop:Op_A_product}
	Assume the magnetic field is polynomially bounded. Then $\Op^A$ is an involutive linear continuous map $\Schwartz'(\Pspace) \longrightarrow \mathcal{L} \bigl ( \Schwartz(\R^d) , \Schwartz'(\R^d) \bigr )$ satisfying $\Op^A (F \magW g) = \Op^A(F) \, \Op^A(g)$ and $\Op^A (g \magW F) = \Op^A(g) \, \Op^A(F)$ for all $F \in \Schwartz'(\Pspace)$ and $g \in \Schwartz(\Pspace)$. 
\end{prop}
\begin{proof}
	In Proposition~\ref{magWQ:extension:prop:extension_Op_A}, it has been established that $\Op^A$ is a linear topological isomorphism between $\Schwartz'(\Pspace)$ and $\mathcal{L} \bigl ( \Schwartz(\R^d) , \Schwartz'(\R^d) \bigr )$. The involution is defined as usual via the scalar product (which is an antiduality) and we conclude $\Op^A(F)^* = \Op^A(F^*)$ for $F \in \Schwartz'(\Pspace)$. In other words, the adjoint in the sense of operators becomes complex conjugation. The equality involving the products follows from a simple approximation argument. 
\end{proof}
%


\subsection{The magnetic Moyal algebra} 
\label{magWQ:extension:magnetic_moyal_algebra}
The next step is to isolate a class of distributions with `good' composition properties. The ideas in \cite{Mantoiu_Purice:magnetic_Weyl_calculus:2004} stem from Gracia-Bondía and Várilly \cite{Gracia_Bondia_Varilly:distributions_phasespace_1:1988,Gracia_Bondia_Varilly:distributions_phasespace_2:1988}. In essence, we would like to be able to multiply two distributions. However, so far, we have only achieved to replace \emph{one} of the factors in 
\begin{align*}
	f_1 \magW \cdots \magW f_n && f_j \in \Schwartz(\Pspace), \; j = 1 , \ldots , n 
\end{align*}
by a tempered distribution. The distributions with good composition properties are in the 
\begin{defn}[Magnetic Moyal algebra $\mathcal{M}^B(\Pspace)$\index{Moyal algebra!magnetic}]\label{magWQ:extension:defn:mag_Moyal_algebra}
	The spaces of distributions 
	\begin{align*}
		\mathcal{M}^B_L(\Pspace) &:= \bigl \{ F \in \Schwartz'(\Pspace) \; \vert \; F \magW g \in \Schwartz(\Pspace) \; \forall g \in \Schwartz(\Pspace) \bigr \} 
		\\
		\mathcal{M}^B_R(\Pspace) &:= \bigl \{ F \in \Schwartz'(\Pspace) \; \vert \; g \magW F \in \Schwartz(\Pspace) \; \forall g \in \Schwartz(\Pspace) \bigr \} 
	\end{align*}
	are left and right magnetic Moyal algebra. Their intersection 
	\begin{align*}
		\mathcal{M}^B(\Pspace) := \mathcal{M}^B_L(\Pspace) \cap \mathcal{M}^B_R(\Pspace)
	\end{align*}
	is called the \emph{magnetic Moyal algebra}. 
\end{defn}
Left- and right Moyal algebra are related by involution ${}^*$. Later on, we will see that
\begin{align*}
	\Op^A \bigl ( \mathcal{M}^B_L(\Pspace) \bigr ) = \mathcal{L} \bigl ( \Schwartz(\R^d) \bigr ) \end{align*}
and
\begin{align*}
	\Op^A \bigl ( \mathcal{M}^B_R(\Pspace) \bigr ) = \mathcal{L} \bigl ( \Schwartz'(\R^d) \bigr ) = \mathcal{L} \bigl ( \Schwartz(\R^d) \bigr )^* 
	. 
\end{align*}
Elements of the magnetic Moyal algebra are elements of $\mathcal{L} \bigl ( \Schwartz(\R^d) \bigr )$ that can be continuously extended to $\mathcal{L} \bigl ( \Schwartz'(\R^d) \bigr )$. 

A quick verification of the claims verifies that $\mathcal{M}^B(\Pspace)$ really deserves to be called an algebra. The product of $F , G \in \mathcal{M}^B(\Pspace)$ is defined via duality, 
\begin{align}
	\bigl ( F \magW G , h \bigr ) := \bigl ( F , G \magW h \bigr ) 
	&& \forall h \in \Schwartz(\Pspace) 
	, 
	\label{magWQ:extension:eqn:extension_magnetic_product}
\end{align}
and the involution is the extension of complex conjugation to distributions. 
\begin{prop}
	The triple $\bigl ( \mathcal{M}^B(\Pspace) , \magW , {}^* \bigr )$ forms a unital $*$-algebra of tempered distributions that contains $\Schwartz(\Pspace)$ as a selfadjoint two-sided ideal. 
\end{prop}
Furthermore, the image of the Moyal algebra under $\Op^A$ has a concise characterization in terms of continuous operators. 
\begin{prop}\label{magWQ:extension:magnetic_moyal_algebra:prop:mag_Moyal_algebra_L_S_L_Sprime}
	$\Op^A : \mathcal{M}^B(\Pspace) \longrightarrow \mathcal{L} \bigl ( \Schwartz(\R^d) \bigr ) \cap \mathcal{L} \bigl ( \Schwartz'(\R^d) \bigr )$ is an isomorphism between $*$-algebras. 
\end{prop}
\begin{proof}
	First of all, we can view $\mathcal{L} \bigl ( \Schwartz(\R^d) \bigr )$ as a subspace of $\mathcal{L} \bigl ( \Schwartz(\R^d) , \Schwartz'(\R^d) \bigr )$. The closed graph theorem ensures that any $T \in \mathcal{L} \bigl ( \Schwartz(\R^d) , \Schwartz'(\R^d) \bigr )$ with $T \Schwartz(\R^d) \subseteq \Schwartz(\R^d)$ is also continuous as a map in $\mathcal{L} \bigl ( \Schwartz(\R^d) \bigr )$. 
	
	Similarly, elements in $\mathcal{L} \bigl ( \Schwartz'(\R^d) \bigr )$ can be thought of as being composed of those $T \in \mathcal{L} \bigl ( \Schwartz(\R^d) , \Schwartz'(\R^d) \bigr )$ that admit a continuous extension to $\Schwartz'(\R^d)$. As the involution ${}^*$ is defined via the antiduality $\scpro{\cdot}{\cdot}$, once we write out the definition of the adjoint, we see that $\mathcal{L} \bigl ( \Schwartz(\R^d) \bigr )^* = \mathcal{L} \bigl ( \Schwartz'(\R^d) \bigr )$ as well as $\mathcal{L} \bigl ( \Schwartz'(\R^d) \bigr )^* = \mathcal{L} \bigl ( \Schwartz(\R^d) \bigr )$. Hence, $\mathcal{L} \bigl ( \Schwartz(\R^d) \bigr ) \cap \mathcal{L} \bigl ( \Schwartz'(\R^d) \bigr )$ is a $*$-algebra. 
	
	Now let $T \in \Op^A \bigl ( \mathcal{M}^B(\Pspace) \bigr )$. Then, by definition, we have 
	\begin{align*}
		&T \mathcal{L} \bigl ( \Schwartz'(\R^d) , \Schwartz(\R^d) \bigr ) \subseteq \mathcal{L} \bigl ( \Schwartz'(\R^d) , \Schwartz(\R^d) \bigr ) 
		, 
		\\
		&\mathcal{L} \bigl ( \Schwartz'(\R^d) , \Schwartz(\R^d) \bigr ) T \subseteq \mathcal{L} \bigl ( \Schwartz'(\R^d) , \Schwartz(\R^d) \bigr ) 
		, 
	\end{align*}
	where the latter is equivalent to 
	\begin{align*}
		T^* \mathcal{L} \bigl ( \Schwartz'(\R^d) , \Schwartz(\R^d) \bigr ) \subseteq \mathcal{L} \bigl ( \Schwartz'(\R^d) , \Schwartz(\R^d) \bigr ) 
		. 
	\end{align*}
	However, these implications are only satisfied if and only if $T \in \mathcal{L} \bigl ( \Schwartz(\R^d) \bigr )$. This again follows from the closed graph theorem and the fact that to each $u \in \Schwartz(\R^d)$, there exists a distribution $U \in \Schwartz'(\R^d)$ and a map $T \in \mathcal{L} \bigl ( \Schwartz'(\R^d) , \Schwartz(\R^d) \bigr )$ such that 
	\begin{align*}
		u = T U
		. 
	\end{align*}
	The relations involving $T^*$ completes the argument and we conclude $\Op^A \bigl ( \mathcal{M}^B(\Pspace) \bigr )$ coincides with $\mathcal{L} \bigl ( \Schwartz(\R^d) \bigr ) \cap \mathcal{L} \bigl ( \Schwartz'(\R^d) \bigr )$. 
	\medskip
	
	\noindent
	It remains to show that $\magW$ is the counterpart of operator multiplication on the level of the magnetic Moyal algebra. Let $F , G \in \mathcal{M}^B(\Pspace)$, $h \in \Schwartz(\Pspace)$ and $U \in \Schwartz'(\R^d)$. Then 
	\begin{align*}
		\Op^A (F \magW G) \, \bigl ( \Op^A(h) U \bigr ) &= \Op^A (F \magW G \magW h) U = \Op^A(F) \, \bigl ( \Op^A(G \magW h) U \bigr ) 
		\\
		&= \bigl ( \Op^A(F) \, \Op^A(G) \bigr ) \, \Op^A(h) U 
	\end{align*}
	by Proposition~\ref{magWQ:extension:prop:Op_A_product}. As any $u \in \Schwartz(\R^d)$ can be written as a product of the form $\Op^A(h) U$ with $h \in \Schwartz(\Pspace)$ and $U \in \Schwartz'(\R^d)$, the product $F \magW G$ maps $u \in \Schwartz(\R^d)$ onto another Schwartz function. This concludes the proof. 
\end{proof}
The next proposition gives an idea what kind of elements are contained in $\mathcal{M}^B(\Pspace)$: 
\begin{prop}
	For polynomially bounded magnetic fields, one has $\Schwartz'(\Pspace) \magW \Schwartz(\Pspace) \subsetneq \mathcal{M}^B_R(\Pspace)$ and $\Schwartz(\Pspace) \magW \Schwartz'(\Pspace) \subsetneq \mathcal{M}^B_L(\Pspace)$. 
\end{prop}
\begin{proof}
	This follows from the previous Proposition and Proposition~\ref{magWQ:extension:prop:extension_Op_A}: 
	\begin{align*}
		\Op^A \bigl ( \Schwartz'(\Pspace) \magW \Schwartz(\Pspace) \bigr ) &= \Op^A \bigl ( \Schwartz'(\R^d) \bigr ) \, \Op^A \bigl ( \Schwartz(\R^d) \bigr ) 
		\\
		&= \mathcal{L} \bigl ( \Schwartz(\R^d) , \Schwartz'(\R^d) \bigr ) \, \mathcal{L} \bigl ( \Schwartz'(\R^d) , \Schwartz(\R^d) \bigr ) 
		\\
		&\subsetneq \mathcal{L} \bigl ( \Schwartz'(\R^d) \bigr ) = \Op^A \bigl ( \mathcal{M}^B_R(\Pspace) \bigr ) 
	\end{align*}
	Thus, $\Schwartz'(\Pspace) \magW \Schwartz(\Pspace) \subsetneq \mathcal{M}^B_R(\Pspace)$. The other claim is proven in the same manner. 
\end{proof}
%


\subsection{Important subclasses} 
\label{magWQ:extension:important_subclasses}
%
In this subsection, we quote results from \cite{Mantoiu_Purice:magnetic_Weyl_calculus:2004} and \cite{Iftimie_Mantiou_Purice:magnetic_psido:2006} regarding important classes of functions which are contained in the magnetic Moyal algebra. For proofs, we refer to the original publications. 
\begin{defn}[Uniformly polynomially bounded functions\index{uniformally polynomially bounded functions $\Cont^{\infty}_{\mathrm{pol \, u}}(\Pspace)$}]
	The space $\Cont^{\infty}_{\mathrm{pol \, u}}(\Pspace)$ consists of smooth functions with uniform polynomial growth at infinity, \ie for each $f \in \Cont^{\infty}_{\mathrm{pol} \, u}(\Pspace)$ we can find $m \in \R$, $m \geq 0$, such that for all multiindices $a,\alpha \in \N_0^d$ there is a $C_{a \alpha} > 0$ with 
	\begin{align*}
		\babs{\partial_{\xi}^{\alpha} \partial_x^{a} f(x,\xi)} < C_{a \alpha} \expval{x}^m \expval{\xi}^m
		, 
		&& 
		\forall (x,\xi) \in \Pspace
		. 
	\end{align*}
\end{defn}
\begin{thm}[\cite{Mantoiu_Purice:magnetic_Weyl_calculus:2004}]\label{magWQ:extension:important_subclasses:uniformly_poly_bounded}
	For polynomially bounded magnetic fields, $\Cont^{\infty}_{\mathrm{pol \, u}}(\Pspace)$ functions are in the magnetic Moyal algebra $\mathcal{M}^B(\Pspace)$. 
\end{thm}
\begin{proof}
	Since the proof is rather technical, we refer to \cite{Mantoiu_Purice:magnetic_Weyl_calculus:2004} for details. One uses a family of Fréchet spaces $\{ R^m_0 \}_{m \in \R}$ whose family of seminorms is defined in terms of $L^1$ rather than $L^{\infty}$ norms. The space of functions $\tilde{S}^m_0$ which is uniformly bounded by a polynomial of $m$th degree in $x$ and $\xi$ is sandwiched between $R^m_0$ and $R^{m + 2N + \eps}_0$. 
\end{proof}
In case $\delta > 0$, Hörmander class symbols $\Hoerrd{m}$ may not be uniformly polynomially bounded and something remains to be proven. We refer to~\cite{Iftimie_Mantiou_Purice:magnetic_psido:2006} for details. 
\begin{thm}[\cite{Iftimie_Mantiou_Purice:magnetic_psido:2006}]\label{magWQ:extension:important_subclasses:Hoermander}
	For bounded magnetic fields $B$ and $0 \leq \delta < \rho \leq 1$ or $\delta = 0 = \rho$, $\Hoerrd{m} \subset \mathcal{M}^B(\Pspace)$. 
\end{thm}
\begin{remark}
	This theorem does not ensure that the product of two Hörmander symbols of order $m_1$ and $m_2$ is again a symbol of order $m_1 + m_2$. 
\end{remark}
\begin{remark}
	There are many operators which are \emph{not} in the magnetic Moyal algebra, \eg the rank-one operator $\sopro{u}{u}$ for $u \in L^2(\R^d) \setminus \Schwartz(\R^d)$. 
\end{remark}
For practical applications, the next theorem on the composition of Hörmander symbols is essential. 
\begin{thm}[\cite{Iftimie_Mantiou_Purice:magnetic_psido:2006}]
	Assume the components of the magnetic field $B$ are of class $\BCont^{\infty}$. Then $\Hoerrd{m_1} \magW \Hoerrd{m_2} \subset \Hoerrd{m_1 + m_2}$ holds. 
\end{thm}
We will give an independent proof in Chapter~\ref{asymptotics:expansions} for the case $\delta = 0$. 



\section{Important results} 
\label{magWQ:important_results}
Most of the standard results of usual Weyl calculus have been transcribed to the magnetic context, most of which can be found in \cite{Iftimie_Mantiou_Purice:magnetic_psido:2006} and \cite{Iftimie_Mantoiu_Purice:commutator_criteria:2008}. For convenience of the reader, we quickly present some of them here. The properties of the operators usually only depend on properties of the magnetic field rather than the vector potential which is highly arbitrary.

\subsection{$L^2$-continuity and selfadjointness} 
\label{magWQ:important_results:continuity_and_selfadjointness}
Properties of certain pseudodifferential operators can be related to properties of the functions. One of the most basic theorems of this sort is the Caldéron-Vaillancourt theorem which ensures the boundedness of the quantizations of $\BCont^{\infty}(\Pspace)$ functions (among others). 
\begin{thm}[Magnetic Caldéron-Vaillancourt theorem\index{Caldéron-Vaillancourt theorem!magnetic} \cite{Iftimie_Mantiou_Purice:magnetic_psido:2006}]\label{magWQ:important_results:continuity_and_selfadjointness:thm:magnetic_Calderon_Vaillancourt}
	Assume $B$ is a bounded magnetic field and $A$ is a polynomially bounded vector potential. Then the magnetic quantization of $f \in \Hoermrd{0}{\rho}{\delta}$, $0 \leq \rho = \delta < 1$ or $0 \leq \delta < \rho \leq 1$, is bounded, $\Op^A(f) \in \BopL$ and the operator norm can be bounded by 
	\begin{align*}
		\bnorm{\Op^A(f)}_{\mathcal{B}(L^2(\R^d))} \leq C(d) \, \sup_{\abs{a} , \abs{\alpha} \leq p(d)} \sup_{(x,\xi) \in \Pspace} \expval{\xi}^{\rho \abs{\alpha} - \delta \abs{a}} \, \babs{\partial_x^a \partial_{\xi}^{\alpha}f(x,\xi)} 
	\end{align*}
	where $C(d)$ and $p(d)$ are constants that only depend on the dimension $d$ and can be determined explicitly. 
\end{thm}
An impportant property of symbols is ellipticity which characterizes the behavior of the functions in momentum at infinity. 
\begin{defn}[Elliptic symbol\index{elliptic symbol}]\label{magWQ:important_results:continuity_and_selfadjointness:defn:elliptic_symbol}
	A symbol $f \in \Hoerrd{m}$ is called elliptic if there exist two positive constants $R$ and $C$ such that 
	\begin{align*}
		C \expval{\xi}^m \leq \abs{f(x,\xi)} 
		&& \forall x \in \R^d, \; \abs{\xi} \geq R 
		. 
	\end{align*}
	We also call the associated operator $\Op^A(f)$ elliptic. 
\end{defn}
Natural domains for quantizations of elliptic symbols are (magnetic) Sobolev spaces: 
\begin{defn}[Magnetic Sobolev space $H^m_A(\R^d)$\index{magnetic Sobolev space $H^m_A(\R^d)$}]\label{magWQ:important_results:continuity_and_selfadjointness:defn:magnetic_Sobolev_space}
	For $m > 0$, we define the magnetic Sobolev space associated to the magnetic field $B$ and vector potential $A$ to be 
	\begin{align*}
		H^m_A(\R^d) := \Bigl \{ u \in L^2(\R^d) \; \big \vert \; \Op^A(\expval{\xi}^m) u \in L^2(\R^d) \Bigr \} 
		. 
	\end{align*}
	The associated scalar product on $H^m_A(\R^d)$ is defined as 
	\begin{align*}
		\bscpro{u}{v}_{H^m_A} := \bscpro{u}{v}_{L^2} + \bscpro{\Op^A(\expval{\xi}^m) u}{\Op^A(\expval{\xi}^m) v}_{L^2} 
		. 
	\end{align*}
	We define $H^{-m}_A(\R^d)$ as the anti-dual to $H^m_A(\R^d)$ with norm 
	\begin{align*}
		\snorm{u}_{H^{-m}_A} := \sup_{v \in H^m_A(\R^d) \setminus \{ 0 \}} \frac{\abs{\scpro{v}{u}}}{\norm{v}_{H^m_A}}
		. 
	\end{align*}
	The scalar product is obtained by polarization. We also define $H^{\infty}_A(\R^d) := \bigcap_{m \in \R} H^m_A(\R^d)$ and $H^{-\infty}_A(\R^d) := \bigcup_{m \in \R} H^m_A(\R^d)$ endowed with the projective limit and inductive limit topology, respectively. 
\end{defn}
\begin{remark}
	One can show that this definition is in fact equivalent to the usual one (which can be found in \cite{Lieb_Loss:analysis:2001}, for instance). 
\end{remark}
This allows one to characterize magnetic pseudodifferential operators that are not necessarily bounded in the $L^2(\R^d)$ sense. 
\begin{prop}[Boundedness of pseudodifferential operators]\label{magWQ:important_results:continuity_and_selfadjointness:prop:boundedness_psido}
	If $B$ is a bounded magnetic field, then for any $f \in \Hoerrd{m}$, $m \geq 0$, $0 \leq \delta < \rho \leq 1$, $\Op^A(f) : H^m_A(\R^d) \longrightarrow L^2(\R^d)$ is bounded. More generally, if $m \leq s$, then $\Op^A(f)$ defines a bounded operator from $H^s_A(\R^d)$ to $H^{s-m}_A(\R^d)$. 
\end{prop}
If in addition to being elliptic, the symbol is real-valued, then the associated operator will be selfadjoint. 
\begin{thm}[Selfadjointness of elliptic symbols]\label{magWQ:important_results:continuity_and_selfadjointness:thm:selfadjointness_elliptic_symbols}
	Assume $B$ is a bounded magnetic field and $f \in \Hoerrd{m}$, $0 \leq \delta < \rho \leq 1$ and $m \geq 0$, a real-valued symbol. If $m > 0$, in addition, we assume $f$ to be elliptic. Then $\Op^A(f)$ defines a selfadjoint operator on the domain $\mathcal{D} = H^m_A(\R^d)$ and $\Schwartz(\R^d)$ is a core. 
\end{thm}
Also, lower semiboundedness is preserved for certain types of symbols: 
\begin{thm}[Gårding inequality]\label{magWQ:important_results:continuity_and_selfadjointness:thm:Garding}
	Let $B$ be a bounded magnetic field and $f \in \Hoerrd{m}$, $m \in \R$, $0 \leq \delta < \rho \leq 1$. Assume there exist two constants $R > 0$ and $C > 0$ such that $\Re f(x,\xi) \geq C \abs{\xi}^m$ for $\abs{\xi} \geq R$. Then for all $s \in \R$, there exist two finite constants $K_1, K_2 \in \R^+$ such that 
	\begin{align*}
		\Re \bscpro{u}{\Op^A(f) u} \geq K_1 \norm{u}_{H^{\nicefrac{m}{2}}_A}^2 - K_2 \norm{u}_{H^s_A}^2
	\end{align*}
	for all $u \in H^{\infty}_A(\R^d)$. 
\end{thm}
An immediate consequence is that real-valued elliptic symbols which are bounded from below in the sense of functions are quantized to selfadjoint operators which are bounded from below. 
\begin{cor}
	Under the assumptions of Theorem~\ref{magWQ:important_results:continuity_and_selfadjointness:thm:Garding}, if the real-valued elliptic symbol $f$ is bounded from below, then so is its quantization $\Op^A(f)$. 
\end{cor}
%


\subsection{Commutator criteria} 
\label{magWQ:important_results:commutator_criteria}
%
So far, we have only presented results connecting properties of the function or distribution $f$ with properties of its magnetic quantization $\Op^A(f)$. Can we say something in the reverse direction: if $\Op^A(f)$ has certain properties, can we deduce some properties of $f$? 

In the context of usual Weyl calculus, two standard results are the commutator criteria of Beals \cite{Beals:characterization_psido:1977} and Bony \cite{Bony:characterization_psido:1996}. For any two operators $S , T \in \mathcal{B} \bigl ( L^2(\R^d) \bigr )$, let us define 
%
\begin{align*}
	\adfrak_S(T) := [S , T] = S T - T S 
	. 
\end{align*}
$\adfrak_S$ acts as a derivation. If $S$ and $T$ are unbounded operators, then one needs to be careful as to how to define the above expression. Then the Beals criterion reads as follows: 
\begin{thm}[Beals criterion\index{Beals criterion!non-magnetic} \cite{Beals:characterization_psido:1977}]\label{magWQ:important_results:commutator_criteria:thm:non_mag_Beals}
	An operator $T \in \mathcal{B} \bigl ( L^2(\R^d) \bigr )$ is the (usual) Weyl quantization of $f \in \BCont^{\infty}(\Pspace) = \Hoermrd{0}{0}{0}$ if and only if for all $a , \alpha \in \N_0^d$ the commutators 
	\begin{align}
		\adfrak_{\Qe_1}^{a_1} \cdots \adfrak_{\Qe_d}^{a_d} \adfrak_{\Pe_1}^{\alpha_1} \cdots \adfrak_{\Pe_d}^{\alpha_d}(T) 
	\end{align}
	define bounded operators on $L^2(\R^d)$. 
\end{thm}
For operators of the form $f(\Pe)$ and $g(\Qe)$ where $f$ and $g$ are sufficiently regular, a formal calculation yields 
\begin{align*}
	\adfrak_{\Pe_j} \bigl ( g(\Qe) \bigr ) &= [ \Pe_j , g(\Qe) ] = - i \partial_{x_j} g(\Qe) \\
	\adfrak_{\Qe_j} \bigl ( g(\Qe) \bigr ) &= [ \Qe_j , g(\Qe) ] = 0 
	. 
\end{align*}
Multiplication operators $g(\Qe)$ (after Fourier transform $f(\Pe)$ also becomes a multiplication operator) are bounded if and only if $g \in L^{\infty}(\R^d)$. Hence, if $g(\Qe)$ and the commutator are also bounded, then $g , \partial_{x_j} g \in L^{\infty}(\R^d)$ holds true. Similarly, if arbitrary commutators with $\Pe$ and $\Qe$ are bounded, then $g$ should be bounded with bounded derivatives to any order, \ie $g \in \BCont^{\infty}(\R^d) \subset \BCont^{\infty}(\Pspace)$. The argument for $f$ is analogous. The Beals criterion extends this formal argument to operators that are not necessarily multiplication operators. 

Later on, we will also need 
\begin{align*}
	\adfrak_X(T) := \bigl [ \sigma \bigl ( X,(\Qe,\Pe) \bigr ) , T \bigr ] = \bigl [ \xi \cdot \Qe - x \cdot \Pe , T]
\end{align*}
as a mixed commutator with respect to a linear combination of $\Qe$ and $\Pe$. Since commutators are often difficult to treat for technical reasons, it is necessary to introduce the associated unitary one-parameter groups as well: the operator 
\begin{align*}
	\Adfrak_X(T) := e^{+ i \sigma(X,(\Qe,\Pe))} T e^{- i \sigma(X,(\Qe,\Pe))} 
\end{align*}
is well-defined for all $T \in \mathcal{B} \bigl ( L^2(\R^d) \bigr )$. Clearly, 
\begin{align*}
	i \frac{\partial}{\partial t} \Adfrak_{t X}(T) = i \lim_{t \rightarrow 0} \frac{\Adfrak_{t X}(T) - T}{t} = \adfrak_{X}(T) 
\end{align*}
holds whenever the right-hand side makes sense as a suitable bounded operator on $L^2(\R^d)$.\footnote{$\adfrak_X$ can be seen as a selfadjoint operator on a dense subspace of the Hilbert space $\mathcal{T}^2 \bigl ( L^2(\R^d) \bigr )$, the Hilbert-Schmidt operators on $L^2(\R^d)$. By Stone's theorem, $\Adfrak_{t X}$ is the associated strongly-continuous one-parameter group. } 

Now, a first attempt at writing a Beals theorem for \emph{magnetic} pseudodifferential operators reads 
\begin{thm}[Magnetic Beals criterion\index{Beals criterion!magnetic},  Theorem~1.1 in \cite{Iftimie_Mantoiu_Purice:commutator_criteria:2008}]
	Assume $B$ is a bounded magnetic field. Choose an associated vector potential $A \in \Cont^{\infty}_{\mathrm{pol}}(\R^d,\R^d)$. A linear continuous operator $T : \Schwartz(\R^d) \longrightarrow \Schwartz'(\R^d)$ is a magnetic pseudodifferential operator with symbol of class $\BCont^{\infty}(\Pspace) = \Hoermrd{0}{0}{0}$ if and only if the commutators 
	\begin{align*}
		\adfrak_{\Qe_1}^{a_1} \cdots \adfrak_{\Qe_d}^{a_d} \adfrak_{\Pe^A_1}^{\alpha_1} \cdots \adfrak_{\Pe^A_d}^{\alpha_d}(T) 
	\end{align*}
	define bounded operators on $L^2(\R^d)$ for all multiindices $a , \alpha \in \N_0^d$. 
\end{thm}
For several reasons, this result is more involved than Theorem~\ref{magWQ:important_results:commutator_criteria:thm:non_mag_Beals}: first of all, there is a family of magnetic pseudodifferential operators $\{ \Op^A(f) \}_{\dd A = B}$ associated to a suitable tempered distribution $f$ which are labeled by possible choices of vector potentials $A$ associated to the magnetic field $B = \dd A$. Certainly, conditions that should be placed on the magnetic \emph{field} have to be extracted from an operator that depends on the choice of gauge. Secondly, different components of momenta no longer commute, but produce terms containing components of $B$ as well as derivatives thereof. This makes proofs and derivations rather tedious. 

Hence, it turns out that it is advantageous to rephrase the problem: \cite{Iftimie_Mantoiu_Purice:commutator_criteria:2008} have suggested to look at more fundamental $C^*$-algebras that depend only on the magnetic field and are by construction independent of the choice of vector potential: we consider 
\begin{align}
	\AB := {\Op^A}^{-1} \Bigl ( \mathcal{B} \bigl ( L^2(\R^d) \bigr ) \Bigr ) 
	, 
\end{align}
a $C^*$-algebra composed of tempered distributions, with transported product $\magW$ and involution ${}^*$, 
\begin{align*}
	f \magW g &:= {\Op^A}^{-1} \bigl ( \Op^A(f) \, \Op^A(g) \bigr ) 
	\\
	f^* &:= {\Op^A}^{-1} \bigl ( \Op^A(f)^* \bigr ) 
	, 
\end{align*}
as well as transported norm $\norm{f}_B := \bnorm{\Op^A(f)}_{\mathcal{B}(L^2(\R^d))}$. As the notation suggests, $\mathfrak{C}^B$ depends only on the magnetic field by covariance of $\Op^A$. Furthermore, $f \magW g$ coincides with equation~\eqref{magWQ:magnetic_weyl_calculus:eqn:Fourier_form_magnetic_composition} for two suitable tempered distributions $f$ and $g$, and $f^*$ is the complex conjugate of the distribution $f$, \ie 
\begin{align*}
	\bigl ( f^* , \varphi \bigr ) := \bigl ( f , \varphi^* \bigr )^* 
	&&
	\forall \varphi \in \Schwartz(\Pspace) 
	. 
\end{align*}
$\mathfrak{C}^B$ is a $\ast$-subalgebra of $\Schwartz'(\Pspace)$ with respect to $\magW$ and ${}^*$ as well as a vector subspace. In Chapter~\ref{magWQ:extension:magnetic_moyal_algebra}, we have seen how to extend $\magW$ by duality to 
\begin{align*}
	\magW : \mathcal{M}^B(\Pspace) \times \Schwartz'(\Pspace) \longrightarrow \Schwartz'(\Pspace) \\
	\magW :  \Schwartz'(\Pspace) \times \mathcal{M}^B(\Pspace) \longrightarrow \Schwartz'(\Pspace) 
\end{align*}
where $\mathcal{M}^B(\Pspace)$ is the magnetic Moyal algebra as by Definition~\ref{magWQ:extension:defn:mag_Moyal_algebra}. Since for each $X \in \Pspace$, the function $l_X : Y \mapsto \sigma(X,Y)$ is linear and thus uniformly polynomially bounded, we can make sense of the expression 
\begin{align*}
	\ad_X^B(F) := [l_X , F]_{\magW} = l_X \magW F - F \magW l_X \in \Schwartz'(\Pspace) 
\end{align*}
for any $X \in \Pspace$ and $F \in \mathfrak{C}^B$. The associated exponential $e_X := e^{- i l_X} \in \BCont^{\infty}(\Pspace)$, $X \in \Pspace$, is again an element of the magnetic Moyal algebra $\mathcal{M}^B(\Pspace)$ which is the \emph{algebraic analog of the Weyl system}. In fact, one can define the usual Weyl system as quantization of $e_X$, $\WeylSys^A(X) := \Op^A(e_X)$. This implies $e_X \in \mathfrak{C}^B$ is a Moyal unitary and thus bounded. Hence, the family of \emph{magnetic phase space translations} $\bigl \{ \tau^B_X \bigr \}_{X \in \Pspace}$, 
\begin{align*}
	\tau^B_X(F) := e_{-X} \magW F \magW e_X 
	, 
	&& 
	F \in \Schwartz'(\Pspace)
	, 
\end{align*}
is a collection of automorphisms on $\mathfrak{C}^B$. In other words, $\tau^B_X$ substitutes for conjugating with the magnetic Weyl system which is composed of translations in real and reciprocal space as well as multiplication by a phase (which contains a magnetic contribution). Not surprisingly, $e_X$ obeys essentially the same composition law as $\WeylSys^A(X)$, \ie 
%
\begin{align*}
	e_X \magW e_Y = e^{\frac{i}{2} \sigma(X,Y)} \, e^{- i \Gamma^B(\sexpval{x-\frac{1}{2}(y+z),x + \frac{1}{2}(y-z) , x+\frac{1}{2}(y+z)})} \, e_{X+Y}
\end{align*}
holds for all $X , Y \in \Pspace$ (compared to equation~\eqref{magWQ:magnetic_weyl_calculus:lem:composition_Weyl_system}, the exponential of the magnetic flux through \emph{different} corners enters). In the non-magnetic case, $\tau_X := \tau^0_X$ reducles to the usual translations, 
\begin{align*}
	\bigl ( \tau_X(f) \bigr )(Y) := f(Y - X) 
	&&
	\forall X , Y \in \Pspace 
	. 
\end{align*}
Hence, we can define the Fréchet space suggested by the Beals criterion, namely 
\begin{align*}
	\Cont^{\infty}(\tau^B,\mathfrak{C}^B) := \Bigl \{ F \in \mathfrak{C}^B \; \big \vert \; X \mapsto \tau^B_X(F) \in \Cont^{\infty} \mbox{ in $X = 0$} \Bigr \} 
\end{align*}
endowed with the family of seminorms 
\begin{align*}
	\Bigl \{ \snorm{\cdot}^{\tau^B}_{U_1 , \ldots , U_n} \; \big \vert \; n \in \N_0 , \, U_j \in \Pspace , \, \sabs{U_j} = 1 \, \forall j = 1 , \ldots , n \Bigr \} 
	, 
\end{align*}
each of which being defined as 
\begin{align*}
	\bnorm{F}^{\tau^B}_{U_1 , \ldots , U_n} := \bnorm{\ad^B_{U_1} \cdots \ad^B_{U_n}(F)}_{\mathfrak{C}^B} 
	. 
\end{align*}
Observe that in the non-magnetic case, $\Cont^{\infty}(\tau,\mathfrak{C}^0)$ coincides with $\BCont^{\infty}(\Pspace)$ by the non-magnetic Beals criterion, Theorem~\ref{magWQ:important_results:commutator_criteria:thm:non_mag_Beals}. Written in this language, the magnetic version reads 
\begin{thm}[Magnetic Beals criterion\index{Beals criterion!magnetic} \cite{Iftimie_Mantoiu_Purice:commutator_criteria:2008}]
	If $B$ is of class $\BCont^{\infty}$, then $f \in \Hoermrd{0}{0}{0} = \BCont^{\infty}(\Pspace)$ if and only if for all $n \in \N_0$ and $U_1 , \ldots , U_n \in \Pspace$ with $\sabs{U_1} = \ldots = \sabs{U_n} = 1$ 
	\begin{align*}
		\ad^B_{U_1} \cdots \ad^B_{U_n}(f) \in \mathfrak{C}^B 
	\end{align*}
	holds. 
\end{thm}
The proof is tedious and technical, and it amounts to showing that $\BCont^{\infty}(\Pspace) = \Hoermrd{0}{0}{0}$ and $\Cont^{\infty}(\tau^B,\mathfrak{C}^B)$ agree as spaces and have isomorphic Fréchet structures. 
\medskip

\noindent
In practical situations, it is often more useful to replace Moyal commutators $\ad^B_X = [ l_X , \cdot]_{\magW}$ with commutators with more general functions. 
\begin{defn}[$S^+_{\rho}$]
	Let $\rho \in [0,1]$. We define the class of symbols $S^+_{\rho}$ as 
	\begin{align*}
		S^+_{\rho} := \Bigl \{ \varphi \in \Cont^{\infty}(\Pspace) \; \big \vert \; \babs{\partial_x^a \partial_{\xi}^{\alpha} \varphi(X)} \leq C_{a \alpha} \sexpval{\xi}^{(1 - \abs{\alpha}) \rho} \, \forall \abs{a} + \abs{\alpha} \geq 1 \Bigr \} 
		. 
	\end{align*}
\end{defn}
For any $\varphi \in S^+_{\rho} \subset \mathcal{M}^B(\Pspace)$, we define the derivation 
\begin{align*}
	\ad^B_{\varphi}(F) := [ \varphi , F ]_{\magW} 
	&& 
	\forall F \in \Schwartz'(\Pspace) 
	 .
\end{align*}
Then the magnetic Bony criterion reads 
\begin{thm}[Magnetic Bony criterion\index{Bony criterion!magnetic} \cite{Iftimie_Mantoiu_Purice:commutator_criteria:2008}]
	Assume the components of $B$ are of class $\BCont^{\infty}$. A distribution $F \in \Schwartz'(\Pspace)$ is a symbol of type $\Hoermr{0}{\rho}$, $\rho \in [0,1]$, if and only if for any $n \in \N_0$ and any family $\{ \varphi_1 , \ldots \varphi_n \} \subset S^+_{\rho}$ 
	\begin{align*}
		\ad^B_{\varphi_1} \cdots \ad^B_{\varphi_n}(F) \in \mathfrak{C}^B 
	\end{align*}
	holds true. 
\end{thm}
The Beals and Bony criteria can be extended to probe whether a distribution is really a Hörmander symbol of type $m \in \R$: for any $m > 0$, we define 
\begin{align*}
	\mathfrak{p}_{m,\lambda}(X) := \expval{\xi}^m + \lambda 
	. 
\end{align*}
It has been proven in \cite[Thm.~1.8]{Mantoiu_Purice_Richard:Cstar_algebraic_framework:2007} that for $\lambda$ large enough, $\mathfrak{p}_{m,\lambda}$ is invertible with respect to the composition law $\magW$ and that its inverse $\mathfrak{p}_{m,\lambda}^{(-1)_B}$ belongs to $\Hoermr{-m}{1}$. So for any $m > 0$ we can fix $\lambda = \lambda(m)$ such that $\mathfrak{p}_{m,\lambda(m)}$ is invertible. Then, for arbitrary $m \in \R$ we set
\begin{align*}
	\mathfrak{r}_m :=
	\left \{
	\begin{matrix}
		\mathfrak{p}_{m , \lambda(m)} & \mbox{ for } m > 0 \\
		1 & \mbox{ for } m = 0 \\
		\mathfrak{p}_{\abs{m},\lambda(\abs{m})}^{(-1)_B} & \mbox{ for } m < 0 \\
	\end{matrix}
	\right .
	. 
\end{align*}
By construction, relation $\mathfrak{r}_m^{(-1)_B} = \mathfrak{r}_{-m}$ holds for all $m \in \R$. The straight-forward generalizations of the Beals and Bony criteria read 
\begin{thm}[Theorem 5.21 in \cite{Iftimie_Mantoiu_Purice:commutator_criteria:2008}]
	Assume the components of $B$ are of class $\BCont^{\infty}$. A distribution $F \in \Schwartz'(\Pspace)$ is a symbol of type $\Hoermr{m}{\rho}$, $m \in \R$, $\rho \in [0,1]$, if and only if for any $n,k \in \N_0$ and any collection of vectors $x_1 , \ldots , x_n \in \R^d$, $\xi_1 , \ldots \xi_k \in {\R^d}^*$ the following holds true: 
	\begin{align*}
		\mathfrak{r}_{-(m - k \rho)} \magW \ad^B_{x_1} \cdots \ad^B_{x_n} \ad^B_{\xi_1} \cdots \ad^B_{\xi_k}(F) \in \mathfrak{C}^B 
	\end{align*}
	The two families of seminorms 
	\begin{align*}
		\bnorm{\mathfrak{r}_{-(m - \sabs{\alpha} \rho)} \partial_x^a \partial_{\xi}^{\alpha}F}_{\infty} 
	\end{align*}
	indexed by $a , \alpha \in \N_0^d$ and 
	\begin{align*}
		\norm{\mathfrak{r}_{-(m - k \rho)} \magW \ad^B_{x_1} \cdots \ad^B_{x_n} \ad^B_{\xi_1} \cdots \ad^B_{\xi_k}(F)}_{\mathfrak{C}^B}
	\end{align*}
	indexed by $n , k \in \N_0$ and sets of vectors in $\Pspace$ define equivalent topologies on $\Hoermr{m}{\rho}$. 
\end{thm}
\begin{thm}[Theorem 5.24 in \cite{Iftimie_Mantoiu_Purice:commutator_criteria:2008}]
	Assume the components of $B$ are of class $\BCont^{\infty}$. A distribution $F \in \Schwartz'(\Pspace)$ is a symbol of type $\Hoermr{m}{\rho}$, $m \in \R$, $\rho \in [0,1]$, if and only if for any $n \in \N_0$ and any family $\{ \varphi_1 , \ldots \varphi_n \} \subset S^+_{\rho}$ 
	\begin{align*}
		\mathfrak{r}_{-m} \magW \ad^B_{\varphi_1} \cdots \ad^B_{\varphi_n}(F) \in \mathfrak{C}^B 
	\end{align*}
	holds true. 
\end{thm}
In the next section, we will see how the Beals and Bony criterion can be applied. 


\subsection{Inversion and holomorphic functional calculus} 
\label{magWQ:important_results:inversion}
%
Usually, the theory of pseudodifferential operators is seen either from an analytic or an algebraic point of view. It turns out that one can benefit from making a connection between the two and use them simultaneously to one's advantage. One such notion from the intersection of the two topics is that of a $\Psi^*$-algebra: 
\begin{defn}[$\Psi^*$-algebra\index{$\Psi^*$-algebra}]\label{magWQ:important_results:inversion:defn:Psi_star_algebra}
	Let $\Psi$ be a unital $C^*$-subalgebra of a $C^*$-algebra $\Alg$. We say that $\Psi$ is a $\Psi^*$-algebra if it is \emph{spectrally invariant} (or full), \ie 
	\begin{align*}
		\Psi \cap \Alg^{-1} = \Psi^{-1} 
	\end{align*}
	where $\Alg^{-1}$ and $\Psi^{-1}$ are the groups of invertible elements of $\Alg$ and $\Psi$, repsectively, and if $\Psi$ can be endowed with a Fréchet topology $\tau_{\Psi}$ such that $\Psi \hookrightarrow \Alg$ can be continuously embedded in $\Alg$. 
\end{defn}
In Chapter~\ref{psiDO_reloaded:inversion_and_affiliation}, we will explore this connection in more detail. 

An easy to prove consequence of the magnetic Bony criterion is that resolvents of magnetic pseudodifferential operators -- should they exist -- are again pseudodifferential operators. In Chapter~\ref{psiDO_reloaded}, we will show that even anisotropies are preserved under inversion. 
\begin{thm}[Propositions~6.28 and 6.29 in \cite{Iftimie_Mantoiu_Purice:commutator_criteria:2008}]\label{magWQ:important_results:inversion:thm:inversion_Hoermander}
	Assume $\rho \in [0,1]$. 
	\begin{enumerate}[(i)]
		\item For $m \geq 0$ if $F \in \Hoermr{m}{\rho}$ is invertible in $\mathcal{M}^B(\Pspace)$ with $\mathfrak{r}_m \magW F^{(-1)_B} \in \mathfrak{C}^B$, then $F^{(-1)_B} \in \Hoermr{-m}{\rho}$ holds. 
		\item For $m < 0$ if $F \in \Hoermr{m}{\rho}$ is such that $1 + F$ invertible in $\mathfrak{C}^B$, then $(1 + F)^{(-1)_B} - 1 \in \Hoermr{m}{\rho}$ holds. 
	\end{enumerate}
\end{thm}
Put another way, it was just shown that 
\begin{thm}
	$\Hoermr{0}{\rho} \hookrightarrow \mathfrak{C}^B$ is a $\Psi^*$-algebra for $\rho \in [0,1]$. 
\end{thm}
Some results on $\Psi^*$-algebras will prove very useful in Chapter~\ref{psiDO_reloaded}: 
\begin{thm}[Corollary~2.5 in \cite{Lauter:operator_theoretical_approach_melrose_algebras:1998}]
	\label{magWQ:important_results:inversion:thm:Psi_star_sub_also_Psi_star}
	Let $\Psi \subseteq \Alg$ be a $\Psi^*$-algebra and $\Psi' \subseteq \Psi$ be a closed, symmetric subalgebra of $\Psi$ with unit. Then $\Psi' \hookrightarrow \Alg$ endowed with the restricted topology $\tau_{\Psi'} := \tau_{\Psi} \vert_{\Psi'}$ is again a $\Psi^*$-algebra. 
\end{thm}
Hence, if one wants to check whether a symmetric subalgebra $\Psi' \subseteq \Psi \subseteq \Alg$ of a $\Psi^*$-algebra is spectrally invariant, all that is left to prove is closedness under multiplication and taking limits. One is freed from showing spectral invariance \emph{in addition} which is often technically much more challenging than showing closedness. In Chapter~\ref{psiDO_reloaded:inversion_and_affiliation}, we will use this fact to prove that under certain conditions, Moyal resolvents $(f - z)^{(-1)_B}$ retain the $x$-dependence of the original symbol $f$. 

$\Psi^*$-algebras also have a nice holomorphic functional calculus: let $f \in \Psi \subseteq \Alg$ be an element of a $\Psi^*$-algebra and $\varphi : \C \longrightarrow \C$ be a function which is holomorphic in a neighborhood of the spectrum $\sigma(f) := \bigl \{ z \in \C \; \vert \; f - z \, \id \mbox{ is invertible} \bigr \}$. Then 
\begin{align*}
	\varphi(f) := \frac{1}{2 \pi i} \int_{\Gamma} \dd z \, \varphi(z) \, \bigl ( f - z \, \id \bigr )^{-1} \in \Psi 
\end{align*}
is well-defined and again an element of the $\Psi^*$-algebra ($\Gamma$ is a contour surrounding $\sigma(f)$). Since only the behavior of $\varphi$ on the real axis, \ie $\varphi \vert_{\R}$, is important for the functional calculus of selfadjoint elements of our algebra, one is interested in extensions beyond holomorphic functions. Helffer and Sjöstrand have suggested a formula that initially holds for $\varphi \in \Cont^{\infty}_c(\R)$ since these functions have quasianalytic extensions. 
\begin{defn}[Quasianalytic extension\index{quasianalytic extension}]
	Let $\varphi \in \Cont^{\infty}_c(\R)$. Then $\tilde{\varphi} \in \Cont^{\infty}_c(\C)$ is called a \emph{quasianalytic extension of $\varphi$} iff 
	\begin{enumerate}[(i)]
		\item $\tilde{\varphi}$ and $\varphi$ agree on the real axis, $\tilde{\varphi} \vert_{\R} = \varphi$, and 
		\item for all $N \in \N_0$ there exists $C_N > 0$ such that $\babs{\partial_{\bar{z}} \tilde{\varphi}(z)} \leq C_N \, \babs{\Im z}^N$ holds. 
	\end{enumerate}
\end{defn}
\begin{remark}
	Quasianalytic extensions for $\varphi \in \Cont^{\infty}_c(\R)$ are by no means unique: let $\chi \in \Cont^{\infty}_c(\R,[0,1])$ be such that $\chi(x) = 1$ for all $\abs{x} \leq 1$ and $\chi(x) = 0$ for all $\abs{x} \geq 2$. In \cite{Davies:functional_calculus:1995}, it was shown that for each $n \in \N_0$ 
	\begin{align*}
		\tilde{\varphi}(x + i y) := \sum_{k = 0}^n \frac{\varphi^{(k)}(x)}{k!} \, (i y)^k \, \chi \bigl ( \tfrac{y}{\sexpval{x}} \bigr ) 
		, 
		&&
		x , y \in \R 
		, 
	\end{align*}
	defines a quasianalytic extension of $\varphi$. It was shown that the later construction does not depend on the particular choice of $n$ or $\chi$. Alternatively, one could follow Dimassi's and Sjöstrand's suggestion \cite{Dimassi_Sjoestrand:spectral_asymptotics:1999} use 
	\begin{align*}
		\tilde{\varphi}(x + i y) := \frac{\psi(x)}{\sqrt{2\pi}} \int_{\R} \dd \xi \, e^{i (x + i y) \cdot \xi} \, \chi(y \xi) \, (\Fourier \varphi)(\xi) 
	\end{align*}
	where $\psi \in \Cont^{\infty}_c(\R,[0,1])$ is $1$ in a neighborhood of $\supp(\varphi)$. 
\end{remark}
If we are interested in extending the functional calculus to functions which are not necessarily holomorphic, but, say smooth and compactly supported, we can use the Helffer-Sjöstrand formula \cite[Proposition~7.2]{Helffer_Sjoestrand:mag_Schroedinger_equation:1989}: for any $\varphi \in \Cont^{\infty}_c(\R)$ and $f \in \Hoermr{m}{\rho}$, one sets 
\begin{align}
	\varphi^B(f) := \frac{1}{\pi} \int_{\C} \dd z \, \partial_{\bar{z}} \tilde{\varphi}(z) \, (f - z)^{(-1)_B} 
\end{align}
where $\dd z$ denotes the usual Lebesgue measure on the complex plane $\C \cong \R^2$. This formula can be extended to more general classes of functions. Now Iftimie, Măntoiu and Purice have proven 
\begin{thm}[Proposition 6.33 in \cite{Iftimie_Mantoiu_Purice:commutator_criteria:2008}]
	Assume the components of the magnetic field $B$ are of class $\BCont^{\infty}$. Let $\varphi \in \Cont^{\infty}_c(\R)$ and $f \in \Hoermr{m}{\rho}$, $m \in \R$, $\rho \in [0,1]$. If $m > 0$, assume in addition that $g$ is real-valued and elliptic. Then $\varphi^B(f) \in \Hoermr{-m}{\rho}$ holds and for any $A$ representing the magnetic field $B = \dd A$, we have 
	\begin{align*}
		\Op^A \bigl ( \varphi^B(f) \bigr ) = \varphi \bigl ( \Op^A(f) \bigr ) 
		. 
	\end{align*}
\end{thm}
%



\chapter{Asymptotic Expansions and Semiclassical Limit} 
\label{asymptotics}
%
In many situations, the commutator of position and momentum operator is small, \ie 
\begin{align*}
	i \bigl [ \Pe^A_l , \Qe_j \bigr ] = \eps \delta_{lj} 
	. 
\end{align*}
Usually the small parameter is denoted by $\hbar$, although we shall use $\eps$ instead to indicate that this parameter can have a multitude of physical interpretations. Furthermore, we insist that $\eps$ be dimensionless, so the natural constant $\hbar$ is actually \emph{not} a good small parameter. 

Just like in the case of non-magnetic Weyl calculus, one can expand the magnetic Weyl product asymptotically in the small parameter $\eps$, 
\begin{align*}
	f \magW g \asymp \sum_{n = 0}^{\infty} \eps^n (f \magW g)_{(n)} 
	, 
\end{align*}
which can be used to prove a semiclassical limit (Chapter~\ref{asymptotics:semiclassical_limit}) and derive effective hamiltonians in multiscale systems (Chapter~\ref{magsapt}). Asymptotic expansions of the product and ramifications in other part of the theory of magnetic pseudodifferential operators will be the main focus of this chapter, applications will be postponed to the next. This chapter is based on the publication \cite{Lein:two_parameter_asymptotics:2008}.

\section{Scalings} 
\label{asymptotics:scalings}
In the context of magnetic pseudodifferential operators, an additional small parameter may be of interest, namely small coupling to the magnetic field quantified by $\lambda$. If there is no small coupling, we simply set $\lambda = 1$. Thus, we are interested in position and momentum operators that satisfy the following commutation relations:\index{commutation relations!magnetic (with small parameters)} 
\begin{align}
	i \bigl [ \Qe_l , \Qe_j \bigr ] = 0 
	&& 
	i \bigl [ \Pe^A_l , \Qe_j \bigr ] = \eps \delta_{lj} 
	&& 
	i \bigl [ \Pe^A_l , \Pe^A_j \bigr ] = - \eps \lambda B_{lj}(\Qe) 
	\label{asymptotics:scalings:eqn:commutation_relations}
\end{align}
In one opens any standard textbook on quantum mechanics, \eg \cite{Sakurai:quantum_mechanics:1994}, then one proposes\index{building block operators!usual scaling} 
\begin{align}
	\tilde{\mathsf{Q}} :=& \, \hat{x} 
	\label{asymptotics:scalings:eqn:q_p_usual_scaling} \\
	\tilde{\mathsf{P}}^A \equiv \tilde{\mathsf{P}}^A_{\eps,\lambda} :=& \, - i \eps \nabla_x - \lambda A(\hat{x}) 
	\notag 
\end{align}
as position and kinetic momentum operators in the \emph{usual scaling} on $L^2(\R^d_x)$ where $A$ is some vector potential associated to the magnetic field $B = \dd A$. Formally, these operators satisfy~\eqref{asymptotics:scalings:eqn:commutation_relations}. Here, position is measured in \emph{macroscopic units} and potentials vary on the scale $\order(1)$. 

Alternatively, we may measure distances in \emph{microscopic units} where external, macroscopic potentials vary slowly on the scale $\order(\eps)$. Here, one uses position and kinetic momentum operators in the \emph{adiabatic scaling} on $L^2(\R^d_x)$,\index{building block operators!adiabatic scaling} 
\begin{align}
	\Qe &\equiv \Qe_{\eps} := \eps \hat{x} \\
	\Pe^A &\equiv \Pe^A_{\eps,\lambda} := - i \nabla_x - \lambda A(\eps \hat{x}) 
	\notag 
	. 
\end{align}
This choice of scaling is used to analyze the Bloch electron subjected to a slowly varying electromagnetic field (see Chapter~\ref{magsapt}). Here, the relevant hamiltonian operator is 
\begin{align*}
	\hat{H} = \tfrac{1}{2} \bigl ( - i \nabla_x - \lambda A(\eps \hat{x}) \bigr )^2 + V_{\Gamma}(\hat{x}) + \Phi(\eps \hat{x}) 
\end{align*}
where $V_{\Gamma}$ is a periodic potential and $A$ and $\Phi$ are potentials associated to the external electromagnetic field. Obviously, these two choices of scalings are unitarily equivalent and thus, according to Proposition~\ref{asymptotics:scalings:prop:equivalence_Weyl_calculi} below, the two Weyl calculi are unitarily equivalent and lead to the \emph{same} magnetic Weyl product. 
%
\begin{lem}
	The adiabatic scaling and the usual scaling are related by the unitary $U_{\eps}$, $\bigl ( U_{\eps} \varphi \bigr )(x) := \eps^{- \nicefrac{d}{2}} \varphi \bigl ( \tfrac{x}{\eps} \bigr )$, $\varphi \in L^2(\R^d)$, \ie we have 
	\begin{align*}
		\tilde{\mathsf{Q}} &= U_{\eps} \, \Qe \, U_{\eps}^{-1} \\
		\tilde{\mathsf{P}}^A &= U_{\eps} \, \Pe^A \, U_{\eps}^{-1} 
		. 
	\end{align*}
\end{lem}
\begin{proof}
	Let $\varphi \in L^2(\R^d)$. Then we have for $\Qe$ 
	\begin{align*}
		(U_{\eps} \, \Qe \, U_{\eps}^{-1} U_{\eps} \varphi )(x) &= \bigl ( U_{\eps} \Qe \varphi \bigr ) (x) = \eps^{-\nicefrac{d}{2}} \, \bigl ( \Qe \varphi \bigr ) \bigl ( \tfrac{x}{\eps} \bigr ) 
		\\
		&= \eps^{-\nicefrac{d}{2}} \, \eps \tfrac{x}{\eps} \varphi \bigl ( \tfrac{x}{\eps} \bigr ) 
		= \tilde{\mathsf{Q}} \, \bigl ( U_{\eps} \varphi \bigr ) (x) . 
	\end{align*}
	Similarly, we get for the momentum operators 
	\begin{align*}
		(U_{\eps} \, \Pe^A \, U_{\eps}^{-1} U_{\eps} \varphi )(x) 
		&= \bigl ( U_{\eps} \Pi^A_{\eps,\lambda} \varphi \bigr ) (x) 
		= \eps^{-\nicefrac{d}{2}} \, \bigl ( \Pi^A_{\eps,\lambda} \varphi \bigr ) \bigl ( \tfrac{x}{\eps} \bigr ) 
		\displaybreak[2]
		\\
		&= \eps^{-\nicefrac{d}{2}} \, \bigl ( -i (\nabla_x \varphi) \bigl ( \tfrac{x}{\eps} \bigr ) - \lambda A \bigl (\eps \tfrac{x}{\eps} \bigr ) \, \varphi \bigl ( \tfrac{x}{\eps} \bigr ) \bigr ) 
		\displaybreak[2]
		\\
		&
		= \bigl ( - i \eps \nabla_x - \lambda A(\tilde{\mathsf{Q}}) \bigr ) \, \bigl ( U_{\eps} \varphi \bigr ) (x) 
		= \bigl ( \tilde{\mathsf{P}}^A U_{\eps} \varphi \bigr )(x)
		. 
	\end{align*}
	Hence the building block operators in the two scalings are unitarily equivalent. 
\end{proof}
Next, we will prove that unitarily equivalent building block observables lead to unitarily equivalent quantizations and the \emph{same} product formula. In view of Chapters~\ref{algebraicPOV} and \ref{psiDO_reloaded}, this is not at all surprising, the two `Weyl quantizations' can be seen as equivalent representations of the same fundamental $C^*$-algebra. 
\begin{prop}\label{asymptotics:scalings:prop:equivalence_Weyl_calculi}
	Let $0 < \eps \leq 1$ and $0 \leq \lambda \leq 1$. Assume $(\Qe , \Pe)$ and $(\Qe' , \Pe')$ are position and momentum operators on the separable Hilbert spaces $\Hil$ and $\Hil'$ that are selfadjoint, satisfy the commutation relations 
	\begin{align}
		i \bigl [ \Qe^{(\prime)}_l , \Qe^{(\prime)}_j \bigr ] = 0 
		&&
		i \bigl [ \Pe^{(\prime)}_l , \Qe^{(\prime)}_j \bigr ] = \eps \delta_{lj}
		&&
		i \bigl [ \Pe^{(\prime)}_l , \Pe^{(\prime)}_j \bigr ] = - \eps \lambda B_{lj}(\Qe^{(\prime)}) 
		, 
		\label{asymptotics:scalings:commutation_relations}
	\end{align}
	and are related by some unitary operator $\mathcal{U} : \Hil \longrightarrow \Hil'$ via 
	\begin{align}
		\Qe' &= \mathcal{U} \Qe \mathcal{U}^{-1} 
		\label{asymptotics:scalings:eqn:equivalence_operators} \\
		\Pe' &= \mathcal{U} \Pe \mathcal{U}^{-1} 
		\notag 
		. 
	\end{align}
	Then for all $h \in \Schwartz(\Pspace)$, the operator $\Op(h) := \frac{1}{(2\pi)^d} \int_{\Pspace} \dd X \, (\Fs h)(X) \, e^{- i \sigma(X,(\Qe,\Pe))}$ on $\Hil$ is unitarily equivalent to $\Op'(h) := \frac{1}{(2\pi)^d} \int_{\Pspace} \dd X \, (\Fs h)(X) \, e^{- i \sigma(X,(\Qe',\Pe'))}$ on $\Hil'$, 
	\begin{align*}
		\Op'(h) = \mathcal{U} \Op(h) \mathcal{U}^{-1} 
		. 
	\end{align*}
	Furthermore, the Weyl products which emulate the operator product are in fact identical: let $\Weyl^{(\prime)}$ be such that for any $f , g \in \Schwartz(\Pspace)$ be the product on $\Schwartz(\Pspace)$ such that 
	\begin{align*}
		\Op^{(\prime)}(f) \, \Op^{(\prime)}(g) =: \Op^{(\prime)}(f \Weyl^{(\prime)} g)
	\end{align*}
	holds. Then we have $\Weyl = \Weyl' : \Schwartz(\Pspace) \times \Schwartz(\Pspace) \longrightarrow \Schwartz(\Pspace)$. 
\end{prop}
\begin{remark}
	The unitary operator most relevant to this section is clearly the scaling operator $\mathcal{U}_{\eps} : L^2(\R^d) \longrightarrow L^2(\R^d)$. However, other unitary operators which immediately come to mind are changes of gauge where $\mathcal{U} = e^{+ i \lambda \chi(\Qe)}$ and the Fourier transform $\Fourier : L^2(\R^d) \longrightarrow L^2({\R^d}^*)$. If we quantize with respect to the latter, we get the momentum representation. 
\end{remark}
\begin{proof}
	First of all, by the spectral theorem, equation~\eqref{asymptotics:scalings:eqn:equivalence_operators} implies the unitary equivalence of the Weyl systems\index{Weyl system}, $e^{- i \sigma(X,(\Qe' , \Pe'))} = \mathcal{U} e^{- i \sigma(X,(\Qe,\Pe))} \mathcal{U}^{-1}$, and thus the unitary equivalence of $\Op$ and $\Op'$, 
	\begin{align*}
		\Op'(h) = \mathcal{U} \Op(h) \mathcal{U}^{-1} 
		&& 
		\forall h \in \Schwartz(\Pspace) 
		. 
	\end{align*}
	Now let $\Weyl^{(\prime)}$ be the Weyl product associated to $\Op^{(\prime)}$. Then for any $f , g \in \Schwartz(\Pspace)$, we have 
	\begin{align*}
		\Op'(f \Weyl' g) &= \Op'(f) \, \Op'(g) = \mathcal{U} \Op(f) \mathcal{U}^{-1} \mathcal{U} \Op(g) \mathcal{U}^{-1} 
		= \mathcal{U} \Op(f \Weyl g) \mathcal{U}^{-1} \\
		&= \Op'(f \Weyl g) 
		, 
	\end{align*}
	and the two Weyl products are in fact identical. 
\end{proof}
\begin{remark}
	This equivalence immediately extends to larger classes of functions and tempered distributions whenever one can make sense of the above relations. 
\end{remark}
This means once we know $\Weyl$ in one of the realizations, we know it in all other unitarily equivalent realizations. The same holds for asymptotic expansions of $\Weyl$ and we need not worry about the particular realization to pick. 
\medskip

\noindent
Now let us state some conventions we will use throughout this chapter. First of all, without loss of generality, \textbf{we will work in the adiabatic scaling.} Furthermore, for simplicity, we will use Einstein's summation convention, \ie \textbf{repeated indices in a product are always summed over from $1$ to $d$}. And we will always make the following assumptions on the magnetic field and associated vector potentials: 
\begin{assumption}\label{asymptotics:assumption:bounded_fields}
	We assume that the components of the magnetic field $B = \dd A$ and associated vector potentials $A$ satisfy $B_{kl} \in \mathcal{BC}^{\infty}(\R^d,\R)$ and $A_l \in \mathcal{C}^{\infty}_{\mathrm{pol}}(\R^d,\R)$, respectively, for all $1 \leq k,l \leq d$. 
\end{assumption}
\begin{remark}
	If a magnetic field $B$ satisfies the above assumption, it is always possible to choose a polynomially bounded vector potential, \eg we may use the transversal gauge (equation~\eqref{appendix:eqn:transversal_gauge}). It is also clear that if $B$ and $A$ satisfy this assumption, then so do the scaled field $B^{\eps,\lambda}(x) := \dd \Ael(x) = \eps \lambda B(\eps x)$ and scaled potential $\Ael(x) := \lambda A(\eps x)$. 
\end{remark}
%


\section{Magnetic Weyl quantization} 
\label{asymptotics:weyl_quantization}
Now that we have decided on our building block operators, we can introduce the associated Weyl system\index{Weyl system} 
\begin{align*}
	\WeylSysAel(X) := e^{-i \sigma(X,(\Qe,\Pe^A))} 
	, 
	&& 
	X \in \Pspace 
	, 
\end{align*}
which contains the small parameters $\eps$ and $\lambda$. A simple modification of the proof of Lemma~\ref{magWQ:magnetic_weyl_calculus:lem:action_Weyl_system} shows that it acts on $u \in L^2(\R^d_x)$ via 
\begin{align*}
	\bigl ( \WeylSysAel(Y) u \bigr )(x) &= e^{- i \frac{\lambda}{\eps} \Gamma^A([\eps x,\eps x + \eps y])} e^{- i \eta \cdot \eps (x + \nicefrac{y}{2})} \, u(x + y) 
	\\
	&
	=: e^{- i \lambda \Gamma_{\eps}^A([x,x + y])} e^{- i \eta \cdot (x + \frac{\eps}{2} y)} \, u(x + y) 
	. 
\end{align*}
Then, as usual, we define the Weyl quantization of $f \in \Schwartz(\Pspace)$ in the strong sense as 
\begin{align*}
	\Op^A_{\eps,\lambda}(f) u \equiv \Op^A(f) u := \frac{1}{(2\pi)^d} \int_{\Pspace} \dd X \, (\Fs f)(X) \, \WeylSysAel(X) u 
	&&
	\forall u \in \Schwartz(\R^d) 
	. 
\end{align*}
Revisiting the arguments found in Chapter~\ref{magWQ:magnetic_weyl_calculus:covariant_quantization}, we find that the action of $\Op^A_{\eps,\lambda}(f)$ on $u \in \Schwartz(\R^d)$ is explicitly given by 
\begin{align}
	\bigl ( \Op^A_{\eps,\lambda}(f) u \bigr )(x) &= \frac{1}{(2\pi)^d} \int_{\R^d} \dd y \int_{{\R^d}^*} \dd \eta \, e^{- i (y - x) \cdot \eta} \, e^{- i \lambda \Gamma^A_{\eps}([x,y])} \, f \bigl ( \tfrac{\eps}{2} (x+y) , \eta \bigr ) \, u(y) 
	\\
	&= \frac{1}{(2\pi)^{\nicefrac{d}{2}}} \int_{\R^d} \dd y \, e^{- i \lambda \Gamma^A_{\eps}([x,y])} \, (\Fourier_2 f) \bigl ( \tfrac{\eps}{2} (x+y) , y - x \bigr ) \, u(y) 
	\notag \\
	&
	=: \frac{1}{(2\pi)^{\nicefrac{d}{2}}} \int_{\R^d} \dd y \, K_f^A(x,y) \, u(y) 
	\notag 
	. 
\end{align}
All results on magnetic Weyl calculus found in the literature (see Chapter~\ref{magWQ:important_results} for a small selection) hold in this context as well since proofs carry over after obvious modifications due to the presence of $\eps$ and $\lambda$. 

For convenience, we give the composition law of the magnetic Weyl system, because it enters into the proof of the product formula: for any $X,Y \in \Pspace$, 
\begin{align}
	\WeylSysAel(X) \WeylSysAel(Y) &= e^{i \frac{\eps}{2} \sigma(X,Y)} \, \oBel (\Qe;x,y) \, \WeylSysAel(X+Y) 
\end{align}
holds where $\oBel(q;x,y) := e^{- i \frac{\lambda}{\eps} \Gamma^B(\sexpval{q,q+\eps x , q + \eps x + \eps y})}$ is the exponential of the scaled magnetic flux. 


\section{Semiclassical symbols and precision} 
\label{asymptotics:semiclSymbolsPrecision}

The Hörmander symbol classes $\Hoerr{m}$ are Fréchet spaces whose topology can be defined by the usual family of seminorms 
\begin{align*}
	\norm{f}_{m a \alpha} := \sup_{(x,\xi) \in \PSpace} \expval{\xi}^{-m + \sabs{\alpha} \rho} \babs{\partial_x^a \partial_{\xi}^{\alpha} f(x,\xi)}  
	, 
	&& 
	a , \alpha \in \N_0^d 
	. 
\end{align*}
One important notion is that of a \emph{semiclassical symbol} \cite{PST:sapt:2002}, \ie it is a symbol which admits an expansion in $\eps$ and $\lambda$ which is in some sense uniform. 
\begin{defn}[Semiclassical two-parameter symbol\index{semiclassical symbol!two-parameter}]\label{asymptotics:defn:2ParameterSemiSymbol}
	A map 
	\begin{align*}
		f : [0,\eps_0) \times [0,\lambda_0) \longrightarrow \Hoerr{m}
		, \; 
		(\eps,\lambda) \mapsto f^{\eps,\lambda} 
	\end{align*}
	is called \emph{semiclassical two-parameter symbol} of order $m$ with weight $\rho \in [0,1]$, \ie $f \in \SemiHoerr{m}$, if there exists a sequence $\{ f_{n,k} \}_{n,k \in \N_0}$, $f_{n,k} \in \Hoerr{m - (n+k)\rho}$ for all $n,k \in \N_0$, such that 
	\begin{align*}
		f^{\eps,\lambda} - \sum_{l = 0}^{N} \sum_{n + k = l} \eps^n \lambda^k f_{n,k} \in \Hoerr{m - (N + 1)\rho} && \forall N \in \N_0 
	\end{align*}
	uniformly in the following sense: for each $j \in \N_0$ there exists a constant $C_{N,m,j} > 0$ (independent of $\eps$ and $\lambda$) such that 
	\begin{align*}
		\biggl \lVert f^{\eps,\lambda} - \sum_{l = 0}^{N} \sum_{n + k = l} \eps^n \lambda^k f_{n,k} \biggr \rVert_{m,j} < C_{N,m,j} \, \max \{ \eps , \lambda \}^{N+1} . 
	\end{align*}
	holds for all $\eps \in [0,\eps_0)$ and $\lambda \in [0,\lambda_0)$. 
\end{defn}
Since $\eps$ and $\lambda$ vary independently, we also have to introduce a more sophisticated concept of precision. This is a technicality, but a definition is necessary to prove that expanding $f \magWel g$ first with respect to $\eps$ and then $\lambda$ yields the same asymptotics as when the product is expanded with respect to $\lambda$ and then with respect to $\eps$ (Theorem~\ref{asymptotics:thm:lambdaEpsExpansion}). If there were only one small parameter, say $\eps$, then $f - g = \mathcal{O}(\eps^n)$ for symbols $f,g \in \Hoerr{m}$ implies two things: (i) the difference between $f$ and $g$ is `numerically small' and (ii) we have associated a symbol class $\Hoerr{m - n \rho}$ to the `number' $\eps^n$. In case of two independent parameters, such a simple concept will not do and we have to introduce an association between a \emph{third} number $\eprec \ll 1$ and a certain symbol class. Although it seems artificial at first to introduce yet another small parameter, in physical applications, this is quite natural: say, we are interested in the dynamics generated by a two-parameter symbol $H^{\eps,\lambda}$ on times of order $\mathcal{O}(\nicefrac{1}{\eprec})$, \ie $e^{- i \frac{t}{\eprec} \, H^{\eps,\lambda}}$. Then we need to include all terms in our expansion for which $\eps^n \lambda^k \leq \eprec$. Even if we choose $\eprec = \eps$, for instance, we still cannot avoid this abstract definition as $\lambda$ is independent of $\eps$. 
\begin{defn}[Precision $\ordereprec$\index{precision $\ordereprec$}]\label{asymptotics:defn:precision}
	Let $\eps \ll 1$, $\lambda \ll 1$. For $\eprec \ll 1$, we define critical exponents $n_c,k_c,N \in \N_0$ such that 
	\begin{align*}
		\eps^{n_c + 1} < \eprec \leq \eps^{n_c}
		, 
		&& 
		\lambda^{k_c + 1} < \eprec \leq \lambda^{k_c} 
	\end{align*}
	and $N \equiv N(\eps,\lambda,\eprec) := \max \{ n_c , k_c \}$ as maximum of these two critical exponents. We say that a finite resummation $\sum_{n = 0}^{N_{\eps}} \sum_{k = 0}^{N_{\lambda}} \eps^n \lambda^k \, f_{n,k}$ of a semiclassical symbol $f^{\eps,\lambda} \in \SemiHoerr{m}$ is $\ordereprec$-close, 
	\begin{align*}
		f^{\eps,\lambda} - \sum_{n = 0}^{N_{\eps}} \sum_{k = 0}^{N_{\lambda}} \eps^n \lambda^k \, f_{n,k} = \ordereprec , 
	\end{align*}
	iff $f^{\eps,\lambda} - \sum_{n = 0}^{N_{\eps}} \sum_{k = 0}^{N_{\lambda}} \eps^n \lambda^k f_{n,k} \in \Hoerr{m-(N+1)\rho}$ and $N_{\eps} , N_{\lambda} \geq N$. 
\end{defn}
%


\section{The Magnetic Wigner transform} 
\label{asymptotics:mag_Wigner_transform}
The Wigner transform plays a central role because it can be used to relate states (density operators) to pseudo-probability measures on phase space. We will need it to show the equivalence of two integral formulas for the magnetic Weyl product $\magWel$. 
\begin{defn}[Magnetic Wigner transform\index{Wigner transform}]
	For any $\varphi,\psi \in \Schwartz(\R^d)$, the magnetic Wigner transform  $\mathcal{W}^A(\varphi,\psi)$\index{Wigner transform!magnetic} is defined as 
	\begin{align*}
		\WignerTrafoel(\varphi,\psi)(X) := \eps^d \, \bigl ( \Fs \bscpro{\varphi}{\WeylSysAel(\cdot) \psi} \bigr )(-X) . 
	\end{align*}
\end{defn}
\begin{lem}\label{magBel:WignerTransformExpectation}
	The Wigner transform $\WignerTrafoel(\varphi,\psi)$ with respect to $\varphi,\psi \in \Schwartz(\R^d)$ is given by 
	\begin{align*}
		\WignerTrafoel(\varphi,\psi)(X) &= \int_{\R^d} \dd y \, e^{- i y \cdot \xi} e^{-i \lambda \GAe([\nicefrac{x}{\eps} - \nicefrac{y}{2} ,\nicefrac{x}{\eps} + \nicefrac{y}{2} ])} \, \varphi^{\ast} \bigl ( \tfrac{x}{\eps} - \tfrac{y}{2} \bigr ) \, \psi \bigl ( \tfrac{x}{\eps} + \tfrac{y}{2} \bigr ) 
	\end{align*}
	and maps $\Schwartz(\R^d) \times \Schwartz(\R^d)$ unitarily onto $\Schwartz(\Pspace)$. 
\end{lem}
\begin{proof}
	Formally, the result follows from direct calculation. The existence is part of the second claim, $\WignerTrafoel(\varphi,\psi) \in \Schwartz(\Pspace)$ follows from $e^{- i \lambda \GAe([\nicefrac{x}{\eps} - \nicefrac{y}{2} , \nicefrac{x}{\eps} + \nicefrac{y}{2}])} \,  \, \varphi^{\ast} \bigl ( \tfrac{x}{\eps} - \tfrac{y}{2} \bigr ) \, \psi \bigl ( \tfrac{x}{\eps} + \tfrac{y}{2} \bigr ) \in \Schwartz(\R^d \times \R^d)$ and the fact that the partial Fourier transformation is a unitary on $\Schwartz(\Pspace)$. 
\end{proof}
\begin{remark}
	The Wigner transform can be easily extended to a map from $L^2(\R^d \times \R^d)$ into $L^2(\Pspace) \cap \mathcal{C}_{\infty}(\Pspace)$ where $\mathcal{C}_{\infty}(\Pspace)$ is the space of continuous functions which decay at $\infty$. For more details, see \cite[Proposition~1.92]{Folland:harmonic_analysis_hase_space:1989}, for example. 
\end{remark}
Let $\mathcal{C}^{\infty}_{\mathrm{pol} \, u}(\Pspace)$ be the space of smooth functions with uniform polynomial growth at infinity, \ie for each $f \in \mathcal{C}^{\infty}_{\mathrm{pol} \, u}(\Pspace)$ we can find $m \in \R$, $m \geq 0$, such that for all multiindices $a , \alpha \in \N_0^d$ there is a $C_{a \alpha} > 0$ with 
\begin{align*}
	\babs{\partial_x^{a} \partial_{\xi}^{\alpha} f(x,\xi)} < C_{a \alpha} \expval{\xi}^m 
	, 
	&& 
	\forall (x,\xi) \in \Pspace
	. 
\end{align*}
\begin{lem}
	For $\varphi,\psi \in \Schwartz(\R^d)$ and $f \in \mathcal{C}^{\infty}_{\mathrm{pol} \, u}(\Pspace) \subset \Schwartz'(\Pspace)$ we have 
	\begin{align*}
		\bscpro{\varphi}{\OpAel(f) \psi} = \frac{1}{(2\pi)^d} \int_{\Pspace} \dd X \, f(X) \, \WignerTrafoel(\varphi,\psi)(X) 
		. 
	\end{align*}
\end{lem}
\begin{proof}
	Since $f \in \mathcal{C}^{\infty}_{\mathrm{pol} \, u}(\Pspace) \subseteq \Schwartz'(\Pspace)$, it is in the magnetic Moyal algebra $\mathcal{M}^B(\Pspace)$ defined in \cite[Section~V.D.]{Mantoiu_Purice:magnetic_Weyl_calculus:2004} and thus its quantization is a continuous operator $\Schwartz(\R^d) \longrightarrow \Schwartz(\R^d)$. Hence, the integral exists and we get the claim by direct computation. 
\end{proof}
The Wigner transform also leads to a `magnetic dequantization' -- once we know the operator kernel, we can reconstruct the distribution. We do not strive for full generality here. In particular, unless the operator has special properties, we cannot conclude that $f$ is in any Hörmander class. More sophisticated techniques are needed, \eg a Beals-type criterion \cite{Iftimie_Mantoiu_Purice:commutator_criteria:2008}. 
\begin{lem}\label{asymptotics:mag_Wigner_transform:inverse_Weyl_quantization}
	Assume $B$ and $A$ satisfy Assumption~\ref{asymptotics:assumption:bounded_fields} and $T \in \mathcal{B} \bigl ( L^2(\R^d) \bigr )$ is a bounded linear operator whose operator kernel $K_T$ is in $\Schwartz(\R^d \times \R^d)$. Then the \emph{inverse magnetic quantization} is given by 
	\begin{align}
		{\OpAel}^{-1}(T)(X) :=& \WignerTrafoel K_{T}(X) 
		\notag \\
		=& \int_{\R^d} \dd y \, e^{-i y \cdot \xi} \, e^{-i \lambda \GAe([\nicefrac{x}{\eps} - \nicefrac{y}{2} ,\nicefrac{x}{\eps} + \nicefrac{y}{2} ])} \,  K_{T} \bigl ( \tfrac{x}{\eps} - \tfrac{y}{2} , \tfrac{x}{\eps} + \tfrac{y}{2} \bigr ) 
		. 
	\end{align}
	This formula extends to operators with distributional kernels $K_T \in \Schwartz'(\R^d \times \R^d)$, \ie the kernels associated to continuous maps $\Schwartz(\R^d) \longrightarrow \Schwartz'(\R^d)$. 
\end{lem}
\begin{proof}
	If $T = \OpAel(f_T)$ is the magnetic quantization of $f_T \in \Schwartz(\Pspace)$, then $\WignerTrafoel K_T = f_T \in \Schwartz(\Pspace)$ follows from direct calculation, using the explicit form of the Wigner transform, Lemma~~\ref{magBel:WignerTransformExpectation}. Similarly, we confirm that $T = \OpAel \bigl ( \WignerTrafoel(K_T) \bigr )$ holds and $\WignerTrafoel K_T \in \Schwartz(\Pspace)$ follows from $K_T \in \Schwartz(\R^d \times \R^d)$. 
	
	If the kernel of $T$ is a tempered distribution, then we can extend the formulas for $\OpAel$ and $\WignerTrafoel$ to $\Schwartz'(\Pspace)$: Fourier transform, multiplication by a phase factor whose phase function is of tempered growth and a linear change of variables can all be extended to $\Schwartz'(\Pspace)$ and thus it makes sense to write $\WignerTrafoel K_T$ after a suitable reinterpretation. Then $\WignerTrafoel K_T = f_T \in \Schwartz'(\Pspace)$ is such that $\OpAel(f_T) = T : \Schwartz(\R^d) \longrightarrow \Schwartz'(\R^d)$. 
\end{proof}
%


\section{Asymptotic expansions of the product} 
\label{asymptotics:expansions}
It turns out that the integral formula for the product found in \cite{Mantoiu_Purice:magnetic_Weyl_calculus:2004,Iftimie_Mantiou_Purice:magnetic_psido:2006} is not amenable to the derivation of an asymptotic expansion in $\eps$ and $\lambda$. Although an asymptotic expansion for $\eps = 1 = \lambda$ has been derived in \cite{Iftimie_Mantiou_Purice:magnetic_psido:2006}, \emph{calculating} each term has proven to be very tedious and it is not obvious how to collect terms of the same power in $\eps$ and $\lambda$. Thus, we will use an equivalent formula for the magnetic Weyl product. From this, we derive closed formulas for the $(n,k)$ term by expanding the `twister' of the convolution. This result is an extension of Theorem~\ref{magWQ:magnetic_weyl_calculus:thm:equivalence_product} and has first been proven by Iftimie, Măntoiu and Purice for $\eps = 1$ and $\lambda = 1$. We only mention it here to have a formula with $\eps$ and $\lambda$ in the proper places that serves as a starting point for the derivation of the asymptotic expansion. 
\begin{thm}[\cite{Iftimie_Mantiou_Purice:magnetic_psido:2006}]\label{asymptotics:thm:equivalenceProduct}
	Assume the magnetic field $B$ satisfies Assumption~\ref{asymptotics:assumption:bounded_fields}. Then for two symbols $f \in \Hoerr{m_1}$ and $g \in \Hoerr{m_2}$, the magnetic Weyl product $f \magWel g$\index{Weyl product!magnetic} is in symbol class $\Hoerr{m_1 + m_2}$ and given by the oscillatory integral 
	\begin{align}
		(f \magWel g)(X) 
		&= \frac{1}{(2 \pi)^{2d}} \int_{\Pspace} \dd Y \int_{\Pspace} \dd Z \, e^{+i \sigma(X,Y+Z)} \, e^{i \frac{\eps}{2} \, \sigma(Y,Z)} \, 
		\cdot \notag \\
		&\qquad \qquad \qquad \qquad \cdot 
		\oBel \bigl ( x-\tfrac{\eps}{2}(y+z),x + \tfrac{\eps}{2}(y-z) , x+\tfrac{\eps}{2}(y+z) \bigr ) \, 
		\cdot \notag \\
		&\qquad \qquad \qquad \qquad \cdot 
		\bigl ( \Fs f \bigr )(Y) \, \bigl ( \Fs g \bigr )(Z) \label{asymptotics:expansions:eqn:Fourier_form_magnetic_Weyl_product} \\ 
		&= \frac{1}{(\pi\eps)^{2d}} \, \int_{\Pspace} \dd \tilde{Y} \int_{\Pspace} \dd \tilde{Z} \, e^{- i \frac{2}{\eps} \, \sigma(\tilde{Y} - X , \tilde{Z} - X)} \, 
		\cdot \notag \\
		&\qquad \qquad \qquad \qquad \cdot 
		\oBel \bigl ( x - \tilde{y} + \tilde{z} , -x + \tilde{y} + \tilde{z} , x + \tilde{y} + \tilde{z} \bigr ) \,  f(\tilde{Y}) \, g(\tilde{Z}) \notag 
	\end{align}
	where $\oBel(x,y,z) := e^{- i \frac{\lambda}{\eps} \Gamma^B(\sexpval{x,y,z})}$ is the exponential of the magnetic flux through the triangle with corners $x$, $y$ and $z$. 
\end{thm}
\begin{proof}
	The Weyl product is defined implicitly by 
	\begin{align*}
		\OpAel(f) \, \OpAel(g) &=: \OpAel(f \magWel g) 
	\end{align*}
	and its quantization maps $\Schwartz(\R^d)$ to itself \cite[Proposition~21]{Mantoiu_Purice:magnetic_Weyl_calculus:2004}. Combined with Theorem~\ref{asymptotics:mag_Wigner_transform:inverse_Weyl_quantization}, this immediately yields 
	\begin{align*}
		(f \magWel g)(X) = \WignerTrafoel \bigl ( K_{\OpAel(f) \, \OpAel(g)} \bigr )(X) 
	\end{align*}
	where $K_{\OpAel(f) \, \OpAel(g)}$ is the kernel of $\OpAel(f) \, \OpAel(g)$. Here, we have chosen a vector potential $A$ which is associated to $B$ that also satisfies Assumption~\ref{asymptotics:assumption:bounded_fields}. Although it is \emph{a priori} not clear that there must exist a \emph{symbol} $f \magWel g$, we will start with formal calculations and then use Corollary~\ref{appendix:asymptotics:existence_osc_int:lem:remainder} to show that integral~\eqref{asymptotics:expansions:eqn:Fourier_form_magnetic_Weyl_product} exists and is in the correct symbol class.
	
	For the computation of the product formula, one has to revisit the proof of Theorem~\ref{magWQ:magnetic_weyl_calculus:thm:equivalence_product} and add $\eps$ and $\lambda$ to the right places. One then arrives at 
	%
	\begin{align*}
		(f &\magWel g)(X) = 
		\\
		&= \frac{1}{(2\pi)^{2d}} \int_{\Pspace} \dd Z \, \int_{\Pspace} \dd Y \, \bigl ( \Fs f \bigr )(Y) \, \bigl ( \Fs g \bigr )(Z-Y) \;  e^{i \frac{\eps}{2} \, \sigma(Y,Z)} \, L_{\eps,\lambda}(y,Z;X) 
		\displaybreak[2]
		\notag \\ 
		&= \frac{1}{(2\pi)^{2d}} \int_{\Pspace} \dd Y \, \int_{\Pspace} \dd Z \, e^{i \sigma(X,Y+Z)} \,  e^{i \frac{\eps}{2} \, \sigma(Y,Z)} \, 
		\cdot \notag \\
		&\qquad \qquad \cdot 
		\oBel \bigl ( x - \tfrac{\eps}{2} (y+z),x + \tfrac{\eps}{2} (y-z) , x + \tfrac{\eps}{2} (y+z) \bigr ) 
		\bigl ( \Fs f \bigr )(Y) \, \bigl ( \Fs g \bigr )(Z)
		. 
	\end{align*}
	The integral on the right-hand side satisfies the assumptions of Lemma~\ref{appendix:asymptotics:existence_osc_int:lem:remainder} with $\tau = 1 = \tau'$ (keeping in mind that $\oBel$ satisfies the assumptions on $G_{\tau'}$ by Lemma~\ref{appendix:asymptotics:cor:properties_flux}). Thus, the oscillatory integral exists and is in symbol class $\Hoerr{m_1 + m_2}$. 
\end{proof}
Equation~\eqref{asymptotics:expansions:eqn:Fourier_form_magnetic_Weyl_product} is the starting point for an asymptotic expansion of the product in small parameters $\eps$ \emph{and} $\lambda$. 
\begin{thm}[Asymptotic expansion of the magnetic Moyal product\index{Weyl product!asymptotic expansion in $\eps$ and $\lambda$}]\label{asymptotics:thm:asymptotic_expansion}
	Assume $B$ is a magnetic field whose components are $\mathcal{BC}^{\infty}$ functions and $f \in \Hoerr{m_1}$ as well as $g \in \Hoerr{m_2}$. Then the magnetic Moyal product can be expanded asymptotically in $\eps \ll 1$ and $\lambda \ll 1$: for every $\eprec \ll 1$ we can choose $N \equiv N(\eprec,\eps,\lambda) \in \N_0$ such that 
	\begin{align}
		f \magWel g &= \sum_{n = 0}^{N} \sum_{k = 0}^n \eps^n \lambda^k \, (f \magWel g)_{(n,k)} + \tilde{R}_N 
		\label{intro:epsLambdaExpansion}
	\end{align}
	where the $(n,k)$ term $(f \magWel g)_{(n,k)}$ is in symbol class $\Hoerr{m_1 + m_2 - (n+k) \rho}$ and we have explicit control over the remainder: $\tilde{R}_N$ as given by equation \eqref{asymptotics:eqn:total_remainder} is numerically small and in the correct symbol class, $\Hoerr{m_1 + m_2 - (N+1)  \rho}$, \ie it is of order $\ordereprec$ in the sense of Definition \ref{asymptotics:defn:precision}. The $(n,k)$ term of the expansion, 
	\begin{align}
		(f \magWel g)_{(n,k)}(X) &= \sum_{\substack{k_0 + \sum_{j = 1}^{n} j k_j = n \notag \\
		\sum_{j=1}^{n} k_j = k}} \frac{i^{k + k_0}}{k_0! \, k_1 ! \cdots k_{n}!} \, \mathcal{L}_0^{k_0} \bigl ( (\partial_{\eta},\partial_y) , (\partial_{\zeta},\partial_z) \bigr ) 
		\cdot \\
		&\qquad \qquad \qquad \cdot 
		 \prod_{j = 1}^{n} \mathcal{L}_j^{k_j}(x,-i \partial_{\eta},-i \partial_{\zeta}) \bigr ) f(Y) g(Z) \Bigl . \Bigr \rvert_{Y = X = Z} \notag 
		, 
	\end{align}
	is defined in terms of a family of differential operators $\mathcal{L}_j$, $j \in \N_0$, 
	\begin{align}
		\mathcal{L}_0(Y,Z) &:= \tfrac{1}{2} \sigma(Y,Z) = \tfrac{1}{2} \bigl ( \eta \cdot z - y \cdot \zeta \bigr ) 
		\displaybreak[2]
		\\
		\mathcal{L}_j(x,y,z) &:= - \frac{1}{j!} \, \sum_{m_1 , \ldots , m_{j-1} = 1}^d \partial_{x_{m_1}} \cdots \partial_{x_{m_{j-1}}} B_{kl}(x) \, y_k \, z_l \, \left ( - \frac{1}{2} \right )^{j+1} \frac{1}{(j+1)^2} \cdot \notag \\
		&\negmedspace \negmedspace \negmedspace \negmedspace \negmedspace \negmedspace \negmedspace \negmedspace \negmedspace \negmedspace \negmedspace \negmedspace \negmedspace \negmedspace 
		\cdot \sum_{c = 1}^j \noverk{j+1}{c} \bigl ( ( 1 - (-1)^{j+1} ) c - ( 1 - (-1)^{c} ) (j+1) \bigr ) \, y_{m_1} \cdots y_{m_{c-1}} z_{m_{c}} \cdots z_{m_{j-1}} \notag \\
		&=: - \sum_{\abs{\alpha} + \abs{\beta} = j-1} C_{j,\alpha,\beta} \, \partial_x^{\alpha} \partial_x^{\beta} B_{kl}(x) \, y_k z_l \, y^{\alpha} \, z^{\beta} 
		. 
	\end{align}
\end{thm}
To obtain an asymptotic expansion, we adapt an idea by Folland to the present case \cite[p 108 f.]{Folland:harmonic_analysis_hase_space:1989}: we expand the exponential of the \emph{twister} 
\begin{align*}
	e^{i \frac{\eps}{2} \sigma(Y,Z) - i \lambda \gamma^B_{\eps}(x,y,z)} &= e^{i T_{\eps,\lambda}(x,Y,Z)} \notag 
	\\
	&\asymp \sum_{n = 0}^{\infty} \sum_{k = 0}^{\infty} \eps^n \lambda^k \sum C_{n,k,a,\alpha,b,\beta}(x) \, y^{a} \eta^{\alpha} \, z^{b} \zeta^{\beta} \notag 
\end{align*}
as a polynomial in $y$, $\eta$, $z$ and $\zeta$ with coefficients $C_{n,k,a,\alpha,b,\beta} \in \mathcal{BC}^{\infty}(\R^d)$ that are bounded functions in $x$ with bounded derivatives to all orders. Here, $\gamma^B_{\eps}$ stands for the scaled magnetic flux 
\begin{align}
	\gamma^B_{\eps}(x,y,z) := \tfrac{1}{\eps} \Gamma^B \bigl ( \bexpval{x-\tfrac{\eps}{2}(y+z),x + \tfrac{\eps}{2}(y-z) , x+\tfrac{\eps}{2}(y+z)} \bigr ) 
	. 
	\label{asymptotics:expansions:eqn:gamma_B_eps}
\end{align}
Then we can rewrite equation \eqref{asymptotics:expansions:eqn:Fourier_form_magnetic_Weyl_product} as a convolution of derivatives of $f$ and $g$. Furthermore, we can show that there are always sufficiently many derivatives with respect to momenta so that each of the terms has the correct decay properties. 

The difficult part of the proof is to show the existence of certain oscillatory integrals. To clean up the presentation of the proof, we have moved these parts to Appendix \ref{appendix:asymptotics:existence_osc_int}. For simplicity, we also introduce the following nomenclature: 
\begin{defn}[Number of $q$s and $p$s\index{numbers of $p$s and $q$s}]
	Let $B \in \mathcal{BC}^{\infty} \bigl ( \R^d_x , \mathcal{C}^{\infty}_{\mathrm{pol}}(\Pspace_Y \times \Pspace_Z) \bigr )$ be a function which can be decomposed into a finite sum of the form 
	\begin{align*}
		B(x,Y,Z) &= \sum_{\substack{\sabs{a} + \sabs{b} = n\\
		\sabs{\alpha} + \sabs{\beta} = k}} b_{a \alpha b \beta}(x,Y,Z) \, y^{a} \, \eta^{\alpha} \, z^{b} \, \zeta^{\beta} 
	\end{align*}
	where all $b_{a \alpha b \beta}$ smooth \emph{bounded} functions that depend on the multiindices $a , \alpha , b , \beta \in \N_0^d$. We then say that $B$ has $n$ $q$s (total number of factors in $y$ and $z$) and $k$ $p$s (total number of factors in $\eta$ and $\zeta$). 
\end{defn}
In the appendix we show how to convert $q$s into derivatives with respect to \emph{momentum} and $p$s into derivatives with respect to \emph{position}. Monomials of $x$ and $\xi$ multiplied with the symplectic Fourier transform of a Schwarz function $\varphi \in \Schwartz(\Pspace)$ can be written as the symplectic Fourier transform of derivatives of $\varphi$ in $\xi$ and $x$: 
\begin{align*}
	x^{a} \xi^{\alpha} (\Fs \varphi)(X) &= \Fs \bigl ( (-i \partial_{{\xi}})^{a} (i \partial_{{x}})^{\alpha} \varphi \bigr ) (X)
\end{align*}
This manipulation can be made rigorous for symbols of class $m$ with weight $\rho$. We see that derivatives with respect to momentum \emph{improve} decay by $\rho$ while those with respect to position do not alter the decay. In this sense, the decay properties of the integrals are determined by the number of $q$s and $p$s. 
\begin{proof}[Theorem~\ref{asymptotics:thm:asymptotic_expansion}]
	Let $\eprec \ll 1$. Then we take $N \equiv N(\eprec,\eps,\lambda) \in \N_0$ to be as in the first part of Definition \ref{asymptotics:defn:precision}, \ie $\eps^{N+1} < \eprec$ and $\lambda^{N+1} < \eprec$ hold. We will show that $f \magWel g - \sum_{n = 0}^N \sum_{k = 0}^n \eps^n \lambda^k \, (f \magWel g)_{(n,k)} = \ordereprec$. 
	\medskip
	
	\noindent
	\textbf{Step 1: Formal expansion of the twister. }
	We expand the exponential of the twister $e^{i \frac{\eps}{2} \sigma(Y,Z)} e^{- i \lambda \gBe(x,y,z)} = e^{i T_{\eps,\lambda}(x,Y,Z)}$ up to the $N$th term, 
	\begin{align*}
		e^{i T_{\eps,\lambda}(x,Y,Z)} &= \sum_{n = 0}^N \frac{i^n}{n!} \bigl ( T_{\eps,\lambda}(x,Y,Z) \bigr )^n + R_N(x,Y,Z) 
		. 
	\end{align*}
	The remainder 
	\begin{align}
		R_N(x,Y,Z) :&= \frac{1}{N!} \int_0^1 \dd \tau \, (1 - \tau)^N \partial_{\tau}^{N+1} e^{\tau u} \big \vert_{u = i T_{\eps,\lambda}(x,Y,Z)} \notag \\
		&= \frac{i^{N+1}}{N!} \bigl ( T_{\eps,\lambda}(x,Y,Z) \bigr )^{N+1} \int_0^1 \dd \tau \, (1 - \tau)^N \, e^{i \tau T_{\eps,\lambda}(x,Y,Z)} 
		\label{asymptotics:eqn:remainder1}
	\end{align}
	is treated in Step 3, right now we are only concerned with the first term. If we plug in the asymptotic expansion of the flux $\gBe$ derived in Lemma~\ref{appendix:asymptotics:expansion_twister:lem:expansion_twister} up to $N'$th order with $N' \geq N$, then we obtain 
	\begin{align}
		\bigl ( T_{\eps,\lambda}&(x,Y,Z) \bigr )^n = 
		\Bigl ( \tfrac{\eps}{2} \sigma(Y,Z) + \lambda \mbox{$\sum_{n' = 1}^{N'}$} \eps^{n'} \mathcal{L}_{n'}(x,y,z) + \lambda R_{N'}[\gBe](x,y,z) \Bigr )^n \notag \\
		&= \sum_{l = 0}^n \noverk{n}{l} \negthinspace \Bigl ( \tfrac{\eps}{2} \sigma(Y,Z) + \lambda \mbox{$\sum_{n' = 1}^{N'}$} \eps^{n'} \mathcal{L}_{n'}(x,y,z) \Bigr )^{n - l} \negmedspace \negthinspace \bigl ( \lambda R_{N'}[\gBe](x,y,z) \bigr )^l \notag \\
		&=: \Bigl ( \tfrac{\eps}{2} \sigma(Y,Z) + \lambda \mbox{$\sum_{n' = 1}^{N'}$} \eps^{n'} \mathcal{L}_{n'}(x,y,z) \Bigr )^{n} + R_{N' \, n}[T_{\eps,\lambda}](x,Y,Z) 
		. 
		\label{asymptotics:eqn:remainder2}
	\end{align}
	Again, we focus on the first term of the expansion and treat the remainder separately in Step~3: 
	\begin{align*}
		\Bigl ( \tfrac{\eps}{2} \sigma(Y,Z) + \lambda \mbox{$\sum_{n' = 1}^{N'}$} \eps^{n'} &\mathcal{L}_{n'}(x,y,z) \Bigr )^n = \sum_{k = 0}^n \sum_{\sum_{j = 1}^{N'} k_j = k} \eps^{(n - k) + \sum_{j = 1}^{N'} j k_j} \lambda^k 
		\cdot \\
		&\cdot 
		\frac{n!}{(n - k)! \, k_1 ! \cdots k_{N'}!} \Bigl ( \tfrac{1}{2} \sigma(Y,Z) \Bigr )^{n-k} \, \prod_{j = 1}^{N'} \mathcal{L}_j^{k_j}(x,y,z) 
	\end{align*}
	Now we define $\mathcal{L}_0(Y,Z) := \tfrac{1}{2} \sigma(Y,Z)$ to clean up the presentation, include the sum over $n$ again and sort by powers of $\eps$ and $\lambda$, 
	\begin{align*}
		\sum_{n = 0}^N &\frac{i^n}{n!} \Bigl ( \tfrac{\eps}{2} \sigma(Y,Z) + \lambda \mbox{$\sum_{n' = 1}^{N'}$} \eps^{n'} \mathcal{L}_{n'}(x,y,z) \Bigr )^n 
		= \\
		&= \sum_{n = 0}^N \frac{i^n}{n!} \sum_{\sum_{j = 0}^{N'} k_j = n} \eps^{k_0 + \sum_{j = 1}^{N'} j k_j} \lambda^{n - k_0}  \, \frac{n!}{k_0! \, k_1 ! \cdots k_{N'}!} \mathcal{L}_0^{k_0}(Y,Z) \prod_{j = 1}^{N'} \mathcal{L}_j^{k_j}(x,Y,Z) 
		\displaybreak[2]
		\\
		&= \sum_{n = 0}^{N \, N'} \sum_{k = 0}^n \eps^{n} \lambda^k 
		\sum_{\substack{k_0 + \sum_{j = 1}^{N'} j k_j = n \\
		\sum_{j=1}^{N'} k_j = k}} \frac{i^{k + k_0}}{k_0! \, k_1 ! \cdots k_{N'}!} \, \mathcal{L}_0^{k_0}(Y,Z) \prod_{j = 1}^{N'} \mathcal{L}_j^{k_j}(x,Y,Z) 
		\\
		&
		=: \sum_{n = 0}^{N \, N'} \sum_{k = 0}^n \eps^{n} \lambda^k \mathcal{T}_{n,k}(x,Y,Z) 
		. 
	\end{align*}
	\textbf{Step 2: Existence of the $(n,k)$ term. }
	The properties of the $(n,k)$th term of the product 
	\begin{align}
		(f \magWel g)_{(n,k)}(X) &= \frac{1}{(2\pi)^{2d}} \int_{\Pspace} \dd Y \int_{\Pspace} \dd Z \, e^{i \sigma(X,Y+Z)} \, \mathcal{T}_{n,k}(x,Y,Z) 
		\cdot \notag \\
		&\qquad \qquad \qquad \qquad \qquad \cdot
		(\Fs f)(Y) \, (\Fs g)(Z) 
		\label{magBel:integralnkTerm}
	\end{align}
	can be deduced from the properties of $\mathcal{T}_{n,k}$: we proceed by showing that $\mathcal{T}_{n,k}$ is a polynomial with $x$-dependent prefactors that contains $n+k$ $q$s (powers of $y$ and $z$) and \emph{at most} $n-k$ $p$s (powers of $\eta$ and $\zeta$). $\mathcal{L}_0$ is the non-magnetic symplectic form and contains $1$ $q$ and $1$ $p$. Hence, the $k_0$th power of $\mathcal{L}_0$ contributes $k_0$ $q$s and an equal amount of $p$s. By Lemma~\ref{appendix:asymptotics:expansion_twister:lem:expansion_twister}, the magnetic terms $\mathcal{L}_j$, $j \geq 1$, contribute $j+1$ $q$s and no $p$s. In this sense, magnetic terms improve decay. By conditions imposed on the indices appearing in the definition of $\mathcal{T}_{n,k}$, we deduce there are 
	\begin{align*}
		k_0 + \sum_{j = 1}^{N'} (j + 1) k_j = k_0 + \sum_{j = 1}^{N'} j k_j + \sum_{j = 1}^{N'} k_j = n + k
	\end{align*}
	$q$s and $k_0$ $p$s. As $0 \leq k_0 \leq n-k$, Lemma~\ref{appendix:asymptotics:existence_osc_int:Lemma2} implies the existence of integral \eqref{magBel:integralnkTerm} and that it belongs to the correct symbol class, \ie $(f \magWel g)_{(n,k)} \in \Hoerr{m_1 + m_2 - (n+k) \rho}$. 
	\medskip
	
	\noindent
	\textbf{Step 3: Existence of remainders. }
	There are two remainders we need to control, equations~\eqref{asymptotics:eqn:remainder1} and \eqref{asymptotics:eqn:remainder2}: the first one stems from the Taylor expansion of the exponential, the second one has its origins in the expansion of the magnetic flux, 
	\begin{align*}
		R_N^{\Sigma}(x,Y,Z) := R_{N}(x,Y,Z) + \sum_{n = 1}^N \frac{i^n}{n!} R_{N' \, n}[T_{\eps,\lambda}](x,Y,Z) 
		. 
	\end{align*}
	The remainder of the product is obtained after integration, 
	\begin{align}
		\tilde{R}_N(X) &:= \frac{1}{(2\pi)^{2d}} \int_{\Pspace} \dd Y \int_{\Pspace} \dd Z \, e^{i \sigma(X,Y+Z)} \, R_N^{\Sigma}(x,Y,Z) 
		\cdot \notag \\
		&\qquad \qquad \qquad \qquad \qquad \cdot
		(\Fs f)(Y) \, (\Fs g)(Z) 
		. 
		\label{asymptotics:eqn:total_remainder}
	\end{align}
	We have to show that (i) the integral exists, (ii) it is in the correct symbol class and (iii) it is of the right order in $\eps$ and $\lambda$. Points (i) and (ii) are the content of Lemma~\ref{appendix:asymptotics:existence_osc_int:lem:remainder} and we have to show that each of the two contributions to the remainder satisfies the assumptions. 
	\medskip
	
	\noindent
	The first contribution to $\tilde{R}_N$ stems from the Taylor expansion of the exponential, 
	\begin{align*}
		\frac{1}{(2\pi)^{2d}} &\int_{\Pspace} \dd Y \int_{\Pspace} \dd Z \, e^{i \sigma(X,Y+Z)} \, 
		\cdot \\
		&\qquad \quad \cdot 
			\frac{1}{N!} \int_0^1 \dd \tau \, (1 - \tau)^N \partial_{\tau}^{N+1} e^{\tau u} \big \vert_{u = i T_{\eps,\lambda}(x,Y,Z)}
			\, (\Fs f)(Y) \, (\Fs g)(Z) = \\
		&= \frac{1}{(2\pi)^{2d}} \int_0^1 \dd \tau \, (1 - \tau)^N \int_{\Pspace} \dd Y \int_{\Pspace} \dd Z \, e^{i \sigma(X,Y+Z)} \, 
			\frac{i^{N+1}}{N!} \bigl ( T_{\eps,\lambda}(x,Y,Z) \bigr )^{N+1} \cdot 
			\\
			&\qquad \qquad \qquad \qquad \qquad \cdot e^{i \tau \frac{\eps}{2} \sigma(Y,Z)} e^{-i \tau \lambda \gBe(x,y,z)} 
			\, (\Fs f)(Y) \, (\Fs g)(Z) 
			. 
	\end{align*}
	The first factor containing $\eps$ and $\lambda$, $\bigl ( T_{\eps,\lambda}(x,Y,Z) \bigr )^{N+1}$, can be expanded in powers of $\sigma(Y,Z)$ and $\gBe(x,y,z)$: 
	\begin{align*}
		\bigl ( T_{\eps,\lambda}(x,Y,Z) \bigr )^{N+1} &= \eps^{N+1} \sum_{l = 0}^{N+1} \noverk{N+1}{l} \lambda^l \, \bigl ( \tfrac{1}{2} \sigma(Y,Z) \bigr )^{N+1-l} \bigl ( \underbrace{\tfrac{1}{\eps} \gBe(x,y,z)}_{= \orderone} \bigr )^l  
	\end{align*}
	As $\eps^{N+1} < \eprec$ holds by definition of $N$, the first term of the remainder is of the correct order. The decay properties are dominated by $\bigl ( \sigma(Y,Z) \bigr )^{N+1}$ with $N+1$ $p$s and $N+1$ $q$s. All other terms contribute less than $N+1$ $p$s and more than $N+1$ $q$s since by Lemma~\ref{appendix:asymptotics:expansion_twister:lem:expansion_twister}, $\gBe$ is of order $\eps$ and contributes 2 $q$s and no $p$s. Furthermore, Lemma~\ref{appendix:asymptotics:properties_mag_flux:lem:boundedness_mag_flux} gives polynomial bounds of derivatives of $\gBe$: 
	\begin{align*}
		\babs{\partial_x^{a} \gBe(x,y,z)} \leq C_{a} \, \bigl ( \sexpval{y} + \sexpval{z} \bigr )^{\sabs{a}} 
	\end{align*}
	A similar bound holds for the exponential of the flux (Corollary~\ref{appendix:asymptotics:cor:properties_flux}):
	\begin{align*}
		\babs{\partial_x^{a} e^{- i \lambda \gBe(x,y,z)}(x,y,z)} \leq C_{a} \, \sexpval{y}^{\sabs{a}} \sexpval{z}^{\sabs{a}} 
		&&
		\forall a \in \N_0^d 
	\end{align*}
	Altogether, $\bigl ( T_{\eps,\lambda}(x,Y,Z) \bigr )^{N+1} \, e^{-i \tau \lambda \gBe(x,y,z)}$ satisfies the conditions placed on $G_{\tau'}$ in Lemma~\ref{appendix:asymptotics:existence_osc_int:lem:remainder} (with $\tau = \tau'$) which implies 
	\begin{align*}
		\frac{1}{(2\pi)^{2d}} \int_0^1 \dd \tau \, (1 - \tau)^N \, \int_{\Pspace} \dd Y \int_{\Pspace} \dd Z \, &e^{i \sigma(X,Y+Z)} \, 
			\frac{i^{N+1}}{N!} \bigl ( T_{\eps,\lambda}(x,Y,Z) \bigr )^{N+1} 
			\cdot \\
			&\cdot e^{-i \tau \lambda \gBe(x,y,z)} \, e^{i \tau \frac{\eps}{2} \sigma(Y,Z)}
			\, (\Fs f)(Y) \, (\Fs g)(Z) 
	\end{align*}
	exists as an oscillatory integral and belongs to symbol class $\Hoerr{m_1 + m_2 - (N+1)\rho}$. 
	\medskip
	
	\noindent
	The second contribution which involves 
	\begin{align*}
		R_{N' \, n}&[T_{\eps,\lambda}](x,Y,Z) = 
		\\
		&= \sum_{l = 1}^n \noverk{n}{l} \Bigl ( \tfrac{\eps}{2} \sigma(Y,Z) + \lambda \mbox{$\sum_{n' = 1}^{N'}$} \eps^{n'} \mathcal{L}_{n'}(x,y,z) \Bigr )^{n - l} \, \bigl ( \lambda R_{N'}[\gBe](x,y,z) \bigr )^l
	\end{align*}
	can be estimated analogously: by Lemma~\ref{appendix:asymptotics:expansion_twister:lem:expansion_twister}, $R_{N'}[\gBe]$ is of order $\ordere{N'+1}$ (the largest prefactor is $\eps^{N'+1} \lambda < \eprec$) and contains $N' + 2$ $q$s. So the terms in the above sum contain at least $N' + 1 \geq N + 1$ more $q$s than $p$s and another application of Lemma~\ref{appendix:asymptotics:existence_osc_int:lem:remainder} (with $\tau = 0$) implies that the second contribution to $\tilde{R}_N$ exists as an oscillator integral and is of symbol class $\Hoerr{m_1 + m_2 - (N'+1) \rho} \subseteq \Hoerr{m_1 + m_2 - (N+1) \rho}$. 
	\medskip
	
	\noindent
	Altogether, we conclude that $\tilde{R}_N$ exists pointwise, is of symbol class $\Hoerr{m_1 + m_2 - (N+1) \rho}$ as long as $N' \geq N$ and hence $f \magWel g - \sum_{n = 0}^N \sum_{k = 0}^n \eps^n \lambda^k \, (f \magWel g)_{(n,k)}  = \ordereprec$. This concludes the proof.
\end{proof}
Since for fixed power of $\eps$, the sum in powers of $\lambda$ is finite, one immediately has the following 
\begin{cor}[Asymptotic expansion in $\eps$\index{Weyl product!asymptotic expansion in $\eps$}]\label{asymptotics:cor:asymptotic_expansion_eps}
	If the assumptions of Theorem~\ref{asymptotics:thm:asymptotic_expansion} are modified by taking $\lambda = 1$, then the $\eps$ expansion of the product $f \magWel g$ of two Hörmander symbols can be recovered from the two-parameter expansion: the $n$th order term in $\eps$ then reads 
	\begin{align*}
		(f \magWel g)_{(n)} &= \sum_{k = 0}^n \, (f \magWel g)_{(n,k)} \in \Hoerrd{m_1 + m_2 - n \rho} 
	\end{align*}
	where the $(f \magWel g)_{(n,k)}$ are taken from Theorem~\ref{asymptotics:thm:asymptotic_expansion}. 
\end{cor}
If we do not have a separation of spatial scales, \ie $\eps = 1$, but weak coupling to the magnetic field, we can still expand the product $\magWel$ as a power series in $\lambda$. This is also the starting point of the $\lambda$-$\eps$ expansion which coincides with the $\eps$-$\lambda$ expansion. 
\begin{thm}[Asymptotic expansion in $\lambda$\index{Weyl product!asymptotic expansion in $\lambda$}]\label{asymptotics:thm:lambda_expansion}
	Assume the magnetic field $B$ satisfies Assumption~\ref{asymptotics:assumption:bounded_fields}; then for $\lambda \ll 1$ and $\eps \leq 1$, we can expand the $\lambda$ Weyl product of $f \in \Hoerr{m_1}$ and $g \in \Hoerr{m_2}$ asymptotically in $\lambda$ such that 
	\begin{align*}
		f \magWel g - \sum_{k = 0}^N \lambda^k (f \magWel g)_{(k)} \in \Hoerr{m_1 + m_2 - 2(N+1) \rho} 
	\end{align*}
	where $(f \magWel g)_{(k)} \in \Hoerr{m_1 + m_2 - 2k \rho}$ is given by equation~\eqref{asymptotics:eqn:lambdaExpansion:kth_term}. In particular, the zeroth-order term reduces to the \emph{non-magnetic} Weyl product, $(f \magWel g)_{(0)} = f \Weyle g$. We have explicit control over the remainder (equation~\eqref{aysmpExp:eqn:lambda_expansion:remainder}): if we expand the product up to $N$th order in $\lambda$, the remainder is of order $\orderl{N+1}$ and in symbol class $\Hoerr{m_1 + m_2 - 2(N+1)\rho}$. 
\end{thm}
\begin{proof}
	Assume we want to expand up to $N$th order in $\lambda$. We will show 
	\begin{align*}
		f \magWel g - \sum_{k = 0}^N \lambda^k (f \magWel g)_{(k)} = \orderl{N+1} 
	\end{align*}
	and that the difference is in $\Hoerr{m_1 + m_2 - 2(N+1) \rho}$. 
	\medskip
	
	\noindent
	\textbf{Step 1: Expansion of exponential flux. }
	If $\eps$ is not necessarily small, we cannot expand the magnetic flux integral $\gBe$ in powers of $\eps$ anymore. However, we will keep $\eps$ as a \emph{bookkeeping device}. Expanding the exponential of the magnetic flux, we get 
	\begin{align*}
		e^{i T_{\eps,\lambda}(x,Y,Z)} &= e^{i \frac{\eps}{2} \sigma(Y,Z)} e^{- i \lambda \gBe(x,y,z)} \\
		&= e^{i \frac{\eps}{2} \sigma(Y,Z)} \Bigl ( \mbox{$\sum_{k = 0}^N$} \lambda^k \tfrac{(-i)^k}{k!} \bigl ( \gBe(x,y,z) \bigr )^k + R_N(x,y,z)  \Bigr ) 
		. 
	\end{align*}
	The remainder is of order $\lambda^{N+1}$ and has $2(N+1)$ $q$s, 
	\begin{align*}
		R_N(x,y,z) &= \frac{1}{N!} \bigl ( -i \lambda \gBe(x,y,z) \bigr )^{N+1} \int_0^1 \dd \tau' \, (1 - \tau')^N \, e^{-i \lambda \tau' \gBe(x,y,z)} 
		. 
	\end{align*}
	This can be seen more readily once we define $- \eps \Bte_{lj}(x,y,z) \, y_l z_j := \gBe(x,y,z)$ to emphasize that $\gBe$ contains $\eps$ as a prefactor and $2$ $q$s. Using the antisymmetry of $B_{lj}$, there is a simple explicit expression for $\Bte_{lj}$ (see proof of Lemma~\ref{appendix:asymptotics:expansion_twister:lem:expansion_twister}): 
	\begin{align*}
		\Bte_{lj}(x,y,z) 
		&= \frac{1}{2} \int_{-\nicefrac{1}{2}}^{+\nicefrac{1}{2}} \dd t \int_0^1 \dd s \, s \, \bigl [ B_{lj} \bigl ( x + \eps s (t y - \nicefrac{z}{2}) \bigr )  + B_{lj} \bigl ( x + \eps s (\nicefrac{y}{2} + t z) \bigr ) \bigr ] 
		\\
		&
		= \orderone 
	\end{align*}
	\textbf{Step 2: Existence of $k$th-order term. }
	Then the expansion can be rewritten so that we can separate off factors of $y$, $z$ and $\eps$. The $k$th order term contains $2k$ $q$s and no $p$s,  
	\begin{align*}
		\frac{(-i)^k}{k!} \bigl ( \gBe(x,y,z) \bigr )^k &= \eps^k \, \frac{i^k}{k!} \prod_{m = 1}^k \Bte_{l_m j_m}(x,y,z) \, y_{l_m} z_{j_m} 
		. 
	\end{align*}
	By Lemma~\ref{appendix:asymptotics:existence_osc_int:lem:remainder} (with $\tau = 1 = \tau'$) the $k$th order term of the product 
	\begin{align}
		(f \magWel g)_{(k)}&(X) := \frac{\eps^k}{(2\pi)^{2d}} \int_{\Pspace} \dd Y \int_{\Pspace} \dd Z \, e^{i \sigma(X,Y+Z)} \, e^{i \frac{\eps}{2} \sigma(Y,Z)} 
		\cdot \notag \\
		&\qquad \qquad \qquad \cdot 
		\biggl ( \frac{i^k}{k!} \prod_{m = 1}^k \Bte_{l_m j_m}(x,y,z) \, y_{l_m} z_{j_m}  \biggr ) 
		(\Fs f)(Y) \, (\Fs g)(Z) 
		\notag \\
		&= \frac{\eps^k}{(2\pi)^{2d}} \int_{\Pspace} \dd Y \int_{\Pspace} \dd Z \, e^{i \sigma(X,Y+Z)} \, e^{i \frac{\eps}{2} \sigma(Y,Z)} \, \biggl ( \frac{i^{3k}}{k!} \prod_{m = 1}^k \Bte_{l_m j_m}(x,y,z) \biggr ) 
		\cdot \notag \\
		&\qquad \qquad \qquad \cdot \bigl ( \Fs (\partial_{\tilde{\eta}_{j_1}} \cdots \partial_{\tilde{\eta}_{j_k}} f) \bigr )(Y) \, \bigl ( \Fs (\partial_{\tilde{\zeta}_{j_1}} \cdots \partial_{\tilde{\zeta}_{j_k}} g) \bigr )(Z) 
		\label{asymptotics:eqn:lambdaExpansion:kth_term}
	\end{align}
	exists and is of symbol class $\Hoerr{m_1 + m_2 - 2k \rho}$. 
	\medskip

	\noindent
	\textbf{Step 3: Existence of remainder. }
	The remainder is of order $\lambda^{N+1}$ and has $2(N+1)$ $q$s. It contains $\eps^{N+1}$ as a prefactor as well which will be of importance in the proof of the next theorem. By Lemma~\ref{appendix:asymptotics:properties_mag_flux:lem:boundedness_mag_flux} and Corollary~\ref{appendix:asymptotics:cor:properties_flux}, the integral in $R_N$ over the exponential of the magnetic flux is bounded and its derivatives can be bounded polynomially in $y$ and $z$, 
	\begin{align*}
		R_N(x,y,z) &= \lambda^{N+1} \frac{\eps^{N+1}}{N!} \bigl ( \Bte_{lj}(x,y,z) \, y_l z_j \bigr )^{N+1} \int_0^1 \dd \tau' \, (1 - \tau')^N \, e^{-i \lambda \tau' \gBe(x,y,z)} 
		. 
	\end{align*}
	This means $R_N$ satisfies the conditions on $G_{\tau'}$ in Lemma~\ref{appendix:asymptotics:existence_osc_int:lem:remainder} (with $\tau = 1$) and we conclude that 
	\begin{align}
		\tilde{R}_N(X) := \frac{1}{(2\pi)^{2d}} \int_{\Pspace} \dd Y \int_{\Pspace} \dd Z \, e^{i \sigma(X,Y+Z)} \, e^{i \frac{\eps}{2} \sigma(Y,Z)} \, R_N(x,y,z) \, (\Fs f)(Y) \, (\Fs g)(Z) 
		\label{aysmpExp:eqn:lambda_expansion:remainder}
	\end{align}
	exists and is in symbol class $\Hoerr{m_1 + m_2 - 2(N+1) \rho}$. 
\end{proof}
The next statement is central to this paper, because it tells us we can speak of \emph{the} two-parameter expansion of the product. 
\begin{thm}\label{asymptotics:thm:lambdaEpsExpansion}
	Assume that the magnetic field $B$ satisfies Assumption~\ref{asymptotics:assumption:bounded_fields} and $\eps \ll 1$ in addition to $\lambda \ll 1$. Then we can expand each term of the $\lambda$ expansion of $f \magWel g$ in $\eps$, $f \in \Hoerr{m_1}$, $g \in \Hoerr{m_2}$, and obtain the same as in Theorem~\ref{asymptotics:thm:asymptotic_expansion}. Hence we can speak of \emph{the} two-parameter expansion of the product $\magWel$. 
\end{thm}
\begin{proof}
	\textbf{Step 1: Precision of expansion. }
	Assume we have expanded the magnetic product $\magWel$ up to $N_0$th power in $\lambda$. Then for the remainder of the proof, we fix $N \equiv N(\eprec,\eps,\lambda) \in \N_0$ as in Definition~\ref{asymptotics:defn:precision} for $\eprec = \lambda^{N_0}$. 
	\medskip

	\noindent
	\textbf{Step 2: Equality of $(n,k)$ terms of expansion. }
	Now to the expansion itself. The two terms we need to expand are the non-magnetic twister $e^{i \frac{\eps}{2} \sigma(Y,Z)}$ and the $k$th power of the magnetic flux integral $\gBe$ in $\eps \ll 1$: we choose $N' , N'' \geq N$ and write the $k$th order of the $\lambda$ expansion as 
	\begin{align*}
		(f \magWel g)_{(k)}(X) 
		&= \frac{1}{(2\pi)^{2d}} \int_{\Pspace} \dd Y \int_{\Pspace} \dd Z \, e^{i \sigma(X,Y+Z)} \, e^{i \frac{\eps}{2} \sigma(Y,Z)} 
		\cdot \\
		&\qquad \qquad \qquad \cdot 
		\frac{(-i)^k}{k!} \bigl ( \gBe(x,y,z) \bigr )^k \, (\Fs f)(Y) \, (\Fs g)(Z) 
		\\ 
		&= \frac{1}{(2\pi)^{2d}} \int_{\Pspace} \dd Y \int_{\Pspace} \dd Z \, e^{i \sigma(X,Y+Z)} 
		\, (\Fs f)(Y) \, (\Fs g)(Z)
		\cdot \\
		&\qquad \qquad \qquad \cdot 
		\Bigl ( \mbox{$\sum_{n = 0}^{N'}$} \eps^n \tfrac{i^n}{n!} \bigl ( \tfrac{1}{2} \sigma(Y,Z) \bigr )^n + R_{N'}[\sigma](Y,Z) \Bigr ) 
		\cdot \\
		&\qquad \qquad \qquad \cdot 
		\frac{(-i)^k}{k!} 
	\Bigl ( \bigl ( \mbox{$\sum_{j = 1}^{N''} \eps^j \mathcal{L}_j(x,y,z)$} \bigr )^k + R_{N'' \, k}[\mathcal{L}R](x,y,z) \Bigr ) 
		. 
	\end{align*}
	The remainders are given explicitly in Step 3, equations~\eqref{asymptotics:eqn:proof_equivalence_expansions:RNsigma} and \eqref{asymptotics:eqn:proof_equivalence_expansions:RNgBe}. The $(n,k)$ terms of the expansion originate from the first of these terms, \ie we need to look at 
	\begin{align*}
		\sum_{n = 0}^{N'} &\eps^n \frac{i^n}{n!} \Bigl ( \tfrac{1}{2} \sigma(Y,Z) \Bigr )^n \, \Bigl ( \mbox{$\sum_{j = 1}^{N''} \eps^j \mathcal{L}_j(x,y,z)$} \Bigr )^k = 
		\\
		&= \sum_{n = 0}^{N'} \sum_{\sum_{j = 1}^{N''} k_j = k} \eps^{n + \sum_{j = 1}^{N''} j k_j} \frac{i^{n+k}}{n! k_1! \cdots k_{N''}!} 
		\bigl ( \tfrac{1}{2} \sigma(Y,Z) \bigr )^n \, \prod_{j = 1}^{N''} \mathcal{L}_j^{k_j}(x,y,z) 
	\end{align*}
	to obtain the $(n,k)$ term of this expansion. The remaining three terms define the remainder which will be treated in the last step. We define $\mathcal{L}_0(Y,Z) := \tfrac{1}{2} \sigma(Y,Z)$, $k_0 := n$ and recognize the result from Theorem~\ref{asymptotics:thm:asymptotic_expansion}, the terms match: 
	\begin{align*}
		\sum_{n = k}^{N' \, N''}& \sum_{\substack{k_0 + \sum_{j = 1}^{N''} j k_j = n\\
		\sum_{j = 1}^{N''} k_j = k}} \eps^{n} \frac{i^{k+k_0}}{k_0! k_1! \cdots k_{N''}!} 
		\mathcal{L}_0^{k_0}(Y,Z) \, \prod_{j = 1}^{N''} \mathcal{L}_j^{k_j}(x,y,z)
	\end{align*}
	Obviously, the arguments made in the proof of Theorem~\ref{asymptotics:thm:asymptotic_expansion} can be applied here as well, and we conclude that the $(n,k)$ term exists and is in the correct symbol class, $\Hoerr{m_1 + m_2 - (n+k)\rho}$. 
	\medskip
	
	\noindent
	\textbf{Step 3: Existence of remainders. }
	The remainders of the expansions of $e^{i \frac{\eps}{2} \sigma(Y,Z)}$ and $\bigl ( \gBe(x,y,z) \bigr )^k$, 
	\begin{align}
		R_{N'}[\sigma](Y,Z) &= \eps^{N'+1} \, \frac{i^{N'+1}}{N'!} \bigl ( \tfrac{1}{2} \sigma(Y,Z) \bigr )^{N'+1} \, \int_0^1 \dd \tau ( 1 - \tau )^{N'} e^{i \tfrac{\eps}{2} \tau \sigma(Y,Z)} 
		\label{asymptotics:eqn:proof_equivalence_expansions:RNsigma}
	\end{align}
	and 
	\begin{align}
		R_{N'' \, k}[\mathcal{L}R](x,y,z) &= \sum_{l = 1}^k \noverk{k}{l} \bigl ( \mbox{$\sum_{j = 1}^{N''} \eps^j \mathcal{L}_j(x,y,z)$} \bigr )^{k - l} \, \bigl ( R_{N''}[\gBe](x,y,z) \bigr )^l
		\label{asymptotics:eqn:proof_equivalence_expansions:RNgBe}
	\end{align}
	with $R_{N''}[\gBe](x,y,z)$ as in Lemma~\ref{appendix:asymptotics:expansion_twister:lem:expansion_twister}, lead to three terms in the total remainder: 
	\begin{align*}
		R_{N N' N'' \, k}^{\Sigma}(x,Y,Z) &= R_{N'}[\sigma](Y,Z) \Bigl ( \bigl ( \mbox{$\sum_{j = 1}^{N''} \eps^j \mathcal{L}_j(x,y,z)$} \bigr )^k 
		+ R_{N'' \, k}[\mathcal{L}R](x,y,z) \Bigr ) 
			+ \\
			&\qquad 
		+ \Bigl ( \mbox{$\sum_{n = 0}^{N'}$} \eps^n \tfrac{i^n}{n!} \bigl ( \tfrac{1}{2} \sigma(Y,Z) \bigr )^n \Bigr ) \, R_{N'' \, k}[\mathcal{L}R](x,y,z) 
	\end{align*}
	Going through the motions of the proof to Theorem~\ref{asymptotics:thm:asymptotic_expansion},  we count $p$s and $q$s, and then apply Lemma~\ref{appendix:asymptotics:existence_osc_int:lem:remainder}. The first remainder, $R_{N'}[\sigma](Y,Z)$, is of order $\eps^{N'+1} < \eprec$ in $\eps$ and contributes $N'+1$ $q$s and $p$s. By Lemma~\ref{appendix:asymptotics:expansion_twister:lem:expansion_twister}, $R_{N''}[\gBe]$ contributes at least $N''+2$ $q$s and all prefactors are less than or equal to $\eps^{N''+1} < \eprec$. Thus the terms in $R_{N'' \, k}[\mathcal{L}R]$ contain at least $N'' + 2$ $q$s (for all $k \leq N$) and prefactors that are at most $\eps^{N''+1} < \eprec$. Hence, the total remainder exists as an oscillatory integral, is $\ordereprec$ small and in symbol class $\Hoerr{m_1 + m_2 - (N+1)\rho}$. 
\end{proof}
\begin{remark}
	The asymptotic expansion of $\magWel$ can be immediately extended to an expansion of products of \emph{semiclassical two-parameter symbols} (see Definition~\ref{asymptotics:defn:2ParameterSemiSymbol}). 
\end{remark}
%

\section{Semiclassical limit} 
\label{asymptotics:semiclassical_limit}
%
%
An immediate application of the asymptotic expansions of the magnetic Weyl product is the proof of an Egorov-type theorem\index{Egorov-type theorem} which connects the quantization of a classically evolved observable with the quantum mechanical Heisenberg observable. The premise which magnetic Weyl calculus is based on is that the magnetic field modifies the \emph{geometry} of phase space, namely it enters into the magnetic symplectic form $\omega^B = \dd x_j \wedge \dd \xi_j - \lambda \, B_{kj} \, \dd x_k \wedge \dd x_j$. On the other hand, there is no need to modify observables -- regardless of whether or not there is a magnetic field, the observables position and momentum are still $x$ and $\xi$. Also from a technical perspective, this approach is simpler than minimal substitution, because at not point do worse-behaved vector potentials enter into the discussion. In this sense, our point of view is not just more natural, but also technically less involved. 

We do not strive for a formulation in its utmost generality, but instead try to give a proof under rather simple and straightforward assumptions. Later on, we will indicate what obstacles need to be overcome if this result is to be generalized to more general hamiltonian symbols and observables. For simplicity, we set $\lambda = 1$ in this section and focus on the behavior of the dynamics as the semiclassical parameter $\eps$ tends to $0$. 
\begin{thm}[Semiclassical limit for observables\index{Egorov theorem}\index{semiclassical limit!observables}]\label{asymptotics:semiclassical_limit:thm:semiclassical_limit_observables}
	Assume the components of $B$ are of class $\BCont^{\infty}$. Let $h \in \Hoermr{2}{0}$ be a real-valued symbol such that $\partial_x^a \partial_{\xi}^{\alpha} h \in \Hoermr{0}{0} = \BCont^{\infty}(\Pspace)$ for all $\abs{a} + \abs{\alpha} \geq 2$ and assume $\Op^A(h)$ defines a selfadjoint operator on $\mathcal{D}^A := \mathcal{D} \bigl ( \Op^A(h) \bigr ) \subseteq L^2(\R^d)$. Furthermore, let $f \in \Hoermr{0}{0} = \BCont^{\infty}(\Pspace)$ a real-valued symbol and that the flow $\phi_t$ associated to the hamiltonian equations of motion 
	\begin{align}
		\left (
		\begin{matrix}
			B & - \id \\
			+ \id & 0 \\
		\end{matrix}
		\right ) 
		\left (
		\begin{matrix}
			\dot{x} \\
			\dot{\xi} \\
		\end{matrix}
		\right ) = \left (
		\begin{matrix}
			\nabla_x h \\
			\nabla_{\xi} h \\
		\end{matrix}
		\right ) 
		. 
	\end{align}
	has bounded derivatives to any order for all $t \in \R$. 
	
	Then $\phi_t$ exists globally in time and for any $T > 0$ and all $\abs{t} \leq T$, the quantization of the classically evolved observable 
	\begin{align}
		F_{\mathrm{cl}}(t) := \Op^A \bigl ( f(t) \bigr ) = \Op^A \bigl ( f \circ \phi_t \bigr ) 
	\end{align}
	is $\order(\eps^2)$-close in the operator norm to the Heisenberg observable 
	\begin{align}
		F_{\mathrm{qm}}(t) := e^{+ i \frac{t}{\eps} \Op^A(h)} \Op^A(f) e^{- i \frac{t}{\eps} \Op^A(h)} 
	\end{align}
	in the sense that there exists a $C_T > 0$ such that for all $t \in [-T , +T]$, we have 
	\begin{align}
		\bnorm{F_{\mathrm{qm}}(t) - F_{\mathrm{cl}}(t)}_{\mathcal{B}(L^2(\R^d))} \leq C_T \eps^2 
		. 
	\end{align}
\end{thm}
\begin{remark}
	The implicit assumption on the flow $\phi_t$ is necessary in the magnetic case, because under the remaining assumptions, its derivatives need not be bounded. In case $B = 0$, the boundedness of the derivatives of the flow $\phi_t$ \cite[Lemma~IV.9]{Robert:tour_semiclassique:1987} follows from writing down the equations of motion for its first-order derivatives
	\begin{align*}
		\partial_t \bigl ( \nabla_x^{\mathrm{T}} \phi_t \vert \nabla_{\xi}^{\mathrm{T}} \phi_t \bigr ) = A \bigl ( \nabla_x^{\mathrm{T}} \phi_t \vert \nabla_{\xi}^{\mathrm{T}} \phi_t \bigr )
	\end{align*}
	where the matrix-valued function $A$ contains only \emph{second-order} derivatives of $h$ and applying Gronwall's Lemma. Then one proceeds by induction to ensure that higher-order derivatives are also bounded. 
	
	If one tries to imitate Robert's proof in the magnetic case, however, one term appearing in $A$ is of the form $\sum_{j = 1}^d \partial_{x_k} B_{lj} \, \partial_{\xi_j} h$, \ie it contains \emph{first}-order derivatives of $h$. The right-hand side is not necessarily in $\BCont^{\infty}(\Pspace)$ -- unless the magnetic field is constant or \emph{first}-order derivatives of the hamiltonian symbol $h$ are already bounded. 
	
	From a physical perspective, the formulation of Theorem~\ref{asymptotics:semiclassical_limit:thm:semiclassical_limit_observables} is not satisfactory: even the simplest physically relevant case, $h(x,\xi) = \tfrac{1}{2} \xi^2$ and non-constant $B \in \BCont^{\infty}$, is not covered automatically. Certainly, this topic deserves more attention. 
\end{remark}
It turns out the proof can be written down a little more tidily if we separate off this little lemma and prove it first: 
\begin{lem}\label{asymptotics:semiclassical_limit:lem:composition_symbol_flow}
	Assume that all the derivatives of the diffeormorphism $\phi \in \Cont^{\infty}(\Pspace,\Pspace)$ are bounded, \ie for all $a , \alpha \in \N_0^d$ with $\abs{a} + \abs{\alpha} \geq 1$, we have 
	\begin{align*}
		\bnorm{\partial_x^a \partial_{\xi}^{\alpha} \phi_j(x,\xi)}_{\infty} < \infty 
	\end{align*}
	for all $j \in \{ 1 , \ldots , 2d \}$, and let $f \in \Hoermr{m}{0}$, $m \in \R$. Then $f \circ \phi \in \Hoermr{m}{0}$ and for all $a , \alpha \in \N_0^d$, the $a\alpha$th seminorm for $f \circ \phi$ can be estimated from above by 
	\begin{align}
		\bnorm{f \circ \phi}_{m , a \alpha} \leq \sum_{\substack{\abs{b} \leq \abs{a} \\ \abs{\beta} \leq \abs{\alpha}}} \bnorm{\varphi^{b \beta}}_{\infty} \, \bnorm{f}_{m , b \beta} < \infty 
		\label{asymptotics:semiclassical_limit:eqn:composition_symbol_flow}
	\end{align}
	where $\varphi^{b \beta}$ is a product of derivatives of $\phi$. 
\end{lem}
\begin{proof}
	We have to show that all seminorms $\{ \snorm{\cdot}_{m , a \alpha} \}_{a , \alpha \in \N_0^d}$ of $f \circ \phi$ are bounded. Since $\phi : \Pspace \longrightarrow \Pspace$ is a diffeomorphism, we have 
	\begin{align*}
		\snorm{f \circ \phi}_{00} = \sup_{(x,\xi) \in \Pspace} \babs{f \circ \phi(x,\xi)} = \sup_{(x,\xi) \in \Pspace} \babs{f(x,\xi)} = \snorm{f}_{00} 
		. 
	\end{align*}
	For terms involving derivatives, we proceed by induction: consider $a = 0 \in \N_0^d$ and $\alpha^k := (\delta_{1k} ,\ldots , \delta_{dk}) \in \N_0^d$, for instance. Then by the chain rule, we can calculate the first-order derivative: 
	\begin{align*}
		\partial_{\xi_k} \bigl ( f \circ \phi \bigr ) &= \sum_{j = 1}^d \Bigl ( \partial_{x_j} f \circ \phi \; \partial_{\xi_k} \phi_j + \partial_{\xi_j} f \circ \phi \; \partial_{\xi_k} \phi_{d + j} \Bigr ) 
	\end{align*}
	Hence, we can estimate the $0 \alpha^k$th seminorm by 
	\begin{align*}
		\bnorm{f \circ \phi}_{0 \alpha^k} &= \sup_{(x,\xi) \in \Pspace} \babs{\sexpval{\xi}^{-m} \, \partial_{\xi_k} \bigl ( f \circ \phi \bigr )(x,\xi)} 
		\\
		&
		\leq \sum_{j = 1}^d \Bigl ( \bnorm{\sexpval{\xi}^{-m} \, \partial_{x_j} f \circ \phi}_{\infty} \, \bnorm{\partial_{\xi_k} \phi_j}_{\infty} 
		+ \Bigr . \\
		&\qquad \qquad \Bigl . 
		+ \bnorm{\sexpval{\xi}^{-m} \, \partial_{\xi_j} f \circ \phi}_{\infty} \, \bnorm{\partial_{\xi_k} \phi_{d + j}}_{\infty} \Bigr ) 
		\\
		&
		\leq \sum_{j = 1}^d \Bigl ( \bnorm{\sexpval{\xi}^{-m} \, \partial_{x_j} f}_{\infty} \, \bnorm{\partial_{\xi_k} \phi_j}_{\infty} + \bnorm{\sexpval{\xi}^{-m} \, \partial_{\xi_j} f}_{\infty} \, \bnorm{\partial_{\xi_k} \phi_{d + j}}_{\infty} \Bigr ) 
		\\
		&
		\leq \sum_{j = 1}^d \Bigl ( \bnorm{f}_{m,a^j 0} \, \bnorm{\partial_{\xi_k} \phi_j}_{\infty} + \bnorm{f}_{m,0\alpha^j} \, \bnorm{\partial_{\xi_k} \phi_{d + j}}_{\infty} \Bigr ) 
		< \infty 
	\end{align*}
	Here $a^j = (\delta_{1j} , \ldots , \delta_{dj})$ and $\alpha^j = (\delta_{1j} , \ldots , \delta_{dj})$ are the multiindices whose entries are all $0$ except for the $j$th which is $1$. 
	
	Similarly, we can estimate the seminorm of $f \circ \phi$ associated to $a^k = (\delta_{1k},\ldots,\delta_{dk}) \in \N_0^d$ and $\alpha = 0 \in \N_0^d$ by 
	\begin{align*}
		\bnorm{f \circ \phi}_{m , a^k 0} \leq \sum_{j = 1}^d \Bigl ( \bnorm{f}_{m,a^j 0} \, \bnorm{\partial_{x_k} \phi_j}_{\infty} + \bnorm{f}_{m,0\alpha^j} \, \bnorm{\partial_{x_k} \phi_{d + j}}_{\infty} \Bigr ) < \infty 
		. 
	\end{align*}
	Now we proceed by induction: let $\abs{a} + \abs{\alpha} \geq 1$. Then $\partial_x^a \partial_{\xi}^{\alpha} \bigl ( f \circ \phi \bigr )$ is a sum of terms of the type 
	\begin{align*}
		\varphi^{b \beta} \; \partial_x^b \partial_{\xi}^{\beta} f \circ \phi 
	\end{align*}
	where $\varphi^{b \beta}$ is a product of derivatives of $\phi$ and $\abs{b} \leq \abs{a}$, $\abs{\beta} \leq \abs{\alpha}$. Thus, the $a \alpha$ seminorm of $f \circ \phi$ can be estimated by 
	\begin{align*}
		\bnorm{f \circ \phi}_{m , a \alpha} \leq \sum_{\substack{\abs{b} \leq \abs{a} \\ \abs{\beta} \leq \abs{\alpha}}} \bnorm{\varphi^{b \beta}}_{\infty} \, \bnorm{f}_{m , b \beta} < \infty 
	\end{align*}
	and $f \circ \phi \in \Hoermr{m}{0}$. 
\end{proof}
\begin{proof}[Theorem~\ref{asymptotics:semiclassical_limit:thm:semiclassical_limit_observables}]
	Let $T > 0$. By the assumptions on $h$ and $B$, the hamiltonian vector field 
	\begin{align*}
		X_h^B := \left (
		\begin{matrix}
			0 & + \id \\
			- \id & B \\
		\end{matrix}
		\right ) \left (
		\begin{matrix}
			\nabla_x h \\
			\nabla_{\xi} h \\
		\end{matrix}
		\right )
	\end{align*}
	%
	satisfies a global Lipschitz condition. Thus, by the Picard-Lindelöf theorem the hamiltonian flow $\phi_t$ exists globally in time and inherits the smoothness of $X_h^B$ in $(x,\xi) \in \Pspace$ and $t \in \R$ \cite{Arnold:ode:2006}. 
	
	Since we have assumed all derivatives of $\phi_t$ to be bounded, we can apply Lemma~\ref{asymptotics:semiclassical_limit:lem:composition_symbol_flow} to conclude $f(t) = f \circ \phi_t \in \Hoermr{0}{0} = \BCont^{\infty}(\Pspace)$ for all $t \in \R$. Then its quantization $F_{\mathrm{cl}}(t) = \Op^A \bigl ( f(t) \bigr )$ defines a bounded selfadjoint operator on $L^2(\R^d)$ (Theorem~\ref{magWQ:important_results:continuity_and_selfadjointness:thm:magnetic_Calderon_Vaillancourt}) whose norm can be estimated from above by a finite number of seminorms of $f(t)$. 
	
	The selfadjoint operator $\Op^A(h)$ generates the strongly continuous one-parameter evolution group $e^{-i \frac{t}{\eps} \Op^A(h)}$ \cite[Theorem~VIII.8]{Reed_Simon:M_cap_Phi_1:1972}. Thus, the quantum observable $F_{\mathrm{qm}}(t) = e^{+i \frac{t}{\eps} \Op^A(h)} \, \Op^A(f) \, e^{-i \frac{t}{\eps} \Op^A(h)}$ as composition of bounded linear maps on $L^2(\R^d)$ is also bounded. 
	
	Both, $F_{\mathrm{qm}}(t) , F_{\mathrm{cl}}(t) : \Schwartz(\R^d) \subset L^2(\R^d) \longrightarrow L^2(\R^d) \subset \Schwartz'(\R^d)$ can be seen as elements of $\mathcal{L} \bigl ( \Schwartz(\R^d) , \Schwartz'(\R^d) \bigr )$, the linear continuous operators between $\Schwartz(\R^d)$ and $\Schwartz'(\R^d)$. This point of view allows us to make the following formal manipulations rigorous: we rewrite the difference $F_{\mathrm{qm}}(t) - F_{\mathrm{cl}}(t)$ as an integral over a derivative (the Duahmel trick), 
	\begin{align*}
		&F_{\mathrm{qm}}(t) - F_{\mathrm{cl}}(t) = \int_0^t \dd s \, \frac{\dd}{\dd s} \Bigl ( e^{+ i \frac{s}{\eps} \Op^A(h)} \, \Op^A \bigl ( f(t-s) \bigr ) \, e^{-i \frac{s}{\eps} \Op^A(h)} \Bigr ) 
		\\
		&\;
		= \int_0^t \dd s \, e^{+ i \frac{s}{\eps} \Op^A(h)} \Bigl ( \tfrac{i}{\eps} \bigl [ \Op^A(h) , \Op^A \bigl ( f(t-s) \bigr ) \bigr ] + \tfrac{\dd}{\dd s} \Op^A \bigl ( f(t-s) \bigr ) \Bigr ) \, e^{-i \frac{s}{\eps} \Op^A(h)} 
		\displaybreak[2]
		\\
		&\;
		= \int_0^t \dd s \, e^{+ i \frac{s}{\eps} \Op^A(h)} \Op^A \Bigl ( \tfrac{i}{\eps} [ h , f(t-s) ]_{\magW} - \{ h , f(t-s) \}_B \Bigr ) \, e^{-i \frac{s}{\eps} \Op^A(h)} 
		\\
		&\;
		= \order(\eps^2 \abs{t})
		. 
	\end{align*}
	The integrand is the derivative of 
	\begin{align*}
		I(t,s) := e^{+ i \frac{s}{\eps} \Op^A(h)} \, \Op^A \bigl ( f(t-s) \bigr ) \, e^{- i \frac{s}{\eps} \Op^A(h)} \in \mathcal{L} \bigl ( \Schwartz(\R^d) , \Schwartz'(\R^d) \bigr ) 
		. 
	\end{align*}
	This operator is also a bounded as a map from $L^2(\R^d)$ to itself. We need to establish that $s \mapsto I(t,s)$ is in $\Cont^1$ in the sense of $\mathcal{L} \bigl ( \Schwartz(\R^d) , \Schwartz'(\R^d) \bigr )$, \ie for all $u , v \in \Schwartz(\R^d)$, the map 
	\begin{align}
		[-T , +T] \ni s \mapsto \bigl ( I(s,t) v , u^* \bigr ) = \scpro{u}{I(s,t) v} \in \C 
		\label{asymptotics:semiclassical_limit:eqn:I_s_t_map}
	\end{align}
	is in $\Cont^1$. Here, $(\cdot , \cdot)$ denotes the duality bracket on $\Schwartz'(\R^d)$ and $\scpro{\cdot}{\cdot}$ the scalar product on $L^2(\R^d)$. 
	
	Combining the magnetic Caldéron-Vaillancourt theorem (Theorem~\ref{magWQ:important_results:continuity_and_selfadjointness:thm:magnetic_Calderon_Vaillancourt}) with inequality~\eqref{asymptotics:semiclassical_limit:eqn:composition_symbol_flow} of Lemma~\ref{asymptotics:semiclassical_limit:lem:composition_symbol_flow} and the smoothness of $t \mapsto \phi_t$, we find a bound of the operator norm of $\Op^A \bigl ( f(t-s) \bigr )$ where the constants depend $\bnorm{\varphi^{b\beta}_{t-s}}$ smoothly on $s$, 
	\begin{align*}
		\bnorm{\Op^A \bigl ( f(t-s) \bigr )}_{\mathcal{B}(L^2(\R^d))} \leq C(d) \max_{\abs{a} , \abs{\alpha} \leq p(d)} \sum_{\substack{\abs{b} \leq \abs{a} \\ \abs{\beta} \leq \abs{\alpha}}} \bnorm{\varphi^{b \beta}_{t-s}}_{\infty} \, \bnorm{f}_{0 , b \beta} 
		. 
	\end{align*}
	Since $h \in \Hoermr{2}{0}$ and $f(t-s) \in \Hoermr{0}{0}$ are in the magnetic Moyal algebra, they are continuous as operators in $\mathcal{L} \bigl ( \Schwartz(\R^d) \bigr )$ by Proposition~\ref{magWQ:extension:magnetic_moyal_algebra:prop:mag_Moyal_algebra_L_S_L_Sprime}. This means $\Schwartz(\R^d) \subset \mathcal{D} \bigl ( \Op^A(h) \bigr )$ is a core and $s \mapsto e^{- i \frac{s}{\eps} \Op^A(h)}$ is strongly $\Cont^1$ on $\Schwartz(\R^d)$. This establishes the continuity of expression~\eqref{asymptotics:semiclassical_limit:eqn:I_s_t_map}. 
	
	For notational simplicity, let us define $u(s) := e^{- i \frac{s}{\eps} \Op^A(h)} u \in \mathcal{D} \bigl ( \Op^A(h) \bigr )$ and $v(s) := e^{- i \frac{s}{\eps} \Op^A(h)} v \in \mathcal{D} \bigl ( \Op^A(h) \bigr )$. We need to regularize expression~\eqref{asymptotics:semiclassical_limit:eqn:I_s_t_map} before differentiating it with respect to $s$. For $\kappa > 0$, define 
	\begin{align*}
		R_{\kappa}^A := \bigl ( i \kappa \Op^A(h) + 1 \bigr )^{-1} = \tfrac{1}{i \kappa} \bigl ( \Op^A(h) + \tfrac{i}{\kappa} \bigr )^{-1} 
	\end{align*}
	which maps $L^2(\R^d)$ onto $\mathcal{D} \bigl ( \Op^A(h) \bigr )$. 
	The operators $\Op^A(h)$, $e^{- i \frac{s}{\eps} \Op^A(h)}$ and $R_{\kappa}^A$ all commute amongst each other by definition and from standard arguments, we conclude $R_{\kappa}^A$ converges strongly to the identity as $\kappa \rightarrow 0$. Since $R_{\kappa}^A$ is essentially the resolvent of a Hörmander class symbol, the Moyal resolvent exists and is a Hörmander symbol as well (Theorem~\ref{magWQ:important_results:inversion:thm:inversion_Hoermander}), 
	\begin{align*}
		r_{\kappa}^B := \tfrac{1}{i \kappa} \bigl ( h - \tfrac{i}{\kappa} \bigr )^{(-1)_B} \in \Hoermr{-2}{0} 
		. 
	\end{align*}
	The Moyal resolvent inherits the commutativity properties of $R_{\kappa}^A = \Op^A(r_{\kappa}^B)$, and using the composition law of Hörmander symbols (Theorems~\ref{asymptotics:thm:equivalenceProduct}) and the magnetic Caldéron-Vaillancourt theorem (Theorem~\ref{magWQ:important_results:continuity_and_selfadjointness:thm:magnetic_Calderon_Vaillancourt}), the Weyl quantization of $r_{\kappa}^B \magW h = h \magW r_{\kappa}^B \in \Hoermr{0}{0}$ of the product is bounded on $L^2(\R^d)$. Thus, using $u , v \in \Schwartz(\R^d) \subset \mathcal{D} \big ( \Op^A(h) \bigr )$, we compute 
	\begin{align}
		\frac{\dd}{\dd s} &\Bscpro{R^A_{\kappa} u(s)}{\Op^A \bigl ( f(t-s) \bigr ) R^A_{\kappa} v(s)} 
		= \notag \\
		&\qquad \qquad = \Bscpro{R^A_{\kappa} \tfrac{\dd}{\dd s} u(s)}{\Op^A \bigl ( f(t-s) \bigr ) R^A_{\kappa} v(s)} 
		+ \notag \\
		&\qquad \qquad \qquad \qquad 
		+ \Bscpro{R^A_{\kappa} u(s)}{\Bigl ( \tfrac{\dd}{\dd s} \Op^A \bigl ( f(t-s) \bigr ) \Bigr )  R^A_{\kappa} v(s)} 
		+ \notag \\
		&\qquad \qquad \qquad \qquad 
		+ \Bscpro{R^A_{\kappa} u(s)}{\Op^A \bigl ( f(t-s) \bigr ) R^A_{\kappa} \tfrac{\dd}{\dd s} v(s)} 
		\displaybreak[2]
		\notag \\
		&\qquad \qquad = \frac{i}{\eps} \Bscpro{\Op^A(r_{\kappa}^B \magW h) u(s)}{\Op^A \bigl ( f(t-s) \magW r_{\kappa}^B \bigr ) v(s)} 
		+ \notag \\
		&\qquad \qquad \qquad \qquad -
		\frac{i}{\eps} \Bscpro{\Op^A(r_{\kappa}^B) u(s)}{\Op^A \bigl ( f(t-s) \bigr ) \Op^A(r_{\kappa}^B \magW h) v(s)} 
		+ \notag \\
		&\qquad \qquad \qquad \qquad + 
		\Bscpro{u(s)}{\Bigl ( \tfrac{\dd}{\dd s} \Op^A \bigl ( {r^B_{\kappa}}^* \magW f(t-s) \magW r_{\kappa}^B \bigr ) \Bigr ) v(s)} 
		\displaybreak[2]
		\notag \\
		&\qquad \qquad = \Bscpro{u(s)}{\Op^A \Bigl ( {r^B_{\kappa}}^* \magW \bigl ( h \magW f(t-s) - f(t-s) \magW h \bigr ) \magW r_{\kappa}^B \Bigr ) v(s)}
		+ \notag \\
		&\qquad \qquad \qquad \qquad + 
		\Bscpro{u(s)}{\Bigl ( \tfrac{\dd}{\dd s} \Op^A \bigl ( {r^B_{\kappa}}^* \magW f(t-s) \magW r_{\kappa}^B \bigr ) \Bigr ) v(s)} 
		\label{asymptotics:semiclassical_limit:eqn:regularized_derivative_I_t_s}
		. 
	\end{align}
	In the last step, we were able to shove $\Op^A (r_{\kappa}^B \magW h)$ into the other argument of the scalar product since it defines a bounded operator. 
	
	Now we compute the remaining derivative in equation~\eqref{asymptotics:semiclassical_limit:eqn:regularized_derivative_I_t_s}: for any $u , v \in \Schwartz(\R^d)$, the magnetic Wigner transform $\WignerTrafo^A(v,u)$ is an element of $\Schwartz(\Pspace)$ (Lemma~\ref{magWQ:magnetic_weyl_calculus:lem:Wigner_transform}) and rewriting the quantum expectation value as a phase space average  (Lemma~\ref{magWQ:magnetic_weyl_calculus:magnetic_wigner_transform:lem:qm_exp_val_phase_space_average}), we can define the magnetic Weyl quantization for any $F \in \Schwartz'(\Pspace)$ via 
	\begin{align*}
		\bigl ( \Op^A(F) v , u^* \bigr ) := \bigl ( F , \WignerTrafo^A(v,u) \bigr )_{\Schwartz'(\Pspace)} 
	\end{align*}
	where the right-hand side is the duality bracket on $\Schwartz'(\Pspace)$. For $F \equiv f(t-s)$, the time-derivative with respect to $s$ is of tempered growth, \ie we have the estimate 
	\begin{align*}
		\babs{\{ h , f(t-s) \}_B(x,\xi) } \leq C(s) \sexpval{\xi}^2 \leq C \sexpval{\xi}^2 
		, 
	\end{align*}
	with a constant $C < \infty$ that is independent of $s$ due to equation~\eqref{asymptotics:semiclassical_limit:eqn:composition_symbol_flow} and the smoothness of $t \mapsto \phi_t$. Hence, the integrand implicit in 
	\begin{align*}
		\Babs{\Bigl ( \{ h , f(t-s) \}_B , \WignerTrafo^A(v,u) \Bigr )_{\Schwartz'(\Pspace)}} \leq \int_{\Pspace} \dd X \, C \sexpval{\xi}^2 \, \babs{\bigl ( \WignerTrafo^A(v,u) \bigr )(X)} < \infty 
	\end{align*}
	can be bounded by something integrable independent of $s \in [-T,+T]$ and we can interchange differentiation with respect to $s$ and integration by Dominated Convergence, 
	\begin{align*}
		\frac{\dd}{\dd s} \Bigl ( \Op^A &\bigl ( f(t-s) \bigr ) v , u^* \Bigr ) = \frac{\dd}{\dd s} \Bigl ( f(t-s) , \WignerTrafo^A(v,u) \Bigr )_{\Schwartz'(\Pspace)} 
		\\
		&
		= \Bigl ( \tfrac{\dd}{\dd s} f(t-s) , \WignerTrafo^A(v,u) \Bigr )_{\Schwartz'(\Pspace)} 
		= - \Bigl ( \{ h , f(t-s) \}_B , \WignerTrafo^A(v,u) \Bigr )_{\Schwartz'(\Pspace)} 
		\\
		&
		= \Bigl ( - \Op^A \bigl ( \{ h , f(t-s) \}_B \bigr ) v , u^* \Bigr ) 
		. 
	\end{align*}
	Finally, we will show that the difference 
	\begin{align*}
		\tfrac{i}{\eps} \bigl [ \Op^A(h) , \Op^A \bigl ( f(t-s) \bigr ) \bigr ] + &\tfrac{\dd}{\dd s} \Op^A \bigl ( f(t-s) \bigr ) 
		= \\
		&
		= \Op^A \Bigl ( \tfrac{i}{\eps} [h , f(t-s)]_{\magW} - \{ h , f(t-s) \}_B \Bigr )
	\end{align*}
	actually defines a \emph{bounded} operator on $L^2(\R^d)$: we use Theorem~\ref{asymptotics:thm:asymptotic_expansion} to expand the Moyal commutator asymptotically, 
	\begin{align*}
		\tfrac{i}{\eps} [h , f(t-s)]_{\magW} = \{ h , f(t-s) \}_B + \order(\eps^2) 
		. 
	\end{align*}
	Under the assumptions on $h$ and $f$, we can show that the remainder (which is the symmetrized version of the remainder in Theorem~\ref{asymptotics:thm:asymptotic_expansion}) is in $\BCont^{\infty}(\Pspace)$ and thus its quantization is a bounded operator on $L^2(\R^d)$ (Theorem~\ref{magWQ:important_results:continuity_and_selfadjointness:thm:magnetic_Calderon_Vaillancourt}): in the language of the last section, the remainder as given by equation~\eqref{asymptotics:eqn:total_remainder} is comprised of two terms: the first contribution to the remainder, $R_2(x,\cdot,\cdot)$ (see equation~\eqref{asymptotics:eqn:remainder1}), contributes $3$ $q$s and $3$ $p$s in total. These are converted into third-order derivatives of both $f$ and $h$ -- which are by assumption $\BCont^{\infty}(\Pspace)$ functions. According to Lemma~\ref{appendix:asymptotics:existence_osc_int:lem:remainder}, the corresponding oscillatory integral yields a $\BCont^{\infty}(\Pspace)$ function. 
	
	The second contribution contains at least $(1) + (2+2) = 5$ $q$s (which are converted into derivatives with respect to momentum), and at least $1$ $p$. Of those derivatives, at least $2$ act on $h$ and at least $2$ act on $f$. Hence, by assumption on $f$ and $g$, the second contribution to the remainder is in symbol class $\BCont^{\infty}(\Pspace)$. 
	
	Hence, we have shown that $\tfrac{i}{\eps} [h , f(t-s)]_{\magW} - \{ h , f(t-s) \}_B = \order(\eps^2) \in \BCont^{\infty}(\Pspace)$. Finishing the regularization argument above by letting $\kappa \rightarrow 0$ in equation~\eqref{asymptotics:semiclassical_limit:eqn:regularized_derivative_I_t_s} and using magnetic Caldéron-Vaillancourt theorem (Theorem~\ref{magWQ:important_results:continuity_and_selfadjointness:thm:magnetic_Calderon_Vaillancourt}), we then get that 
	\begin{align*}
		\frac{\dd}{\dd s} &\Bscpro{u}{e^{+ i \frac{s}{\eps} \Op^A(h)} \, \Op^A \bigl ( f(t-s) \bigr ) \, e^{-i \frac{s}{\eps} \Op^A(h)} v} = 
		\\
		&\qquad 
		= \Bscpro{u}{e^{+ i \frac{s}{\eps} \Op^A(h)} \, \Op^A \Bigl ( \tfrac{i}{\eps} [ h , f(t-s) ]_{\magW} - \{ h , f(t-s) \}_B \Bigr ) \, e^{-i \frac{s}{\eps} \Op^A(h)} v} 
		\\
		&\qquad 
		= \order(\eps^2) 
		. 
	\end{align*}
	Plugged into the original equation and invoking Theorem~\ref{magWQ:important_results:continuity_and_selfadjointness:thm:magnetic_Calderon_Vaillancourt} once more, we can estimate the operator norm of the difference $F_{\mathrm{qm}}(t) - F_{\mathrm{cl}}(t)$ by 
	\begin{align*}
		&\bnorm{F_{\mathrm{qm}}(t) - F_{\mathrm{cl}}(t)}_{\mathcal{B}(L^2(\R^d))} \leq 
		\\
		&\qquad \qquad \leq 
		\int_0^t \dd s \, \Bnorm{\Op^A \Bigl ( \tfrac{i}{\eps} [ h , f(t-s) ]_{\magW} - \{ h , f(t-s) \}_B \Bigr )}_{\mathcal{B}(L^2(\R^d))} 
		\\
		&\qquad \qquad 
		\leq \eps^2 \, K \abs{t} \sup_{\abs{s} \leq \abs{t}} \max_{\abs{a} , \abs{\alpha} \leq p(d)} \Bnorm{\tfrac{i}{\eps} [ h , f(t-s) ]_{\magW} - \{ h , f(t-s) \}_B}_{a \alpha} 
		\leq C \eps^2 \abs{t}
	\end{align*}
	for some finite constants $K$ and $C$. This concludes the proof.
\end{proof}
A simple corollary is a semiclassical limit for states: 
\begin{cor}[Semiclassical limit for states\index{semiclassical limit!states}]
	In addition to the assumptions on $h$ and $B$ made in Theorem~\ref{asymptotics:semiclassical_limit:thm:semiclassical_limit_observables}, let $u \in L^2(\R^d)$. Then the classically evolved signed probability measure 
	\begin{align}
		\mu_{\mathrm{cl}}(t) := (2 \pi)^{- \nicefrac{d}{2}} \, \WignerTrafo^A(u,u) \circ \phi_{-t} 
	\end{align}
	approximates the quantum mechanically evolved state 
	\begin{align}
		\mu_{\mathrm{qm}}(t) := (2 \pi)^{- \nicefrac{d}{2}} \, \WignerTrafo^A \bigl ( u(t),u(t) \bigr ) 
	\end{align}
	with $u(t) := e^{- i \frac{t}{\eps} \Op^A(h)} u$ up to errors of $\order(\eps^2)$ in the sense that 
	\begin{align}
		\int_{\Pspace} \dd X \, f(X) \, \Bigl ( \bigl ( \mu_{\mathrm{qm}}(t) \bigr )(X) - \bigl ( \mu_{\mathrm{cl}}(t) \bigr )(X) \Bigr ) = \order(\eps^2)
		&&
		\forall f \in \Schwartz(\Pspace)
	\end{align}
	holds for any $t \in [-T , +T]$.
\end{cor}
\begin{proof}
	The proof rests on a simple extension of Lemma~\ref{magWQ:magnetic_weyl_calculus:magnetic_wigner_transform:lem:qm_exp_val_phase_space_average} to include $u \in L^2(\R^d)$, Liouville's theorem \cite[Propositions~5.2.2 and 5.4.2]{Marsden_Ratiu:intro_mechanics_symmetry:1999} and the semiclassical limit for observables (Theorem~\ref{asymptotics:semiclassical_limit:thm:semiclassical_limit_observables}). 
\end{proof}
%


\section{Relation between magnetic and ordinary Weyl calculus} 
\label{asymptotics:magWQminSub}
In a previous work \cite{Iftimie_Mantiou_Purice:magnetic_psido:2006}, Iftimie et al have investigated the relation between magnetic Weyl quantization and regular Weyl quantization combined with minimal substitution, the `usual' recipe to couple a quantum system to a magnetic field. However, since there were no small parameters $\eps$ and $\lambda$, we have to revisit their statements and adapt them to the present case. 

Let us define $\minsAl(X) := \xi - \lambda A(x)$ as coordinate transformation which relates momentum and kinetic momentum. With a little abuse of notation, we will also use $f \circ \minsAl(X) := f(x,\minsAl(X))$ to transform functions. In general, $\OpAel(f) \neq \Ope(f \circ \minsAl)$ since the latter is \emph{not} manifestly covariant. However, we would like to be able to compare results obtained with magnetic Weyl calculus to those obtained with usual Weyl calculus and minimal substitution. To show how the two calculi are connected, we need to make slightly stronger assumptions on the magnetic \emph{vector potential}. This may appear contrary to the spirit of the rest of the paper where it has been emphasized that restrictions should be placed on the magnetic \emph{field}. The necessity arises, because usual, non-magnetic Weyl calculus is used in this section. 
\begin{assumption}\label{asymptotics:magWQminSub:assumption:stricterBA}
	We assume that the magnetic field $B$ is such that we can find a vector potential $A$ whose components satisfy 
	\begin{align*}
		\babs{\partial_x^{a} A_l(x)} \leq C_{a} 
		, 
		&& \forall 1 \leq l \leq d, \, \sabs{a} \geq 1, \, a \in \N_0^d 
		. 
	\end{align*}
\end{assumption}
In particular, this implies that the magnetic field $B = \dd A$ satisfies Assumption~\ref{asymptotics:assumption:bounded_fields}, \ie its components are  $\mathcal{BC}^{\infty}$~functions. It is conceptually useful to introduce the line integral 
\begin{align}
	\Gamma^A(x,y) := \int_0^1 \dd s \, A \bigl (x + s(y - x) \bigr )
\end{align}
which is related to the circulation $\Gamma^A([x,y]) = (y - x) \cdot \Gamma^A(x,y)$; similarly, the scaled line integral is defined as $\GAe([x,y]) =: (y - x) \cdot \GAe(x,y)$. This allows us to rewrite the integral kernel of a magnetic pseudodifferential operator $\OpAel(f)$ for $f \in \Hoerr{m}$ as 
\begin{align}
	K_{f,\eps,\lambda}(x,y) = \int_{{\R^d}^*} \dd \eta \, e^{- i y \cdot \eta} f \bigl ( \tfrac{\eps}{2} (x+y) , \eta - \lambda \GAe(x,y) \bigr ) . \label{asymptotics:magWQminSub:magneticIntKernelRewritten}
\end{align}
If we had used minimal substitution instead, then we would have to replace the line integral $\GAe(x,y)$ by its mid-point value $A \bigl ( \tfrac{\eps}{2} (x+y) )$. 

\begin{thm}[\cite{Iftimie_Mantiou_Purice:magnetic_psido:2006}]\label{asymptotics:thm:magWQusualWQ}
	Assume the magnetic field satisfies Assumption~\ref{asymptotics:magWQminSub:assumption:stricterBA}. Then for any $f \in \Hoerr{m}$ there exists a \emph{unique} $g \in \Hoerr{m}$ such that $\OpAel(f) = \Ope(g \circ \minsAl)$. $g$ can be expressed as an asymptotic series $g \asymp \sum_{n = 0}^{\infty} \sum_{k = 1}^n \eps^n \lambda^k \, g_{n,k}$, where $g_{n,k} \in \Hoerr{m - (n+k) \rho}$ for all $n \geq 1$, and 
	\begin{align}
		\sum_{k = 1}^n \lambda^k \, &g_{n,k}(x,\xi) 
		= \\
		&
		= \eps^{-n} \sum_{\sabs{a} = n} \frac{1}{a !} \, \bigl ( i \partial_y \bigr )^{a} \Bigl ( \partial_{\xi}^{a} f \bigl ( x , \xi - \lambda \Gamma^A(x + \tfrac{\eps}{2} \, y , x - \tfrac{\eps}{2} \, y) + \lambda A(x) \bigr ) \Bigr ) \Bigl . \Bigr \rvert_{y = 0} 
		. 
		\notag 
	\end{align}
	Only terms with even powers of $\eps$ contribute, \ie $g_{n,k} = 0$ for all $n \in 2 \N_0 + 1$, $1 \leq k \leq n$. In particular we have $g_{0,0} = f$, $g_{1,0} = 0$, $g_{1,1} = 0$ and $f - g \in \Hoerr{m - 3 \rho}$. 
\end{thm}
\begin{remark}
	The reason that only even powers of $\eps$ contribute can be traced back to the symmetry of $\GAe(x,y) = + \GAe(y,x)$. Note that this is consistent with what was said in the introduction, $\GAe([x,y]) = (y - x) \cdot \GAe([x,y])$ is indeed odd. 
\end{remark}
\begin{proof}
	The proof is virtually identical to the proof of Proposition~6.7 in \cite{Iftimie_Mantiou_Purice:magnetic_psido:2006}; we will only specialize the formal part to the present case, the rigorous justification can be found in the reference. 
	
	For a symbol $f \in \Hoerr{m}$, the integral kernel of its magnetic quantization is given by equation~\eqref{asymptotics:magWQminSub:magneticIntKernelRewritten}. On the other hand, it is clear how to invert $\OpAel$ for $\lambda = 0$, $A \equiv 0$: we apply the non-magnetic Wigner transform $\mathcal{W}_{\eps} := \mathcal{W}^{A \equiv 0}_{\eps,\lambda=0}$ to the magnetic integral kernel of $\OpAel(f)$: 
	\begin{align*}
		\mathcal{W}_{\eps} K_{f,\eps,\lambda} (X) &= \int_{\R^d} \dd y \, e^{-i y \cdot \xi} \, K_{f,\eps,\lambda} \bigl ( \tfrac{x}{\eps} + \tfrac{y}{2} , \tfrac{x}{\eps} - \tfrac{y}{2} \bigr ) \\ 
		&= \int_{\R^d} \dd y \, \int_{{\R^d}^*} \dd \eta \, e^{i y \cdot \eta} \, f \bigl ( x , \eta + \xi - \lambda \Gamma^A \bigl ( x + \tfrac{\eps}{2} \, y , x - \tfrac{\eps}{2} \, y \bigr ) \bigr ) 
	\end{align*}
	Since we have a separation of scales, we can expand $\Gamma^A \bigl (x + \tfrac{\eps}{2} \, y , x - \tfrac{\eps}{2} \, y \bigr )$ in powers of $\eps$ up to some \emph{even} $N$. We will find that only \emph{even} powers of $\eps$ survive -- which immediately explains the absence of the first-order correction, 
	\begin{align*}
		\Gamma^A \bigl ( x + \tfrac{\eps}{2} \, y , &x - \tfrac{\eps}{2} \, y \bigr ) = \int_{-\nicefrac{1}{2}}^{+\nicefrac{1}{2}} \dd s \, \Biggl ( \sum_{n = 0}^N \eps^n s^n \, \sum_{\sabs{a} = n} \partial_x^{a} A(x) \, y^{a} + R_N(s,x,y) \Biggr ) \\
		&= \sum_{n = 0}^{\nicefrac{N}{2}} \eps^{2n} \left ( \frac{1}{2} \right )^{2n} \frac{1}{2n + 1} \, \sum_{\sabs{a} = 2n} \partial_x^{a} A(x) \, y^{a} + \int_{-\nicefrac{1}{2}}^{+\nicefrac{1}{2}} \dd s \, R_N(s,x,y) 
		. 
	\end{align*}
	The remainder is bounded since it is the integral of a $\mathcal{C}^{\infty}_{\mathrm{pol}}$ function over the compact set $[-\nicefrac{1}{2} , + \nicefrac{1}{2}] \times [0,1]$. In any event, The exact value will not matter if we choose $N$ large enough as we set $y = 0$ in the end. 
	
	A Taylor expansion of $f \bigl ( x , \eta + \xi - \lambda \Gamma^A \bigl ( x + \tfrac{\eps}{2} \, y , x - \tfrac{\eps}{2} \, y \bigr ) \bigr )$ around $\eta - \lambda \Gamma^A$ and some elementary integral manipulations formally yield for the $n$th term of the expansion 
	\begin{align}
		\eps^n \, \sum_{k = 1}^n \lambda^k \, &g_{n,k}(x,\xi - \lambda A(x)) 
		\label{asymptotics:eqn:expansion_f_g_n_k}
		\\
		&= \sum_{\sabs{a} = n} 
		\frac{1}{a !} \, \bigl ( i \partial_y \bigr )^{a} \Bigl ( \partial_{\xi}^{a} f \bigl ( x , \xi - \lambda \Gamma^A(x + \tfrac{\eps}{2} \, y , x - \tfrac{\eps}{2} \, y) \bigr ) \Bigr ) \Bigl . \Bigr \rvert_{y = 0} 
		\notag 
	\end{align}
	where we substitute the expansion above for $\Gamma^A$. Each derivative in $y$ will give one factor of $\eps$, \ie we will have $n$ altogether. On the other hand, we have \emph{at least} $1$ and \emph{at most} $n$ factors of $\lambda$. Only even powers in $\eps$ contribute, because the expansion of $\Gamma^A \bigl ( x + \tfrac{\eps}{2} \, y , x - \tfrac{\eps}{2} \, y \bigr )$ contains only \emph{even} powers of $y$. Furthermore, all terms in this sum are bounded functions in $x$, because derivatives of $A$ are bounded by assumption. 
	\medskip

	\noindent
	To show that $g_{n,k}$ is in symbol class $\Hoerr{m - (n+k) \rho}$, we need to have a closer look at equation~\eqref{asymptotics:eqn:expansion_f_g_n_k}: the only possibility to get $k$ factors of $\lambda$ is to derive $\partial_{\xi}^{a} f \bigl ( x , \xi - \lambda \Gamma^A(x + \tfrac{\eps}{2} \, y , x - \tfrac{\eps}{2} \, y) \bigr )$ $k$ times with respect to $y$. Each of these $y$ derivatives becomes an additional derivative of $\partial_{\xi}^{a}f$ with respect to momentum. Hence, there is a total of $\abs{a} + k = n + k$ derivatives with respect to $\xi$. 
	\medskip

	\noindent
	The rigorous justification that these integrals exist can be found in \cite[Proposition~6.7]{Iftimie_Mantiou_Purice:magnetic_psido:2006}. 
\end{proof}
\begin{remark}\label{asymptotics:magWQminSub:epsExpansion}
	If we are interested in a one-parameter expansion in $\eps$ only, then 
	\begin{align*}
		g_n(X) := \eps^{-n} \sum_{\sabs{a} = n} \frac{1}{a !} \, \bigl ( i \partial_y \bigr )^{a} \Bigl ( \partial_{\xi}^{a} f \bigl ( x , \xi - \lambda \Gamma^A(x + \tfrac{\eps}{2} \, y , x - \tfrac{\eps}{2} \, y) + \lambda A(x) \bigr ) \Bigr ) \Bigl . \Bigr \rvert_{y = 0} 
	\end{align*}
	gives the $n$th order correction in $\eps$. 
\end{remark}
\begin{prop}[\cite{Iftimie_Mantiou_Purice:magnetic_psido:2006}]\label{asymptotics:thm:usualWQmagWQ}
	The converse statement also holds: if the magnetic field satisfies Assumption~\ref{asymptotics:magWQminSub:assumption:stricterBA}, then for each $g \in \Hoerr{m}$ there exists a unique $f \in \Hoerr{m}$ such that $\Ope(g \circ \minsAl) = \OpAel(f)$, $f \asymp \sum_{n = 0}^{\infty} \sum_{k = 1}^n \eps^n \lambda^k \, f_{n,k}$, $f_{n,k} \in \Hoerr{m - (n+k) \rho}$, can be expressed as a formal power series in $\eps$ where the $n$th term is given by 
	\begin{align}
		\sum_{k = 1}^n \lambda^k \, &f_{n,k}(x,\xi) = 
		\\
		&= \eps^{-n} \sum_{\abs{a} = n} \frac{1}{a !} \, (i \partial_{y})^{a} \bigl ( \partial_{\xi}^{a} f \bigr ) \bigl ( x , \xi + \lambda \GAe(x - \nicefrac{y}{2} , x + \nicefrac{y}{2}) - \lambda A(x) \bigr ) \bigl . \bigr \rvert_{y = 0} 
		\notag 
	\end{align}
	In particular we have $f_{0,0} = g$, $f_{1,0} = 0$, $f_{1,1} = 0$ and $g - f \in \Hoerr{m - 3 \rho}$. 
\end{prop}
\begin{proof}
	This proof works along the same lines: one magnetically Wigner-transforms the kernel of the operator $\Ope(f \circ \minsAl)$, we refer to  \cite[Proposition~6.9]{Iftimie_Mantiou_Purice:magnetic_psido:2006} for details. 
\end{proof}
%


\chapter{Magnetic Space-adiabatic Perturbation Theory} 
\label{magsapt}
The motivation to rigorously derive an asymptotic expansion of the magnetic Weyl product came from the side of applications: in 2002, Panati, Spohn and Teufel have presented a generic recipe to derive effective dynamics\index{effective dynamics} for systems with two inherent scales \cite{PST:sapt:2002}. Their technique, space-adiabatic perturbation theory\index{space-adiabatic perturbation theory}, has been successfully applied to numerous physical problems, \eg Born-Oppenheimer systems\index{Born-Oppenheimer systems} \cite{PST:sapt:2002,PST:Born-Oppenheimer:2007}, the non- and semirelativistic limit of the Dirac equation\index{Dirac equation} \cite{Teufel:adiabatic_perturbation_theory:2003,Lein:two_parameter_asymptotics:2008,Fuerst_Lein:scaling_limits_Dirac:2008} and piezoelectricity \cite{Lein:polarization:2005,Panati_Sparber_Teufel:polarization:2006}. 

The derivation uses usual, non-magnetic Weyl calculus to derive an effective hamiltonian order-by-order in an adiabatic parameter $\eps$ via recursion relations and the asymptotic expansion of the Weyl product. This effective hamiltonian\index{effective hamiltonian} then generates the effective dynamics on some smaller reference Hilbert space. In many cases, the electromagnetic field \emph{is} the perturbation and usual Weyl calculus is not well-adapted to this situation since its formalism is oblivious to the magnetic field. In order to make derivations into theorems, additional technical conditions have to be placed on $B$: one often assumes that the magnetic field is such that there exists a vector potential $A$ whose components are of class $\BCont^{\infty}(\R^d)$. If one tries a little harder, one can weaken this assumption: $B$ must be such that there exists a smooth vector potential whose derivatives are all bounded, \ie for all $a \in \N_0^d$, $\abs{a} \geq 1$, there exists $C_a > 0$ such that $\abs{\partial_x^a A_j(x)} \leq C_a$ holds for all $x \in \R^d$. This allows the vector potential to grow linearly and one can cover the case of constant magnetic field. These conditions imposed on $A$ are unnecessary and -- from the point of view of magnetic Weyl calculus -- unnatural. 

Replacing usual Weyl calculus with its magnetic variant solves this problem. The pseudodifferential parts of the proofs in \cite{PST:sapt:2002} hinge on the following facts: 
\begin{enumerate}[(i)]
	\item The ability to quantize Hörmander-class symbols \cite{Mantoiu_Purice:magnetic_Weyl_calculus:2004}. 
	\item A magnetic version of the Caldéron-Vaillancourt theorem \cite{Iftimie_Mantiou_Purice:magnetic_psido:2006}. 
	\item An asymptotic expansion of the magnetic Weyl product for Hörmander-class symbols (Chapter~\ref{asymptotics} and \cite{Lein:two_parameter_asymptotics:2008}). 
\end{enumerate}
The application of magnetic pseudodifferential techniques to space-adiabatic perturbation theory has been worked out in detail for two particular cases: the non- and semirelativistic limit of the Dirac equation was the first example. Expanding upon Fürst's diploma thesis \cite{Fuerst:Dirac:2008} where traditional Weyl calculus was still used, Fürst and I have shown how one can derive the non- and semirelativistic limit of the Dirac equation from first principles \cite{Fuerst_Lein:scaling_limits_Dirac:2008}. A choice of scaling in the momentum operator determines which limit one obtains in the end. 

The second application which I want to present here in detail is a joint work with Giuseppe De Nittis \cite{DeNittis_Lein:Bloch_electron:2009}. It is concerned with a standard model of a crystalline solid: if one neglects electron-electron interactions, one can start with a simple single-particle model where an electron moves in the electric field that is generated by the ionic cores and all other electrons. This interaction is subsumed by a lattice-periodic potential. One is then interested in the currents induced by external, macroscopic electromagnetic fields.

\section{The model} 
\label{magsapt:intro}
The model hamiltonian we are considering is given by 
\begin{align}
	\hat{H} \equiv \hat{H}(\eps,\lambda) := \tfrac{1}{2} \bigl ( -i \nabla_x - \lambda A(\eps \hat{x}) \bigr )^2 + V_\Gamma(\hat{x}) + \phi(\eps \hat{x}) 
	\label{magsapt:intro:eqn:full_hamiltonian}
\end{align}
and acts on $L^2(\R^d_x)$. Just as in Chapter~\ref{asymptotics}, $\eps$ is the semiclassical parameter and $\lambda$ quantifies the coupling to the magnetic field. We will always make the following assumptions on the external fields: 
\begin{assumption}[Electromagnetic fields]\label{magsapt:intro:assumption:em_fields}
	We assume that the components of the external (macroscopic) magnetic field $B$ and the electric potential $\phi$ are $\BCont^{\infty}(\R^d)$ functions, \ie smooth, bounded functions with bounded derivatives to any order. 
\end{assumption}

\begin{remark}
	All vector potentials $A$ associated to magnetic fields $B = \dd A$ with components in $\BCont^{\infty}(\R^d)$ are always assumed to have components in $\Cont^{\infty}_{\mathrm{pol}}(\R^d)$. This is always possible as one could pick the transversal gauge, 
	\begin{align*}
		A_k(x) := - \sum_{j = 1}^n \int_0^1 \dd s \, B_{kj}(s x) \, s x_j 
		. 
	\end{align*}
	This is to be contrasted with the original work of Panati, Spohn and Teufel \cite{PST:effective_dynamics_Bloch:2003} where the \emph{vector potential} had to have components in $\BCont^{\infty}(\R^d)$. 
\end{remark}
The potential generated by the nuclei and all other electrons $V_{\Gamma}$ is periodic with respect to the crystal lattice \cite{Cances_Deleurence_Lewin:modelling_local_defects_crystals:2008,Cances_Deleurence_Lewin:non-perurbative_defects_crystal:2008}
\begin{align}
	\Gamma := \Bigl \{ \gamma \in \R^d \; \vert \; \mbox{$\gamma = \sum_{j = 1}^d \alpha_j e_j$} , \; \alpha_j \in \Z \Bigr \}
	\label{magsapt:intro:eqn:lattice}
\end{align}
and assumed to be infinitesimally bounded with respect to $- \tfrac{1}{2} \Delta_x$. By Theorem~XIII.96 in \cite{Reed_Simon:M_cap_Phi_4:1978}, this is ensured by the following 
\begin{assumption}[Periodic potential]\label{magsapt:intro:assumption:V_Gamma}
	We assume that $V_{\Gamma}$ is $\Gamma$-periodic, \ie $V_{\Gamma}(\cdot + \gamma) = V_{\Gamma}$ for all $\gamma \in \Gamma$, and $\int_{\WS} \dd y \, \abs{V_{\Gamma}(y)} < \infty$. 
\end{assumption}
Under these assumptions, $\hat{H}$ defines an essentially selfadjoint operator on $\Cont^{\infty}_0(\R^d_x) \subset L^2(\R^d_x)$. 

The dual lattice $\Gamma^*$ is spanned by the \emph{dual basis} $\{e_1^\ast,\ldots,e_d^\ast \}$, \ie the set of vectors which satisfy $e_j \cdot e_k^{\ast} = 2 \pi \delta_{kj}$. The assumption on $V_{\Gamma}$ ensures the unperturbed periodic hamiltonian 
\begin{align}
	\hat{H}_{\mathrm{per}} = \tfrac{1}{2} (- i \nabla_x)^2 + V_{\Gamma}
	\label{magsapt:intro:eqn:per_ham}
\end{align}
defines a selfadjoint operator on the second Sobolev space $H^2(\R^d)$ and gives rise to Bloch bands in the usual manner (see Chapter~\ref{magsapt:rewriting:Zak}): the Bloch-Floquet-Zak transform (\BZak~transform) fibers $\hat{H}_{\mathrm{per}}$ into 
\begin{align*}
	\Zak \hat{H}_{\mathrm{per}} \Zak^{-1} =: \hat{H}_{\mathrm{per}}^{\Zak} = \int_{\BZ}^{\oplus} \dd k \, \Hper^{\Zak}(k) := \int_{\BZ}^{\oplus} \dd k \, \Bigl ( \tfrac{1}{2} \bigl ( - i \nabla_y + k \bigr )^2 + V_{\Gamma}(y) \Bigr )
\end{align*}
where we have introduced the Brillouin zone\index{Brillouin zone} 
\begin{align}
	\BZ := \Bigl \{ k \in \R^d \; \vert \; \mbox{$k = \sum_{j = 1}^d \alpha_j e^\ast_j$} , \; \alpha_j \in [-\nicefrac{1}{2},+\nicefrac{1}{2}] \Bigr \}
	\label{magsapt:intro:eqn:BZ} 
\end{align}
as fundamental cell in reciprocal space.\index{Brillouin zone} For each $k \in \BZ$, the eigenvalue equation 
\begin{align*}
	\Hper(k) \varphi_n(k) = E_n(k) \, \varphi_n(k) 
	, 
	&& 
	\varphi_n(k) \in L^2(\T^d_y) := L^2(\R^d / \Gamma)
	, 
\end{align*}
is solved by the Bloch function associated to the $n$th band. Assume for simplicity we are given a band $\Eb$ which does not intersect or merge with other bands (\ie there is a \emph{local} gap in the sense of Assumption~\ref{magsapt:mag_sapt:mag_wc:defn:gap_condition}). Then common lore is that transitions to other bands are exponentially suppressed and the effective dynamics for an initial state localized in the eigenspace associated to $\Eb$ is generated by $\Eb(- i \nabla_x)$ \cite{Grosso_Parravicini:solid_state_physics:2003,Ashcroft_Mermin:solid_state_physics:2001}. 
\medskip

\noindent
If we switch on a constant magnetic field, no matter how weak, the Bloch bands are gone as there is no \BZak decomposition with respect to $\Gamma$ for hamiltonian~\eqref{magsapt:intro:eqn:full_hamiltonian}. As a matter of fact, the spectrum of $\hat{H}$ is a Cantor set \cite{Gruber:noncommutative_Bloch:2001} if the flux through the Wigner-Seitz cell\index{Wigner-Seitz cell} 
\begin{align}
	\WS := \Bigl \{ y \in \R^d \; \vert \; \mbox{$y = \sum_{j = 1}^d \alpha_j e_j$} , \; \alpha_j \in [-\nicefrac{1}{2},+\nicefrac{1}{2}] \Bigr \}
	\label{magsapt:intro:eqn:WS} 
\end{align}
is irrational. Even if the flux through the unit cell is rational, we recover only \emph{magnetic} Bloch bands that are associated to a larger lattice $\Gamma' \supset \Gamma$. A natural question is if it is at all possible to see signatures of \emph{nonmagnetic} Bloch bands if the applied magnetic field is weak?

Our main result, Theorem~\ref{magsapt:modulation_field:effective_dynamics:thm:adiabatic_decoupling}, answers this question in the positive in the following sense: if the electromagnetic field varies on the macroscopic level, \ie $\eps \ll 1$, then to leading order the dynamics is still generated by $\Eb \bigl ( -i \nabla_x - \lambda A(\eps \hat{x}) \bigr ) + \phi \bigl (\eps \hat{x} \bigr )$ (defined as the magnetic Weyl quantization of $\Eb(k) + \Phi(r)$ via equation~\eqref{magsapt:rewrite:magnetic_weyl_calculus:eqn:mag_wq_torus}). Hence, the dynamics are dominated by the Bloch bands even in the presence of a weak magnetic field. Furthermore, we can derive corrections to any order in $\eps$ in terms of Bloch bands, Bloch functions, the magnetic \emph{field} and the electric \emph{potential}. We do not need to choose a `nice' vector potential for $B$, in fact, in all of the calculations only the magnetic field $B$ enters. 

Let us now explain why we have chosen to include three expansion parameters in the initial hamiltonian $\hat{H} \equiv \hat{H}(\eps,\lambda)$. Our goal is to model an experimental setup that applies an external, \ie macroscopic electric and magnetic field. The parameter $\eps \ll 1$ relates the microscopic scale as given by the crystal lattice to the scale on which the external fields vary. \emph{We always assume $\eps$ to be small.} It is quite easy to fathom an apparaturs where electric and magnetic field can be regulated separately by, say, two dials. We are interested in the case where we can \emph{selectively switch off the magnetic field}. Assume, we can regulate the strength of the magnetic field by varying the relative amplitude $\lambda \leq 1 $, 
\begin{align*}
	B^{\eps,\lambda}(x) :=& \eps \lambda B(\eps x) 
	\\
	\El^{\eps}(x) :=& \eps \El(\eps x) 
	. 
\end{align*}
Then we can take the limit $B^{\eps,\lambda} \rightarrow 0$ without changing the external electric field $\El^{\eps}$. 


\section{Rewriting the problem} 
\label{magsapt:rewriting}
As a preliminary step, we will rewrite the problem: first, we extract the Bloch band picture via the \BZak transform and then we reinterpret the \BZak-transformed hamiltonian as magnetic quantization of an operator-valued symbol. We insist we only rephrase the problem, \emph{no additional assumptions} are introduced.

\subsection{The Bloch-Floquet-Zak transform} 
\label{magsapt:rewriting:Zak}
Usually, one would exploit lattice periodicity by going to the Fourier basis: each $\Psi \in \Schwartz(\R^d_x) \subset L^2(\R^d_x)$ is mapped onto 
\begin{align*}
	( \mathcal{F} \Psi )(k,y) := \sum_{\gamma \in \Gamma} e^{- i k \cdot y} \, \Psi(y + \gamma) 
\end{align*}
and the corresponding representation is usually called Bloch-Floquet representation\index{Bloch-Floquet transform}. It is easily checked that 
\begin{align*}
	( \mathcal{F} \Psi )(k - \gamma^*,y) &= ( \mathcal{F} \Psi )(k,y) 
	&& \forall \gamma^* \in \Gamma^* 
	\\
	( \mathcal{F} \Psi )(k,y - \gamma) &= e^{- i k \cdot \gamma} ( \mathcal{F} \Psi )(k,y) 
	&& \forall \gamma \in \Gamma 
\end{align*}
holds and $\mathcal{F} \Psi$ can be written as 
\begin{align*}
	( \mathcal{F} \Psi )(k,y) = e^{i k \cdot y} \, u(k,y)
\end{align*}
where $u(k,y)$ is $\Gamma$-periodic in $y$ and $\Gamma^*$-periodic up to a phase in $k$. For technical reasons, we prefer to use a variant of the Bloch-Floquet transform introduced by Zak \cite{Zak:dynamics_Bloch_electrons:1968}\index{Zak transform} which maps $\Psi \in \Schwartz(\R^d_x)$ onto u, 
\begin{align}
	( \Zak \Psi )(k,y) := \sum_{\gamma \in \Gamma} e^{- i k \cdot (y + \gamma)} \, \Psi(y + \gamma) 
	\label{magsapt:rewrite:Zak:eqn:Zak_transform} 
	. 
\end{align}
The \BZak transform has the following periodicity properties: 
\begin{align}
	( \Zak \Psi ) (k - \gamma^*,y) &= e^{+ i \gamma^* \cdot y} \, ( \Zak \Psi ) (k,y) 
	=: \tau(\gamma^*) \, ( \Zak \Psi ) (k,y) 
	&& \forall \gamma^* \in \Gamma^* 
	\\ 
	( \Zak \Psi ) (k,y - \gamma) &= ( \Zak \Psi ) (k,y) 
	&& \forall \gamma \in \Gamma
	\notag 
\end{align}
$\tau$ is a unitary representation of the group of dual lattice translations $\Gamma^*$. By density, $\Zak$ immediately extends to $L^2(\R^d_x)$ and it maps it unitarily onto 
\begin{align}
	\Htau := \Bigl \{ \psi \in L^2_{\mathrm{loc}} \bigl ({\R^d_k},L^2(\T^d_y) \bigr ) \; \big \vert \;  \psi (k-\gamma^\ast) = \tau(\gamma^\ast) \, \psi(k) \mbox{ a.~e. } \forall \gamma^* \in \Gamma^* \Bigr \} 
	\label{magsapt:rewriting:Zak:eqn:defnition_Htau} 
	, 
\end{align}
which is equipped with the scalar product 
\begin{align*}
	\bscpro{\varphi}{\psi}_{\tau} := \int_{\BZ} \dd k \, \bscpro{\varphi(k)}{\psi(k)}_{L^2(\T^d_y)} 
	. 
\end{align*}
It is obvious from the definition that the left-hand side does not depend on the choice of the unit cell $\BZ$ in reciprocal space. The \BZak representation of momentum and position operator on $L^2(\R^d_x)$, equipped with the obvious domains, can be computed directly, 
\begin{align}
	\Zak (- i \nabla_x) \Zak^{-1} &= \id_{L^2(\BZ)} \otimes (- i \nabla_y) + \hat{k} \otimes \id_{L^2(\T^d_y)} 
	\equiv -i \nabla_y + k 
	\label{magsapt:rewrite:Zak:eqn:Zak_transformed_building_blocks}
	\\
	\Zak \hat{x} \Zak^{-1} &= i \nabla_k^{\tau} 
	, 
	\notag 
\end{align}
where we have used the identification $\Htau \cong L^2(\BZ) \otimes L^2(\T^d_y)$. The superscript $\tau$ on $i \nabla_k^{\tau}$ indicates that the operator's domain $\Htau \cap H^1_{\mathrm{loc}} \bigl ( \R^d , L^2(\T^d_y) \bigr )$ consists of $\tau$-equivariant functions (see equation~\eqref{magsapt:rewriting:Zak:eqn:defnition_Htau}). The \BZak transformed domain for momentum $-i \nabla_y + k$ is $L^2(\BZ) \otimes H^1(\T^d_y)$. Since the phase factor $\tau$ depends on $y$, the \BZak transform of $\hat{x}$ does not factor --- unless we consider $\Gamma$-periodic functions, then we have 
\begin{align*}
	\Zak V_{\Gamma}(\hat{x}) \Zak^{-1} = \id_{L^2(\BZ)} \otimes V_{\Gamma}(\hat{y}) \equiv V_{\Gamma}(\hat{y}) 
	. 
\end{align*}
Equations~\eqref{magsapt:rewrite:Zak:eqn:Zak_transformed_building_blocks} immediately give us the \BZak transform of $\hat{H}$, namely 
\begin{align}
	\hat{H}^{\Zak} := \Zak \hat{H} \Zak^{-1} = \tfrac{1}{2} \bigl ( -i \nabla_y + k - \lambda A(i \eps \nabla_k^{\tau}) \bigr )^2 + V_{\Gamma}(\hat{y}) + \phi(i \eps \nabla_k^{\tau}) 
	\label{magsapt:rewrite:Zak:eqn:Zak_transformed_hamiltonian} 
	, 
\end{align}
which defines an essentially selfadjoint operator on $\Zak \Cont^{\infty}_0(\R^d_x)$. If the external electromagnetic field vanishes, the hamiltonian 
\begin{align}
	\hat{H}_{\mathrm{per}}^{\Zak} := \Zak \hat{H}_{\mathrm{per}} \Zak^{-1} = \int_{\BZ}^{\oplus} \dd k \, \Hper^{\Zak}(k) 
\end{align}
fibers into a family of operators on $L^2(\T^d_y)$ indexed by crystal momentum $k \in \BZ$. $\tau$-equivariance relates $\Hper^{\Zak}(k - \gamma^*)$ and $\Hper^{\Zak}(k)$ via 
\begin{align*}
	\Hper^{\Zak}(k - \gamma^*) = \tau(\gamma^*) \, \Hper^{\Zak}(k) \, \tau(\gamma^*)^{-1} 
	&& \forall \gamma^* \in \Gamma^* 
\end{align*}
which, among other things, ensures that Bloch bands $\{ E_n \}_{n \in \N}$, \ie the solutions to the eigenvalue equation 
\begin{align*}
	\Hper^{\Zak}(k) \varphi_n(k) = E_n(k) \, \varphi_n(k) 
	, 
	&& \varphi_n(k) \in L^2(\T^d_y)
	, 
\end{align*}
are $\Gamma^*$-periodic functions. Standard arguments show that $\Hper^{\Zak}(k)$ has purely discrete spectrum for all $k \in \BZ$ and if Bloch bands are ordered by magnitude, they are smooth functions away from band crossings. Similarly, the Bloch functions\index{Bloch function} $k \mapsto \varphi_n(k)$ are smooth if the associated energy band $E_n$ does not intersect with or touch others \cite{Reed_Simon:M_cap_Phi_4:1978}. 

The next subsection shows that the effect of introducing an external electromagnetic field can be interpreted as ``replacing'' the direct integral with the \emph{magnetic} quantization of $\Hper^{\Zak} + \phi$. 


\subsection{Equivariant magnetic Weyl calculus} 
\label{magsapt:rewriting:magnetic_weyl_calculus}
For technical reasons, we must adapt magnetic Weyl calculus to deal with equivariant, unbounded operator-valued functions. We follow the general strategy outlined in \cite{PST:effective_dynamics_Bloch:2003}, but we need to be more careful as the roles of $\Qe$ and $\PA$ are not interchangeable if $B \neq 0$. We would like to reuse results for Weyl calculus on $T^* \R^d$ -- in particular, the two-parameter expansion of the product (equation~\eqref{asymptotics:expansions:eqn:Fourier_form_magnetic_Weyl_product}). Consider the building block kinetic operators macroscopic position $\Reps$ and crystal momentum $\KA$, 
\begin{align}
	\Reps &= i \eps \nabla_k \otimes \id_{L^2(\T^d_y)} \equiv i \eps \nabla_k 
	\label{magsapt:rewriting:magnetic_weyl_calculus:eqn:building_blocks_Td} 
	\\
	\KA &= \hat{k} - \lambda A(\Reps) 
	\notag 
	, 
\end{align}
in momentum representation: they define selfadjoint operators whose domains are dense in $L^2_{\tau'} \bigl ( \R^d_k , L^2(\T^d_y) \bigr )$ where $\tau'$ stands for either $\tau : \gamma^* \mapsto e^{- i \gamma^* \cdot \hat{y}}$ or $1 : \gamma^* \mapsto 1$. The elements of this Hilbert space can be considered as vector-valued tempered distributions with special properties as $L^2_{\tau'} \bigl ( \R^d_k , L^2(\T^d_y) \bigr )$ can be continuously embedded into $\Schwartz' \bigl ( \R^d_k , L^2(\T^d_y) \bigr )$. For simplicity, let us ignore questions of domains and assume that $h \in \BCont^{\infty} \bigl ( T^* \R^d_x , \mathcal{B} \bigl ( L^2(\T^d_y) \bigr ) \bigr )$ is a \emph{bounded operator}-valued function. Then its magnetic Weyl quantization\index{Weyl quantization!equivariant magnetic} 
\begin{align}
	\Opk^A (h) := \frac{1}{(2\pi)^d} \int \dd r \int \dd k \, (\Fs h)(r,k) \, \WeylSys^A(r,k) 
	\label{magsapt:rewrite:magnetic_weyl_calculus:eqn:mag_wq_torus}
\end{align}
defines a continuous operator from $\Schwartz \bigl ( \R^d_k , L^2(\T^d_y) \bigr )$ to itself which has a continuous extension as an operator from $\Schwartz' \bigl ( \R^d_k , L^2(\T^d_y) \bigr )$ to itself \cite[Proposition~21]{Mantoiu_Purice:magnetic_Weyl_calculus:2004}. Here, the corresponding Weyl system 
\begin{align*}
	\WeylSys^A(r,k) := e^{- i \sigma((r,k),(\Reps,\KA))} \otimes \id_{L^2(\T^d_y)} \equiv e^{- i (k \cdot \Reps - r \cdot \KA)}
\end{align*}
is defined in terms of the building block operators $\KA$ and $\Reps$ and acts trivially on $L^2(\T^d_y)$. The Weyl product $f \magW g$ of two suitable distributions associated to the quantization $\Opk^A$ is also given by a suitable reinterpretation of equation~\eqref{asymptotics:expansions:eqn:Fourier_form_magnetic_Weyl_product} as $f$ and $g$ are now operator-valued functions. Furthermore, we can also develop $f \magW g$ asymptotically in $\eps$ and $\lambda$, see~Theorem~\ref{asymptotics:thm:asymptotic_expansion}. To see this, we remark that the difference between the products associated to $\Opx^A$ and $\Opk^A$ is two-fold: first of all, $\Opx^A$ is a position representation while $\Opk^A$ is a momentum representation. Let $\Opk^{\prime \, A}$ be the magnetic Weyl quantization defined with respect to $\Reps' := \mathfrak{F}^{-1} \Reps \mathfrak{F} = \eps \hat{r}$ and $\mathsf{K}^{\prime \, A} := \mathfrak{F}^{-1} \KA \mathfrak{F} = - i \nabla_r - \lambda A(\eps \hat{r})$, \ie the position representation. By Proposition~\ref{asymptotics:scalings:prop:equivalence_Weyl_calculi}, these two are unitarily equivalent representations of the same algebra of observables and thus have the same Weyl product. 

Secondly, the functions which are to be quantized by $\Opx^A$ and $\Opk^A$ take values in $\C$ and the bounded operators on $L^2(\T^d_y)$, respectively. The interested reader may check the proofs regarding the various properties of the product $\magW$ in \cite{Mantoiu_Purice:magnetic_Weyl_calculus:2004,Iftimie_Mantiou_Purice:magnetic_psido:2006} and Chapter~\ref{asymptotics} can be generalized to accommodate operator-valued functions, including Hörmander symbols.\index{Hörmander symbols!operator-valued} 
\begin{defn}[Hörmander symbols $\Hoerr{m}
 \bigl ( \mathcal{B}(\Hil_1,\Hil_2) \bigr )$\index{Hörmander symbols!operator-valued}]
	Let $m \in \R$, $\rho \in [0,1]$ and $\Hil_1$, $\Hil_2$ be separable Hilbert spaces. Then a function $f$ is said to be in $\Hoerr{m}
 \bigl ( \mathcal{B}(\Hil_1,\Hil_2) \bigr )$ if and only if for all $a , \alpha \in \N_0^d$ the seminorms 
	\begin{align*}
		\bnorm{f}_{m , a \alpha} := \sup_{(x,\xi) \in \PSpace} \sqrt{1 + \xi^2}^{-(m - \sabs{\alpha} \rho)} \bnorm{\partial_x^a \partial_{\xi}^{\alpha} f(x,\xi)}_{\mathcal{B}(\Hil_1,\Hil_2)} < \infty 
	\end{align*}
	are finite where $\snorm{\cdot}_{\mathcal{B}(\Hil_1,\Hil_2)}$ denotes the operator norm on $\mathcal{B}(\Hil_1,\Hil_2)$. In case $\rho = 1$, one also writes $\Hoer{m} := \Hoermr{m}{1}$
\end{defn}
Hörmander symbols which have an expansion in $\eps$ that is uniform in the small parameter are called semiclassical. 
\begin{defn}[Semiclassical symbols $\SemiHoermr{m}{\rho} \bigl ( \mathcal{B}(\Hil_1,\Hil_2) \bigr )$]\label{magsapt:rewriting:magnetic_weyl_calculus:defn:semiclassical_symbol}
	A map $f : [0,\eps_0) \longrightarrow \Hoerr{m}
$, $\eps \mapsto f_{\eps}$ is called a semiclassical symbol of order $m \in \R$ and weight $\rho \in [0,1]$\index{semiclassical symbol!operator-valued}, that is $f \in \SemiHoermr{m}{\rho}$, if there exists a sequence $\{ f_n \}_{n \in \N_0}$, $f_n \in \Hoerr{m - n \rho}{\rho}$, such that for all $N \in \N_0$, one has 
	\begin{align*}
		\eps^{-N} \left ( f_{\eps} - \sum_{n = 0}^{N-1} \eps^n \, f_n \right ) \in \Hoerr{m - N \rho}{\rho} 
	\end{align*}
	uniformly in $\eps$ in the sense that for any $N \in \N_0$ and $a , \alpha \in \N_0^d$, there exist constants $C_{N a \alpha} > 0$ such that 
	\begin{align*}
		\norm{f_{\eps} - \sum_{n = 0}^{N-1} \eps^n \, f_n}_{m - N \rho , a \alpha} \leq C_{N a \alpha} \, \eps^N 
	\end{align*}
	holds for all $\eps \in [0,\eps_0)$. If $\rho = 1$, then one abbreviates $\SemiHoermr{m}{1}$ with $\SemiHoer{m}$. 
\end{defn}
Lastly, in applications, we $\tau$-equivariant symbols are of particular importance. 
\begin{defn}[$\tau$-equivariant symbols $\SemiTau \bigl ( \mathcal{B}(\Hil_1 , \Hil_2) \bigr )$\index{Hörmander symbol!$\tau$-equivariant}]\label{magsapt:rewriting:magnetic_weyl_calculus:defn:tau_equivariant_symbol}
	Assume $\tau_j : \Gamma^* \longrightarrow \mathcal{U}(\Hil_j)$, $j = 1,2$, are unitary $*$-representations of the group $\Gamma^*$. Then $f \in \SemiHoermr{0}{0}$ is $\tau$-equivariant, \ie an element of $\SemiTau \bigl ( \mathcal{B}(\Hil_1,\Hil_2) \bigr )$, if and only if 
	\begin{align*}
		f(k - \gamma^* , r) = \tau_2(\gamma^*) \, f(k,r) \, \tau_1(\gamma^*)^{-1} 
	\end{align*}
	holds for all $k \in \R^d$, $r \in \R^d$ and $\gamma^* \in \Gamma^*$. 
	
\end{defn}
\medskip

\noindent
Now the reader is in a position to translate the results derived in Appendix~B of \cite{Teufel:adiabatic_perturbation_theory:2003} to the context of magnetic Weyl calculus. These results are essential for the rigorous derivation of effective dynamics. We caution that we do not need to have position and momentum switch roles (see Warning in the reference), thus simplifying some of the arguments. 



\section{The magnetic Bloch electron as a space-adiabatic~problem} 
\label{magsapt:mag_sapt}
Our tool of choice to derive effective dynamics is space-adiabatic perturbation theory \linebreak\cite{PST:sapt:2002,PST:effective_dynamics_Bloch:2003,Teufel:adiabatic_perturbation_theory:2003} which uses pseudodifferential techniques to derive perturbation expansions order-by-order in a systematic fashion. We adapt their results by replacing ordinary Weyl calculus with \emph{magnetic} Weyl calculus. Adiabatic decoupling only hinges on $\eps \ll 1$ and does not rely on $\lambda$ to be small.

\subsection{Slow variation: the adiabatic point of view} 
\label{magsapt:mag_sapt:slow_var}
The insight of \cite{PST:effective_dynamics_Bloch:2003} was that the slow variation of the external electromagnetic field (quantified by $\eps \ll 1$) leads to a decoupling into slow (macroscopic) and fast (microscopic) degrees of freedom. This is characteristic of adiabatic systems\index{characteristics of adiabatic systems} whose three main features are\index{adiabatic trinity} 
\begin{enumerate}[(i)]
	\item A distinction between \emph{slow}\index{degrees of freedom!slow} and \emph{fast degrees of freedom}\index{degrees of freedom!fast}: the original (physical) Hilbert space $\Hil = L^2(\R^d_x)$ is decomposed unitarily into a slow and a fast component, $\Hslow \otimes \Hfast := L^2(\BZ) \otimes L^2(\T^d_y)$, in which the unperturbed hamiltonian is block diagonal (see diagram~\eqref{magsapt:mag_sapt:diagram:unperturbed}). The fast dynamics happen within a unit cell $\WS$ whereas the slow dynamics describe the motion across unit cells.
	\item A \emph{small, dimensionless parameter $\eps$} that quantifies the separation of spatial scales. In our situation, $\eps \ll 1$ relates the variation of the external electromagnetic field to the microscopic scale as given by the lattice constant. In addition, we have an additional parameter $\lambda$ which quantifies the coupling to the magnetic field. However, only the semiclassical parameter $\eps$ is crucial for adiabatic decoupling and $\lambda$ may even be set equal to $1$. 
	\item A \emph{relevant part of the spectrum}, \ie a subset of the spectrum which is separated from the remainder by a gap. We are interested in the dynamics associated to a family of Bloch bands $\{ E_n \}_{n \in \mathcal{I}}$ that does not intersect or merge with bands from the remainder of the spectrum.
\end{enumerate}
\begin{assumption}[Gap condition\index{gap condition}]\label{magsapt:mag_sapt:mag_wc:defn:gap_condition}
The spectrum of $\hat{H}_{\mathrm{per}}^{\Zak}$ satisfies the \emph{gap condition}, \ie \linebreak there exists a family of Bloch bands $\{ E_n \}_{n \in \mathcal{I}}$, $\mathcal{I} = [I_- , I_+] \cap \N_0$ such that\index{gap condition} 
 \begin{align*}
     \inf_{k \in \BZ} \mathrm{dist} \Bigl ( \bigcup_{n \in \mathcal{I}} \{ E_n(k) \} , \bigcup_{j \not\in \mathcal{I}} \{ E_j(k) \} \Bigr ) =: C_g > 0
     .
 \end{align*}
\end{assumption}
The spectral gap ensures that transitions from and to the relevant part of the spectrum are exponentially suppressed. Band crossings within the relevant part of the spectrum are admissible, though.
\medskip

\noindent
In the original publication, an additional assumption was made on the existence of a \emph{smooth}, $\tau$-equivariant basis, a condition that is equivalent to the triviality of a certain $U(N)$ bundle over the torus $\T^d_k$ where $N := \abs{\mathcal{I}}$ is the number of bands including multiplicity. At least for the physically relevant cases, \ie $d \leq 3$, Panati has shown that this is always possible for nonmagnetic Bloch bands \cite{Panati:triviality_Bloch_bundle:2006}. For $d \geq 4$, our results still hold if we add 
\begin{assumption}[Smooth frame ($d \geq 4$)]\label{magsapt:mag_sapt:mag_wc:assumption:smooth_frame}
	If $d \geq 4$, we assume there exists an ortho\-normal basis (called \emph{smooth frame}) $\{ \varphi_j(\cdot) \}_{j=1, \ldots , \abs{\mathcal{I}}}$ of whose elements are smooth and $\tau$-equi\-variant with respect to $k$, \ie $\varphi_j(\cdot-\gamma^{\ast}) = \tau(\gamma^{\ast}) \varphi_j(\cdot)$ for all $\gamma^{\ast} \in \Gamma^{\ast}$ and for all  $j \in \{ 1 , \ldots , \abs{\mathcal{I}} \}$. 
\end{assumption}
\medskip

\noindent
Let us consider the unperturbed case, \ie in the absence of an external electromagnetic field. Then the dynamics on $\Htau$ is generated by $\hat{H}_{\mathrm{per}}^{\Zak} = \int_{\BZ}^{\oplus} \dd k \, \Hper^{\Zak}(k)$. Each fiber hamiltonian $\Hper^{\Zak}(k)$ is an operator on the fast Hilbert space $\Hfast = L^2(\T^d_y)$. Then $\hat{\pi}_0 = \int_{\BZ}^{\oplus} \dd k \, \pi_0(k)$ is the projection onto the relevant part of the spectrum, \ie
\begin{align*}
 \pi_0(k) := \sum_{n \in \mathcal{I}} \sopro{\varphi_n(k)}{\varphi_n(k)}
 .
\end{align*}
Even though the $\varphi_n(k)$ may not be continuous at eigenvalue crossings, the projection $k \mapsto \pi_0(k)$ is due to the spectral gap. Associated to the relevant band is a (non-unique) unitary $\hat{u}_0 = \int_{\BZ}^{\oplus} \dd k \, u_0(k)$ which ``straightens'' $\Htau$ into $L^2(\BZ_k) \otimes L^2(\T^d_y)$: for each $k \in \BZ$, we define
\begin{align*}
 u_0(k) := \sum_{n \in \mathcal{I}} \sopro{\chi_n}{\varphi_n(k)} + u_0^{\perp}(k)
\end{align*}
where $\chi_n \in L^2(\T^d_y)$, $n \in \Index$, are \emph{fixed} vectors \emph{independent of $k$} and $u_0^{\perp}(k)$ (also non-unique) acts on the complement of $\mathrm{ran} \, \pi_0(k)$ and is such that $\hat{u}_0$ is a proper unitary. Even though this means $u_0$ is not unique, the specific choices of the $\{ \chi_n \}_{n \in \Index}$ and $u_0^{\perp}$ will not enter the derivation. Then we can put all parts of the puzzle into a diagram:
\begin{align}
	\bfig
		\node L2R(-900,0)[L^2(\R^d_x)]
		\node piL2R(-900,-600)[\Zak^{-1} \hat{\pi}_0 \Zak L^2(\R^d_x)]
		\node Htau(0,0)[\Htau]
		\node piHtau(0,-600)[\hat{\pi}_0 \Htau]
		\node HslowHfast(900,0)[L^2(\BZ) \otimes L^2(\T^d_y)]
		\node Href(900,-600)[L^2(\BZ) \otimes \C^N]
		\arrow[L2R`Htau;\Zak]
		\arrow/-->/[L2R`piL2R;\Zak^{-1} \hat{\pi}_0 \Zak]
		\arrow/-->/[piL2R`piHtau;]
		\arrow[Htau`piHtau;\hat{\pi}_0]
		\arrow[Htau`HslowHfast;\hat{u}_0]
		\arrow[HslowHfast`Href;\Piref]
		\arrow/-->/[piHtau`Href;]
		\Loop(-900,0){L^2(\R^d_x)}(ur,ul)_{e^{-i \frac{t}{\eps} \hat{H}}} 
		\Loop(0,0){\Htau}(ur,ul)_{e^{-i \frac{t}{\eps} \hat{H}^{\Zak}}} 
		\Loop(900,-600){L^2(\BZ) \otimes \C^N}(dr,dl)^{e^{-i \frac{t}{\eps} \hat{h}_{\mathrm{eff} \, 0}}} 
	\efig
	\label{magsapt:mag_sapt:diagram:unperturbed}
\end{align}
The reference projection $\Piref = \id_{L^2(\BZ)} \otimes \piref$ acts trivially on the first factor, $L^2(\BZ)$, and projects via $\piref = \sum_{j = 1}^N \sopro{\chi_j}{\chi_j} = u_0(k) \, \pi_0(k) \, u_0^*(k)$ onto an $N$-dimensional subspace of $L^2(\T^d_y)$. We will identify $\piref L^2(\T^d_y)$ with $\C^N$ when convenient. 

The dynamics in the lower-right corner is generated by the effective hamiltonian
\begin{align*}
 \hat{h}_{\mathrm{eff} \, 0} := \Piref \, \hat{u}_0 \, \hat{H}_{\mathrm{per}} \, \hat{u}_0^* \, \Piref
\end{align*}
which reduces to $E_n(\hat{k})$ if the relevant part of the spectrum consists of an isolated Bloch band. 


\subsection{Effective quantum dynamics to any order} 
\label{magsapt:mag_sapt:adiabatic_decoupling}
Now the question is whether a similar diagram exists in the presence of the perturbation, \ie whether there exist a projection $\Pi$, a unitary $U$ and an effective hamiltonian $\hat{h}_{\mathrm{eff}}$ that take the place of $\hat{\pi}_0$, $\hat{u}_0$ and $\hat{h}_{\mathrm{eff} \, 0}$? This has been answered in the positive for magnetic fields that admit $\BCont^{\infty}(\R^d,\R^d)$ vector potentials in \cite{PST:effective_dynamics_Bloch:2003} where these objects are explicitly constructed by recursion. We replace standard Weyl calculus used in the original publication with its magnetic variant (see Chapter~\ref{magsapt:rewriting:magnetic_weyl_calculus}) and use the two-parameter asymptotic expansion of the magnetic Weyl product recently derived in Chapter~\ref{asymptotics}. This technique naturally allows for the treatment of magnetic fields with components in $\BCont^{\infty}$ and lifts the previous restriction that the magnetic \emph{vector potential} must have components in $\BCont^{\infty}$. 

The construction is a ``defect construction'' where recursion relations derived from
\begin{align*}
	{\Pi}^2 &= \Pi
	&&
	\bigl [ \Pi , \hat{H}^{\Zak} \bigr ] = 0 \\
	{U}^* \, U &= \id_{\Htau}, \; U \, {U}^* = \id_{L^2_{\mathrm{per}}(\BZ) \otimes L^2_{\mathrm{per}}(\WS)}
	&&
	U \, \Pi \, {U}^* = \Piref 
\end{align*}
relate the $n$th term to all previous terms. These four conditions merely characterize that $\Pi$ and $U$ are still a projection and a unitary (first column) and adapted to the problem (second column). These equations can be translated via magnetic Weyl calculus to
\begin{align}
	\pi \magW \pi &= \pi + \ordere{\infty}
	&&
	\bigl [ \pi , H^{\Zak} \bigr ]_{\magW} = \ordere{\infty}
	\label{magsapt:mag_sapt:eqn:condition_corrected_objects} \\
	u \magW {u}^* &= 1 + \ordere{\infty} = {u}^* \magW u
	&&
	u \magW \pi \magW {u}^* = \pi_{\mathrm{eff}} + \ordere{\infty}
	\notag 
	. 
\end{align}
For technical reasons, $\Opk^A(u)$ and $U$, for instance, agree only up to an error that is arbitrarily small in $\eps$ with respect to the operator norm,
\begin{align*}
	U = \Opk^A(u) + \ordern(\eps^{\infty})
	.
\end{align*}
Let us now state the main result of the paper:



\section{Adiabatic decoupling} 
\label{magsapt:modulation_field}
We now consider case~1 in detail: here, the strength of the magnetic field can be changed by varying the coupling constant $\lambda \leq 1$. The hamiltonian we start with is the Zak transform of $\hat{H}$, \ie equation~\eqref{magsapt:rewrite:Zak:eqn:Zak_transformed_hamiltonian}. 

The aforementioned ``defect construction'' yields the tilted projection $\pi$ and the intertwining unitary $u$ as asymptotic expansion in $\eps$ and $\lambda$. It is important that the decoupling is \emph{solely due to the separation of spatial scales} quantified by $\eps$ and \emph{independent} of $\lambda$ which regulates the strength of the magnetic field. The limits $\lambda \rightarrow 0$ and $\eps \rightarrow 0$ are physically very different: $\lambda \rightarrow 0$ (with $\eps \ll 1$ small, but non-zero) corresponds to selectively switching off the magnetic field, $\eps \rightarrow 0$ implies that both, the magnetic and the electric field go to $0$.

\subsection{Effective quantum dynamics} 
\label{magsapt:modulation_field:effective_dynamics}
We will quickly explain how the corrections are computed order-by-order in $\eps$ and $\lambda$. We adapt the general recipe explained in \cite{Teufel:adiabatic_perturbation_theory:2003} to incorporate two parameters: since the decoupling is due to the separation of spatial scales quantified by $\eps \ll 1$, we will order corrections in power of $\eps$. Expanding the magnetic Weyl product to zeroth order, we can check 
\begin{align*}
	\pi_0 \magW \pi_0 &= \pi_0 + \ordere{}
	&&
	\bigl [ \pi_0 , H^{\Zak} \bigr ]_{\magW} = \ordere{} 
	\\
	u_0 \magW {u_0}^* &= 1 + \ordere{} = {u_0}^* \magW u_0
	&&
	u_0 \magW \pi_0 \magW u_0^* = \pi_{\mathrm{eff}} + \ordere{} 
	. 
\end{align*}
Here, $\bigl [ \pi_0 , H^{\Zak} \bigr ]_{\magW} := \pi_0 \magW H^{\Zak} - H^{\Zak} \magW \pi_0$ denotes the magnetic Weyl commutator. The asymptotic expansion of the product is key to deriving corrections in a systematic manner: the $\ordere{}$ terms can be used to infer $\pi_1$ and $u_1$, the subprincipal symbols. Then, one proceeds by recursion: if $\pi^{(n)} := \sum_{l = 0}^n \eps^l \pi_l$ and $u^{(n)} := \sum_{l = 0}^n \eps^l u_l$ satisfy equations~\eqref{magsapt:mag_sapt:eqn:condition_corrected_objects} up to errors of order $\eps^{n+1}$, then we can compute $\pi_{n+1}$ and $u_{n+1}$. The construction of $\pi$ and $u$ follows exactly from Lemma 3.8 and Lemma 3.15 of \cite{Teufel:adiabatic_perturbation_theory:2003}; it is \emph{purely algeraic} and only uses that we have a recipe to expand the Moyal product in terms of the semiclassical parameter $\eps$. Let us define 
\begin{align}
	\pi^{(n)} \magW \pi^{(n)} - \pi^{(n)} &=: \eps^{n+1} G_{n+1} + \ordere{n+2} 
	\label{magsapt:modulation_field:effective_dynamics:eqn:projection_defects} \\
	\bigl [ H^{\Zak} , \pi^{(n)} + \eps^{n+1} \pi_{n+1}^{\mathrm{d}} \bigr ]_{\magW} &=: \eps^{n+1} F_{n+1} + \ordere{n+2} 
	\notag 
\end{align}
as projection and commutation defects and 
\begin{align}
	u^{(n)} \magW {u^{(n)}}^* - 1 &=: \eps^{n+1} A_{n+1} + \ordere{n+2} 
	\label{magsapt:modulation_field:effective_dynamics:eqn:unitarity_defects} \\
	\bigl ( u^{(n)} + \eps^{n+1} a_{n+1} u_0 \bigr ) \magW \pi^{(n+1)} \magW \bigl ( u^{(n)} + \eps^{n+1} a_{n+1} u_0 \bigr )^* &=: \eps^{n+1} B_{n+1} + \ordere{n+2} 
	\notag
\end{align}
as unitarity and intertwining defects. The diagonal part of the projection $\pi_{n+1}^{\mathrm{d}}$ can be computed from $G_{n+1}$ via 
\begin{align}
	\pi_{n+1}^{\mathrm{D}} := - \pi_0 G_{n+1} \pi_0 + (1 - \pi_0) G_{n+1} (1 - \pi_0) 
	\label{magsapt:modulation_field:effective_dynamics:eqn:pi_n1_diag} 
	. 
\end{align}
The term 
\begin{align}
	a_{n+1} = - \tfrac{1}{2} A_{n+1} 
	\label{magsapt:modulation_field:effective_dynamics:eqn:u_a} 
\end{align}
stems from the ansatz $u_{n+1} = (a_{n+1} + b_{n+1}) u_0$ where $a_{n+1}$ and $b_{n+1}$ are symmetric and antisymmetric, respectively. One can solve the second equation for 
\begin{align}
	b_{n+1} = \bigl [ \piref , B_{n+1} \bigr ] 
	\label{magsapt:modulation_field:effective_dynamics:eqn:u_b} 
	. 
\end{align}
This equation fixes only the off-diagonal part of $b_{n+1}$ as $\piref B_{n+1} \piref = 0 = (1 - \piref) B_{n+1} (1 - \piref)$ and in principle one is free to choose the diagonal part of $b_{n+1}$. This means, there is a freedom that allows arbitrary unitary transformations within $\piref L^2(\T^d_y)$ as well as its orthogonal complement. In general, it is not possible to solve 
\begin{align}
	\bigl [ H^{\Zak} , \pi_{n+1}^{\mathrm{OD}} \bigr ] = - F_{n+1} 
\end{align}
explicitly since Bloch functions at band crossings \emph{within} the relevant part of the spectrum (which are admissible) are no longer differentiable. In that case, one has to construct 
\begin{align}
	\pi = \frac{i}{2 \pi} \int_{C(k_0,r_0)} \dd z \, (H^{\Zak} - z)^{(-1)_B} + \ordere{\infty} 
\end{align}
in the vicinity of the point $(r_0,k_0)$ by asymptotically expanding the Moyal resolvent $(H^{\Zak} - z)^{(-1)_B}$, \ie the symbol which satisfies 
\begin{align*}
	(H^{\Zak} - z) \magW (H^{\Zak} - z)^{(-1)_B} = 1 = (H^{\Zak} - z)^{(-1)_B} \magW (H^{\Zak} - z) 
	. 
\end{align*}
A recent result by Iftime, Măntoiu and Purice \cite{Iftimie_Mantoiu_Purice:commutator_criteria:2008} suggests that under these circumstances ($H^{\Zak}$ is elliptic and selfadjoint operator-valued) $(H^{\Zak} - z)^{(-1)_B}$ always exists and is a Hörmander symbol even in the presence of a magnetic field. We reckon their result extends to the case of operator-valued symbols, but seeing how tedious the proof is, we simply stick to the procedure used by Panati, Spohn and Teufel \cite[Lemma~5.17]{Teufel:adiabatic_perturbation_theory:2003}. This construction uniquely fixes the tilted Moyal projection $\pi$ uniquely, but not the Moyal unitary $u$. 

As the \emph{two-parameter expansion of the product} 
\begin{align*}
	f \magW g \asymp \sum_{n = 0}^{\infty} \sum_{k = 0}^n \eps^n \lambda^k \, (f \magW g)_{(n,k)} 
	, 
\end{align*}
contributes only \emph{finitely many terms} in $\lambda$ for fixed power of $n$ of $\eps$ \cite{Lein:two_parameter_asymptotics:2008}, we can order the terms of the expansion of $\pi$ and $u$ in powers of $\lambda$, \eg 
\begin{align*}
	\pi_n = \sum_{k = 0}^n \lambda^k \, \pi_{(n,k)} 
	. 
\end{align*}
The magnetic Weyl product as well as its asymptotic expansion are defined in terms of \emph{oscillatory integrals}, \ie integrals which exist in the distributional sense. If we take the limit $\lambda \rightarrow 0$ of $f \magW g$, we can interchange oscillatory integration and limit procedure \cite[p.~90]{Hoermander:Fourier_integral_operators_1:1971} and conclude $\lim_{\lambda \rightarrow 0} f \magW g = f \Weyl g$ where $\Weyl$ is the usual Moyal product. Similarly, we can apply this reasoning to the asymptotic expansion: for any fixed $N \in \N_0$, we may write the product as 
\begin{align*}
	f \magW g = \sum_{n = 0}^N \eps^n \biggl ( \sum_{k = 0}^n \lambda^k \, (f \magW g)_{(n,k)} \biggr ) + \eps^{N+1} \, R_{N+1}^B(f,g) 
\end{align*}
and taking the limit $\lambda \rightarrow 0$ means only the nonmagnetic terms $(f \magW g)_{(n,0)}$ survive. The remainder also behaves nicely when taking the limit as it is also just another oscillatory integral and $\lim_{\lambda \rightarrow 0} R_{N+1}^B(f,g)$ is exactly the remainder of the nonmagnetic Weyl product. 
\medskip

\noindent
Hence, we can now prove the main result of this paper: 
\begin{thm}\label{magsapt:modulation_field:effective_dynamics:thm:adiabatic_decoupling}
	Let Assumptions~\ref{magsapt:intro:assumption:V_Gamma}, \ref{magsapt:intro:assumption:em_fields} and~\ref{magsapt:mag_sapt:mag_wc:defn:gap_condition} be satisfied. Furthermore, if $d \geq 4$, we add Assumption~\ref{magsapt:mag_sapt:mag_wc:assumption:smooth_frame}. Then there exist 
	\begin{enumerate}[(i)]
		\item an orthogonal projection $\Pi \in \mathcal{B}(\Htau)$, 
		\item a unitary map $U$ which intertwines $\Htau$ and $L^2(\BZ) \otimes L^2(\T^d_y)$, and 
		\item a selfadjoint operator $\Opk^A(\heff) \in \mathcal{B} \bigl ( L^2(\BZ) \otimes \C^N \bigr)$, $N := \abs{\Index}$ 
	\end{enumerate}
	such that 
	\begin{align*}
		\bnorm{\bigl [ \hat{H}^{\Zak} , \Pi \bigr ]} = \ordere{\infty} 
	\end{align*}
	and
	\begin{align}
		\Bnorm{\bigl ( e^{- i t \hat{H}^{\Zak}} - U^* e^{- i t \Opk^A(\heff)} U \bigr ) \Pi}_{\mathcal{B}(\Htau)} = \mathcal{O} \bigl ( \eps^{\infty}(1 + \abs{t}) \bigr ) 
		\label{magsapt:modulation_field:effective_dynamics:eqn:adiabatic_decoupling_dynamics}
		. 
	\end{align}
	The effective hamiltonian is the magnetic quantization of the $\Gamma^*$-periodic function 
	\begin{align}
		\heff := \piref \, u \magW H^{\Zak} \magW u^* \, \piref \asymp \sum_{n = 0}^{\infty} \eps^n \, \heff_n \in \SemiHoer{0} \bigl ( \mathcal{B}(\C^N) \bigr ) 
	\end{align}
	whose asymptotic expansion can be computed to any order in $\eps$ and $\lambda$. To each order in $\eps$, only finitely many terms in $\lambda$ contribute, $\heff_n = \sum_{k = 0}^n \lambda^k \, \heff_{(n,k)}$.
\end{thm}
The proof of this result amounts to showing (i)-(iii) separately. 
\begin{prop}\label{magsapt:modulation_field:effective_dynamics:prop:projection}
	Under the assumptions of Theorem~\ref{magsapt:modulation_field:effective_dynamics:thm:adiabatic_decoupling} there exists and orthogonal projection $\Pi \in \mathcal{B}(\Htau)$ such that 
	\begin{align}
		\bigl [ \hat{H}^{\Zak} , \Pi \bigr ] = \ordern(\eps^{\infty}) 
	\end{align}
	and $\Pi = \Opk^A(\pi) + \ordern(\eps^{\infty})$ where $\Opk^A(\pi)$ is the magnetic Weyl quantization of a $\tau$-equivariant semiclassical symbol 
	\begin{align*}
		\pi \asymp \sum_{n = 0}^{\infty} \eps^n \, \pi_n \in \SemiTau \bigl ( \mathcal{B}(\Hfast) \bigr )
	\end{align*}
	whose principal part $\pi_0(k,r)$ coincides with the spectral projection of $H^{\Zak}(k,r)$ onto the subspace corresponding to the given isolated family of Bloch bands $\{ E_n \}_{n \in \Index}$. Each term in the expansion can be written as a finite sum 
	\begin{align*}
		\pi_n = \sum_{k = 0}^n \lambda^k \, \pi_{(n,k)} \in \SemiTau \bigl ( \mathcal{B}(\Hfast) \bigr ) 
	\end{align*}
	ordered by powers of $\lambda$. For $\lambda \rightarrow 0$, \ie letting $B^{\eps,\lambda} \rightarrow 0$ while keeping the electric field fixed, the projection $\pi$ reduces to the nonmagnetic projection $\pi^0 \asymp \sum_{n = 0}^{\infty} \eps^n \, \pi_{(n,0)}$. 
\end{prop}
\begin{proof}[Sketch]
	The proof relies on a well-developed \emph{magnetic} Weyl calculus and the gap condition. In particular, one needs a magnetic Caldéron-Vaillancourt theorem, composition and quantization of Hörmander symbols \cite{Iftimie_Mantiou_Purice:magnetic_psido:2006} and finally, an asymptotic two-parameter expansion of the magnetic Weyl product $\magW$ \cite{Lein:two_parameter_asymptotics:2008}. The interested reader may check line-by-line that the original proof can be transliterated to the magnetic context with obvious modifications. If we were using standard Weyl calculus, the major obstacle would be to control derivatives of $\pi$ since vector potentials may be unbounded. In magnetic Weyl calculus the vector potential at no point enters the calculuations and the assumptions on the magnetic field assure that $\pi \in \SemiTau \bigl ( \mathcal{B}(\Hfast) \bigr )$ is a proper $\tau$-equivariant semiclassical Hörmander-class symbol. 
	
	The fact that we can write all of the $\pi_n$ as \emph{finite} sum of terms ordered by powers of $\lambda$ stems from the fact that calculating $\pi_n$ involves the expansion of the product up to $n$th power in $\eps$, \eg for the projection defect, we find 
	\begin{align*}
		\pi^{(n-1)} \magW \pi^{(n-1)} &- \pi^{(n-1)} = \eps^n \negmedspace \negmedspace \negmedspace \negmedspace \sum_{a + b + c = n} \negmedspace \negmedspace (\pi_a \magW \pi_b \bigr )_{(c)} + \ordere{n+1}
		\\
		&= \eps^n \negmedspace \negmedspace \negmedspace \negmedspace \sum_{a + b + c = n} \sum_{a' = 0}^a \sum_{b' = 0}^b \sum_{c' = 0}^c \lambda^{a' + b' + c'} \, (\pi_{(a,a')} \magW \pi_{(b,b')} \bigr )_{(c,c')} + \ordere{n+1} 
		. 
	\end{align*}
	Certainly, the exponent of $\lambda$ is always bounded by $n \geq a' + b' + c'$. And since the sum is finite, this clearly defines a semiclassical symbol in $\eps$ (see Definition~\ref{magsapt:rewriting:magnetic_weyl_calculus:defn:semiclassical_symbol}) is shown. Similar arguments for the commutation defect in conjunction with the comments in the beginning of this section show $\pi$ to be a semiclassical two-parameter symbol. 
	It is well-behaved under the $\lambda \rightarrow 0$ limit and reduces to the projection associated to the case $B = 0$. 
	\medskip
	
	\noindent
	Lastly, to make the almost projection $\Opk^A(\pi)$ into a true projection, we define $\Pi$ to be the spectral projection onto the spectrum in the vicinity of $1$, 
	\begin{align*}
		\Pi := \int_{\abs{z - \nicefrac{1}{2}} = 1} \dd z \, \bigl ( \Opk^A(\pi) - z \bigr )^{-1} 
		. 
	\end{align*}
	This concludes the proof. 
\end{proof}
Similarly, one can modify the proof found in \cite{Teufel:adiabatic_perturbation_theory:2003} to show the existence of the unitary. 
\begin{prop}\label{magsapt:modulation_field:effective_dynamics:prop:unitary}
	Let $\{ E_n \}_{n \in \Index}$ be a family of bands separated by a gap from the others and let Assumption~\ref{magsapt:intro:assumption:V_Gamma} be satisfied. If $d > 3$, assume $u_0 \in \Hoermr{0}{0} \bigl ( \mathcal{B}(\Hfast) \bigr )$. Then there exists a unitary operator $U : \Htau \longrightarrow L^2(\BZ) \otimes L^2(\T^d_y)$ such that $U = \Opk^A(u) + \ordern(\eps^{\infty})$ where 
	\begin{align*}
		u \asymp \sum_{n = 0}^{\infty} \eps^n u_n \in \SemiHoer{0} \bigl ( \mathcal{B}(\Hfast) \bigr )
	\end{align*}
	is right-$\tau$-covariant at any order and has principal symbol $u_0$. Each term in the expansion can be written as a finite sum 
	\begin{align*}
		u_n = \sum_{k = 0}^n \lambda^k u_{(n,k)} 
	\end{align*}
	ordered by powers of $\lambda$. For $\lambda \rightarrow 0$, \ie letting $B^{\eps,\lambda} \rightarrow 0$ while keeping the electric field fixed, the unitary $u$ reduces to the nonmagnetic unitary $u^0 \asymp \sum_{n = 0}^{\infty} \eps^n \, u_{(n,0)}$. 
\end{prop}
\begin{proof}[Sketch]
	Equations~\eqref{magsapt:modulation_field:effective_dynamics:eqn:u_a} and \eqref{magsapt:modulation_field:effective_dynamics:eqn:u_b} give us $a_{n+1}$ and $b_{n+1}$ which combine to $u_{n+1} = (a_{n+1} + b_{n+1}) u_0$; by Theorem~1.1 from \cite{Lein:two_parameter_asymptotics:2008} it is also in the correct symbol class, namely $\Hoermr{0}{0} \bigl ( \mathcal{B}(\Hfast) \bigr )$. The right $\tau$-covariance is also obvious from the ansatz. 
	
	Lastly, the true unitary $U$ is obtained via the Nagy formula as described in \cite{Teufel:adiabatic_perturbation_theory:2003}. 
\end{proof}
\begin{proof}[Theorem~\ref{magsapt:modulation_field:effective_dynamics:thm:adiabatic_decoupling} (Sketch)]
	The existence of $\Pi$ and $U$ have been the subject of Propositions~\ref{magsapt:modulation_field:effective_dynamics:prop:projection} and \ref{magsapt:modulation_field:effective_dynamics:prop:unitary}. By right-$\tau$-covariance of $u$, $\heff$ is a $\Gamma^*$-periodic function and since it is the magnetic Weyl product of $\BCont^{\infty} \bigl ( T^* \R^d , \mathcal{B} \bigl ( L^2(\T^d_y) \bigr ) \bigr )$ functions, Theorem~1.1 from \cite{Lein:two_parameter_asymptotics:2008} ensures that the product and the terms of its asymptotic two-parameter expansion are in $\BCont^{\infty} \bigl ( T^* \R^d , \mathcal{B} \bigl ( L^2(\T^d_y) \bigr ) \bigr )$ as well. 
\end{proof}
%


\subsection{Effective dynamics for a single band} 
\label{magsapt:modulation_field:effective_dynamics_for_a_single_band}
In case the relevant part of the spectrum consists of a single non-degenerate band $\Eb$\index{effective single band hamiltonian}, we can calculate the first-order correction to $\heff$ explicitly: the magnetic Weyl product reduces to the pointwise product to zeroth order in $\eps$. Thus, we can directly compute 
\begin{align*}
	\heff_0 = \piref \, u_0 \, H_0 \, u_0^* \, \piref =: \piref \, h_0 \, \piref = \Eb + \phi 
	. 
\end{align*}
For the first order, we use the recursion formula \cite[eq.~(3.35)]{Teufel:adiabatic_perturbation_theory:2003} and the fact that $\heff_0$ is a scalar-valued symbol: 
\begin{align*}
	\heff_1 &= \Bigl ( u_1 \, H_0 - h_0 \, u_1 + (u_0 \magW H_0)_{(1)} - (h_0 \magW u_0)_{(1)} \Bigr ) u_0^* 
	\\
	&= \piref \, \bigl [ u_1 u_0^* , h_0 \bigr ] \, \piref + (u_0 \magW H_0)_{(1)} u_0^* - (h_0 \magW u_0)_{(1)}  u_0^*
	\\
	&= - \tfrac{i}{2} \piref \, \Bigl ( \bigl \{ u_0 , H_0 \bigr \}_{\lambda B} - \bigl \{ h_0 , u_0 \bigr \}_{\lambda B} \Bigr ) \, \piref 
\end{align*}
The term with the magnetic Poisson bracket can be easily computed: 
\begin{align*}
	\piref \, \Bigl ( &\bigl \{ u_0 , H_0 \bigr \}_{\lambda B} - \bigl \{ h_0 , u_0 \bigr \}_{\lambda B} \Bigr ) \, u_0^* \, \piref 
	= 
	\\
	&= \piref \, \Bigl ( \partial_{k_l} u_0 \, \partial_{r_l} H_0 - \lambda B_{lj} \, \partial_{k_l} u_0 \, u_0^* \, 
	\bigl ( \partial_{k_j} h_0 \, u_0 - \partial_{k_j} u_0 \, u_0^* \, h_0 - h_0 \, u_0 \, \partial_{k_j} u_0^* \bigr )
	+ \\
	&\qquad \qquad 
	+ \partial_{r_l} h_0 \, \partial_{k_l} u_0 + \lambda B_{lj} \, \partial_{k_l} h_0 \, \partial_{k_j} u_0 \Bigr ) \, u_0^* \, \piref 
	\displaybreak[2]
	\\
	&= 2 i \bigl ( \partial_{r_l} \phi - \lambda B_{lj} \, \partial_{k_j} \Eb \bigr ) \, \BerryC_l 
	+ \\
	&\qquad \qquad 
	+ \lambda B_{lj} \, \piref \, \partial_{k_l} u_0 \, u_0^* \, \bigl ( \partial_{k_j} u_0 \, u_0^* \, h_0 + h_0 \, u_0 \, \partial_{k_j} u_0^* \bigr ) \, \piref 
	\\
	&= 2 i \bigl ( \partial_{r_l} \phi - \lambda B_{lj} \, \partial_{k_j} \Eb \bigr ) \, \BerryC_l + \lambda B_{lj} \, \bscpro{\partial_{k_l} \varphi_{\mathrm{b}}}{\bigl ( \Hper - \Eb \bigr ) \partial_{k_j} \varphi_{\mathrm{b}}} 
\end{align*}
The first term combines to a Lorentz force term, the second one -- which is purely imaginary -- yields the Rammal-Wilkinson term. The components of the magnetic field are $\BCont^{\infty}(\R^d)$ functions and hence, principal and subprincipal symbol are in $\BCont^{\infty}(T^* \R^d)$ as well. 

This means, we have proven the following corollary to Theorem~\ref{magsapt:modulation_field:effective_dynamics:thm:adiabatic_decoupling}: 
\begin{cor}\label{magsapt:modulation_field:effective_dynamics:cor:adiabatic_decoupling_explicit}
	Under the assumptions of Theorem~\ref{magsapt:modulation_field:effective_dynamics:thm:adiabatic_decoupling}, the principal and subprincipal symbol of the effective hamiltonian for a single non-degenerate Bloch band $\Eb$ are given by 
	\begin{align}
		\heff_0 &= \Eb + \phi
		\\
		\heff_1 &= - \bigl ( - \partial_{r_l} \phi + \lambda B_{lj} \, \partial_{k_j} \Eb \bigr ) \, \BerryC_l - \lambda B_{lj} \, \mathcal{M}_{lj} 
		\notag \\
		&=: - F_{\mathrm{Lor} \, l} \, \BerryC_l - \lambda B_{lj} \, \mathcal{M}_{lj} 
		\notag 
		. 
	\end{align}
	where 
	\begin{align*}
		\BerryC_l(k) := i \scpro{\varphi_{\mathrm{b}}(k)}{\nabla_k \varphi_{\mathrm{b}}(k)} 
	\end{align*}
	and 
	\begin{align*}
		\mathcal{M}_{lj}(k) = \Re \Bigl ( \tfrac{i}{2} \bscpro{\partial_{k_l} \varphi_{\mathrm{b}}}{\bigl ( \Hper(k) - \Eb(k) \bigr ) \, \partial_{k_j} \varphi_{\mathrm{b}}} \Bigr )
	\end{align*}
	are the Berry connection and the so-called Rammal-Wilkinson term, respectively. 
\end{cor}
%


\subsection{Semiclassical equations of motion} 
\label{magsapt:modulation_field:eom}
Now that we have approximated the full quantum evolution by an effective \emph{quantum} evolution on a smaller reference space, namely $L^2(\BZ) \otimes \C^N$, we will further simplify the problem of finding approximate dynamics by taking the semiclassical limit.\index{semiclassical limit} Conceptually, this is a two-step process: if we reconsider the diagram of spaces, 
\begin{align}
	\bfig
		\node L2R(-900,0)[L^2(\R^d_x)]
		\node piL2R(-900,-600)[\Zak^{-1} \Pi \Zak L^2(\R^d_x)]
		\node Htau(0,0)[\Htau]
		\node piHtau(0,-600)[\Pi \Htau]
		\node HslowHfast(900,0)[L^2(\BZ) \otimes L^2(\T^d_y)]
		\node Href(900,-600)[L^2(\BZ) \otimes \C^N]
		\arrow[L2R`Htau;\Zak]
		\arrow/-->/[L2R`piL2R;\Zak^{-1} \Pi \Zak]
		\arrow/-->/[piL2R`piHtau;]
		\arrow[Htau`piHtau;\Pi]
		\arrow[Htau`HslowHfast;U]
		\arrow[HslowHfast`Href;\Piref]
		\arrow/-->/[piHtau`Href;]
		\Loop(-900,0){L^2(\R^d_x)}(ur,ul)_{e^{-i \frac{t}{\eps} \hat{H}}} 
		\Loop(0,0){\Htau}(ur,ul)_{e^{-i \frac{t}{\eps} \hat{H}^{\Zak}}} 
		\Loop(900,-600){L^2(\BZ) \otimes \C^N}(dr,dl)^{e^{-i \frac{t}{\eps} \Opk^A(\heff)}} 
	\efig
	\label{magsapt:modulation_field:eom:diagram:perturbed}
\end{align}
we notice that our physical observables live on the upper-left space $L^2(\R^d_x)$ -- or equivalently on $\Htau$. The effective evolution generated by $\heff$ approximates the dynamics if the initial states are localized in the tilted eigenspace associated to the relevant bands. In this section, we always assume the relevant part of the spectrum consists of a \emph{single non-degenerate band} $\Eb$ and thus $L^2(\BZ) \otimes \C^1 \cong L^2(\BZ)$. 

In a first step, we need to connect the semiclassical dynamics in the left column of the diagram with those in the lower-right corner. The second, much simpler step is to establish an Egorov-type theorem on the reference space.

\subsubsection{Connection of the effective dynamics with the original dynamics} 
\label{magsapt:modulation_field:eom:connecting_dynamics}
Since we are concerned with the semiclassical dynamics of a particle in an electromagnetic field, the magnetic field must enter in the classical equations of motion. There are two ways: either one uses minimal coupling, \ie one writes down the equations of motion for position $r$ and momentum $k - \lambda A(r)$ with respect to the usual symplectic form. Or alternatively, the classical flow which enters the Egorov theorem is generated by 
\begin{align}
	\left (
	\begin{matrix}
		\lambda B(r) & - \id \\
		+ \id & 0 \\
	\end{matrix}
	\right )
	\left (
	\begin{matrix}
		\dot{r} \\
		\dot{k} \\
	\end{matrix}
	\right )
	= 
	\left (
	\begin{matrix}
		\nabla_r \\
		\nabla_k \\
	\end{matrix}
	\right )
	H^{\Zak}(k,r)
	\label{magsapt:modulation_field:eom:eqn:initial_eom}
\end{align}
where the appearance of $B$ in the matrix representation of the symplectic form is due to the fact that $k$ is \emph{kinetic} momentum. 
What constitutes a suitable observable? Physically, we are interested in measurements on the \emph{macroscopic} scale, \ie the observable should be independent of the \emph{microscopic} degrees of freedom. On the level of symbols, this means $f(k,r)$ has to commute pointwise with the hamiltonian $H^{\Zak}(k,r)$ for all $k$ and $r$. Hence, such an observable is a constant of motion with respect to the \emph{fast} dynamics. In the simplest case, the observables are scalar-valued. This also ensures we are able to ``separate'' the contributions to the full dynamics band-by-band. Note that this \emph{by no means} implies $\Opk^A(f)$ commutes with $\Opk^A(H^{\Zak})$, but rather that all of the non-commutativity is contained in the slow variables $(k,r)$. 
\begin{defn}[Macroscopic semiclassical observable\index{macroscopic semiclassical observable}]\label{magsapt:modulation_field:eom:defn:semiclassical_observable}
	A macropscopic observable $f$ is a scalar-valued semiclassical symbol (see Definition~\ref{magsapt:rewriting:magnetic_weyl_calculus:defn:semiclassical_symbol}) $\SemiHoermr{0}{0}(\C)$ which is $\Gamma^*$-periodic in $k$, $f(k + \gamma^*,r) = f(k,r)$ for all $(k,r) \in T^* \R^d$, $\gamma^* \in \Gamma^*$. 
\end{defn}
Our assumption that our dynamics lives on the almost-invariant subspace $\Pi \Htau$ modifies the classical dynamics to first order in $\eps$ as well: instead of using $\KA$ and $\Reps$ as building block observables, the proper observables should be $\Pi \KA \Pi$ and $\Pi \Reps \Pi$. Equivalently, we can switch to the reference space representation and use the magnetic quantization of 
\begin{align}
	\keff &:= \piref \, u \magW k \magW u^* \, \piref = k + \eps \lambda B(r) \BerryC(k) + \ordere{2}
	\label{magsapt:modulation_field:eom:eqn:effective_building_blocks} \\
	\reff &:= \piref \, u \magW r \magW u^* \, \piref = r + \eps \BerryC(k) + \ordere{2} 
	\notag 
\end{align}
The crucial proposition we will prove next says that for suitable observables $f$, the effect of going to the effective representation is, up to errors of order $\eps^2$ at least, equivalent to replacing the arguments $k$ and $r$ by $\keff$ and $\reff$, 
\begin{align*}
	\Piref U \Opk^A(f) U^* \Piref &= \Opk^A \bigl ( \piref \, u \magW f \magW u^* \, \piref \bigr ) + \ordern(\eps^{\infty}) 
	\\
	&
	=: \Opk^A(f_{\mathrm{eff}}) + \ordern(\eps^{\infty}) 
	. 
\end{align*}
Then it follows that the effective observable $f_{\mathrm{eff}}$ coincides with the original observable $f$ after a change of variables up to errors of order $\ordere{2}$, 
\begin{align}
	f_{\mathrm{eff}} :=& \piref \, u \magW f \magW u^* \, \piref = f \circ T_{\mathrm{eff}} + \ordere{2}
	, 
	\label{magsapt:modulation_field:eom:eqn:effective_sym_replacement}
\end{align}
where the map $T_{\mathrm{eff}} : (k,r) \mapsto (\keff,\reff)$ maps the observables $k$ and $r$ onto the effective observables $\keff$ and $\reff$ defined via equations~\eqref{magsapt:modulation_field:eom:eqn:effective_building_blocks} . 
\begin{prop}\label{magsapt:modulation_field:eom:prop:effective_sym_building_block}
	Let $f$ be a macroscopic semiclassical observable. Then up to errors of order $\eps^2$ equation~\eqref{magsapt:modulation_field:eom:eqn:effective_sym_replacement} holds, \ie 
	\begin{align}
		\Piref \, U \, \Opk^A(f) \, U^* \, \Piref = \Opk^A \bigl ( f_{\mathrm{eff}} \bigr ) + \ordern(\eps^{\infty}) 
		= \Opk^A \bigl ( f \circ T_{\mathrm{eff}} \bigr ) + \ordern(\eps^2) 
		\label{magsapt:modulation_field:eom:eqn:effective_op_replacement} 
		. 
	\end{align}
\end{prop}
\begin{proof}
	The equivalence of the left-hand sides of equations~\eqref{magsapt:modulation_field:eom:eqn:effective_op_replacement} and \eqref{magsapt:modulation_field:eom:eqn:effective_sym_replacement} follows from $U = \Opk^A(u) + \ordern(\eps^{\infty})$ and the fct that we only need ot consider the first two terms in the $\eps$ expansion. With the help of Theorem~1.1 from \cite{Lein:two_parameter_asymptotics:2008}, we conclude $f_{\mathrm{eff}} \in \SemiHoer{0}$ is also a semiclassical symbol of order $0$. The left-hand side of \eqref{magsapt:modulation_field:eom:eqn:effective_sym_replacement} can be computed explicitly: to zeroth order, nothing changes as $f$ commutes pointwise with $u$ and $u^*$, 
	\begin{align*}
		f_{\mathrm{eff} \, 0} = \piref \, u_0 f u_0^* \, \piref = f_0 
		. 
	\end{align*}
	To first order, we have 
	\begin{align*}
		f_{\mathrm{eff} \, 1} &= \piref \, \Bigl ( u_0 f_1 + u_1 f_0 - f_{\mathrm{eff} \, 0} u_1 + (u_0 \magW f_0)_{(1)} - (f_{\mathrm{eff} \, 0} \magW u_0)_{(1)} \Bigr ) u_0^* \, \piref 
		\\
		&= \piref \, u_0 f_1 u_0^* \, \piref - \tfrac{i}{2} \Bigl ( \bigl \{ u_0 , f_0 \bigr \}_{\lambda B} - 
		\bigl \{ f_{\mathrm{eff} \, 0} , u_0 \bigr \}_{\lambda B} \Bigr ) 
		\\
		&= f_1 - i \bigl ( \partial_{r_j} f_0 + \lambda B_{lj}(r) \, \partial_{k_l} f_0 \bigr ) \, \piref \, \partial_{k_j} u_0 \, u_0^* \, \piref 
		\\
		&=  f_1 + \bigl ( \partial_{r_j} f_0 + \lambda B_{lj}(r) \, \partial_{k_l} f_0 \bigr ) \, \BerryC_j 
		. 
	\end{align*}
	On the other hand, if we Taylor expand $f \circ T_{\mathrm{eff}} = f \bigl ( \keff , \reff \bigr )$ to first order in $\eps$, we get 
	\begin{align*}
		f \bigl ( \keff , \reff \bigr ) &= f_0 \bigl ( k + \eps \lambda B(r) \BerryC(k) + \ordere{2} , r + \eps \BerryC(k) + \ordere{2} \bigr ) + 
		\\
		&\quad 
		+ \eps f_1 \bigl ( k + \eps \lambda B(r) \BerryC(k) + \ordere{2} , r + \eps \BerryC(k) + \ordere{2} \bigr ) + \ordere{2} 
		\\
		&= f_0(k,r) 
		+ \\
		&\quad 
		+ \eps \Bigl ( f_1(k,r) + \lambda \partial_{k_l} f_0(k,r) \, B_{lj}(r) \, \BerryC_j(k) + \partial_{r_j} f_0(k,r) \, \BerryC_j(k) \Bigr ) 
		+ \ordere{2} 
	\end{align*}
	which coincides with $f_{\mathrm{eff}}$ up to $\ordere{2}$. 
\end{proof}
Now if the equations of motion~\eqref{magsapt:modulation_field:eom:eqn:initial_eom} are an approximation of the full quantum dynamics, what are the equations of motion with respect to the effective variables? The classical evolution generated by $\heff$ with respect to the magnetic symplectic form (equation~\eqref{magsapt:modulation_field:eom:eqn:initial_eom}) can be rewritten in terms of effective variables, 
\begin{align}
	\Phi^{\mathrm{macro}}_t = T_{\mathrm{eff}} \circ \Phi^{\mathrm{eff}}_t \circ T_{\mathrm{eff}}^{-1} + \ordere{2} 
	. 
	\label{magsapt:modulation_field:eom:eqn:macro_flow}
\end{align}
The right-hand side \emph{does not serve} as a definition for the flow of the macroscopic observables, but it is a consequence: $\Phi^{\mathrm{macro}}_t$ is the flow associated to a modified symplectic form and a modified hamiltonian. The modified symplectic form includes the Berry curvature associated to $\Eb$ acting as a pseudo-magnetic field on the position variables. 
\begin{prop}\label{magsapt:modulation_field:eom:prop:macro_flow}
	Let $\Phi_t^{\mathrm{macro}}$\index{semiclassical equations of motion} be the flow on $T^* \R^d$ generated by 
	\begin{align}
		\left (
		\begin{matrix}
			\lambda B(\reff) & - \id \\
			+ \id & \eps \Omega(\keff) \\
		\end{matrix}
		\right )
		\left (
		\begin{matrix}
			\dot{r}_{\mathrm{eff}} \\
			\dot{k}_{\mathrm{eff}} \\
		\end{matrix}
		\right )
		= 
		\left (
		\begin{matrix}
			\nabla_{\reff} \\
			\nabla_{\keff} \\
		\end{matrix}
		\right )
		h_{\mathrm{sc}}(\keff,\reff)
		\label{magsapt:modulation_field:eom:eqn:macro_eom}
	\end{align}
	where the \emph{semiclassical hamiltonian} is given by 
	\begin{align}
		h_{\mathrm{sc}}(\reff,\keff) :=& \, \heff \circ T_{\mathrm{eff}}^{-1}(\reff,\keff) 
		\notag \\
		=& \, \bigl ( \Eb(\keff) + \phi(\reff) \bigr ) - \eps \lambda B(\reff) \cdot \mathcal{M}(\keff) + \ordere{2} 
		. 
	\end{align}
	Then equation~\eqref{magsapt:modulation_field:eom:eqn:macro_flow} holds, ie 
	$\Phi_t^{\mathrm{macro}}$ and $T_{\mathrm{eff}} \circ \Phi_t^{\mathrm{eff}} \circ T_{\mathrm{eff}}^{-1}$ agree up to errors of order $\eps^2$. 
\end{prop}
\begin{proof}
	We express $k$ and $r$ in terms of $\keff$ and $\reff$ in \eqref{magsapt:modulation_field:eom:eqn:initial_eom} since, for $\eps$ small enough, $T_{\mathrm{eff}} : (k,r) \mapsto (\keff,\reff)$ is a bijection. For instance, the semiclassical hamiltonian $\heff \circ T_{\mathrm{eff}}^{-1}$ simplifies to 
	\begin{align*}
		h_{\mathrm{sc}}(\keff,\reff) :=& \, \heff \bigl ( \keff - \eps \lambda B(\reff) \BerryC(\keff) , \reff - \eps \BerryC(\keff) \bigr ) + \ordere{2} 
		\notag \\ 
		=& \, \bigl ( \Eb(\keff) + \phi(\reff) \bigr ) - \eps \lambda B(\reff) \cdot \mathcal{M}(\keff) + \ordere{2} 
		. 
	\end{align*}
	The symplectic form can be easily expanded to 
	\begin{align*}
		\left (
		\begin{matrix}
			\lambda B(\reff - \eps \BerryC(\keff)) & - \id \\
			+ \id & 0 \\
		\end{matrix}
		\right ) 
		&= 
		\left (
		\begin{matrix}
			\lambda B(\reff) & - \id \\
			+ \id & 0 \\
		\end{matrix}
		\right ) - \eps \negmedspace \left (
		\begin{matrix}
			\lambda \partial_{{\reff}_l} B(\reff) \, \BerryC(\keff) & 0 \\
			0 & 0 \\
		\end{matrix}
		\right ) 
		+ \\
		&\qquad 
		+ \ordere{2} 
		. 
	\end{align*}
	The other two terms, the time derivatives and gradients of $\keff$ and $\reff$ have slightly more complicated expansions, but they can be worked out explicitly. Then if we put all of them together, we arrive at the modified symplectic form~\eqref{magsapt:modulation_field:eom:eqn:macro_eom}. This proves the first claim. 
	
	Hence, the hamiltonian vector fields agree up to $\ordere{2}$ and Lemma~5.24 in \cite{Teufel:adiabatic_perturbation_theory:2003} implies that also the flows differ only by $\ordere{2}$. 
\end{proof}
\begin{remark}
	These equations of motion have first been proposed in the appendix of \linebreak\cite{PST:effective_dynamics_Bloch:2003} and we have derived them in a more systematic fashion. The effective coordinates $r_{\mathrm{eff}}$ and $k_{\mathrm{eff}}$ are associated to the noncommutative manifold $T^* \R^d$ \cite{Doplicher_Fredenhagen_Roberts:quantum_structure_spacetime:2003}: from equation~\eqref{magsapt:modulation_field:eom:eqn:macro_eom}, one can read off that the Poisson bracket\index{Poisson bracket} with respect to $r_{\mathrm{eff}}$ and $k_{\mathrm{eff}}$ is given by 
	\begin{align*}
		\bigl \{ f , g \bigr \}_{\lambda B , \eps \Omega} = \bigl ( \partial_{\xi_l} f \, \partial_{x_l} g - \partial_{x_l} f \, \partial_{\xi_l} g \bigr ) - \bigl ( \lambda B_{lj} \, \partial_{\xi_l} f \, \partial_{\xi_j} g - \eps \Omega_{lj} \, \partial_{x_l} f \, \partial_{x_j} g \bigr )
	\end{align*}
	and thus different components of position $r_{\mathrm{eff}}$ no longer commute, 
	\begin{align*}
		\bigl \{ r_{\mathrm{eff} \, l} , r_{\mathrm{eff} \, j} \bigr \}_{\lambda B , \eps \Omega} &= - \eps \Omega_{lj} 
		. 
	\end{align*}
	Hence, $\Omega$ acts as a pseudomagnetic field that is due to quantum effects. 
\end{remark}
Now we proceed and prove the semiclassical limit. 


\subsubsection{An Egorov-type theorem} 
\label{magsapt:modulation_field:eom:egorov}
The semiclassical approximation hinges on an Egorov-type theorem\index{Egorov-type theorem} which we first prove on the level of effective dynamics: 
\begin{thm}\label{magsapt:modulation_field:eom:egorov:thm:semiclassics}
	Let $\heff$ be the effective hamiltonian\index{Egorov-type theorem} as given by Theorem~\ref{magsapt:modulation_field:effective_dynamics:thm:adiabatic_decoupling} associated to an isolated, non-degenerate Bloch band $\Eb$. Then for any $\Gamma^*$-periodic semiclassical observable $f \in \SemiHoer{0}$, $f = f_0 + \eps f_1$, the flow $\Phi^{\mathrm{eff}}_t$ generated by $\heff$ with respect to the magnetic symplectic form (equation~\eqref{magsapt:modulation_field:eom:eqn:initial_eom}) approximates the quantum evolution uniformly for all $t \in [-T , +T]$, 
	\begin{align}
		\Bnorm{e^{+ i \frac{t}{\eps} \Opk^A(\heff)} \, \Opk^A(f) \, e^{- i \frac{t}{\eps} \Opk^A(\heff)} - \Opk^A \bigl ( f \circ \Phi^{\mathrm{eff}}_t \bigr )}_{\mathcal{B}(L^2(\BZ))} \leq C \eps^2 
		. 
	\end{align}
\end{thm}
\begin{proof}
	Since $\heff \in \BCont^{\infty}(T^* \R^d)$, the flow inherits the smoothness and $f \circ \Phi^{\mathrm{eff}}_t , \frac{\dd}{\dd t} \bigl ( f \circ \Phi^{\mathrm{eff}}_t \bigr ) \in \SemiHoer{0}(\C)$ remain also $\Gamma^*$-periodic in the momentum variable. To compare the two time-evlutions, we use the usual Duhammel trick which yields 
	\begin{align}
		e^{+ i \frac{t}{\eps} \Opk^A(\heff)} \, &\Opk^A(f) \, e^{- i \frac{t}{\eps} \Opk^A(\heff)} - \Opk^A \bigl ( f \circ \Phi^{\mathrm{eff}}_t \bigr ) =
		\notag \\
		&= \int_0^t \dd s \, \frac{\dd}{\dd s} \Bigl ( e^{+ i \frac{s}{\eps} \Opk^A(\heff)} \, \Opk^A \bigl ( f \circ \Phi^{\mathrm{eff}}_{t-s} \bigr ) \, e^{- i \frac{s}{\eps} \Opk^A(\heff)} \Bigr )
		\notag \\
		&= \int_0^t \dd s \, e^{+ i \frac{s}{\eps} \Opk^A(\heff)} \cdot
		\Bigl ( \tfrac{i}{\eps} \bigl [ \Opk^A(\heff) , \Opk^A \bigl ( f \circ \Phi^{\mathrm{eff}}_{t-s} \bigr ) \bigr ] 
		+ \notag \\
		&\qquad \qquad \qquad \qquad \qquad \qquad 
		- \Opk^A \bigl ( \tfrac{\dd}{\dd s} \bigl ( f \circ \Phi^{\mathrm{eff}}_{t-s} \bigr ) \bigr ) \Bigr ) \, e^{- i \frac{s}{\eps} \Opk^A(\heff)} 
		\notag \\
		&= \int_0^t \dd s \, e^{+ i \frac{s}{\eps} \Opk^A(\heff)} \, \Opk^A \Bigl ( \tfrac{i}{\eps} \bigl [ \heff , f \circ \Phi^{\mathrm{eff}}_{t-s} \bigr ]_{\magW} 
		+ \notag \\
		&\qquad \qquad \qquad \qquad \qquad \qquad 
		- \bigl \{ \heff , f \circ \Phi^{\mathrm{eff}}_{t-s} \bigr \}_{\lambda B} \Bigr ) \, e^{- i \frac{s}{\eps} \Opk^A(\heff)} 
		. 
		\label{magsapt:modulation_field:eom:eqn:Egorov_Duhammel}
	\end{align}
	The magnetic Moyal commutator -- to first order -- agrees with the magnetic Poisson bracket, 
	\begin{align*}
		\tfrac{i}{\eps} \bigl [ \heff , f \circ \Phi^{\mathrm{eff}}_{t-s} \bigr ]_{\magW} &= \bigl \{ \heff , f \circ \Phi^{\mathrm{eff}}_{t-s} \bigr \}_{\lambda B} + \ordere{2} 
		\\
		&= \bigl ( \partial_{k_l} \heff \, \partial_{r_l} f - \partial_{r_l} \heff \, \partial_{k_l} f \bigr ) - \lambda B_{lj} \, \partial_{k_l} \heff \, \partial_{k_j} f + \ordere{2} 
		. 
	\end{align*}
	Hence, the term to be quantized in equation~\eqref{magsapt:modulation_field:eom:eqn:Egorov_Duhammel} vanishes up to first order in $\eps$, 
	\begin{align*}
		\mbox{r.h.s.~of \eqref{magsapt:modulation_field:eom:eqn:Egorov_Duhammel}} &= \int_0^t \dd s \, e^{+ i \frac{s}{\eps} \Opk^A(\heff)} \, \Opk^A \bigl ( 0 + \ordere{2} \bigr ) \, e^{- i \frac{s}{\eps} \Opk^A(\heff)} = \ordern(\eps^2)
		. 
	\end{align*}
	This finishes the proof. 
\end{proof}
The main result combines Proposition~\ref{magsapt:modulation_field:eom:prop:effective_sym_building_block} with the Egorov theorem we have just proven: 
\begin{thm}[Semiclassical limit\index{semiclassical limit}]\label{magsapt:modulation_field:eom:thm:semiclassical_limit}
	Let $\hat{H}$ satisfy Assumptions~\ref{magsapt:intro:assumption:V_Gamma}, \ref{magsapt:intro:assumption:em_fields}, \ref{magsapt:mag_sapt:mag_wc:defn:gap_condition} and if $d \geq 4$ also Assumption~\ref{magsapt:mag_sapt:mag_wc:assumption:smooth_frame}. Furthermore, let us assume the relevant part of the spectrum consists of a single non-degenerate Bloch band $\Eb$. Then for all macroscopic semiclassical observables $f$ (Definition~\ref{magsapt:modulation_field:eom:defn:semiclassical_observable}) the full quantum evolution can be approximated by the hamiltonian flow $\Phi^{\mathrm{macro}}_t$ as given in Proposition~\ref{magsapt:modulation_field:eom:prop:macro_flow} if the initial state is localized in the corresponding tilted eigenspace $\Zak^{-1} \Pi \Zak L^2(\R^d_x)$, 
	\begin{align}
		\Bnorm{\Zak^{-1} \Pi \Zak \, \Bigl ( e^{+ i \frac{t}{\eps} \hat{H}} \, \Opk^A(f) \, e^{- i \frac{t}{\eps} \hat{H}} - \Opk^A \bigl ( f \circ \Phi^{\mathrm{macro}}_t \bigr ) \Bigr ) \, \Zak^{-1} \Pi \Zak}_{\mathcal{B}(L^2(\R^d_x))} \leq C_T \eps^2 
		. 
	\end{align}
\end{thm}
\begin{proof}
	We now combine all of these results to approximate the dynamics: let $f$ be a macroscopic observable. Then if we start with a state in the tilted eigenspace $\Zak^{-1} \Pi \Zak$, the time-evolved observable can be written as 
	\begin{align*}
		\Zak^{-1} \, \Pi \, &\Zak e^{- i \frac{t}{\eps} \hat{H}^{\eps}} \, \Zak^{-1} \Opk^A(f) \Zak \, e^{+ i \frac{t}{\eps} \hat{H}^{\eps}} \Zak^{-1} \Pi \Zak = 
		\\
		&= \Zak^{-1} \, \Pi e^{-i \frac{t}{\eps} \hat{H}^{\Zak}} \Opk^A(f) e^{+ i \frac{t}{\eps} \hat{H}^{\Zak}} \Pi \, \Zak 
		\\ 
		&= \Zak^{-1} \, U^{-1} \Piref U U^{-1} e^{- i \frac{t}{\eps} \hat{h}} U \Opk^A(f) U^{-1} e^{+ i \frac{t}{\eps} \hat{h}} U U^{-1} \Piref U \, \Zak + \ordern(\eps^{\infty})
		\\
		&= \Zak^{-1} \, U^{-1} \Piref e^{- i \frac{t}{\eps} \hat{h}} U \Opk^A(f) U^{-1} e^{+ i \frac{t}{\eps} \hat{h}} \Piref U \, \Zak + \ordern(\eps^{\infty})
		\\
		&= \Zak^{-1} \, U^{-1} \Piref e^{- i \frac{t}{\eps} \hat{h}_{\mathrm{eff}}} \Piref U \Opk^A(f) U^{-1} \Piref e^{+ i \frac{t}{\eps} \hat{h}_{\mathrm{eff}}} \Piref U \, \Zak + \ordern(\eps^{\infty})
		. 
	\end{align*}
	After replacing $U$ with $\Opk^A(u)$ (which adds another $\ordern(\eps^{\infty})$ error) and $\Piref$ with $\Opk^A(\piref)$, the term in the middle combines to the quantization of the effective observable $f_{\mathrm{eff}} = \piref \, u \magW f \magW u^* \, \piref$. We apply Proposition~\ref{magsapt:modulation_field:eom:prop:effective_sym_building_block} and the Egorov theorem involving $\heff$ and obtain 
	\begin{align*}
		\ldots &= \Zak^{-1} \, U^{-1} \Piref e^{- i \frac{t}{\eps} \hat{h}_{\mathrm{eff}}} \, \Opk^A \bigl (\piref \, u \magW f \magW u^* \, \piref \bigr ) \, e^{+ i \frac{t}{\eps} \hat{h}_{\mathrm{eff}}} \Piref U \, \Zak + \ordern(\eps^{\infty})
		\\
		&= \Zak^{-1} \, U^{-1} \Piref e^{- i \frac{t}{\eps} \hat{h}_{\mathrm{eff}}} \, \Opk^A \bigl (f_{\mathrm{eff}} \bigr ) \, e^{+ i \frac{t}{\eps} \hat{h}_{\mathrm{eff}}} \Piref U \, \Zak + \ordern(\eps^{\infty})
		\\
		&= \Zak^{-1} \, \Pi \, U^{-1} \, \Opk^A \bigl (f \circ T_{\mathrm{eff}} \circ \Phi^{\mathrm{eff}}_t \bigr ) \, U \, \Pi \, \Zak + \ordern(\eps^{2})
		. 
	\end{align*}
	Since two flows are $\ordere{2}$ close if the corresponding hamiltonian vector fields are \cite[Lemma~5.24]{Teufel:adiabatic_perturbation_theory:2003}, we conclude 
	\begin{align*}
		\ldots &= \Zak^{-1} \, \Pi \, U^{-1} \, \Opk^A \bigl (f \circ T_{\mathrm{eff}} \circ \Phi^{\mathrm{eff}}_t \circ T_{\mathrm{eff}}^{-1} \circ T_{\mathrm{eff}} \bigr ) \, U \, \Pi \, \Zak + \ordern(\eps^{2})
		\\
		&= \Zak^{-1} \, \Pi \, U^{-1} \, \Opk^A \bigl (f \circ \Phi^{\mathrm{macro}}_t \circ T_{\mathrm{eff}} \bigr ) \, U \, \Pi \, \Zak + \ordern(\eps^{\infty})
		\\
		&= \Zak^{-1} \, \Pi \, \Opk^A \bigl (f \circ \Phi^{\mathrm{macro}}_t \bigr ) \, \Pi \, \Zak + \ordern(\eps^{2})
		. 
	\end{align*}
	This finishes the proof. 
\end{proof}

\section{Physical relevance for the quantum Hall effect} 
\label{magsapt:physics}
Theorem~\ref{magsapt:modulation_field:eom:egorov:thm:semiclassics} relates the quantum dynamics associated to initial conditions in a relevant band to the semiclassical equations of motion 
\begin{align}
	\dot{r}_{\mathrm{eff}} &= + \nabla_{\keff} \heff - \eps \Omega \, \dot{k}_{\mathrm{eff}} + \order(\eps^2)
	\label{magsapt:physics:eqn:eom_reff}
	\\
	\dot{k}_{\mathrm{eff}} &= - \nabla_{\reff} \heff + \lambda B \, \dot{r}_{\mathrm{eff}} + \order(\eps^2)
	\label{magsapt:physics:eqn:eom_keff}
	. 
\end{align}
Since our main contribution is an extension of \cite{PST:effective_dynamics_Bloch:2003} to magnetic fields with components in $\BCont^{\infty}$, it includes the setting of the quantum Hall effect where a uniform field $B$ is applied to a quasi-two-dimensional sample. Can we explain the quantization of the Hall current? 

To put things into perspective, let us start with a few experimental facts: first of all, in the experimental setting of the quantum Hall effect, the magnetic fields are very weak. Typical crystals have lattice constants of $\unit[3.5 \sim 11]{Å}$ \cite{Grosso_Parravicini:solid_state_physics:2003} while the magnetic fields range from $\unit[1 \sim 12]{T}$ \cite{von_Klitzing_Dorda_Pepper:quantum_hall_effect:1980}. Thus, the magnetic flux is in silicon (lattice constant $\approx \unit[5]{Å}]$) is typically smaller than 
\begin{align*}
	\Phi \simeq \unit[5^2 \cdot 10]{T Å^2} \approx \unit[2.5 \cdot 10^{-18}]{T m^2} 
	, 
\end{align*}
\ie it is roughly $10^{-3}$ times the flux quantum $\Phi_0 \approx \unit[2 \cdot 10^{-15}]{T m^2}$, and we are \emph{always in the weak field regime}. The magnetic flux scales with $\eps$ in the relevant equations (\eg the magnetic Weyl product, equation~\eqref{asymptotics:expansions:eqn:Fourier_form_magnetic_Weyl_product}), this means $\eps \lambda \simeq 10^{-3}$ and thus, our mathematical description of the model is consistent with the physics. To be clear: we \emph{need not} assume that the magnetic flux through the unit cell $\WS$ is rational. 

Unfortunately, our result \emph{predicts no quantization} of the Hall current: in our model, the Hall quantization is linked to the Chern number which is associated to the relevant bands and propotional to $\int_{\BZ} \dd k \, \Omega(k)$. To see that, we plug \eqref{magsapt:physics:eqn:eom_keff} into equation~\eqref{magsapt:physics:eqn:eom_reff} and solve for $\dot{r}_{\mathrm{eff}}$, 
\begin{align*}
	(1 + \eps \lambda \Omega B) \dot{r}_{\mathrm{eff}} &= \nabla_{\keff} \heff + \eps \, \Omega \, \nabla_{\reff} \phi + \order(\eps^2) 
	. 
\end{align*}
If $\eps \lambda$ is small enough, then the factor left-hand side (which is really a scalar-multiple of the identity) is invertible and using the geometric series, we can obtain an explicit expression for 
\begin{align*}
	\dot{r}_{\mathrm{eff}} &= (1 + \eps \lambda \, \Omega \, B)^{-1} \bigl ( \nabla_{\keff} \heff + \eps \, \Omega \nabla_{\reff} \phi \bigr ) + \order(\eps^2) 
	\\
	&=  \nabla_{\keff} \heff - \eps \, \bigl (-  \Omega \, \nabla_{\reff} \phi + \lambda \, \Omega \, B \, \nabla_{\keff} \Eb \bigr ) + \order(\eps^2) 
	. 
\end{align*}
Assuming the band $\Eb$ is completely filled\footnote{The materials used in the quantum Hall effect are semiconductors. If they are gapped, the Fermi level $E_{\mathrm{F}}$ lies in the gap and all bands below $E_{\mathrm{F}}$ are completely filled. } we need to integrate with respect to the constant density $\sabs{\BZ}^{-1}$ in crystal momentum to get an averaged current, 
\begin{align}
	j(\reff) :=& \, \frac{1}{\eps} \frac{1}{\sabs{\BZ}} \int_{\BZ} \dd \keff \; \dot{r}_{\mathrm{eff}}(\reff,\keff)
	\notag \\
	=& \, \frac{1}{\eps} \frac{1}{\sabs{\BZ}} \int_{\BZ} \dd \keff \, \Bigl ( \nabla_{\keff} \heff - \eps \, \bigl (-  \Omega \, \nabla_{\reff} \phi + \lambda \, \Omega \, B \, \nabla_{\keff} \Eb \bigr ) \Bigr ) + \order(\eps) 
	\notag \\
	=& \, \frac{1}{\sabs{\BZ}} \underbrace{\left ( \int_{\BZ} \dd \keff \, \Omega \right )}_{= 0} \, \nabla_{\reff} \phi - \frac{\lambda}{\sabs{\BZ}} \int_{\BZ} \dd \keff \, \Omega \, B \, \nabla_{\keff} \Eb + \order(\eps) 
	. 
\end{align}
The terms involving the gradient of $\heff$ are $\Gamma^*$-periodic and thus their integral over the Brillouin zone vanishes. We are left with two terms: the first one is a multiple of the first Chern class. Since the Berry curvature $\Omega$ involves Bloch functions associated to the \emph{non}-magnetic hamiltonian, the first Chern class always vanishes \cite{Panati:triviality_Bloch_bundle:2006}. The second term is independent of $\reff$ and grows linearly in the strength of the magnetic field as $\Omega$ and $\Eb$ are determined by the unperturbed hamiltonian where the external fields vanish and thus independent of $\lambda$. 

The space-adiabatic point of view interprets the quantum Hall system as a small perturbation of the field-free case. Unfortunately, the predictions based on this model do not agree with experimental facts. This suggests that impurities and defects are indeed crucial to explain the quantization of transverse conductance: according to experiment, the quantum Hall effect is very robust and independent of sample size, sample geometry and details of the defect distribution \cite{Tscheuschner_et_al:robustness_quantum_Hall_effect:1998}. In fact, impurities are necessary to observe the effect \cite{Koch_Haug_von_Klitzing:impurities_quantum_Hall_effect:1991,Stoermer:robustness_quantum_Hall_effect:1998,von_Klitzing:quantum_Hall_effect:2004}. 


\chapter{An Algebraic Point of View} 
\label{algebraicPOV}
%
Magnetic pseudodifferential theory on $\R^d$ and $\T^d \subset \R^d$ can be recast in the language of twisted crossed products. These twisted crossed products are $C^*$-algebras whose elements are preimages of certain magnetic $\Psi$DO on $L^2(\R^d)$ under quantization and encapsulate many of the features of the corresponding operator algebras. It conceptually separates the algebraic object from the representation on Hilbert spaces. In this language, $\OpA$ is a \emph{representation} of more fundamental algebra of quantum observables -- more fundamental, because no reference is made to a Hilbert space or a vector potential $A$ associated to the magnetic field $B = \dd A$. This fundamental algebra depends only on the magnetic field $B$; the necessity to choose a vector potential associated to $B$ arises when one wants to represent this abstract algebra on a Hilbert space. 

In the algebraic approach, the behavior of the `potentials' is encoded in a $C^*$-algebra $\Alg \subseteq \BCont_u(\R^d)$. This algebra needs to be stable under translations, \ie for all $\varphi \in \Alg$ and $x \in \R^d$ the translated function $\varphi(\cdot + x) \in \Alg$ is still in the algebra. For instance, periodicity or decay properties can be encoded in $\Alg$. We shall henceforth call this algebra `anisotropy.' 
\medskip

\noindent
So let us present the general setup in the next two sections following  \cite{Mantoiu_Purice_Richard:twisted_X_products:2004}: let $\Alg$ be a separable, \emph{abelian} $C^*$-algebra with a group action\index{group action} $\theta : \Xgroup \longrightarrow \Aut(\Alg)$. Here, $\Xgroup$ is an abelian, second countable, locally compact group. The idea is to start with a crossed product similar to that in group theory and then to introduce a magnetic twist. The \emph{a priori} justification for this procedure is that we recover the usual pseudodifferential theory for $\Xgroup = \R^d$. The theory of twisted crossed products can be stated in much more generality, see  \cite{Packer_Raeburn:twisted_X_products_1:1989,Packer_Raeburn:twisted_X_products_2:1990}, for instance. In particular, neither $\Xgroup$ nor $\Alg$ need to be commutative. However, if $\Alg$ is commutative, we can completely characterize it by Gelfand theory.

\section{Twisted crossed products} 
\label{algebraicPOV:twisted_crossed_products}
Let us start with standard, non-magnetic pseudodifferential theory: the Weyl quantization 
\begin{align*}
	\bigl ( \Op(f) u \bigr )(x) :=& \frac{1}{(2\pi)^d} \int_{\R^d} \dd y \int_{{\R^d}^*} \dd \eta \, e^{-i \eta \cdot (y - x)} \, f \bigl ( \tfrac{1}{2} (x + y) , \eta \bigr ) \, u(y) 
	\notag \\
	=& \frac{1}{(2\pi)^{\nicefrac{d}{2}}} \int_{\R^d} \dd y  \, (\Fourier f) \bigl ( \tfrac{1}{2} (x + y) , y - x \bigr ) \, u(y) 
\end{align*}
%
of suitable functions $f : \R^d \times {\R^d}^* \longrightarrow \C$ is a representation of a `quantum algebra' on $L^2(\R^d)$. If $\Fourier f \in L^1(\R^d ; \BCont_u(\R^d))$ and $u \in L^2(\R^d)$, for instance, the above integral is absolutely convergent and we do not need oscillatory integral techniques as in Chapters~\ref{magWQ} and \ref{asymptotics}. However, as there is no simple characterization of $\Fourier^{-1} L^1(\R^d ; \BCont_u(\R^d)) \subset \Cont_{\infty}({\R^d}^* ; \BCont_u(\R^d))$, we introduce another representation 
\begin{align}
	\bigl ( \Rep(f) u \bigr ) (x) := \int_{\R^d} \dd y \, f \bigl ( \tfrac{1}{2}(x + y) , y - x \bigr ) \, u(y) 
\end{align}
where $f \in L^1(\R^d ; \BCont_u(\R^d))$ now and $(2 \pi)^{- \nicefrac{d}{2}}$ has been absorbed into the measure for convenience. The representations $\Op$ and $\Rep$ are related by partial Fourier transform, $\Op = \Rep \circ \Fourier$. Similar to the Weyl product $\Weyl$, there is an induced product $\repcom$ such that 
\begin{align}
	\Rep(f) \, \Rep(g) &= \Rep (f \repcom g) 
	\notag \\
	(f \repcom g)(x,y) &= \int_{\R^d} \dd x' \, f \bigl (y + \tfrac{1}{2} (x - x') , x' \bigr ) \, g \bigl ( y + \tfrac{1}{2} x' , x - x' \bigr ) 
	. 
	\label{algebraicPOV:twisted_crossed_products:eqn:untwisted_product_Rd}
\end{align}
This composition law is a sort of twisted convolution and it is easily proven that 
\begin{align*}
	\repcom : L^1(\R^d ; \BCont_u(\R^d)) \times L^1(\R^d ; \BCont_u(\R^d)) \longrightarrow L^1(\R^d ; \BCont_u(\R^d)) 
	. 
\end{align*}
If we add $f^{\repcom}(x) := f(-x)^{\ast}$ as involution, $(L^1(\R^d ; \BCont_u(\R^d)) , \repcom , {}^{\repcom})$ forms a Banach-$*$ algebra\index{Banach-$*$ algebra}. Completions of Banach-$*$ algebras of this type with respect to a natural $C^*$-norm are called \emph{crossed products $\Alg \rtimes_{\theta} \R^d$}. 

Before we continue, let us quickly recall some basic facts of Gelfand theory first.

\subsection{Gelfand theory} 
\label{algebraicPOV:twisted_crossed_products:gelfand_theory}

Abelian $C^*$-algebras $\Alg$ are \emph{completely classified by Gelfand theory}: simply put, the Gel\-fand-Naimark theorem says that they are always isomorphic to $\Cont_{\infty}(\Omega)$, the space of continuous complex-valued functions which decay at infinity, where $\Omega$ is a suitable locally convex space. With multiplication and involution declared pointwise in the usual manner and $\sup$ norm, this space is indeed a commutative $C^*$-algebra. 

For a $C^*$-algebra $\Alg$, \emph{which we always take to be separable}, we define the \emph{Gelfand spectrum}\index{Gelfand spectrum} 
\begin{align}
	\Salg := \bigl \{ h : \Alg \longrightarrow \C \; \vert \; h (\varphi \, \psi) = h(\varphi) \, h(\psi) \; \forall \varphi , \psi \in \Alg \bigr \} 
\end{align}
as the set of all morphisms. They are in one-to-one correspondence with the set of maximal ideals. Endowed with the topology of simple convergence (the weak-$\ast$ topology), $\Salg$ is a locally convex space. It is \emph{compact} if and only if $\Alg$ is \emph{unital}; in that case $\Cont_{\infty}(\Salg)$ coincides with $\BCont(\Salg)$. 

The Gelfand isomorphism\index{Gelfand isomorphism} $\Gelf : \Alg \longrightarrow \Cont_{\infty}(\Salg)$, $(\Gelf (\varphi))(\kappa) := \kappa(\varphi)$, establishes the equivalence between abstract abelian $C^*$-algebras $\Alg$ and $C^*$-algebras of the form $\Cont_{\infty}(\Omega)$ for some locally compact $\Omega$. 
The Gelfand map $\Gelf$ is really an isomorphism: the weak-$*$ topology separates points in $\Salg$ and hence we can apply the Stone-Weierstrass theorem. 

If $\Alg$ is not unital, then we need to define the multiplier algebra $\mathcal{M}(\Alg)$ as the class of all \emph{double centralizers} \cite{Dixmier:C_star_algebras:1977}. $\mathcal{M}(\Alg)$ contains $\Alg$ densely with respect to the \emph{strict topology}, \ie the locally convex topology induced by the family of seminorms $\bigl \{ \snorm{\cdot}_{\varphi} \bigr \}_{\varphi \in \Alg}$, $\snorm{m}_{\varphi} := \snorm{m \, \varphi}_{\Alg}$. In case $\Alg$ does not have a unit, we define the set of unitary elements of $\Alg$ as $\mathcal{U}(\Alg) := \mathcal{U} \bigl ( \mathcal{M}(\Alg) \bigr )$. 

In view of the Gelfand isomorphism, we can view the multiplier algebra $\mathcal{M}(\Alg)$ as $\BCont(\Salg)$ \cite{Raeburn_Williams:Morita_equivalence:1998} and the unitary group $\mathcal{U}(\Alg)$ can be seen as $\Cont(\Salg;\T)$ where $\T := \bigl \{ z \in \C \; \vert \; \abs{z} = 1 \bigr \}$ is the unit circle. On $\Cont(\Salg ; \T)$, the strict topology coincides with the topology of uniform convergence on compact subsets. 

Of particular importance are $\ast$-subalgebras of $\BCont_u(\Xgroup)$, the bounded, uniformly continuous functions, $\Alg \subseteq \BCont_u(\Xgroup)$. We will always assume that $\Alg$ is translation-invariant: for all $\varphi \in \Alg$ and $x \in \Xgroup$, $\varphi(\cdot + x)$ is in $\Alg$. For such algebras, we may include $\Xgroup$ into $\Salg$ via $\imath_{\Alg} : \Xgroup \longrightarrow \Salg$, $x \mapsto \imath_{\Alg}(x) := \delta_x$, although in general $\imath_{\Alg}$ is neither injective nor surjective. The Gelfand isomorphism $\Gelf : \Alg \longrightarrow \Cont_{\infty}(\Salg)$ maps each $\varphi \in \Alg$ onto $\tilde{\varphi} \in \Cont_{\infty}(\Salg)$ which are related via $\varphi = \tilde{\varphi} \circ \imath_{\Alg}$. As a matter of fact, $\imath_{\Alg}$ is injective if and only if $\Cont_{\infty}(\Xgroup) \subseteq \Alg$, a consequence of the Stone-Weierstrass theorem. We can say even more: 
%
\begin{lem}[\cite{Lein_Mantoiu_Richard:anisotropic_mag_pseudo:2009}]\label{algebraicPOV:twisted_crossed_products:lem:C0_in_or_out}
	Let $\Alg$ be $C^*$-subalgebra of $\BCont(\Xgroup)$ which is stable under translations. Then $\Cont_{\infty}(\Xgroup) \cap \Alg$ is either $\{ 0 \}$ or $\Cont_{\infty}(\Xgroup)$. 
\end{lem}
%
%
\begin{proof}
	Assume there exists a non-zero $\varphi \in \Cont_{\infty}(\Xgroup) \cap \Alg$. Since both, $\Cont_{\infty}(\Xgroup)$ and $\Alg$ are stable under translations, the $\ast$-subalgebra generated by $\{ \theta_x[\varphi] \; \vert \; x \in \Xgroup \}$ is dense in $\Cont_{\infty}(\Xgroup) \cap \Alg$ and separates points. By the Stone-Weierstrass theorem, this family is dense in $\Cont_{\infty}(\Xgroup)$, and by taking the closure, it follows that $\Cont_{\infty}(\Xgroup)$ is contained in $\Alg$. 
\end{proof}
In case $\Alg$ is unital and contains $\Cont_{\infty}(\Xgroup)$, we can view $\Salg$ as compactification of $\Xgroup$. The points $\Salg \setminus \imath_{\Alg}(\Xgroup) =: \SalgComp$ are then the points `located at infinity' which encode the asymptotic behavior of $\varphi$. 


\subsection{Crossed Products} 
\label{algebraicPOV:twisted_crossed_products:crossed_products}

Assume $\Xgroup$ is an abelian, locally compact, second countable group. Then we can naturally associate a $C^*$-algebra to it, namely 
\begin{align*}
	C^*(\Xgroup) := \bigl ( \overline{L^1(\Xgroup)}^{\norm{\cdot}} , \ast , {}^{\ast} \bigr ) 
	. 
\end{align*}
Here $\norm{\cdot}$ denotes the closure of the space of absolutely integrable complex-valued functions with respect to some suitable $C^*$-norm (given by equation~\eqref{algebraicPOV:twisted_crossed_products:eqn:Cstar_norm}), $f^{\ast}(x) := f(-x)^{\ast}$ is the involution and the convolution 
\begin{align}
	(f \ast g)(x) := \int_{\Xgroup} \dd y \, f(x - y) \, g(y) 
\end{align}
as product. We always integrate with respect to the Haar measure of $\Xgroup$. From standard theory, we know $\ast : L^1(\Xgroup) \times L^1(\Xgroup) \longrightarrow L^1(\Xgroup)$ and thus the convolution $\ast$ extends nicely to all of $C^*(\Xgroup)$. This $C^*$-algebra is in one-to-one correspondence with $\Cont_{\infty}(\dualX)$ equipped with the pointwise product and $\sup$ norm. $\dualX$ is the dual group of $\Xgroup$ (see Definition~\ref{algebraicPOV:generalized_weyl_calculus:defn:dual_group}) and the $C^*$-isomorphism is given by the Fourier transform. 

This definition can be generalized to $C^*$-dynamical systems which are the building blocks for `quantum algebras.' 
\begin{defn}[$C^*$-dynamical system\index{$C^*$-dynamical system}]
	A $C^*$-dynamical system is a triple $(\Alg,\theta, \Xgroup)$ formed by 
	\begin{enumerate}[(i)]
		\item an abelian, second countable locally compact group $\Xgroup$, 
		\item an abelian, separable $C^*$-algebra $\Alg$, and 
		\item a group morphism $\theta : \Xgroup \longrightarrow \mathrm{Aut}(\Alg)$ of $\Xgroup$ into the group of automorphisms on $\Alg$ such that for any $\varphi \in \Alg$ and $x \in \Xgroup$, the map $x \mapsto \theta_x[\varphi]$ is (norm-)continuous. 
	\end{enumerate}
\end{defn}
The crossed product is once again the completion of an $L^1$-space with a $\theta$-dependent product. The space $L^1(\Xgroup ; \Alg)$ is the space of Bochner-integrable functions with norm 
\begin{align}
	\snorm{f}_{L^1} := \int_{\Xgroup} \dd x \, \snorm{f(x)}_{\Alg} 
	. 
\end{align}
For any endomorphism $\tau \in \mathrm{End}(\Xgroup)$, we define 
\begin{align}
	(f \reptau g)(x) := \int_{\Xgroup} \dd y \, \theta_{\tau(y-x)} \bigl [ f(y) \bigr ] \, \theta_{(\id - \tau)(y)} \bigl [ g(x - y) \bigr ] \in \Alg
	\label{algebraicPOV:twisted_crossed_products:eqn:crossed_prod_tau}
	. 
\end{align}
In case of $\Xgroup = \R^d$, a choice of $\tau$ corresponds to a choice in `operator ordering:'\index{operator ordering} $\tau = \nicefrac{1}{2}$ gives the symmetric Weyl ordering. It can be checked that this `twisted' convolution maps $L^1(\Xgroup ; \Alg) \times L^1(\Xgroup ; \Alg)$ onto $L^1(\Xgroup ; \Alg)$; hence $L^1(\Xgroup ; \Alg)$ together with $f^{\reptau}(x) := \theta_{(\id - 2 \tau)(x)} \bigl [ f(-x)^{\ast} \bigr ]$ as involution, the triple $\bigl ( L^1(\Xgroup ; \Alg) , \reptau , {}^{\reptau} \bigr )$ forms a Banach-$\ast$ algebra. We will show later on for the twisted case that different choices of $\tau \in \mathrm{End}(\Xgroup)$ will lead to \emph{isomorphic algebras} (Lemma~\ref{algebraicPOV:twisted_crossed_products:cor:equivalence_twistedXproducts_tau}) and hence we will often suppress the $\tau$-dependence. $L^1(\Xgroup ; \Alg)$ is also an $A^*$-algebra\index{$A^*$-algebra}, meaning we can make it into a $C^*$-algebra by taking the completion with respect to the $C^*$-norm 
\begin{align}
	\snorm{f}_{\sXprod} := \sup \Bigl \{ \bnorm{\pi(f)}_{\mathcal{B}(\Hil)} \; \big \vert \; \mbox{$\pi$ non-degenerate representation on $\Hil$} \Bigr \} 
	. 
	\label{algebraicPOV:twisted_crossed_products:eqn:Cstar_norm}
\end{align}
Thus we define the 
\begin{defn}[Crossed product\index{crossed product}]
	The envelopping $C^*$-algebra $\Xprod \equiv \sXprod$ of the Banach-$\ast$ algebra $\bigl ( L^1(\Xgroup ; \Alg) , \reptau , {}^{\reptau} \bigr )$ will be called the \emph{crossed product of $\Alg$ by $\Xgroup$} associated with the action $\theta$ and the endomorphism $\tau$. 
\end{defn}
Certainly, by definition $L^1(\Xgroup ; \Alg)$ is dense in $\Xprod$ as are all dense subspaces of $L^1(\Xgroup ; \Alg)$. 

In principle, these crossed products are, in a suitable sense, prototypical, because \emph{twisted} crossed products that will be the topic of the next section can be written as untwisted crossed products of suitably enlarged algebras via a stabilization trick \cite{Packer_Raeburn:twisted_X_products_1:1989}. 
\medskip

\noindent
Crossed products can be \emph{covariantly represented} on a Hilbert spaces $\Hil$: the representation $r : \Alg \longrightarrow \mathcal{B}(\Hil)$ of the algebra $\Alg$ and the unitary representation $T : \Xgroup \longrightarrow \mathcal{U}(\Hil)$ of the group $\Xgroup$, need to intertwine correctly: 
\begin{defn}[Covariant representation\index{representation!covariant}]
	A covariant representation of a $C^*$-dyna\-mical system $(\Alg , \theta , \Xgroup)$ is a triple $(\Hil , r , T)$ where $\Hil$ is a separable Hilbert space and $r : \Alg \longrightarrow \mathcal{B}(\Hil)$ and $T : \Xgroup \longrightarrow \mathcal{U}(\Hil)$ are maps with the following properties: 
	\begin{enumerate}[(i)]
		\item $r$ is a non-degenerate $\ast$-representation, 
		\item $T$ is a strongly continuous unitary representation of $\Xgroup$ in $\Hil$, and 
		\item $T(x) \, r(\varphi) \, T(x)^* = r \bigl ( \theta_x [\varphi] \bigr )$ for all $x \in \Xgroup$ and $\varphi \in \Alg$, \ie $T$ and $\theta$ are intertwined via $r$. 
	\end{enumerate}
\end{defn}
\begin{lem}
	If the triple $(\Hil , r , T)$ is a covariant representation of the $C^*$-dynamical system $(\Alg , \theta , \Xgroup)$ and $\tau \in \mathrm{End}(\Xgroup)$, then $r \rtimes_{\tau} T$ defined on $L^1(\Xgroup ; \Alg)$ by 
	\begin{align}
		r \rtimes_{\tau} T (f) := \int_{\Xgroup} \dd x \, \theta_{\tau(x)} \bigl [ f(x) \bigr ] \, T(x) 
	\end{align}
	extends to a representation of $\Xprod$, called the \emph{integrated form of $(r,T)$}. 
\end{lem}
Commonly, $\Alg$ is a $C^*$-subalgebra of $\BCont_u(\Xgroup)$; then a particularly important representation $\Rep_{\tau}$ is that on $\Hil = L^2(\Xgroup)$: for any $u \in L^2(\Xgroup)$, it reads 
\begin{align}
	\bigl ( \Rep_{\tau}(f) u \bigr )(x) := \int_{\Xgroup} \dd y \, f \bigl ( (\id - \tau)(x) - \tau(y) ; y - x \bigr ) \, u(y) 
	. 
	\label{algebraicPOV:twisted_crossed_products:eqn:Rep_tau}
\end{align}
One can directly check that in case of $\Xgroup = \R^d$, we recover Weyl quantization if we set $\tau = \nicefrac{1}{2}$; $\tau = 0$ and $\tau = 1$ correspond to standard and anti-standard ordering (all derivatives to the right or left, respectively). Lastly, we would like to point out that by definition of the various norms, we can estimate the norm of \emph{any} covariant representation $r \rtimes_{\tau} T$ of $f \in \Xprod$ by 
\begin{align*}
	\bnorm{r \rtimes_{\tau} T(f)} \leq \bnorm{f}_{\sXprod} \leq \bnorm{f}_{L^1}
\end{align*}
where the right-most norm may be infinite. 


\subsection{Twisted crossed products} 
\label{algebraicPOV:twisted_crossed_products:twisted_crossed_products}
If we compare the product formula for crossed products, equation~\eqref{algebraicPOV:twisted_crossed_products:eqn:crossed_prod_tau}, and the magnetic product as in the introduction, we see that we can insert a factor of modulus $1$ that \emph{twists} the product even more, 
\begin{align*}
	(f \repom g)(x) := \int_{\Xgroup} \dd y \, \theta_{\tau(y-x)} \bigl [ f(y) \bigr ] \, \theta_{(\id - \tau)(y)} \bigl [ g(x - y) \bigr ] \, \theta_{- \tau(x)} \bigl [ \omega^B(y , x - y) \bigr ] 
	. 
\end{align*}
From Chapters~\ref{magWQ} and \ref{asymptotics}, we remember that $\omega^B(q ; x , y) := e^{- i \Gamma^B(\expval{q , q + x , q + x + y})}$ is the exponential of the flux through the triangle with corners $q$, $q + x$ and $q + x + y$. However, mathematically, this is not the only admissible choice of a twist, $\omega^B$ is just a special case of a so-called $2$-cocycle\index{$2$-cocycle}. Nevertheless, let us review the magnetic case in detail: Stoke's theorem allows us to rewrite the flux as a sum of circulations along the edges with respect to a vector potential, 
\begin{align*}
	\Gamma^B&(\expval{q , q + x , q + x + y}) =
	\\
	&\qquad \qquad 
	= \Gamma^A([q , q + x]) + \Gamma^A([q + x , q + x + y]) 
	+ \Gamma^A([q + x + y , q]) , 
\end{align*}
and hence 
\begin{align}
	\omega^B(q ; x , y) &= e^{- i \Gamma^A([q , q + x])} \, e^{- i \Gamma^A([q + x , q + x + y])} \, e^{+ i \Gamma^A([q , q + x + y])} 
	\notag \\
	&=: \lambda^A(q ; x) \, \theta_x [ \lambda^A(q ; y) ] \, {\lambda^A(q ; x + y)}^{-1}
	\label{algebraicPOV:twisted_crossed_products:eqn:magnetic_2_cocycle} 
	. 
\end{align}
Furthermore, products of $\omega^B$ can be interpreted as summing up the fluxes through the various triangles, \eg 
\begin{align}
	\omega^B(q ; x,y) \, \omega^B(q ; x + y , z) &= \omega^B(q +x ; y,z) \, \omega^B(q ; x , y + z) 
	\notag \\
	\Leftrightarrow
	\omega^B(x,y) \, \omega^B(x + y , z) &= \theta_x[\omega^B(y,z)] \, \omega^B(x , y + z) 
	\label{algebraicPOV:twisted_crossed_products:eqn:sum_of_flux_triangles}
\end{align}
%
has a simple geometric interpretation: there are exactly two ways to bisect the quadrangle with corners $q$, $q + x$, $q + x + y$ and $q + x + y + z$, but either way, the total magnetic flux is the same. This so-called \emph{cocycle condition} also ensures that the twisted composition law is associative. We call $\lambda^A$ a pseudotrivialization of $\omega^B$ as we can write the latter as a product of circulations. If we interpret the magnetic field $B$ as $2$-form and $A$ as $1$-form (in the sense of alternating differential forms), then on $\Xgroup = \R^d$, we know that every closed $2$-form ($\dd B = 0$) is also exact and there always exists a (highly non-unique) $1$-form $A$ with $\dd A = B$. This fact is encoded in the \emph{second cohomology group}\index{second cohomology group} which is defined as the quotient of closed $2$-forms and exact $2$-forms. As the name already suggests, the theory of twisted crossed products also has a \emph{coholomogical flavor}. We will only give a brief introduction and refer the interested reader to \cite[Section~2.3]{Mantoiu_Purice_Richard:twisted_X_products:2004}. Just like the exterior derivative maps $1$-forms onto $2$-forms, the coboundary map 
\begin{align}
	\bigl ( \delta^1(\lambda^A) \bigr )(x,y) := \lambda^A(x) \, \theta_x [ \lambda^A(y) ] \, {\lambda^A(x + y)}^{-1}
\end{align}
maps $1$-coboundaries onto $2$-cocycles (the equivalent of closed $2$-forms in differential geometry). The coboundary map $\delta^2$ which maps $2$-cochains onto $3$-cocycles satisfies $\delta^2 \circ \delta^1 = 1$ and hence any $\omega^B = \delta^1(\lambda^A)$ satisfies the $2$-cocycle condition 
\begin{align*}
	\bigl ( \delta^2(\omega^B) \bigr )(x,y,z) = \theta_x[\omega^B(y,z)] \, {\omega^B(x + y , z)}^{-1} \, \omega^B(x , y + z) \, {\omega^B(x,y)}^{-1} = 1
	. 
\end{align*}
We recognize this equation to be equivalent to \eqref{algebraicPOV:twisted_crossed_products:eqn:sum_of_flux_triangles}. In addition, we also see that the magnetic $2$-cocycle is \emph{normalized}, \ie 
\begin{align*}
	\omega^B(0,x) = 1 = \omega^B(x,0) 
	. 
\end{align*}
Let us now consider the general case and we shall drop $B$ in all equations to indicate that this construction is much more general. 
\begin{defn}[Twisted abelian $C^*$-dynamical system\index{$C^*$-dynamical system!twisted abelian}]
	A twisted $C^*$-dynamical \linebreak system is a quadruplet $(\Alg , \theta , \omega , \Xgroup)$ where 
	\begin{enumerate}[(i)]
		\item $\Xgroup$ is an abelian, second countable locally compact group, 
		\item $\Alg$ is an abelian, separable $C^*$-algebra, 
		\item $\theta : \Xgroup \longrightarrow \mathrm{Aut}(\Alg)$ is a group morphism from $\Xgroup$ to the group of automorphisms of $\Alg$ such that $x \mapsto \theta_x[\varphi]$ is (norm-)continuous for all $\varphi \in \Alg$, and 
		\item $\omega$ is a strictly continuous, normalized 2-cocycle with values in $\mathcal{U}(\Alg)$, the unitary group of the multiplier algebra of $\mathcal{A}$. 
	\end{enumerate}
\end{defn}
Twisted dynamical systems can also be covariantly represented: 
\begin{defn}[Covariant representation of a twisted $C^*$-dynamical system\index{representation!covariant}]\label{algebraicPOV:twisted_crossed_products:defn:covariant_rep_twisted_Cstar_dyn_sys}
	\makebox{}\linebreak For a given twisted $C^*$-dynamical system, a covariant representation consists of a Hilbert space $\Hil$ and two maps $r : \Alg \longrightarrow \mathcal{B}(\Hil)$, $T : \Xgroup \longrightarrow \mathcal{U}(\Hil)$ with the following properties: 
	\begin{enumerate}[(i)]
		\item $r$ is a non-degenerate representation of $\Alg$, 
		\item $T$ is strongly continuous and $T(x) \, T(y) = r \bigl ( \omega(x,y) \bigr ) \, T(x + y)$, and 
		\item $T(x) \, r(\varphi) \, T(x)^* = r \bigl ( \theta_x[\varphi] \bigr )$ for all $\varphi \in \Alg$ and $x \in \Xgroup$. 
	\end{enumerate}
\end{defn}
Just like in the previous section, the twisted crossed product is the completion of $L^1(\Xgroup ; \Alg)$ with respect to an abstract $C^*$-norm. 

\begin{defn}[Twisted crossed product\index{twisted crossed product}]\label{algebraicPOV:twisted_crossed_products:defn:twisted_X_product}
	The twisted crossed product $\twistedXprod \equiv \stwistedXprod$ is the completion of $\bigl ( L^1(\Xgroup;\Alg) , \repom , {}^{\repom} \bigr )$ with product and involution given by 
	\begin{align}
		(f \repom g)(x) &:= \int_{\Xgroup} \dd y \, \theta_{\tau(y-x)} \bigl [ f(y) \bigr ] \, \theta_{(\id - \tau)(y)} \bigl [ g(x - y) \bigr ] \, \theta_{- \tau(x)} \bigl [ \omega(y , x - y) \bigr ] 
		\label{algebraicPOV:twisted_crossed_products:eqn:twisted_composition_law}
		\\
		f^{\repom}(x) &:= \theta_{- \tau(x)} \bigl [ {\omega(x , -x)}^{-1} \bigr ] \, \theta_{(\id - 2 \tau)(x)} \bigl [ f(-x)^{\ast} \bigr ] 
		\label{algebraicPOV:twisted_crossed_products:eqn:twisted_involution}
	\end{align}
	with respect to the $C^*$-norm
	\begin{align*}
		\snorm{f} := \sup \Bigl \{ \snorm{\pi(f)}_{\mathcal{B}(\Hil)} \; \big \vert \; \; \mbox{$\pi$ is a non-degenerate representation on $\Hil$} \Bigr \}
		. 
	\end{align*}
\end{defn}
%
The fact that twisted crossed products are well-defined follows from the analogous statement for the underlying Banach-$\ast$ algebras: 
\begin{lem}\label{algebraicPOV:twisted_crossed_products:lem:Banach_star_algebra_well-defined}
	For two functions $f , g \in L^1(\Xgroup ; \Alg)$ and $\tau \in \mathrm{End}(\Xgroup)$, the product $f \repom g$ is again in $L^1(\Xgroup ; \Alg)$. Thus $\bigl ( L^1(\Xgroup ; \Alg) , \repom , {}^{\repom} \bigr )$ forms a Banach-$\ast$ algebra. For different $\tau \in \mathrm{End}(\Xgroup)$, the Banach-$\ast$ algebras are isomorphic. 
\end{lem}
\begin{proof}
	First, we show $\repom : L^1(\Xgroup ; \Alg) \times L^1(\Xgroup ; \Alg) \longrightarrow L^1(\Xgroup ; \Alg)$: as $\snorm{\theta_x[\varphi]}_{\Alg} = \snorm{\varphi}_{\Alg}$, we can estimate the $L^1$-norm explicitly: 
	\begin{align*}
		\bnorm{f \repom g}_{L^1} &= \int_{\Xgroup} \dd x \, \bnorm{(f \repom g)(x)}_{\Alg} 
		\\
		&
		\leq 
		\int_{\Xgroup} \dd x \int_{\Xgroup} \dd y \, \Bnorm{\theta_{\tau(y - x)} \bigl [ f(y) \bigr ] \, \theta_{(\id - \tau)(y)} \bigl [ g(x - y) \bigr ] \, \theta_{- \tau(y)} \bigl [ \omega(y,x-y) \bigr ]}_{\Alg} 
		\\
		&\leq 
		\int_{\Xgroup} \dd x \int_{\Xgroup} \dd y \, \snorm{f(y)}_{\Alg} \, \snorm{g(x - y)}_{\Alg} 
		= \snorm{f}_{L^1} \, \snorm{g}_{L^1} 
		. 
	\end{align*}
	The associativity of $\repom$ can be checked explicitly and relies on the $2$-cocycle property of $\omega$. The remaining properties are proven via routine calculations. Finally, we note that 
	\begin{align*}
		m_{\tau,\tau'} : &L^1(\Xgroup ; \Alg) \longrightarrow L^1(\Xgroup ; \Alg) , 
		\\
		&f(x) \mapsto \bigl ( m_{\tau,\tau'}(f) \bigr )(x) := \theta_{(\tau' - \tau)(x)}[f(x)]
	\end{align*}
	defines an isomorphism between $\bigl ( L^1(\Xgroup ; \Alg) , \repom , {}^{\repom} \bigr )$ and $\bigl ( L^1(\Xgroup ; \Alg) , \repcom^{\omega}_{\theta,\tau'} , {}^{\repcom^{\omega}_{\theta,\tau'}} \bigr )$. 
\end{proof}
\begin{cor}\label{algebraicPOV:twisted_crossed_products:cor:equivalence_twistedXproducts_tau}
	The twisted crossed product $\twistedXprod$ is well-defined and forms a $C^*$-algebra. For different $\tau \in \mathrm{End}(\Xgroup)$, the twisted crossed products are isomorphic. 
\end{cor}
Covariant representations admit an integrated form analogously to the untwisted case; the twist is implicitly contained in $T$ (see property (ii) in Definition~\ref{algebraicPOV:twisted_crossed_products:defn:covariant_rep_twisted_Cstar_dyn_sys}). 
\begin{lem}
	If $(\Hil , r , T)$ is a covariant representation of the twisted abelian $C^*$-dynamical system $(\Alg , \theta , \omega , \Xgroup)$ and $\tau \in \mathrm{End}(\Xgroup)$, then $r \rtimes_{\tau} T$ defined on $L^1(\Xgroup ; \Alg)$ by 
	\begin{align}
		r \rtimes_{\tau} T (f) := \int_{\Xgroup} \dd x \, \theta_{\tau(x)} \bigl [ f(x) \bigr ] \, T(x) 
	\end{align}
	extends to a representation of $\twistedXprod$ called the \emph{integrated form of $(r,T)$}. For any two $\tau , \tau' \in \mathrm{End}(\Xgroup)$, one has $r \rtimes_{\tau'} T = r \rtimes_{\tau} T \circ m_{\tau,\tau'}$ where $m_{\tau,\tau'}$ is the isomorphism from the proof of Lemma~\ref{algebraicPOV:twisted_crossed_products:lem:Banach_star_algebra_well-defined}. 
\end{lem}
%


\subsection{Special case of $\Xgroup$-algebras} 
\label{algebraicPOV:twisted_crossed_products:special_case_of_xgroup_algebras}

A particularly important choice for algebras $\Alg$ are those composed of certain classes of complex-valued functions on $\Xgroup$. They allow one to define derivatives -- indispensable for generalizations of pseudodifferential calculus (Chapter~\ref{psiDO_reloaded}). 
\begin{defn}[$\Xgroup$-algebra\index{$\Xgroup$-algebra}]\label{algebraicPOV:twisted_crossed_products:defn:X-algebra}
	Let $\Xgroup$ be an abelian, second countable, locally compact \linebreak group. We call a $C^*$-algebra $\Alg$ composed of bounded, uniformly continuous functions on $\Xgroup$ which is stable under translations, 
	\begin{align*}
		\theta_x[\varphi] = \varphi(\cdot + x) \in \Alg 
		, 
		&& 
		\forall \varphi \in \Alg , x \in \Xgroup
		, 
	\end{align*}
	an $\Xgroup$-algebra. 
\end{defn}
\begin{remark}
	If $\Cont_{\infty}(\Xgroup) \subseteq \Alg$, then we can embed $\Xgroup$ into the Gelfand spectrum $\Salg$ via $\imath_{\Alg} : \Xgroup \longrightarrow \Salg$ and $\imath_{\Alg}(\Xgroup) \subseteq \Salg$ is dense. If in addition $1 \in \Alg$, then $\Salg$ is a \emph{compactification of $\Xgroup$}. In any case, $\Alg \cong \Cont_{\infty}(\Salg)$, meaning to any $\varphi \in \Alg$ there exists a $\tilde{\varphi}$ such that $\varphi = \tilde{\varphi} \circ \imath_{\Alg}$. 
\end{remark}
The notion of equivalence of two cocycles is naturally provided from the framework of cohomology: 
\begin{defn}
	Let $\omega , \omega' \in C^2(\Xgroup ; \mathcal{U}(\Alg))$ be two $2$-cocycles. They are called cohomologous if they belong to the same class of cohomology, \ie if there exists $\lambda \in C^1(\Xgroup ; \mathcal{U}(\Alg))$ with 
	\begin{align*}
		\omega = \delta^1(\lambda) \, \omega' 
		. 
	\end{align*}
	$\omega$ is called trivial if it is cohomologous to the trivial $2$-cocycle $1$. 
\end{defn}
\begin{lem}\label{algebraicPOV:twisted_crossed_products:lem:twisted_crossed_products_cohomologou_isomorphic}
	If $\omega$ and $\omega'$ are two cohomologous $2$-cocycles,\index{$2$-cocycle!cohomologous} then $\twistedXprod$ and $\Alg \rtimes^{\omega'}_{\theta,\tau} \Xgroup$ are isomorphic. 
\end{lem}
\begin{proof}
	Assume $\lambda \in C^1(\Xgroup ; \mathcal{U}(\Alg))$ is such that 
	\begin{align*}
		\omega = \delta^1(\lambda) \, \omega' 
		. 
	\end{align*}
	Then $i^{\lambda}_{\tau} : L^1(\Xgroup ; \Alg) \longrightarrow L^1(\Xgroup ; \Alg)$, $\bigl ( i^{\lambda}_{\tau}(f) \bigr )(x) := \theta_{- \tau(x)}[\lambda(x)] \, f(x)$ extends to an isomorphism between $\twistedXprod$ and $\Alg \rtimes^{\omega'}_{\theta,\theta} \Xgroup$. The products are intertwined via 
	\begin{align*}
		f \repcom^{\omega'}_{\theta,\tau} g = {i^{\lambda}_{\tau}}^{-1} \Bigl ( \bigl ( i^{\lambda}_{\tau} (f) \bigr ) \repom \bigl ( i^{\lambda}_{\tau} (g) \bigr ) \Bigr ) 
	\end{align*}
	which we confirm by direct calculation: for $f , g \in L^1(\Xgroup ; \Alg)$, we calculate 
	\begin{align*}
		\Bigl ( \bigl ( i^{\lambda}_{\tau} &(f) \bigr ) \repom \bigl ( i^{\lambda}_{\tau} (g) \bigr ) \Bigr ) (x) 
		= \\
		&= \int_{\Xgroup} \dd y \, \theta_{\tau(y - x)} \bigl [ \bigl ( i^{\lambda}_{\tau} (f) \bigr ) (y) \bigr ] \, \theta_{(\id - \tau)(y)} \bigl [ \bigl ( i^{\lambda}_{\tau} (g) \bigr ) (x - y) \bigr ] \, \theta_{- \tau(x)} \bigl [ \omega(y , x - y) \bigr ] 
		\\
		&= \int_{\Xgroup} \dd y \, \theta_{\tau(y - x) - \tau(y)}[\lambda(y)] \, \theta_{\tau(y - x)}[f(y)] \cdot 
		\\
		&\qquad \qquad \cdot \theta_{(\id - \tau)(y) - \tau(x - y)}[\lambda(x - y)] \, \theta_{(\id - \tau)(y)}[g(x - y)] \, \theta_{- \tau(x)}[\omega(y , x - y)] 
		\\
		&= \int_{\Xgroup} \dd y \, \theta_{\tau(y - x)}[f(y)] \, \theta_{(\id - \tau)(y)}[g(x - y)] \cdot 
		\\
		&\qquad \qquad \cdot \theta_{y - \tau(x)}[\lambda(x - y)] \, \theta_{- \tau(x)}[\lambda(y)] \, \theta_{- \tau(x)}[\omega(y , x - y)] 
		\\
		&= \int_{\Xgroup} \dd y \, \theta_{\tau(y - x)}[f(y)] \, \theta_{(\id - \tau)(y)}[g(x - y)] \cdot 
		\\
		&\qquad \qquad \cdot \theta_{- \tau(x)} \bigl [ \lambda(y) \, \theta_y[\lambda(x - y)] \, \omega(y,x - y) \bigr ] 
		\\
		&= \theta_{- \tau(x)} \bigl [ {\lambda(x)}^{-1} \bigr ] \, \int_{\Xgroup} \dd y \, \theta_{\tau(y - x)}[f(y)] \, \theta_{(\id - \tau)(y)}[g(x - y)] 
		\, 
		\theta_{- \tau(x)} \bigl [ \omega'(y , x - y) \bigr ] 
		\\
		&= \theta_{- \tau(x)} \bigl [ {\lambda(x)}^{-1} \bigr ] \, \Bigl ( \bigl ( i^{\lambda}_{\tau} (f) \bigr ) \repcom^{\omega'}_{\theta,\tau} \bigl ( i^{\lambda}_{\tau} (g) \bigr ) \Bigr ) (x)
		. 
	\end{align*}
	By density, this immediately extends to $f , g \in \twistedXprod$. It is immediate from the definitions that $\bigl ( i^{\lambda}_{\tau} (f) \bigr )^{\repom} = f^{\repcom^{\omega'}_{\theta,\tau}}$. 
\end{proof}
%
If $\Alg$ is an $\Xgroup$-algebra, it is useful to introduce the notion of \emph{pseudo}triviality. We need an abstract non-sense lemma first (which is a corollary of \cite[Lemma~2.9]{Mantoiu_Purice_Richard:twisted_X_products:2004}): 
\begin{lem}\label{algebraicPOV:twisted_crossed_products:lem:triviality_2_cocycle}
	If $\Cont_{\infty}(\Xgroup) \subseteq \Alg$, then every $n$-cocycle, $n \geq 1$, is trivial. 
\end{lem}
First of all, if $\Alg' \subseteq \Alg$, then certainly we have $\mathcal{U}(\Alg') \subseteq \mathcal{U}(\Alg)$. Hence, if $\Cont_{\infty}(\Xgroup) \subseteq \Alg$, then $\mathcal{U}(\Cont_{\infty}(\Xgroup)) \cong \Cont(\Xgroup ; \T) \subseteq \mathcal{U}(\Alg)$. Then every $2$-cocycle is trivial, \ie there exists a $1$-cochain $\lambda$ such that $\omega = \delta^1(\lambda)$ and $\omega$ is cohomologous to $1$. 

If $\Cont_{\infty}(\Xgroup) \cap \Alg = \{ 0 \}$ (and those are the only two cases by Lemma~\ref{algebraicPOV:twisted_crossed_products:lem:C0_in_or_out}), then there still exists $\lambda \in C^1 \bigl ( \Xgroup ; \mathcal{U}(\BCont_u(\Xgroup)) \bigr ) \cong C^1(\Xgroup ; \Cont(\Xgroup ; \T))$. However, in general $\lambda(x)$ \emph{fails to be in $\mathcal{U}(\Alg)$} and $\omega$ is trivial only if we enlarge $\Alg$. Hence, the twisted crossed products associated to two non-cohomologous $2$-cocycles $\omega$ and $\omega'$ will \emph{still be different}, but they can both be embedded in the same, larger and trivially twisted crossed product, 
%
\begin{align*}
	\twistedXprod , \Alg \rtimes^{\omega'}_{\theta,\tau} \Xgroup \subset \BCont_u(\Xgroup) \rtimes^{\omega}_{\theta,\tau} \Xgroup 
	\cong \BCont_u(\Xgroup) \rtimes^{\omega'}_{\theta,\tau} \Xgroup 
	\cong \BCont_u(\Xgroup) \rtimes^{\id}_{\theta,\tau} \Xgroup 
	. 
\end{align*}
A $1$-cochain $\lambda \in C^1(\Xgroup ; \Cont(\Xgroup ; \T))$ for which $\omega = \delta^1(\lambda)$ holds in an enlarged algebra is a \emph{pseudotrivialization of $\omega$}.\index{$2$-cocycle!pseudotrivialization of} 
\begin{prop}
	Let $(\Alg,\Xgroup,\theta,\omega)$ be a twisted $C^*$-dynamical system where $\Alg$ is an $\Xgroup$-algebra. Then $\omega$ is pseudotrivial. 
\end{prop}
\begin{proof}
	We will now make the comment above rigorous, \ie that we can regard the unitary group of any non-trivial $\Alg \subseteq \BCont_u(\Xgroup)$ as a subgroup of $\Cont(\Xgroup ; \T)$, the group of unitary elements of the multiplier algebra of $\Cont_{\infty}(\Xgroup)$. This is mostly a matter of topologies. 
	
	First of all, we will show that the multiplier algebra $\mathcal{M}(\Alg)$ can be regarded as a $C^*$-subalgebra of $\mathcal{M}(\Cont_{\infty}(\Xgroup)) = \BCont(\Xgroup)$. Unless $\Alg = \{ 0 \}$, the invariance of $\Alg$ under translations implies the non-degeneracy of the natural, faithful representation $r : \Alg \longrightarrow \mathcal{B} \bigl ( L^2(\Xgroup) \bigr )$ of $\Alg$ on $L^2(\Xgroup)$ as multiplication operators. The double commutant of $r(\Alg)$ is contained in $\mathcal{M}(r(\Alg)) \cong \mathcal{M}(\Alg)$ which, in turn, is contained in $L^{\infty}(\Xgroup)$. Continuity of any $\varphi \in \Alg$ and translation invariance of $\Alg$ also imply that any $m \in \mathcal{M}(\Alg) \subseteq L^{\infty}(\Xgroup)$ needs to be continuous: assume $m$ is not continuous at $x_0 \in \Xgroup$. Then there exists $\varphi \in \Alg$ which does not vanish in an open neighborhood of $x_0$ and thus $m \varphi \not \in \Alg$ would no longer be continuous. Hence, $\mathcal{M}(\Alg) \subseteq \BCont(\Xgroup) = \mathcal{M}(\Cont_{\infty}(\Xgroup))$ is a subset of the bounded continuous functions on $\Xgroup$. 
	
	Furthermore, $\mathcal{U}(\Alg) \subset \mathcal{M}(\Alg) \cong \BCont(\Salg)$ is a \emph{bounded} subset of $\BCont(\Xgroup)$ so that the strict topology on $\mathcal{U}(\Alg)$ coincides with the topology of uniform convergence on compact subsets of $\Xgroup$. On the other hand, the strict topology on $\mathcal{M}(\Alg)$ coincides with the topology of uniform convergence on compact subsets of $\Salg$. As $\imath_{\Alg} : \Xgroup \longrightarrow \Salg$ is continuous, the strict topology on $\mathcal{M}(\Alg)$ induces a finer topology on $\mathcal{U}(\Alg)$ than that inherited from $\BCont(\Xgroup)$. 
	
	Hence, we can identify $\mathcal{U}(\Alg)$ with a subgroup of $\Cont(\Xgroup ; \T)$ where the natural topology of $\mathcal{U}(\Alg)$ inherited from $\mathcal{M}(\Alg)$ is finer than the strict topology of $\Cont(\Xgroup ; \T)$. This means we can consider each $2$-cocycle as a function $C^2(\Xgroup ; \Cont(\Xgroup ; \T))$ and then apply Lemma~\ref{algebraicPOV:twisted_crossed_products:lem:triviality_2_cocycle}. This concludes the proof. 
\end{proof}
\begin{example}
	One such example where $2$-cocycles can in general only be pseudotriviallized is $\Xgroup = \R^d$ and $\Alg_{\mathrm{per}} := \bigl \{ \varphi \in \BCont_u(\R^d) \; \vert \; \varphi(\cdot + \gamma) = \varphi(\cdot) \; \forall \gamma \in \Z^d \bigr \}$. Here, magnetically twisted crossed products $\Alg_{\mathrm{per}} \rtimes^{\omega^B}_{\theta,\tau} \R^d$ for different magnetic fields $B$ and $B'$ with components in $\Alg_{\mathrm{per}}$ are only isomorphic if and only if the flux through a unit cell differs by an integer multiple of $2 \pi$. A $2$-cocycle $\omega^B$ is trivial (cohomologous to $1$) if and only if $B$ has zero flux modulo $2 \pi$ through a unit cell. In other words, there exists a vector potential $A$ with components in $\Alg_{\mathrm{per}}$. 
\end{example}
\begin{example}
	In case of $\Xgroup = \Z^d$ and $\Alg = \C$, the twisted crossed product $\C \rtimes^{\omega}_{\id,\tau} \Z^d$ (which is independent of the choice of $\tau$) also cannot be trivialized, $\mathcal{U}(\C) \cong \T$, and is often called the twisted $C^*$-algebra of $\Z^d$ associated with $\omega$. If $\omega$ is given by a constant magnetic field $B$, then $\C \rtimes^{\omega^B}_{\id,\tau} \Z^d \equiv C^*_{\omega^B}(\Z^d)$ is commonly called \emph{non-commutative torus}. 
\end{example}
Another nice aspect of $\Xgroup$-algebras is that they have a natural representation that is large enough to contain all information: the natural Hilbert space is $L^2(\Xgroup)$, the space of functions which are square-integrale with respect to the Haar measure of $\Xgroup$. The representation $r$ of the algebra $\Alg$ maps $\varphi$ onto $r(\varphi) := \varphi(Q)$, the operator of multiplication by $\varphi$. The non-magnetic translations $\bigl ( T(y) u \bigr )(x) = u(x + y)$ are augumented with a \emph{pseudo}trivialization of $\omega = \delta^1(\lambda)$, $\lambda \in C^1(\Xgroup ; \Cont(\Xgroup ; \T))$. With respect to the larger algebra, there is only one cohomology class and any $2$-cocycle is trivial. Thus, we define $T^{\lambda}(y) := r(\lambda(y)) \, T(y)$ which, applied to $u \in L^2(\Xgroup)$, yields 
\begin{align}
	\bigl ( T^{\lambda}(y) u \bigr )(x) := \lambda(x ; y) \, u(x+y) 
	. 
\end{align}
The triple $(L^2(\Xgroup) , r , T^{\lambda})$ is a covariant representation of $\twistedXprod$ and is manifestly gauge-covariant. 
\begin{prop}\label{algebraicPOV:twisted_crossed_products:prop:gauge_covariance_rep}
	\begin{enumerate}[(i)]
		\item $(L^2(\Xgroup) , r , T^{\lambda})$ is a covariant representation of the twisted $C^*$-dynamical system $(\Alg , \Xgroup , \theta , \omega)$ where $\Alg$ is an $\Xgroup$-algebra. 
		\item If $\lambda' \in C^1(\Xgroup ; \Cont(\Xgroup ; \T))$ is another $1$-cochain pseudotrivializing $\omega = \delta^1(\lambda')$, then $\Rep^{\lambda}_{\tau} := r \rtimes_{\tau} T^{\lambda}$, 
		\begin{align}
			\bigl ( \Rep^{\lambda}_{\tau} (f) u \bigr )(x) &= \int_{\Xgroup} \dd y \, f(x + \tau(y) ; y) \, \lambda(x ; y) \, u(x + y) 
			\notag \\
			&= \int_{\Xgroup} \dd y \, f \bigl ( (\id - \tau)(x) + \tau(y) ; y - x) \, \lambda(x ; y - x) \, u(y) 
			, 
		\end{align}
		and $\Rep^{\lambda'}_{\tau}$ are unitarily equivalent:\index{representation!covariant} if $c \in C^0(\Xgroup ; \Cont(\Xgroup ; \T)) \equiv \Cont(\Xgroup ; \T)$ is the $0$-cochain such that $\lambda' = \theta_x[c] \, c^{-1} \, \lambda$, then 
		\begin{align}
			\Rep^{\lambda'}_{\tau}(f) = r(c^{-1}) \, \Rep^{\lambda}_{\tau}(f) \, r(c) 
		\end{align}
		for any $f \in \twistedXprod$. 
	\end{enumerate}
\end{prop}
\begin{proof}
	\begin{enumerate}[(i)]
		\item The claim follows from direct verification of the properties enumerated in Definition~\ref{algebraicPOV:twisted_crossed_products:defn:covariant_rep_twisted_Cstar_dyn_sys}. 
		\item There always exists a $c \in C^0(\Xgroup ; \Cont(\Xgroup ; \T))$ such that $\delta^1(\lambda) = \omega = \delta^1(\lambda')$: as $\lambda$ and $\lambda'$ are cohomologous, $\nicefrac{\lambda'}{\lambda}$ is a $1$-cocycle and by Lemma~\ref{algebraicPOV:twisted_crossed_products:lem:triviality_2_cocycle} there exists $c \in C^0(\Xgroup ; \Cont(\Xgroup ; \T))$ such that $\delta^0(c) = \theta_x[c] \, c^{-1} = \nicefrac{\lambda'}{\lambda}$. In other words, $\lambda' = \theta_x[c] \, c^{-1} \, \lambda$ and we conclude $T^{\lambda'}(y) = r(c^{-1}) \, T^{\lambda}(y) \, r(c)$ for all $y \in \Xgroup$, 
		\begin{align*}
			\bigl ( T^{\lambda'}(y) u \bigr ) (x) &= \lambda'(x ; y) \, u(x + y) 
			\\
			&= \theta_y[c(x)] \, c^{-1}(x) \, \lambda(x ; y) \, u(x + y) 
			= c^{-1}(x) \, \lambda(x ; y) \, c(x + y) \, u(x + y) 
			\\
			&= \bigl [ \bigl ( r(c^{-1}) \, T^{\lambda}(y) r(c) \bigr ) u \bigr ] (x) 
			. 
		\end{align*}
		Thus, the representations $\Rep^{\lambda}_{\tau}$ and $\Rep^{\lambda'}_{\tau}$ are also related by conjugation with $r(c^{-1})$. 
	\end{enumerate}
\end{proof}
Some basic facts of the Schrödinger representation $\Rep^{\lambda}_{\tau}$\index{representation!Schrödinger} on $L^2(\Xgroup)$ are proven in the next proposition: 
\begin{prop}\label{algebraicPOV:twisted_crossed_products:prop:basic_properties_rep}
	\begin{enumerate}[(i)]
		\item $\Rep^{\lambda}_{\tau} \bigl ( \Cont_{\infty}(\Xgroup) \rtimes^{\omega}_{\theta,\tau} \Xgroup \bigr )$ always coincides with $\mathcal{K} \bigl ( L^2(\Xgroup) \bigr )$, the $C^*$-algebra of all compact operators on $L^2(\Xgroup)$.\index{compact operators $\mathcal{K}(L^2(\Xgroup))$} 
		\item If $\Cont_{\infty}(\Xgroup) \subseteq \Alg$, then $\Rep^{\lambda}_{\tau}$ is irreducible.\index{representation!irreducible} 
		\item $\Rep^{\lambda}_{\tau}$ is faithful.\index{representation!faithful} 
	\end{enumerate}
\end{prop}
\begin{proof}
	\begin{enumerate}[(i)]
		\item Since $\Alg = \Cont_{\infty}(\Xgroup)$, Lemma~\ref{algebraicPOV:twisted_crossed_products:lem:triviality_2_cocycle} tells us that all $2$-cocycles are cohomologous to $1$ and hence we can invoke Lemma~\ref{algebraicPOV:twisted_crossed_products:lem:twisted_crossed_products_cohomologou_isomorphic} to conclude 
		\begin{align*}
			\Cont_{\infty}(\Xgroup) \rtimes^{\omega}_{\theta,\tau} \Xgroup \cong \Cont_{\infty}(\Xgroup) \rtimes^{\id}_{\theta,0} \Xgroup 
			. 
		\end{align*}
		On the other hand, Proposition~\ref{algebraicPOV:twisted_crossed_products:prop:gauge_covariance_rep} also gives that $\Rep^{\lambda}_{\tau}$ of $\Cont_{\infty}(\Xgroup) \rtimes^{\omega}_{\theta,\tau} \Xgroup$ and $\Rep^{\id}_{0}$ of $\Cont_{\infty}(\Xgroup) \rtimes^{\id}_{\theta,0} \Xgroup$ are unitarily equivalent. However, from standard theory, we know that $\Rep^{\id}_{0} \bigl ( \Cont_{\infty}(\Xgroup) \rtimes^{\id}_{\theta,0} \Xgroup \bigr ) = \mathcal{K} \bigl ( L^2(\Xgroup) \bigr )$ \cite{Georgescu_Iftimovici:X_products:2002} \cite[Proposition~2.17]{Georgescu_Iftimovici:Cstar_algebras_quantum_hamiltonians:2003}. 
		\item If $\Cont_{\infty}(\Xgroup)$ is contained in $\Alg$, then the irreducibility of 
		\begin{align*}
			\Rep^{\lambda}_{\tau} \bigl ( \twistedXprod \bigr ) \supseteq \mathcal{K} \bigl ( L^2(\Xgroup) \bigr ) 
		\end{align*}
		follows from the irreducibility of $\mathcal{K} \bigl ( L^2(\Xgroup) \bigr )$. 
		\item We will make use of Theorem~3.11 in \cite{Packer_Raeburn:twisted_X_products_1:1989}: as any abelian, locally compact group is also amenable, there exists a faithful representation (called regular representation) of $\twistedXprod$ on $L^2 \bigl (\Xgroup ; L^2(\Xgroup) \bigr )$. We follow Packer and Raeburn and set $r' : \Alg \longrightarrow \mathcal{B} \bigl ( L^2 \bigl ( \Xgroup ; L^2(\Xgroup) \bigr ) \bigr )$
		\begin{align*}
			\bigl ( r'(\varphi) u \bigr )(x) := \theta_x[\varphi] \, u(x) \in L^2(\Xgroup)
			&& \forall \varphi \in \Alg 
		\end{align*}
		as well as $T' : \Xgroup \longrightarrow \mathcal{U} \bigl ( L^2 \bigl ( \Xgroup ; L^2(\Xgroup) \bigr ) \bigr )$ 
		\begin{align*}
			\bigl ( T'(y) u \bigr )(x) := \omega(x,y) \, u(x + y) \in L^2(\Xgroup)
			&& \forall x , y \in \Xgroup 
			. 
		\end{align*}
		The unitary operator $W^{\lambda} : L^2 \bigl ( \Xgroup ; L^2(\Xgroup) \bigr ) \longrightarrow L^2 \bigl ( \Xgroup ; L^2(\Xgroup) \bigr )$, 
		\begin{align*}
			\bigl ( W^{\lambda} u \bigr )(x ; y) :&= \lambda(x ; y) \, u(x ; x + y) 
			, 
		\end{align*}
		intertwines the usual $r'$ and $\id_{L^2(\Xgroup)} \otimes r$ as well as $T'$ and $\id_{L^2(\Xgroup)} \otimes T^{\lambda}$: 
		\begin{align*}
			\bigl ( {W^{\lambda}}^{\ast} \, r'(\varphi) \, W^{\lambda} u \bigr ) (x ; y) &= \varphi(y) \, u(x ; y) 
			= \bigl ( \id_{L^2(\Xgroup)} \otimes r(\varphi) \, u \bigr ) (x ; y)
			\\
			\bigl ( {W^{\lambda}}^{\ast} \, T'(z) \, W^{\lambda} u \bigr ) (x ; y) &= \lambda(y ; z) \, u(x ; y + z) 
			= \bigl ( \id_{L^2(\Xgroup)} \otimes T^{\lambda}(z) \, u \bigr ) (x ; y) 
			. 
		\end{align*}
		The above follows from calculations as well as 
		%
		\begin{align*}
			\bigl ( {W^{\lambda}}^{\ast} u \bigr )(x ; y) = \lambda(x ; y - x)^{-1} \, u(x ; y - x) 
			. 
		\end{align*}
		Hence, the faithful regular representation of $\Alg \rtimes^{\omega}_{\theta,0} \Xgroup$ is unitarily equivalent to the covariant representation $\bigl ( L^2(\Xgroup) \otimes L^2(\Xgroup) , \id_{L^2(\Xgroup)} \otimes r , \id_{L^2(\Xgroup)} \otimes T^{\lambda} \bigr )$. The twisted crossed products $\twistedXprod$ for different choices of $\tau \in \mathrm{End}(\Xgroup)$ are unitarily equivalent (Lemma~\ref{algebraicPOV:twisted_crossed_products:cor:equivalence_twistedXproducts_tau}), so $r \rtimes_{\tau} T^{\lambda} \equiv \Rep_{\tau}^{\lambda}$ inherits the faithfulness of the regular representation. 
	\end{enumerate}
\end{proof}
\begin{remark}
	If $\Cont_{\infty}(\Xgroup) \not\subseteq \Alg$, \ie $\Cont_{\infty}(\Xgroup) \cap \Alg = \{ 0 \}$, then $\Rep^{\lambda}_{\tau}$ need not be irreducible. 
\end{remark}
\begin{remark}
	Lastly, one remark regarding twisted crossed products being isomorphic: if $\Cont_{\infty}(\Xgroup) \subseteq \Alg$, then $\mathcal{U}(\Cont_{\infty}(\Xgroup)) = \Cont(\Xgroup ; \T) = \mathcal{U}(\Alg)$ and all $2$-cocycles are trivial, \ie cohomologuous to $1$  (Lemma~\ref{algebraicPOV:twisted_crossed_products:lem:triviality_2_cocycle}). Then the twisted crossed product is independent of the particular choice of $\omega$. There are several caveats: physically, this does not imply that magnetic and non-magnetic theories agree. In fact, instead of working with $f \in \twistedXprod$, one has to use $i^{\lambda}_{\tau}(f) \in \Alg \rtimes^{\id}_{\theta,\tau} \Xgroup$ -- which depends on the trivializing $1$-coboundary ($A$-dependent) -- to describe the same object. In particular, saying that $\omega^B$ is trivial (cohomologous to $1$) does in no way imply $B = 0$. 
	
	Related to this, there is the notion of pseudotriviality that is necessary to find natural covariant representations on $\Hil = L^2(\Xgroup)$. The proper way to think about this is to embed $\twistedXprod$ into a larger $C^*$-algebra, \eg $\BCont_u(\Xgroup) \rtimes^{\omega}_{\theta,\tau} \Xgroup$, which does not depend on the $2$-cocycle. The representation of the larger algebra can then be restricted to the smaller algebra (Proposition~\ref{algebraicPOV:twisted_crossed_products:prop:gauge_covariance_rep}). 
\end{remark}
Before we turn to generalized Weyl calculus, let us quote without proof two very useful results from \cite{Packer_Raeburn:twisted_X_products_1:1989,Packer_Raeburn:twisted_X_products_2:1990}: these results have been proven for twisted crossed products from non-abelian $C^*$-algebras $\Alg$ and groups $\Xgroup$, but we will stick to the abelian framework.\footnote{Furthermore, they have worked with cohomologies in the category of Borel functions rather than continuous functions. } Assume $\Xgroup$ admits a closed subgroup $N$ (which is automatically normal by commutativity) so that 
\begin{align*}
	\Xgroup \cong N \times {\Xgroup} / {N} 
	. 
\end{align*}
Then one may suspect that also the twisted crossed product decomposes into 
\begin{align*}
	\twistedXprod \cong \bigl ( \Alg \rtimes^{\omega'}_{\theta',\tau} N \bigr ) \rtimes^{\omega''}_{\theta'',\tau} {\Xgroup} / {N} 
	. 
\end{align*}
Packer and Raeburn have answered this in the positive: 
%
\begin{thm}[Theorem 4.1 in \cite{Packer_Raeburn:twisted_X_products_1:1989}]
	Assume $(\Alg , \Xgroup , \theta , \omega)$ is a separable, abelian and twisted $C^*$-dynamical system and $\mathcal{N}$ a closed normal subgroup of $\Xgroup$. Then the twisted crossed product is isomorphic to an iterated twisted crossed product,\index{twisted crossed product!iterated} 
	\begin{align*}
		\twistedXprod \cong \bigl ( \Alg \rtimes^{\omega'}_{\theta',\tau} \mathcal{N} \bigr ) \rtimes^{\omega''}_{\theta'',\tau} \Xgroup / \mathcal{N} 
		,  
	\end{align*}
	where the actions $\theta' : \mathcal{N} \longrightarrow \mathrm{Aut}(\Alg)$ and $\theta'' : \Xgroup / \mathcal{N} \longrightarrow \mathrm{Aut} \bigl ( \Alg \rtimes^{\omega'}_{\theta',\tau} \mathcal{N} \bigr )$ and the $2$-cocycles $\omega' : \mathcal{N} \times \mathcal{N} \longrightarrow \mathcal{U}(\Alg)$ and $\omega'' : \Xgroup / \mathcal{N} \times \Xgroup / \mathcal{N} \longrightarrow \mathcal{U} \bigl ( \Alg \rtimes^{\omega'}_{\theta',\tau} \mathcal{N} \bigr )$ involved in the definition of the right-hand side are known explicitly. 
\end{thm}
Another important theorem concerns the `trivialization' of the $2$-cocycle: Packer and Raeburn show that it is always possible to enlarge the twisted crossed product by tensoring the $C^*$-algebra of compact operators on $L^2(\Xgroup)$ to it and to `straighten out' the twisted crossed product: 
\begin{align*}
	\twistedXprod \otimes \mathcal{K} \bigl ( L^2(\Xgroup) \bigr ) 
	\cong \bigl ( \Alg \otimes \mathcal{K} \bigl ( L^2(\Xgroup) \bigr ) \bigr ) \rtimes^{\omega'}_{\theta',\tau} \Xgroup 
	\cong \bigl ( \Alg \otimes \mathcal{K} \bigl ( L^2(\Xgroup) \bigr ) \bigr ) \rtimes^{\id}_{\theta'',\tau} \Xgroup
\end{align*}
Roughly speaking, this works, because the algebra of compact operators on $L^2(\Xgroup)$ can be thought of as representation of $\Cont_{\infty}(\Xgroup) \rtimes^{\omega}_{\theta,\tau} \Xgroup$ -- which is isomorphic to $\Cont_{\infty} \rtimes^{\id}_{\theta,\tau} \Xgroup$. 

The significance of this result is that it allows to easily extend existing results pertaining to \emph{un}twisted crossed products to \emph{twisted} ones: 
\begin{thm}[Corollary 3.7 in \cite{Packer_Raeburn:twisted_X_products_1:1989}]\label{algebraicPOV:twisted_crossed_products:thm:twisted_equal_untwisted_Packer_Raeburn}
	Let $(\Alg , \Xgroup , \theta , \omega)$ be a separable twisted $C^*$-dynamical system. There there is an action $\theta'$ of $\Xgroup$ on $\Alg \otimes \mathcal{K} \bigl ( L^2(\Xgroup) \bigr )$ such that 
	\begin{align*}
		\bigl ( \twistedXprod \bigr ) \otimes \mathcal{K} \bigl ( L^2(\Xgroup) \bigr ) \cong \bigl ( \Alg \otimes \mathcal{K} \bigl ( L^2(\Xgroup) \bigr ) \bigr ) \rtimes^{\id}_{\theta' , \tau} \Xgroup 
		. 
	\end{align*}
\end{thm}
%



\section{Generalized Weyl calculus} 
\label{algebraicPOV:generalized_weyl_calculus}
%
To motivate twisted crossed products, we have done away with the Fourier transform in the product formula (equation~\eqref{algebraicPOV:twisted_crossed_products:eqn:untwisted_product_Rd}). To connect pseudodifferential theory with twisted crossed products, we need to Fourier transform the twisted crossed product $\twistedXprod$ in the $\Xgroup$-variable. In our setting, \ie $\Xgroup$ being an abelian, locally compact group, we can define the 
%
\begin{defn}[Dual group $\dualX$\index{dual group $\hat{\Xgroup}$}]\label{algebraicPOV:generalized_weyl_calculus:defn:dual_group}
	Let $\Xgroup$ be a locally compact, second countable, abelian group. The set of continuous group morphisms $\xi : \Xgroup \longrightarrow \T$ equipped with pointwise multiplication and uniform convergence on compact subsets forms the dual group $\dualX$. 
\end{defn}
%
Then, for a suitable normalization of the Haar measure, we define the Fourier transform of $L^1(\Xgroup)$ functions,\index{Fourier transform!on groups} 
%
\begin{align}
	\Fourier_{\Xgroup} &: L^1(\Xgroup) \longrightarrow \Cont_{\infty}(\dualX), \; \bigl ( \Fourier_{\Xgroup} f \bigr ) (\xi) := \int_{\Xgroup} \dd x \, {\xi(x)}^* \, f(x) 
	\\
	\bar{\Fourier}_{\Xgroup} &: L^1(\Xgroup) \longrightarrow \Cont_{\infty}(\dualX), \; \bigl ( \bar{\Fourier}_{\Xgroup} f \bigr ) (\xi) := \int_{\Xgroup} \dd x \, \xi(x) \, f(x) \notag 
	. 
\end{align}
It induces unitary maps between $L^2(\Xgroup)$ and $L^2(\dualX)$. The inverses act on $L^1(\dualX) \cap L^2(\dualX)$ as 
\begin{align}
	\bigl ( \bar{\Fourier}_{\dualX} f \bigr ) (\xi) :&= \int_{\dualX} \dd \xi \, \xi(x) \, f(\xi) = \bigl ( \Fourier_{\Xgroup}^{-1} f \bigr )(x) 
	\\
	\bigl ( \Fourier_{\dualX} f \bigr ) (x) :&= \int_{\dualX} \dd \xi \, {\xi(x)}^* \, f(\xi) = \bigl ( \bar{\Fourier}_{\Xgroup}^{-1} f \bigr )(x) 
	. 
\end{align}
We extend this definition to $\Alg$-valued functions via $\Fourier := \id_{\Alg} \otimes \Fourier_{\Xgroup}$ (where we have identified $L^1(\Xgroup ; \Alg)$ with $\Alg \otimes L^1(\Xgroup)$), 
\begin{align*}
	\Fourier : L^1(\Xgroup ; \Alg) \longrightarrow \Cont_{\infty}(\dualX ; \Alg) 
	. 
\end{align*}
Thus, we can transcribe all elements of the Banach-$\ast$ algebra $\bigl ( L^1(\Xgroup ; \Alg) , \repom , {}^{\repom} \bigr )$ onto a subset of $\Cont_{\infty}(\dualX ; \Alg)$ via $\Fourier^{-1}$. For any two $f , g \in \Fourier^{-1} L^1(\Xgroup ; \Alg)$, the induced product $\Weyl^{\omega}_{\tau}$ is related to $\repom$ by intertwining them with $\Fourier$, 
\begin{align}
	\bigl ( f \Weyl^{\omega}_{\tau} g \bigr )(x , \xi) &:= \Fourier^{-1} \bigl ( (\Fourier f) \repom (\Fourier g) \bigr )(x , \xi) 
	\label{algebraicPOV:generalized_weyl_calculus:eqn:repom_product}
	, \\
	f^{\Weyl^{\omega}_{\tau}} &:= \Fourier^{-1} \bigl ( (\Fourier f)^{\repom} \bigr ) 
	, 
	\label{algebraicPOV:generalized_weyl_calculus:eqn:repom_involution}
\end{align}
and the transported norm is likewise defined as 
\begin{align}
	\norm{f} := \norm{\Fourier f}_{L^1} 
	. 
\end{align}
However, there is no direct characterization of $\Fourier^{-1} L^1(\Xgroup ; \Alg)$, so we will usually work on suitable dense subsets. The canonical extension of the inverse Fourier transform $\Fourier^{-1}$ relates the enveloping $C^*$-algebra of $\bigl ( \Fourier^{-1} L^1(\Xgroup ; \Alg) , \Weyl^{\omega}_{\tau} , {}^{\Weyl^{\omega}_{\tau}} \bigr )$, denoted by $\sFXprod$, to the twisted crossed product $\twistedXprod$ and 
\begin{align*}
	\norm{\Fourier f}_{\stwistedXprod} = \norm{f}_{\sFXprod} 
\end{align*}
holds. The space $\Xgroup \times \dualX$ is interpreted as \emph{quantum phase space}: the attribute quantum shall remind the reader that in general $\Xgroup \times \dualX$ does not coincide with $T^* \Xgroup$, the cotangent bundle over $\Xgroup$. If $\Xgroup = \T^d$, then $T^* \T^d \cong \T^d \times \R^d$ whereas $\T^d \times \widehat{\T^d} \cong \T^d \times \Z^d$. In a way, the quantum phase space already knows that the spectrum of the momentum operator is purely discrete. The Fourier transform also induces a covariant representation of $\sFXprod$ on $\Hil = L^2(\Xgroup)$. 
\begin{prop}
	\begin{enumerate}[(i)]
		\item The representation 
		\begin{align*}
			\Op^{\lambda}_{\tau} := \Rep^{\lambda}_{\tau} \circ \Fourier : \sFXprod \longrightarrow \mathcal{B} \bigl ( L^2(\Xgroup) \bigr ) 
		\end{align*}
		is faithful and acts on $f \in \Fourier^{-1} L^1(\Xgroup ; \Alg)$ by 
		\begin{align}
			\bigl ( \Op^{\lambda}_{\tau} (f) u \bigr)(x) = \int_{\Xgroup} \dd y \int_{\dualX} \dd \xi \, &{\xi(y - x)}^* \, \lambda(x ; y - x)  
			\cdot \notag \\
			&\cdot 
			f \bigl ( (\id - \tau)(x) + \tau(y) , \xi \bigr ) \, u(y) 
			, 
			\label{algebraicPOV:generalized_weyl_calculus:eqn:Op_lt}
		\end{align}
		for all $u \in L^2(\Xgroup)$ and $x \in \Xgroup$ where the right-hand side is viewed as an iterated integral. 
		\item If $\lambda' \in C^1 \bigl (\Xgroup ; \Cont(\Xgroup ; \T) \bigr )$ is another $1$-cochain that pseudotrivializes the $2$-cocycle $\omega = \delta^1 \lambda' \in C^2 \bigl ( \Xgroup ; \Cont(\Xgroup ; \T) \bigr )$, then $\lambda' = \delta^0(c) \, \lambda$ for some $c \in \Cont(\Xgroup ; \T)$ and the representations $\Op^{\lambda}_{\tau}$ and $\Op^{\lambda'}_{\tau}$ are unitarily equivalent: 
		\begin{align}
			r(c^{-1}) \, \Op^{\lambda}_{\tau}(f) \, r(c) = \Op^{\lambda'}_{\tau}(f) 
			&& \forall f \in \sFXprod 
			. 
		\end{align}
	\end{enumerate}
\end{prop}
\begin{remark}
	The second property is the principle of \emph{gauge-covariance of generalized Weyl calculus}.\index{gauge-covariance} 
\end{remark}
If we have a look at how we have defined Weyl quantization, we have emphasized the importance of the Weyl system. What is the relation between the Weyl system and the covariant representation $\Op^{\lambda}_{\tau}$ of $\sFXprod$? If $r$ is the usual representation of $\Alg \subseteq \BCont_u(\Xgroup)$ on $\Hil = L^2(\Xgroup)$ mapping $\varphi \in \Alg$ onto multiplication by $\varphi(Q)$, 
\begin{align*}
	\bigl ( r(\varphi) u \bigr )(x) := \varphi(x) \, u(x) 
	, 
\end{align*}
then we can define $V : \dualX \longrightarrow \mathcal{U} \bigl ( L^2(\Xgroup) \bigr )$, 
\begin{align*}
	V(\xi) := r(\xi)^{\ast} 
	. 
\end{align*}
Keep in mind that $\dualX \subseteq \BCont(\Xgroup)$ and hence $r(\xi)$ is a well-defined expression. Now we can define the 
\begin{defn}[Weyl system\index{Weyl system!abstract}]
	The family of operators $\bigl \{ \WeylSyslt(x , \xi) \; \vert \; (x , \xi) \in \Xgroup \times \dualX \bigr \}$, 
	\begin{align}
		\WeylSyslt (x , \xi) := \xi \bigl ( (\id - \tau)(x) \bigr ) \, T^{\lambda}(x) \, V(\xi) = \xi \bigl ( - \tau(x) \bigr ) \, V(\xi) \, T^{\lambda}(x) 
		, 
	\end{align}
	is called the \emph{Weyl system} associated to the pseudotrivialization $\lambda$ and the endomorphism $\tau$. 
\end{defn}
The Weyl system acts on $u \in L^2(\Xgroup)$ as 
\begin{align*}
	\bigl ( \WeylSyslt(x,\xi) u \bigr )(y) = \xi \bigl ( y + \tau(x) \bigr )^* \, \lambda(y ; x) \, u(y + x) , 
	&& (x , \xi) \in \Xgroup \times \dualX 
	. 
\end{align*}
We can now `blow up' the original $C^*$-dynamical system $(\Xgroup , \Alg , \theta , \omega)$ to $(\Pspace , \Alg , \Theta , \Omega_{\tau})$ where $\Pspace := \Xgroup \times \dualX$ is quantum phase space, $\Theta : \Pspace \longrightarrow \Aut(\Alg)$, $\Theta_{(x , \xi)}[\varphi] := \theta_x[\varphi]$ is a trivially extended action and 
\begin{align*}
	\Omega_{\tau} \bigl ( (x , \xi) , (y , \eta) \bigr ) := \xi \bigl ( \tau(y) \bigr ) \, \eta \bigl ( (\tau - \id)(y) \bigr ) \, \omega(x,y) 
\end{align*}
as $2$-cocycle. The Weyl system plays the role of the unitary representation $T^{\lambda} : \Xgroup \longrightarrow \mathcal{U} \bigl ( L^2(\Xgroup) \bigr )$ and it is easy to verify that $\bigl ( L^2(\Xgroup) , r , \WeylSyslt \bigr )$ is a covariant representation of the blown up twisted $C^*$-dyanmical system $(\Pspace , \Alg , \Theta , \Omega_{\tau})$. 
\begin{prop}
	\begin{enumerate}[(i)]
		\item If $\bigl ( L^2(\Xgroup) , r , T^{\lambda} \bigr )$ is a covariant representation of $(\Xgroup , \Alg , \theta , \omega)$, then the triple $\bigl ( L^2(\Xgroup) , r , \WeylSyslt \bigr )$ is a covariant representation of the enlarged twisted $C^*$-dynamical system $(\Pspace , \Alg , \Theta , \Omega_{\tau})$ where 
		\begin{align*}
			\Theta_{(x , \xi)} [\varphi] :=& \; \theta_x [\varphi] 
			, \\
			\Omega_{\tau} \bigl ( (x , \xi) , (y , \eta) \bigr ) :=& \; \xi \bigl ( \tau(y) \bigr ) \, \eta \bigl ( (\tau - \id)(y) \bigr ) \, \omega(x,y) 
			. 
		\end{align*}
		\item If $\lambda'$ is another element of $C^1 \bigl ( \Xgroup ; \Cont(\Xgroup ; \T) \bigr )$ which pseudotrivializes $\delta^1(\lambda') = \omega = \delta^1(\lambda)$, then there exists $c \in \Cont(\Xgroup ; \T)$ such that $W^{\lambda'}_{\tau}(x,\xi) = r(c^{-1}) \, \WeylSyslt(x,\xi) \, r(c)$ for all $(x,\xi) \in \Pspace$. 
		\item For two endomorphisms $\tau , \tau' \in \mathrm{End} (\Xgroup)$, the $2$-cocycles $\Omega_{\tau}$ and $\Omega_{\tau'}$ on $\Pspace$ are cohomologous and the corresponding Weyl systems are connected by $W^{\lambda}_{\tau'}(x,\xi) = \xi \bigl ( (\tau - \tau')(x) \bigr ) \, \WeylSyslt(x,\xi)$ for all $(x,\xi) \in \Pspace$. 
	\end{enumerate}
\end{prop}
\begin{proof}
	\begin{enumerate}[(i)]
		\item One can directly verify that $\bigl ( L^2(\Xgroup) , r , \WeylSyslt \bigr )$ satisfies all the properties in Definition~\ref{algebraicPOV:twisted_crossed_products:defn:covariant_rep_twisted_Cstar_dyn_sys}. 
		\item This follows from the corresponding property of $T^{\lambda}$ (Proposition~\ref{algebraicPOV:twisted_crossed_products:prop:gauge_covariance_rep}). 
		\item The $1$-cochain $\Lambda_{\tau,\tau'}(x,\xi) := \xi \bigl ( (\tau - \tau')(x) \bigr )$ links $\WeylSyslt$ and $W^{\lambda'}_{\tau}$, 
		\begin{align*}
			\Omega_{\tau'} = \delta^1(\Lambda_{\tau,\tau'}) \, \Omega_{\tau} 
			. 
		\end{align*}
	\end{enumerate}
\end{proof}
%
Analogously to the previous section, we may defined extended twisted crossed products of the form $\Alg \rtimes^{\Omega_{\tau}}_{\Theta,\tau} \Pspace$. As a matter of fact, this point of view is used in \cite{Kaschek_Neumaier_Waldmann:complete_positivity:2008}. In that case, the \emph{symplectic} Fourier transform $\Fs$\index{Fourier transform!symplectic} takes the place of ordinary Fourier transform, 
\begin{align}
	(\Fs f)(x , \xi) := \int_{\Xgroup} \dd y \int_{\dualX} \dd \eta \, \xi(y) \, {\eta(x)}^* \, f(y,\eta) , 
	&& f \in L^1(\Pspace) 
	, 
	\label{algebraicPOV:generalized_weyl_calculus:eqn:symplectic_Fourier}
\end{align}
%
and can be written as $\Fs = \mathfrak{S} \circ (\bar{\Fourier}_{\Xgroup} \otimes \Fourier_{\dualX})$ where $(\mathfrak{S} f)(\xi,x) := f(x,\xi)$ swaps arguments. As in the case $\Xgroup = \R^d$, $\Fs$ is its own inverse, $\Fs^{-1} = \Fs$. Armed with these definitions, we can define in analogy with equation~\eqref{magWQ:magnetic_weyl_calculus:eqn:magnetic_Weyl_quantization} 
\begin{align}
	\widetilde{\Op}^{\lambda}_{\tau}(f) := \int_{\Pspace} \dd X \, \bigl ( \Fs f \bigr )(X) \, \WeylSys^{\lambda}(X) 
	\label{algebraicPOV:generalized_weyl_calculus:eqn:Op_tilde}
\end{align}
for $f \in \Fs L^1(\Pspace)$. If we restrict ourselves to the dense subspace $\Fourier_{\dualX} L^1(\dualX) \odot \bar{\Fourier}_{\Xgroup} L^1(\Xgroup)$ interpreted as a subspace of $\Cont_{\infty}(\Xgroup) \odot \Cont_{\infty}(\dualX)$, we see that $\widetilde{\Op}^{\lambda}_{\tau}$ agrees with $\Op^{\lambda}_{\tau}$ as given by equation~\eqref{algebraicPOV:generalized_weyl_calculus:eqn:Op_lt} for all functions $f \in \Fourier_{\dualX} L^1(\dualX) \odot \bar{\Fourier}_{\Xgroup} L^1(\Xgroup) \subseteq \Alg \odot \bar{\Fourier}_{\Xgroup} L^1(\Xgroup)$ as long as $\Cont_{\infty}(\Xgroup) \subseteq \Alg$. 
\begin{prop}
	Assume $\Cont_{\infty}(\Xgroup) \subseteq \Alg$. Then both, $\Op^{\lambda}_{\tau}$ and $\widetilde{\Op}^{\lambda}_{\tau}$ as defined by equations~\eqref{algebraicPOV:generalized_weyl_calculus:eqn:Op_lt} and \eqref{algebraicPOV:generalized_weyl_calculus:eqn:Op_tilde}, respectively, are well-defined on $\Fourier_{\dualX} L^1(\dualX) \odot \bar{\Fourier}_{\Xgroup} L^1(\Xgroup)$ and they coincide on this set. 
\end{prop}
One last thing concerning extensions: we can extend all objects such as $\repom$, $\Rep^{\lambda}_{\tau}$, $\Weyl^{\omega}_{\tau}$ and $\Op^{\lambda}_{\tau}$ to the multiplier algebras of $\twistedXprod$ and $\sFXprod$, respectively. This construction is particularly illuminating for $\Alg = \Cont_{\infty}(\Xgroup)$: in that case 
\begin{align*}
	\Rep^{\lambda}_{\tau} \bigl ( \Cont_{\infty}(\Xgroup) \rtimes^{\omega}_{\theta,\tau} \Xgroup \bigr ) = \mathcal{K} \bigl ( L^2(\Xgroup) \bigr ) = \Op^{\lambda}_{\tau} \bigl ( \Fourier \Calg^{\omega}_{\Cont_{\infty}(\Xgroup)} \bigr )
\end{align*}
and it is immediately clear that $\Rep^{\lambda}_{\tau}$ maps $\mathcal{M} \bigl ( \Cont_{\infty}(\Xgroup) \rtimes^{\omega}_{\theta,\tau} \Xgroup \bigr )$ onto all of $\mathcal{B} \bigl ( L^2(\Xgroup) \bigr )$ -- which is the multiplier algebra of $\mathcal{K} \bigl ( L^2(\Xgroup) \bigr )$! For $\Xgroup = \R^d$, we may use a different route, duality techniques, to achieve the same thing (Corollary~\ref{psiDO_reloaded:relevant_cStar_algebras:cor:MBA_contained_mult_fxb}). 
\begin{prop}
	The space $\Fs \mathbb{M}(\Pspace)$ of all symplectic Fourier transforms of bounded, com\-plex-valued measures is contained in $\mathcal{M} \bigl ( \Cont_{\infty}(\Xgroup) \rtimes^{\omega}_{\theta,\tau} \Xgroup \bigr )$. 
\end{prop}
%


\section{The concept of affiliation} 
\label{algebraicPOV:affiliation}
Affiliation is a concept that allows us to treat possibly unbounded, possibly not densely defined operators in an abstract $C^*$-algebraic setting. The $C^*$-algebra $\Calg$ does not have to be $\twistedXprod$, $\sFXprod$ or some $C^*$-subalgebra of $\mathcal{B}(\Hil)$ and it is easier to make statements in rather broad generality. 

The material contained in this section is taken from \cite[Chapter~8]{Amrein_Boutet_Georgescu:C0_groups_commutator_methods:1996}.

\subsection{Observables affiliated to $C^*$-algebras} 
\label{algebraicPOV:affiliation:affiliated_observables}
The inspiration to define affiliated observables comes from considering resolvent families of selfadjoint or normal operators and functional calculus. We will explain this in the first example right after the definition: 
\begin{defn}[Observable affiliated to $\Calg$\index{affiliation}]
	A selfadjoint observable $H$ (also denoted as $\Phi_H$) affiliated to a $C^*$-algebra $\Calg$ is a morphism 
	\begin{align*}
		\Phi_H : \Cont_{\infty}(\R) \longrightarrow \Calg 
		. 
	\end{align*}
\end{defn}
%
If we choose functions $\varphi \in \Cont_{\infty}(\R)$ whose domain is $\R$, we implicitly restrict ourselves to selfadjoint observables. Similarly, a morphism $\Phi_H : \Cont_{\infty}(\C) \longrightarrow \Calg$ defines a \emph{normal} observable affiliated to $\Calg$. Quite generally, observables may be defined via morphisms $\Phi : \Cont_{\infty}(\mathcal{Y}) \longrightarrow \Calg$ where $\mathcal{Y}$ is a topological, locally compact, second countable space. We shall only be concerned with selfadjoint observables, \ie $\mathcal{Y} = \R$, for a more general setup, we refer to \cite{Amrein_Boutet_Georgescu:C0_groups_commutator_methods:1996}. For this reason, we will frequently omit the word `selfadjoint.' 
\begin{example}
	Let $\Hil$ be a Hilbert space, $H$ a densely defined, selfadjoint operator and $\Calg = \mathcal{B}(\Hil)$ the $C^*$-algebra of bounded operators on $\Hil$. Then $H$ defines an observable $\Phi_H$ via usual functional calculus, 
	\begin{align*}
		\Phi_H : \Cont_{\infty}(\R) \longrightarrow \mathcal{B}(\Hil) , \; \varphi \mapsto \Phi_H(\varphi) := \varphi(H) 
		. 
	\end{align*}
	Clearly, this gives the recipe how to define observables affiliated to twisted crossed products and their Fourier transforms. Also, with a little abuse of notation, we will often use $H$ to denote the morphism $\Phi_H$. 
\end{example}
%
%
There is a one-to-one correspondence between morphisms $\Phi : \Cont_{\infty}(\R) \longrightarrow \Calg$ and observables affiliated to the algebra of bounded operators on a Hilbert space $\Hil$ which are defined densely on some closed subspace $\overline{\mathcal{D}(H)} := \Hil' \subseteq \Hil$. 

One direction is clear: if $(H,\Hil')$ is an operator which is densely defined on a closed subspace of $\Hil$, then we can use functional calculus to define a morphism on $\Phi_H : \Cont_{\infty}(\R) \longrightarrow \mathcal{B}(\Hil')$. If we extend this morphism trivially on ${\Hil'}^{\perp}$, \ie $\Phi_H(\varphi) \vert_{{\Hil'}^{\perp}} := 0$, we get a morphism $\Phi_H : \Cont_{\infty}(\R) \longrightarrow \mathcal{B}(\Hil)$. 

Conversely, any morphism $\Phi : \Cont_{\infty}(\R) \longrightarrow \mathcal{B}(\Hil)$ has an extension to the bounded Borel functions, $\widetilde{\Phi} : \mathcal{BO}(\R) \longrightarrow \mathcal{B}(\Hil)$.\footnote{This extension implies adding $W^*$-algebra structure to the right-hand side. In this case $\mathcal{B}(\Hil)$ already \emph{is} a $W^*$-algebra by definition, but if the morphism maps onto some abstract $C^*$-algebra, this procedure is equivalent to choosing a particular representation \cite[Chapter~7.4]{Pederson:Cstar_algebras:1979}. } With this morphism, we can extract the spectral measure~-- which then uniquely determines the operator. In fact, other well-known facts of standard operator theory \cite{Reed_Simon:M_cap_Phi_1:1972,Reed_Simon:M_cap_Phi_2:1975} combine nicely with this $C^*$-algebraic point of view developed here: 
\begin{prop}
	Let $H_0$ be a densely defined selfadjoint operator in a Hilbert space $\Hil$ affiliated to a $C^*$-subalgebra $\Calg$ of $\mathcal{B}(\Hil)$. 
	\begin{enumerate}[(i)]
		\item If $V$ is a $H_0$-bounded symmetric operator in $\Hil$ with $H_0$-bound strictly less than $1$ and if $\, V \, (H_0 - z_0)^{-1} \in \Calg$ for some $z_0 \in \C \setminus \R$, then $H := H_0 + V$ is a densely defined selfadjoint operator in $\Hil$ affiliated to $\Calg$. 
		\item Assume that $H_0$ is bounded form below and let $V$ be a symmetric sesquilinear form on $\Hil$ which is relatively form-bounded with respect to $H_0$ with relative bound stricly less than $1$. If $(\lambda_0 + H_0)^{-\nicefrac{1}{2}} \, V \, (\lambda_0 + H_0)^{-\nicefrac{1}{2}} \in \Calg$ for some $\lambda_0 > - \inf H_0$, then the operator associated to the form sum $H := H_0 \overset{.}{+} V$ is a densely defined selfadjoint operator in $\Hil$ affiliated to $\Calg$. 
	\end{enumerate}
\end{prop}
\begin{proof}
	\begin{enumerate}[(i)]
		\item With the help of the resolvent identity, we rewrite 
		\begin{align*}
			V \, (H_0 - z)^{-1} &= V \, (H_0 - z_0)^{-1} + (z - z_0) \, V \, (H_0 - z_0)^{-1} \, (H_0 - z)^{-1}
		\end{align*}
		%
		to confirm that $V \, (H_0 - z)^{-1} \in \Calg$ for any $z \in \C \setminus \sigma(H_0)$. Since $V$ is $H_0$-bounded with bound less than $1$, we can expand $(H - z)^{-1}$ in terms of $(H_0 - z)^{-1}$ and $V \, (H_0 - z)^{-1}$ for suitable $z \in \C \setminus \sigma(H_0)$ \cite[Theorem~X.12]{Reed_Simon:M_cap_Phi_2:1975}. 
		\item Define $m := \inf H_0$ as the infimum of the spectrum and pick $\lambda , \lambda_0 > -m$. If $S := (\lambda_0 - \lambda) \, (H_0 + \lambda)^{-1}$, then $S \in \Calg$ and clearly $\id_{\Hil} + S \geq \min \bigl \{ 1 , (\lambda_0 + m) (\lambda + m)^{-1} \bigr \}$. This means we can take the square root of $\id_{\Hil} + S$ in the operator sense which can be written as $\id_{\Hil} + T$ for some suitable bounded operator $T \in \Calg$. Thus, we can rewrite $(H_0 + \lambda)^{- \nicefrac{1}{2}} \, V \, (H_0 + \lambda)^{- \nicefrac{1}{2}}$ in terms of $(H_0 + \lambda_0)^{-1}$: 
		\begin{align*}
			(H_0 + \lambda)^{- \nicefrac{1}{2}} \, &V \, (H_0 + \lambda)^{- \nicefrac{1}{2}} = 
			\\
			&= (\id_{\Hil} + T) \, (H_0 + \lambda_0)^{- \nicefrac{1}{2}} \, V \, (H_0 + \lambda_0)^{- \nicefrac{1}{2}} \, (\id_{\Hil} + T) \in \Calg 
		\end{align*}
		This in turn implies we can express $(H + \lambda)^{-1}$ as 
		\begin{align*}
			(H + \lambda)^{-1} &= (H_0 + \lambda)^{- \nicefrac{1}{2}} \, \bigl ( \id_{\Hil} + (H_0 + \lambda)^{- \nicefrac{1}{2}} \, V \, (H_0 + \lambda)^{- \nicefrac{1}{2}} \bigr ) \, (H_0 + \lambda)^{- \nicefrac{1}{2}} 
			\\
			&\in \Calg 
			. 
		\end{align*}
		This concludes the proof. 
	\end{enumerate}
\end{proof}
\begin{defn}[Resolvent family\index{resolvent family}]\label{algebraicPOV:affiliation:defn:resolvent_family}
	A resolvent family is a family of functions indexed by $z \in \C \setminus \R$, $\{ R(z) \}_{z \in \C \setminus \R} \subset \Cont_{\infty}(\R)$, such that 
	\begin{enumerate}[(i)]
		\item $R(z)^* = R(z^*)$ and 
		\item $R(z) - R(z') = (z - z') \, R(z) \, R(z')$. 
	\end{enumerate}
\end{defn}
\begin{prop}
	A selfadjoint observable $\Phi_H : \Cont_{\infty}(\R) \longrightarrow \Calg$ affiliated to a $C^*$-algebra $\Calg$ is uniquely determined by its resolvent family and $R(z) := \Phi_H(r_z)$, $r_z := (\cdot - z)^{-1}$, holds for $z \in \C \setminus \R$. In fact, $\Phi_H (r_{z_0}) \in \Calg$ for some $z_0 \in \C \setminus \R$ suffices. 
\end{prop}
\begin{proof}
	The translates of $\{ r_z \}_{z \in \C \setminus \R}$ is dense in $\Cont_{\infty}(\R)$ and a simple approximation argument suffices to extend $\Phi_H$ to all functions on $\R$ vanishing at infinity (Theorem of Stone-Weierstrass). 
	
	Assume $R(z_0) = \Phi_H(r_{z_0}) \in \Calg$. Then for all $z$ in some neighborhood $U(z_0)$ of $z_0$, we can write $R(z)$ as 
	\begin{align*}
		R(z) = \sum_{n = 0}^{\infty} (z - z_0)^n \, R(z_0)^{n + 1} \in \Calg
		. 
	\end{align*}
	We can then repeat this procedure as often as needed to write the resolvent for an arbitrary $z \in \C \setminus \R$ which lies in the same half plane as $z_0$. The resolvents for all other $z$ can be obtained by complex conjugation (property (i) of Definition~\ref{algebraicPOV:affiliation:defn:resolvent_family}). 
\end{proof}
%


\subsection{Spectrum of affiliated observables} 
\label{algebraicPOV:affiliation:spectrum_of_affiliated_observables}

The connection of selfadjoint observables affiliated to $C^*$-algebras and resolvents indicates that there is a way to recover the spectrum of the original observable. 
\begin{defn}[Spectrum $\sigma(H)$/$\sigma(\Phi_H)$\index{spectrum!of an affiliated observable}]
	Let $H$ be a selfadjoint observable affiliated to a $C^*$-algebra $\Calg$, $\Phi_H : \Cont_{\infty}(\R) \longrightarrow \Calg$. Then the spectrum of $H$ is defined as 
	\begin{align}
		\sigma(H) \equiv \sigma(\Phi_H) := \Bigl \{ \lambda \in \R \; \big \vert \; \forall \varphi \in \Cont_{\infty}(\R) : \varphi(\lambda) \neq 0 \Rightarrow \Phi_H(\varphi) \equiv \varphi(H) \neq 0 \Bigr \} 
		. 
	\end{align}
\end{defn}
\begin{example}
	To make sense of this definition, we remember that if $\Calg = \mathcal{B}(\Hil)$, we indeed recover the usual spectrum: keeping in mind that $\mathcal{B}(\Hil)$ is a $W^*$-algebra, we can extend $\Phi_H$ to 
	bounded Borel functions on $\R$ by a suitable approximation argument. In particular, we can plug in the characteristic function $\chi_A$ of any Borel subset $A \subseteq \R$. If $A \cap \mathrm{spec} \, (H) = \emptyset$ (where $\mathrm{spec} \, (H)$ is the spectrum in the the usual functional analytic sense), then by the spectral theorem $\Phi_H(\chi_A) \equiv \chi_A(H) = 0$. Hence, $\sigma(H)$ coincides with $\mathrm{spec} \, (H)$. 
\end{example}
\begin{remark}
	The algebraically defined spectrum $\sigma(H)$ is a closed subset of $\R$. For consistency with regular spectral theory, we always view $\R$ as a subset of $\C$. 
\end{remark}
%
%
\begin{remark}
	The $C^*$-algebraic approach has one caveat: we are not able to distinguish continuous spectrum from dense point spectrum, because we cannot separate two infinitesimally close points by a continuous function. Also, spectral information on the type of spectrum are lost (e. g. absolutely continuous and singularly continuous spectrum). However, it \emph{is} possible to distinguish essential from discrete spectrum. 
\end{remark}
In this framework, we can easily transfer some important notions from operator theory. For instance, we say two observables $H_1$ and $H_2$ commute if and only if 
\begin{align*}
	\varphi(H_1) \, \varphi(H_2) = \varphi(H_2) \, \varphi(H_1) 
\end{align*}
for all $\varphi \in \Cont_{\infty}(\R)$ -- which coincides with the standard definition if $\Calg = \mathcal{B}(\Hil)$ \cite[Theorem~VIII.13]{Reed_Simon:M_cap_Phi_1:1972}. The spectral theorem also has a rather natural translation. Let 
\begin{align*}
	f : \sigma(H) \longrightarrow \R 
\end{align*}
%
be a proper\footnote{A map $f : X \longrightarrow Y$ between locally compact spaces is called proper if $f^{-1}(K)$ is compact for any compact subset $K \subseteq Y$. Put another way, $f(x)$ must tend to infinity in $Y$ if $x$ tends to infinity in $X$. } continuous function. Then, canonically, we define 
\begin{align*}
	f^* : \Cont_{\infty}(\R) \longrightarrow \Cont_{\infty}(\R) , \varphi \mapsto f^*(\varphi) := \varphi \circ f 
	. 
\end{align*}
By definition, an analog of the spectral theorem holds, \ie for proper continuous functions, $f(H)$ is an observable affiliated to $\Calg$ and we have 
\begin{align*}
	\sigma \bigl ( f(H) \bigr ) = f \bigl ( \sigma(H) \bigr ) 
	. 
\end{align*}
Again, this result is well-known if $H$ is a selfadjoint operator on a Hilbert space. The next operation has \emph{no Hilbert space analog} and is \emph{the key advantage} of the $C^*$-algebraic framework. 
\begin{defn}[Image of $H$ through $\pi$\index{image of $H$ through $\pi$}]
	Let $\pi : \Calg \longrightarrow \Calg'$ be a morphism between two $C^*$-algebras $\Calg$ and $\Calg'$. If $H$ is an observable affiliated to $\Calg$ via the morphism $\Phi_H : \Cont_{\infty}(\R) \longrightarrow \Calg$, then 
	\begin{align*}
		\pi(H) \equiv \pi(\Phi_H) : \Cont_{\infty}(\R) \longrightarrow \Calg' 
	\end{align*}
	defines an observable affiliated to $\Calg'$ called the image of $H$ through $\pi$, 
	\begin{align}
		\bigl ( \pi(H) \bigr ) (\varphi) := \pi \bigl ( \varphi(H) \bigr ) 
		. 
	\end{align}
\end{defn}
Since we want to stress the importance of the next fact, which is just a direct consequence of the definition, we have made it into a theorem: 
\begin{thm}\label{algebraicPOV:affiliation:thm:spectrum_of_image_of_H_through_pi}
	Let $H$ be a selfadjoint observable affiliated to the $C^*$-algebra $\Calg$ and $\pi : \Calg \longrightarrow \Calg'$ a morphism between $C^*$-algebras. Then the spectrum of the image of $H$ through $\pi$ is contained in the spectrum of $H$, 
	\begin{align}
		\sigma \bigl ( \pi(H) \bigr ) \subseteq \sigma(H) 
		. 
	\end{align}
\end{thm}
\begin{proof}
	Any morphism is norm-decreasing, \ie 
	\begin{align*}
		\bnorm{\bigl ( \pi(H) \bigr ) (\varphi)}_{\Calg'} \equiv \bnorm{\pi \bigl ( \varphi(H) \bigr )}_{\Calg'} \leq \bnorm{\varphi(H)}_{\Calg} \equiv \bnorm{\Phi_H(\varphi)}_{\Calg} 
		, 
	\end{align*}
	and thus $\Phi_H(\varphi) \equiv \varphi(H) = 0$ implies $\bigl ( \pi(H) \bigr )(\varphi) = 0$, 
	\begin{align*}
		0 \leq \bnorm{\bigl ( \pi(H) \bigr ) (\varphi)}_{\Calg'} \leq \bnorm{\varphi(H)}_{\Calg} = 0 
		. 
	\end{align*}
	Hence, we have shown $\sigma \bigl ( \pi(H) \bigr ) \subseteq  \sigma(H)$. 
\end{proof}
Furthermore, if $f : \sigma(H) \longrightarrow \R$ is a proper, continuous function, then 
\begin{align*}
	f \bigl ( \pi(H) \bigr ) = \pi \bigl ( f(H) \bigr ) 
	. 
\end{align*}
The most fruitful class of examples comes from two basic ideas: 
\begin{enumerate}[(i)]
	\item Faithful representations $\pi : \Calg \longrightarrow \mathcal{B}(\Hil)$ of $\Calg$ on some Hilbert space $\Hil$ are of this type. If $\pi$ is faithful, \ie injective, then $\sigma(H) = \sigma \bigl ( \pi(H) \bigr )$ coincides with the usual spectrum of $\pi(H)$ on $\Hil$ in the functional analytic sense. Furthermore, the norm of $f \in \Calg$ can be \emph{calculated} via $\pi$ and $\norm{f}_{\Calg} = \snorm{\pi(f)}_{\mathcal{B}(\Hil)}$. 
	\item Taking quotients with respect to some two-sided, closed, proper ideal $\Ideal \subset \Calg$: ${\Calg} / {\Ideal}$ is another $C^*$-algebra with induced $C^*$-norm, composition law and involution. In this case 
	\begin{align*}
		\pi_{\Ideal} : \Calg \longrightarrow {\Calg} / {\Ideal} 
	\end{align*}
	is the projection onto equivalence classes. We then define 
\end{enumerate}
\begin{defn}[$\Ideal$-essential spectrum $\sigma_{\Ideal}(H)$\index{spectrum!$\Ideal$-essential}]
	For a selfadjoint observable $H$ affiliated to a $C^*$-algebra $\Calg$, we define the $\Ideal$-essential spectrum as 
	\begin{align}
		\sigma_{\Ideal}(H) := \sigma \bigl ( \pi_{\Ideal}(H) \bigr ) = \Bigl \{ \lambda \in \R \; \big \vert \; \forall \varphi \in \Cont_{\infty}(\R) : \varphi(\lambda) \neq 0 \Rightarrow \varphi(H) \not\in \Ideal \Bigr \} 
		. 
	\end{align}
\end{defn}
\begin{example}
	If $\Calg = \mathcal{B}(\Hil)$ and $\Ideal = \mathcal{K}(\Hil)$ is the ideal of compact operators, then $\Calg(\Hil) := {\mathcal{B}(\Hil)} / {\mathcal{K}(\Hil)}$ is the so-called Calkin algebra. In this case, the $\mathcal{K}(\Hil)$-essential spectrum coincides with the usual essential spectrum. 
\end{example}
%


\subsection{Tensor products and families of observables} 
\label{algebraicPOV:affiliation:tensor_product}

If $\Calg$ is `the' tensor product of two $C^*$-algebras $\Alg$ and $\Calg'$, $\Calg \cong \Alg \otimes \Calg'$, then this extra structure can be exploited in the analysis. We have used quotation marks on purpose as there usually is no single $C^*$-norm with respect to which the algebraic tensor product $\Alg \odot \Calg'$ is to be completed, but rather a family of $C^*$-norms with a minimal and a maximal $C^*$-norm, 
\begin{align*}
	\norm{\varphi \otimes \psi}_{\min} \leq \norm{\varphi \otimes \psi} \leq \norm{\varphi \otimes \psi}_{\max} 
	.  
\end{align*}
However, if one of the algebras, say $\Alg$, is abelian, then there is \emph{only one} tensor product ($\norm{\varphi \otimes \psi}_{\min} = \norm{\varphi \otimes \psi}_{\max}$). Furthermore, we can use Gelfand theory to characterize $\Alg$, \ie 
\begin{align*}
	\Alg \cong \Cont_{\infty}(\Salg) 
\end{align*}
where $\Salg$ is the Gelfand spectrum (see~Chapter~\ref{algebraicPOV:twisted_crossed_products:gelfand_theory}). With this, we can identify $\Calg$ with $\Cont_{\infty}(\Salg ; \Calg')$, the functions $\Psi : \Salg \longrightarrow \Calg'$ which take values in $\Calg'$, are continuous and vanish at infinity (if $\Salg$ is not compact). Product, involution and norm are also defined in the natural manner, \eg 
\begin{align}
	\bnorm{\Psi}_{\Calg} &:= \sup_{x \in \Salg} \bnorm{\Psi(x)}_{\Calg'} 
	\label{algebraicPOV:affiliation:eqn:norm_tensor_product}
	, \\
	(\Psi \cdot \Psi')(x) &:= \Psi(x) \cdot \Psi'(x) 
	. 
	\notag 
\end{align}
%
%
\begin{prop}
	Let $\Alg$, $\Calg'$ be $C^*$-algebras and $\Alg$ abelian. Then there exists only one tensor product 
	\begin{align*}
		\Calg := \Alg \otimes \Calg' \cong \Cont_{\infty}(\Salg) \otimes \Calg' \cong \Cont_{\infty}(\Salg ; \Calg') 
		. 
	\end{align*}
\end{prop}
\begin{proof}
	Abelian $C^*$-algebras are of type I and hence the smallest and largest $C^*$-norms coincide \cite{Dixmier:C_star_algebras:1977}. Thus, there is only one way to complete $\Alg \odot \Calg'$ to $\Alg \otimes \Calg'$. 
	
	For the second part of the proof, by density, it suffices to consider elements in $\Alg \odot \Calg' \cong \Cont_{\infty}(\Salg) \odot \Calg'$. Any such element $f \in \Cont_{\infty}(\Salg) \odot \Calg'$ is of the form 
	\begin{align*}
		x \mapsto f(x) = \sum_{j = 1}^n \varphi_j(x) \otimes \psi_j 
		, && 
		\varphi_j \in \Cont_{\infty}(\Salg) , \; \psi_j \in \Calg' , \; j \in \{ 1 , \ldots , n \} 
		, 
	\end{align*}
	which corresponds to 
	\begin{align*}
		x \mapsto \sum_{j = 1}^n \varphi_j(x) \, \psi_j \in \Cont_{\infty}(\Salg ; \Calg') 
		. 
	\end{align*}
	It remains to show that the norms of $\Cont_{\infty}(\Salg) \otimes \Calg'$ and $\Cont_{\infty}(\Salg ; \Calg')$ as given by equation~\eqref{algebraicPOV:affiliation:eqn:norm_tensor_product} coincide. In general, morphisms are norm-decreasing, but faithful representations are even norm-\emph{preserving}. Thus we can \emph{calculate} the norm of $f \in \Cont_{\infty}(\Salg) \odot \Calg' \subset \Cont_{\infty}(\Salg) \otimes \Calg'$ via a faithful representation. Consider 
	\begin{align*}
		\ell^2(\Salg) := \Bigl \{ u : \Salg \longrightarrow \C \; \big \vert \; \mbox{$\sum_{x \in \Salg}$} \abs{u(x)}^2 < \infty \Bigr \} 
		. 
	\end{align*}
	Any $\varphi \in \Cont_{\infty}(\Salg)$ acts on $u \in \ell^2(\Salg)$ by multiplication, 
	\begin{align*}
		\bigl ( r(\varphi) u \bigr )(x) := \varphi(x) \, u(x) 
		, && 
		u \in \ell^2(\Salg) , \; x \in \Salg 
		. 
	\end{align*}
	Clearly, $r : \Cont_{\infty}(\Salg) \longrightarrow \mathcal{B} \bigl ( \ell^2(\Salg) \bigr )$ maps any $\varphi \in \Cont_{\infty}(\Salg)$ onto a bounded operator on $\ell^2(\Salg)$: 
	\begin{align*}
		\bnorm{r(\varphi)}_{\mathcal{B}(\ell^2(\Salg))} 
		= \sup_{\norm{u}_{\ell^2} = 1} \babs{\bscpro{u}{r(\varphi) u}} 
		\leq \sup_{\norm{u}_{\ell^2} = 1} \babs{\bscpro{u}{\snorm{\varphi}_{\infty} u}} 
		= \snorm{\varphi}_{\infty}
	\end{align*}
	In fact, we even have equality, $\snorm{r(\varphi)}_{\mathcal{B}(\ell^2(\Salg))} = \snorm{\varphi}_{\infty}$: as $\varphi$ vanishes at infinity and is continuous, $\abs{\varphi}$ attains its maximum at a point $x_0 \in \Salg$ (which need not be unique). Choosing $u_0 = \delta_{x_0}$ gives us 
	\begin{align*}
		\sabs{\sscpro{u_0}{r(\varphi) u}} &= \sabs{\varphi(x_0)} = \snorm{\varphi}_{\infty} 
		\\
		&\leq \snorm{r(\varphi)}_{\mathcal{B}(\ell^2(\Salg))} \leq \snorm{\varphi}_{\infty} 
		. 
	\end{align*}
	This implies $r \bigl ( \Cont_{\infty}(\Salg) \bigr ) \subseteq \mathcal{B} \bigl ( \ell^2(\Salg) \bigr )$. With trivial modifications, we can can adapt the above argument to treat the case  $\Cont_{\infty}(\Salg) \otimes \Calg'$: we faithfully realize it on $\ell^2(\Salg) \otimes \Hil' \cong \ell^2(\Salg ; \Hil')$; $\Hil'$ is chosen such that $\Calg'$ can be represented faithfully on it. By direct calculation, we confirm 
	\begin{align*}
		\bnorm{f}_{\Cont_{\infty}(\Salg) \otimes \Calg'} &= \bnorm{r(\varphi)}_{\mathcal{B}(\ell^2(\Salg) \otimes \Hil')} 
		= \sup_{x \in \Salg} \bnorm{\bigl ( r(f) \bigr )(x)}_{\mathcal{B}(\Hil')} 
		\\
		&= \sup_{x \in \Salg} \bnorm{f(x)}_{\Calg'}
	\end{align*}
	%
	for any $f \in \Cont_{\infty}(\Salg) \odot \Calg'$. This means the norms coincide on a dense subspace and thus the algebras $\Cont_{\infty}(\Salg) \otimes \Calg'$ and $\Cont_{\infty}(\Salg ; \Calg')$ are isomorphic. 
\end{proof}
In this setting, it is useful to view an observable $H$ affiliated to $\Cont_{\infty}(\Salg ; \Calg')$ as a family of observables $H(x)$ affiliated to $\Calg'$ which is \emph{continuous} in the following sense: for all $\varphi \in \Cont_{\infty}(\R)$, the map 
\begin{align}
	x \mapsto \varphi \bigl ( H(x) \bigr ) 
	\label{algebraicPOV:affiliation:eqn:continuity_of_observable_valued_function}
\end{align}
is continuous in the norm. Continuity of $x \mapsto \bigl ( H(x) - z_0 \bigr )^{-1}$ for some $z_0 \in \C \setminus \R$ suffices. A continuous family of observables is called \emph{proper} if and only if 
\begin{align*}
	\lim_{x \rightarrow \infty} H(x) = \infty \Longleftrightarrow \lim_{x \rightarrow \infty} \bnorm{\varphi \bigl ( H(x) \bigr )}_{\Calg'} = 0 
	&& \forall \varphi \in \Cont_{\infty}(\R) 
	. 
\end{align*}
This is only another way of saying that $H$ is affiliated to the $C^*$-algebra composed of continuous $\Calg'$-valued functions which vanish at $\infty$, $\Cont_{\infty}(\Salg ; \Calg')$. If $\Salg$ is compact, \ie $\Alg \cong \Cont_{\infty}(\Salg)$ is unital, then the latter condition is empty. In general, $H$ may be $\infty$ outside of some open set 
\begin{align*}
	\supp H := \bigl \{ x \in \Salg \; \vert \; H(x) \neq \infty \bigr \} 
	. 
\end{align*}
This implies $\varphi \bigl ( H(x) \bigr ) = 0$ for all $\varphi \in \Cont_{\infty}(\R)$ and $x \not\in \supp H$. 

Sequences of observables affiliated to $\Calg'$ may be defined by considering $\Salg = \N \cup \{ \infty \}$ and $\Cont_{\infty}(\N \cup \{ \infty \} ; \Calg') \equiv \BCont(\N ; \Calg')$. This sense of convergence coincides with convergence in the norm-resolvent sense if $\Calg'$ is a subalgebra of $\mathcal{B}(\Hil)$ for some Hilbert space $\Hil$. 
\begin{prop}\label{algebraicPOV:affiliation:prop:fiber_decomposition_of_spectrum}
	Let $H$ be a selfadjoint observable affiliated to $\Cont_{\infty}(\Salg ; \Calg)$ such that $x \mapsto H(x)$ is continuous (see equation~\eqref{algebraicPOV:affiliation:eqn:continuity_of_observable_valued_function}) and proper. Then 
	\begin{align*}
		\sigma(H) = \bigcup_{x \in \Salg} \sigma \bigl ( H(x) \bigr ) 
	\end{align*}
	and $\sigma(H)$ is a closed subset of $\R \subseteq \C$. 
\end{prop}
We will need an auxiliary Lemma to prove the proposition: 
\begin{lem}
	Let $\{ H(x) \}_{x \in \Salg}$ be a proper family of selfadjoint observables affiliated to $\Calg'$. Then the spectrum is localized near infinity for large $x$ in the following sense: for any compact $K \subset \R$, there is a compact $L \subseteq \Salg$ such that $\sigma \bigl ( H(x) \bigr ) \cap K = \emptyset$ for all $x \not\in L$. 
\end{lem}
\begin{proof}
	Pick $\varphi \in \Cont_{\infty}(\R)$ with $\varphi(\lambda) = 1$ for all $\lambda \in K$. As $x \mapsto H(x)$ is proper, we know $\bnorm{\varphi \bigl ( H(x) \bigr )}_{\Calg'} \rightarrow 0$ as $x \rightarrow \infty$. Thus, we can find a compact set $L \subseteq \Salg$ such that $\bnorm{\varphi \bigl ( H(x) \bigr )}_{\Calg'} < 1$ for $x \in L$. On the other hand, the norm can be expressed in terms of the spectrum, namely 
	\begin{align*}
		\bnorm{\varphi \bigl ( H(x) \bigr )}_{\Calg'} = \sup \bigl \{ \abs{\varphi(\lambda)} \; \vert \; \lambda \in \sigma \bigl ( H(x) \bigr ) \bigr \} < 1 
		. 
	\end{align*}
	Thus $\abs{\varphi(\lambda)} < 1$ for all $\lambda \in \sigma \bigl ( H(x) \bigr )$ as long as $x \not\in L$. 
\end{proof}
\begin{proof}[Proposition~\ref{algebraicPOV:affiliation:prop:fiber_decomposition_of_spectrum}]
	Let us start by showing that $\bigcup_{x \in \Salg} \sigma \bigl ( H(x) \bigr )$ is a closed set. Pick $\lambda \in \overline{\bigcup_{x \in \Salg} \sigma \bigl ( H(x) \bigr )}$; then there are sequences $\lambda_n \rightarrow \lambda$ in $\R$ and $\{ x_n \}_{n \in \N}$ in $\Salg$ such that $\lambda_n \in \sigma \bigl ( H(x_n) \bigr )$ for all $n \in \N$. By the previous Lemma, for a compact subset $K \supset \{ \lambda_n \}_{n \in \N}$ of $\R$, there exists a compact set $L \subseteq \Salg$ such that $x \not \in L$ implies $\sigma \bigl ( H(x) \bigr ) \cap K = \emptyset$. Hence, there is a subsequence $\{ x_{n_k} \}_{k \in \N} \subset L$ which we may take to be convergent to some $x \in L$. 
	Then pick $\varphi \in \Cont_{\infty}(\R)$ and $N \in \N$ large enough with $\varphi(\lambda) \neq 0$ such that 
	\begin{align*}
		\abs{\varphi(\lambda_n)} \geq \tfrac{1}{2} \abs{\varphi(\lambda)} > 0 
	\end{align*}
	holds for all $n \geq N$. Hence, $\bnorm{\varphi \bigl ( H(x_n) \bigr )}_{\Calg'}$ can be bounded from below by $\frac{1}{2} \abs{\varphi(\lambda)}$: 
	\begin{align*}
		\bnorm{\varphi \bigl ( H(x_n) \bigr )}_{\Calg'} &= \sup \bigl \{ \abs{\varphi(\lambda)} \; \vert \; \, \lambda \in \sigma \bigl ( H(x_n) \bigr ) \bigr \} 
		\\ &
		\geq \abs{\varphi(\lambda_n)} \geq \tfrac{1}{2} \abs{\varphi(\lambda)} > 0 
		&&  \forall n \geq N 
	\end{align*}
	%
	By continuity of $x \mapsto \bnorm{\varphi \bigl ( H(x) \bigr )}_{\Calg'}$, this implies $\bnorm{\varphi \bigl ( H(x) \bigr )}_{\Calg'} > 0$, \ie $\varphi \bigl ( H(x) \bigr ) \neq 0$. Hence $\lambda \in \bigcup_{x \in \Salg} \sigma \bigl ( H(x) \bigr )$ and the union is a closed subset of $\R$. 
	
	$\varphi \bigl ( H(x) \bigr ) \neq 0$ also trivially implies $\varphi(H) \neq 0$, and the union of the spectra must be contained in the spectrum of $H$, $\bigcup_{x \in \Salg} \sigma \bigl ( H(x) \bigr ) \subseteq \sigma(H)$. 
	\medskip
	
	\noindent
	Conversely, let $\lambda \not\in \bigcup_{x \in \Salg}$. As each of the $\sigma \bigl ( H(x) \bigr )$ as well as their union is closed, there exists a neighborhood $V$ of $\lambda$ that is disjoint from the spectra of $H(x)$ for all $x \in \supp H$, $V \cap \sigma \bigl ( H(x) \bigr ) = \emptyset$. If we then pick $\varphi_V \in \Cont_{\infty}(\R)$ whose support lies in $V$, then $\varphi_V \bigl ( H(x) \bigr ) = 0$ for all $x \in \supp H$, and hence $\varphi_V(H) = 0$. This means $\lambda \not\in \sigma(H)$ and we have shown $\bigcup_{x \in \Salg} \sigma \bigl ( H(x) \bigr ) \supseteq \sigma(H)$. 
\end{proof}
%
An analogous statement holds if we consider $\Ideal$-essential spectra: let $\Ideal$ be a two-sided, closed ideal of $\Calg'$ and $\pi_{\Ideal} : \Calg' \longrightarrow \Calg' / \Ideal$ the corresponding canonical projection. As we will see, this induces a canonical projection between $\Cont_{\infty}(\Salg ; \Calg')$ and $\Cont_{\infty} \bigl ( \Salg ; {\Calg'} / {\Ideal} \bigr )$. We can identify the latter with the quotient $\Cont_{\infty}(\Salg ; \Calg') / \Cont_{\infty}(\Salg ; \Ideal)$: two functions $f , g \in \Cont_{\infty}(\Salg ; \Calg')$ are equal modulo $\Cont_{\infty}(\Salg ; \Ideal)$, \ie $f = g + h$ for some $h \in \Cont_{\infty}(\Salg ; \Ideal)$, if and only if $x \mapsto \pi_{\Ideal} \bigl ( f(x) \bigr )$ equals $x \mapsto \pi_{\Ideal} \bigl ( g(x) \bigr )$ as functions mapping from $\Salg$ to ${\Calg'} / {\Ideal}$. This identification leads to an injection 
\begin{align*}
	\imath : \Cont_{\infty}(\Salg ; \Calg') / \Cont_{\infty}(\Salg ; \Ideal) \longrightarrow \Cont_{\infty} \bigl ( \Salg ; {\Calg'} / {\Ideal} \bigr ) 
\end{align*}
Hence, the projection $\pi_{\Ideal}$ induces a projection between the larger $C^*$-algebras 
\begin{align*}
	\pi_{\Ideal}^{\Salg} : \Cont_{\infty}(\Salg ; \Calg') \longrightarrow \Cont_{\infty} \bigl ( \Salg ; {\Calg'} / {\Ideal} \bigr ) , \; f \mapsto \pi_{\Ideal}^{\Salg}(f) 
\end{align*}
that acts as $\bigl ( \pi_{\Ideal}^{\Salg}(f) \bigr )(x) := \pi_{\Ideal} \bigl ( f(x) \bigr )$. Not surprisingly, this immediately leads to 
\begin{prop}\label{algebraicPOV:affiliation:prop:decomposition_spectrum}
	Assume $\Ideal$ is a two-sided, closed ideal of a $C^*$-algebra $\Calg'$ and 
	\begin{align*}
		\pi_{\Ideal} : \Calg'  \longrightarrow {\Calg'} / {\Ideal} 
	\end{align*}
	the corresponding canonical projection. Then the $\Ideal$-essential spectrum of a proper, continuous observable affiliated to $\Cont_{\infty}(\Salg ; \Calg')$ can be written as the union of the $\Ideal$-essential spectra of $H(x)$, 
	\begin{align}
		\sigma_{\Ideal}(H) \equiv \sigma \bigl ( \pi_{\Ideal}^{\Salg}(H) \bigr ) = \bigcup_{x \in \supp H} \sigma_{\Ideal} \bigl ( H(x) \bigr ) 
		= \bigcup_{x \in \supp H} \sigma \bigl ( \pi_{\Ideal} \bigl ( H(x) \bigr ) \bigr ) 
		, 
	\end{align}
	where $\pi^{\Salg}_{\Ideal} : \Cont_{\infty}(\Salg ; \Calg') \longrightarrow \Cont_{\infty} \bigl ( \Salg ; {\Calg'} / {\Ideal} \bigr )$ is the projection induced by $\pi_{\Ideal}$. 
\end{prop}
%



%
%
%

\chapter{Pseudodifferential theory revisited} 
\label{psiDO_reloaded}
%
Magnetic operators on $\Xgroup = \R^d$ have been the subject of Chapters~\ref{magWQ} and \ref{asymptotics} where rather hands-on pseudodifferential techniques have been extended to the case of \emph{magnetic} Weyl calculus. One of the main results, Theorem~\ref{asymptotics:thm:equivalenceProduct}, states that the magnetic Weyl product preserves Hörmander classes and is proven via oscillatory integral techniques (see Appendix~\ref{appendix:oscillatory_integrals}). 

On the other hand, the algebraic point of view outlined in Chapter~\ref{algebraicPOV} \cite{Mantoiu_Purice_Richard:twisted_X_products:2004} sheds light on how algebraic properties translate into spectral properties. If $\Alg \subseteq \BCont_u(\R^d)$ is an $\R^d$-algebra (Definition~\ref{algebraicPOV:twisted_crossed_products:defn:X-algebra}) and the components of $B$ are in the multiplier algebra $\mult(\Alg)$, then the Fourier transform of the twisted crossed product $\Alg \rtimes^{\omega^B}_{\theta,\tau} \R^d \equiv \sXprodR$ can be thought of as being composed of tempered distributions. Some elements are functions on phase space $\Pspace = \R^d \times {\R^d}^*$ and the behavior of $x \mapsto h(\xi,x) \in \Alg$ is characterized by the \emph{coefficient algebra $\Alg$}.\index{anisotropy algebra} Such functions are called $\Alg$-anisotropic and the product on $\sFXprodR$ respects the $\Alg$-anisotropy by definition. Smoothness properties, however, are more difficult to extract. 

The purpose of this chapter is to show how to \emph{combine} these two complementary approaches to one's advantage. Smoothness properties inferred from the pseudodifferential point of view indeed combine nicely with the \emph{anisotropy}, \ie~the behavior in $x$ as encoded in the coefficient algebra $\mathcal{A}$ is preserved under multiplication. We caution the reader that this is by no means an immediate consequence of the various definitions. 
\medskip

\noindent
Three main results will be proven in this chapter: 
\begin{enumerate}[(i)]
	\item We use the recently proven fact that $\Hoermrd{0}{\rho}{0}$ with the magnetic Moyal product $\magW$ and complex conjugation as involution forms a so-called $\Psi^*$-algebra \cite{Iftimie_Mantoiu_Purice:commutator_criteria:2008}: they are spectrally invariant $C^*$-subalgebras with Fréchet structure. What makes $\Psi^*$-algebras special is the fact that closed, symmetric subalgebras are again $\Psi^*$-algebras \cite[Cor.~2.5]{Lauter:operator_theoretical_approach_melrose_algebras:1998}. This abstract fact automatically guarantees that Moyal resolvents of real-valued \emph{anisotropic} Hörmander symbols of positive order $m > 0$ are again \emph{anisotropic} symbols or order $- m$ (Theorem~\ref{psiDO_reloaded:PsiD_theory:magnetic_composition:thm:composition_of_anisotropic_symbols}). 
	\item This leads to the second main result, namely a principle of affiliation for real-valued, elliptic, anisotropic symbols of positive order to the abstract algebra $\sFXprodR$ (Theorem~\ref{psiDO_reloaded:inversion_affiliation:thm:affiliation}). We also identify $C^*$-subalgebras composed of smooth functions on which we can apply pseudodifferential techniques again. 
	\item The principle of affiliation is then the starting point for spectral analysis. We show how to decompose the spectrum in terms of `asymptotic operators' which are in some sense located at infinity (Theorem~\ref{psiDO_reloaded:spectral:thm:essential_spectrum}). There is no Hilbert space analog for this decomposition, though, and even for the non-magnetic case $B = 0$, this result is new. 
\end{enumerate}
%
%
Now we specialize the definitions of the previous chapter to the case $\Xgroup = \R^d$ and magnetic $2$-cocycles $\omega \equiv \omega^B$ again to reap the benefits from the algebraic approach. Although much of what we state here holds for more general choices of groups $\Xgroup$, we shall not make this endeavor for the sake of readability. The results presented in this section are the fruits of a collaboration of Marius Măntoiu, Serge Richard and the author \cite{Lein_Mantoiu_Richard:anisotropic_mag_pseudo:2009}. 

\paragraph{Setting and assumptions} 
\label{psiDO_reloaded:setting_and_assumptions}
Let us revisit the definitions of the last chapter: for simplicity, we will always choose Weyl ordering, \ie $\tau = \nicefrac{1}{2}$. The advantages (most importantly, real-valued functions are mapped onto symmetric operators) have already been elaborated upon in Chapter~\ref{magWQ}. Furthermore, for magnetic $2$-cocycles we have $\omega^B(x,-x) = 1$ as the area of the flux triangle vanishes. Hence, the product of $f , g \in \Alg \rtimes^{\omega^B}_{\theta,\nicefrac{1}{2}} \R^d \cap L^1(\R^d ; \Alg) =: \Alg \rtimes^B_{\theta} \R^d \cap L^1(\R^d ; \Alg)$ reads 
\begin{align}
	(f \crossB g)(x) = \frac{1}{(2 \pi)^{\nicefrac{d}{2}}} \int_{\R^d} \dd y \, \theta_{\nicefrac{(y - x)}{2}} [f(y)] \, \theta_{\nicefrac{y}{2}} [g(x - y)] \, \theta_{- \nicefrac{x}{2}} [\omega^B(y , x - y)]
\end{align}
where \emph{in this chapter $\Alg$ always denotes a unital $\R^d$-algebra} in the sense of Definition~\ref{algebraicPOV:twisted_crossed_products:defn:X-algebra} and $\omega^B(q;x,y) := e^{- i \Gamma^B(\expval{q , q + x , q + x + y})}$ is the exponential of magnetic fluxes through triangles. The involution simplifies to $f^{\crossB}(x) = f(-x)^{\ast}$. We have added a prefactor $(2\pi)^{- \nicefrac{d}{2}}$ which has been absorbed into the measure in the last chapter. The condition that $\omega^B(x,y)$ is $\mathcal{U}(\Alg)$-valued implies that the components of $B$ have to be $\Alg$. The condition that $\Alg$ is unital is rather natural: if $h$ and $h'$ are hamiltonian symbols differ only by a constant $E_0$, then their quantizations will essentially have the same properties. The difference $E_0$ (seen as a constant function) corresponds to a different choice of zero energy and we would like both, $x \mapsto h(x,\xi)$ and $x \mapsto h'(x,\xi)$ to be contained in $\Alg$. It also does not change anything for the magnetic fields: if $\Alg$ were not unital, we could admit magnetic fields in the multiplier algebra of $\Alg$ -- which contains the case of constant field. 

Naturally, the twisted crossed product is represented on $L^2(\R^d)$ via $\Rep^A$ where the superscript $A$ indicates that the pseudotrivialization of the $2$-cocycle, 
\begin{align*}
	\lambda^A(q;x) := e^{-i \Gamma^A([q,q+x])} 
	, 
\end{align*}
involves the choice of a vector potential $A$ which represents $B = \dd A$, 
\begin{align}
	\bigl ( \Rep^A(f) u \bigr )(x) = \frac{1}{(2 \pi)^{\nicefrac{d}{2}}} \int_{\R^d} \dd y \, \lambda^A(x ; y - x) \, f \bigl ( \tfrac{1}{2} (x + y) , y - x \bigr ) \, u(y) 
	. 
\end{align}
The above is defined for $f \in L^1(\R^d ; \Alg)$ and $u \in L^2(\R^d)$, although we can immediately extend $\Rep^A$ to $f \in \Alg \rtimes^B_{\theta} \R^d$ or even further. $\Rep^A$ is called the \emph{Schrödinger representation} of $\Alg \rtimes^B_{\theta} \R^d$. 

Equivalently, we may restate this using the physically more natural, Fourier-trans\-formed twisted crossed product $\sFXprodR$. The involution $f^{\magW} = f^{\ast}$ reduces to complex conjugation and the product of two suitable functions $f , g : \R^d \times {\R^d}^* \longrightarrow \C$ is given by 
\begin{align*}
	(f \magW g)(X) &= \frac{1}{(2 \pi)^{2d}} \int_{\Xi} \dd Y \int_{\Xi} \dd Z \, e^{i \sigma(X,Y+Z)} \, e^{\frac{i}{2} \sigma(Y,Z)}
	\cdot \\
	&\qquad \qquad \qquad \qquad \qquad \cdot 
	\omega^B \bigl ( x-\tfrac{1}{2}(y+z) ; x + \tfrac{1}{2}(y-z) , x+\tfrac{1}{2}(y+z) \bigr ) 
	\cdot \\
	&\qquad \qquad \qquad \qquad \qquad \cdot 
	(\Fs^{-1} f)(Y) \, (\Fs^{-1} g)(Z) 
	. 
\end{align*}
This has been defined for Hörmander symbols via oscillatory integral techniques in Chapters~\ref{magWQ} and \ref{asymptotics}. Before we have a detailed look at interesting algebras and subalgebras, we need to define anisotropic Hörmander symbols and revisit the product formula. 


\section{Magnetic composition of anisotropic symbols} 
\label{psiDO_reloaded:magB_anisotropic_symbols}
\subsection{Anisotropic symbol spaces} 
\label{psiDO_reloaded:magB_anisotropic_symbols:anisotropic_symbol_spaces}
The algebra $\Alg$ encodes the behavior of the configurational part of the symbol (`the behavior in $x$') and that of the magnetic field. We will always make the following assumptions: 
\begin{assumption}[on the anisotropy algebra $\Alg$]\label{psiDO_reloaded:assumption:anisotropy_algebra}
	Unless explicitly stated otherwise, \linebreak throughout this chapter we will always assume that $\Alg$ is a \emph{unital} $\R^d$-algebra, \ie a $C^*$-subalgebra of $\BCont_u(\R^d)$ that is stable under translations. 
\end{assumption}
Much of the following can be immediately applied or adapted to the case where $\Alg$ is a general abelian $C^*$-algebra with an $\R^d$-action: 
\begin{defn}
	Let us define $\Alg^{\infty} := \bigl \{ \varphi \in \Alg \; \vert \; \R^d \ni x \mapsto \theta_x(\varphi) \in \Alg \mbox{ is } \Cont^{\infty} \bigr \}$. For $a \in \N_0^d$ we set 
	\begin{enumerate}[(i)] 
		\item $\delta^a : \Alg^{\infty} \longrightarrow \Alg^{\infty}$, $\varphi \mapsto \delta^a(\varphi) := \partial^a_x \bigl ( \theta_x [\varphi]\bigr ) \big \vert_{x=0}$, 
		\item $s^a : \Alg^{\infty} \longrightarrow \R^+$, $\varphi \mapsto s^a(\varphi) := \bnorm{\delta^a(\varphi)}_{\Alg}$.
	\end{enumerate}
\end{defn}
It is known that $\Alg^{\infty}$ is a dense $*$-subalgebra of $\Alg$, as well as a Fréchet $*$-algebra with the family of semi-norms $\{s^a \; \vert \; a \in \N_0^d \}$. But our setting is quite special: $\Alg$ is an abelian $C^*$-algebra composed of bounded and uniformly continuous complex-valued functions defined on the group $\R^d$ itself. 
\begin{lem}\label{psiDO_reloaded:magB_anisotropic_symbols:anisotropic_symbol_spaces:lem:coincidence_of_spaces}
	$\Alg^{\infty}$ coincides with $\Alg^{\infty}_0 := \bigl \{ \varphi \in \Cont^{\infty}(\R^d) \; \vert \; \partial_x^a \varphi \in \Alg, \, \forall a \in \N_0^d \bigr \}$. Furthermore, for any $a \in \N_0^d$ and $\varphi \in \Alg^{\infty}$, one has $\delta^a(\varphi) = \partial^a_x \varphi$.
\end{lem}
\begin{proof}
	For any $\varphi\in \Alg^{\infty}$ and $j\in \{ 1, \ldots, d \}$, there exists $\delta_j \varphi \in \Alg^{\infty}$ such that
	\begin{align*}
		\norm{ \frac{1}{t} \bigl ( \theta_{t e_j}(\varphi) - \varphi \bigr ) - \delta_j \varphi }_{\Alg}
		= \sup_{x \in \R^d} \abs{ \frac{1}{t} \bigl ( \varphi(x + t e_j) - \varphi(x) \bigr ) - (\delta_j \varphi)(x) } \xrightarrow{t \rightarrow 0} 0 
		, 
	\end{align*}
	where $e_1 , \ldots , e_d$ is the canonical basis in $\R^d$. This implies that $\varphi$ is differentiable and that $\delta_j \varphi = \partial_{x_j} \varphi$. By recurrence, we get $\delta^a(\varphi) = \partial^a_x \varphi$, and it clearly follows that $\Alg^{\infty} \subseteq \Alg^{\infty}_0$.

	To prove the converse, let $\varphi \in \Alg^{\infty}_0$. We know that $\varphi \in \Cont^{\infty} (\R^d)$ and that $\partial^a \varphi \in \Alg \subseteq \BCont (\R^d)$ for any $a \in \N_0^d$. Thus, we simply have to show that the derivatives are obtained as uniform limits. Indeed, one has for $t > 0$:
	\begin{align*}
		&\abs{\frac{1}{t} \bigl ( \varphi(x+te_j)-\varphi(x) \bigr ) - (\partial_j \varphi)(x)}
		= \abs{\frac{1}{t} \int^t_0 \dd s \bigl ( (\partial_j \varphi)(x + s e_j) - (\partial_j \varphi)(x) \bigr )} 
		\\
		&\qquad = \abs{\frac{1}{t}\int^t_0 \dd s \int^s_0 \dd u (\partial^2_j \varphi)(x + u e_j)}
		\leq \bnorm{\partial^2_j \varphi}_{\infty} \, \frac{1}{t} \int^t_0 \dd s \, s 
		\\
		&\qquad = \bnorm{\partial^2_j \varphi}_{\infty} \, \frac{t}{2} \xrightarrow{t \rightarrow 0} 0 
		, 
	\end{align*}
	uniformly in $x \in \R^d$, and similarly for $t < 0$. Then this argument can be applied iteratively.
\end{proof}
%
We now introduce the anisotropic version of the Hörmander classes of symbols, confer~also \cite{baaj:twisted_X_products_1:1988,baaj:twisted_X_products_2:1988,Coburn_Moyer_Singer:Cstar_algebras_ap_psiDo:1973,Connes:Cstar_algebras_and_differential_geometry:1980,Shubin:ap_functions_and_partial_do:1978}. For any $f:\Pspace \longrightarrow \C$ and $(x,\xi) \in \Pspace$, we will often write $f(\xi)$ for $f(\cdot, \xi)$. In that situation, $f$ will be seen as a function on ${\R^d}^*$ taking values in some space of functions defined on $\R^d$.
\begin{defn}[Anisotropic Hörmander symbols of order $m$\index{Hörmander symbols!anisotropic}]
	The space of \emph{$\Alg$-an\-iso\-tropic symbols of order $m$ and type $(\rho,\delta)$} is
	\begin{align*}
		S^m_{\rho,\delta}(\syspace) := \Bigl \{ f \in \Cont^{\infty}&(\syspace) \; \vert \; \forall a , \alpha \in \N_0^d \; \exists C_{a \alpha} > 0 : 
		\Bigr . \\
		\Bigl . &
		s^a \bigl [ (\partial_{\xi}^{\alpha} f)(\xi) \bigr ] \leq C_{a \alpha} \expval{\xi}^{m - \rho \abs{\alpha} + \delta \abs{a}} \; \forall \xi \in {\R^d}^* \Bigr \} 
	\end{align*}
	where the Fréchet structure is generated by the family of seminorms $\bigl \{ \snorm{\cdot}_{m,a\alpha} \bigr \}_{a , \alpha \in \N_0^d}$, 
	\begin{align*}
		\snorm{f}_{m,a \alpha} := \sup_{\xi \in {\R^d}^*} \Bigl ( \sexpval{\xi}^{-(m - \abs{\alpha} \rho + \abs{a} \delta)} \, s^a \bigl [ (\partial_{\xi}^{\alpha} f)(\xi) \bigr ] \Bigr )
		. 
	\end{align*}
\end{defn}
Due to the very specific nature of the $C^*$-algebra $\Alg$, there are again some simplifications:
\begin{lem}\label{psiDO_reloaded:magB_anisotropic_symbols:lem:equality_symbol_classes}
	The following equality holds:
	\begin{align}\label{psiDO_reloaded:magB_anisotropic_symbols:eqn:equality_symbol_classes}
		\Hoerrdan{m}
		= \Bigl \{ f \in \Hoerrd{m} \; \vert \; (\partial^a_x \partial_{\xi}^{\alpha} f)(\xi) \in \Alg , \, \forall \xi \in {\R^d}^* \mbox{ and } a , \alpha \in \N_0^d \Bigr \}
		. 
	\end{align}
\end{lem}
\begin{proof}
	First we notice that the conditions
	\begin{align*}
		s^a \bigl [ (\partial^{\alpha}_{\xi} f)(\xi) \bigr ] \leq C_{a \alpha} \expval{\xi}^{m-\rho \abs{\alpha} + \delta \abs{a}}, 
		&& \forall \xi \in {\R^d}^* 
		, 
	\end{align*}
	and 
	\begin{align*}
		\babs{(\partial_x^a\partial^\alpha_\xi f)(x,\xi)} \leq C_{a \alpha} \expval{\xi}^{m - \rho \abs{\alpha} + \delta \abs{a}}, 
		&& \forall (x,\xi) \in \Pspace 
		, 
	\end{align*}
	are identical. On the other hand, by Lemma \ref{psiDO_reloaded:magB_anisotropic_symbols:anisotropic_symbol_spaces:lem:coincidence_of_spaces},
	\begin{align*}
		(\partial^{\alpha}_{\xi} f)(\xi) \in \Alg^{\infty} \, \Longleftrightarrow \, (\partial_x^a \partial^{\alpha}_{\xi} f)(\xi) \in \Alg, 
		\, \forall a \in \N_0^d
		.
	\end{align*}
	It thus follows that $S^m_{\rho,\delta} (\syspace)$ is included in the right-hand~side of \eqref{psiDO_reloaded:magB_anisotropic_symbols:eqn:equality_symbol_classes}, and we are then left with proving that if $f \in S^m_{\rho,\delta}(\Pspace)$ and $(\partial^\alpha_\xi f)(\xi) \in \Alg^{\infty}$ for all $\alpha$ and $\xi$, then $f \in \Cont^{\infty} (\syspace)$.

	We first show that $f : {\R^d}^* \longrightarrow \Alg^{\infty}$ is differentiable, that is for each $a \in \N_0^d$ 
	\begin{align*}
		s^a \left [ \frac{1}{t} \bigl ( f(\xi + t e_j) - f(\xi) \bigr ) - (\partial_{\xi_j}f)(\xi) \right ] \xrightarrow{t \rightarrow 0} 0 , 
		&& \forall j = 1 , \ldots , d ,
	\end{align*}
	holds where $e_1 , \ldots , e_d$ is the canonical basis in ${\R^d}^* \cong \R^d$. Indeed, we have for $t > 0$:
	\begin{align*}
		\sup_{x \in \R^d} &\abs{\frac{1}{t} \bigl ( (\partial_x^a f)(x,\xi + t e_j) - (\partial_x^a f)(x,\xi) \bigr ) - \partial_x^a \partial_{\xi_j} f(x,\xi)} 
		= \\
		&\qquad \qquad \qquad \qquad 
		= \sup_{x \in \R^d} \abs{\frac{1}{t} \int_0^t \dd s \, \int_0^s \dd u \, (\partial_x^a \partial_{\xi_j}^2 f)(x,\xi + u e_j)} 
		\\
		&\qquad \qquad \qquad \qquad 
		\leq \sup_{x \in \R^d} \, \frac{1}{t} \int_0^t \dd s \, \int_0^s \dd u \, C_a \, \expval{\xi + t e_j}^{m - 2 \rho + \delta \abs{a}} 
		\\
		&\qquad \qquad \qquad \qquad 
		\leq C_a' \, \expval{\xi}^{m - 2 \rho + \delta \abs{a}} \, \frac{1}{t} \int_0^t \dd s \, \int_0^s \dd u \, \expval{u}^{\abs{m - 2 \rho + \delta \abs{a}}} 
		\\
		&\qquad \qquad \qquad \qquad 
		\leq C_a'' \, \expval{\xi}^{m - 2 \rho + \delta \abs{a}} \, \frac{1}{t} (t^2 - 0) \xrightarrow{t \rightarrow 0} 0 
	\end{align*}
	and similarly for $t < 0$. We can continue to apply this procedure to the resulting derivative $\partial_{\xi_j} f \in \Hoerrd{m - \rho}$ and finish the proof by recurrence.
\end{proof}
For $\Alg = \BCont_u(\R^d)$, it is readily shown that
\begin{align*}
	\BCont_u(\R^d)^{\infty} &= \bigl \{ \varphi \in \Cont^{\infty}(\R^d) \; \vert \; \partial_x^a \varphi \in \BCont_u(\R^d), \, \forall a \in \N_0^d \bigr \} 
	\\ 
	&= \bigl \{ \varphi \in \Cont^{\infty}(\R^d) \; \vert \; \partial_x^a \varphi \in \BCont(\R^d), \, \forall a \in \N_0^d \bigr \}
	=: \BCont^{\infty}(\R^d)
\end{align*}
and thus 
\begin{align*}
	S^m_{\rho,\delta} \bigl ( {\R^d}^*;\BCont_u(\R^d)^{\infty} \bigr ) 
	\equiv S^m_{\rho,\delta}\bigl ( {\R^d}^*;\BCont^{\infty}(\R^d) \bigr ) 
	= \Hoerrd{m} 
	.
\end{align*}
The symbol class associated to subalgebras $\Alg \subseteq \BCont_u(\R^d)$ are naturally contained in $\Hoerrd{m}$. 
\begin{prop}\label{psiDO_reloaded:PsiD_theory:anisotropic_symbol_classes:prop:properties_anisotropic_symbol_classes}
	\begin{enumerate}[(i)]
		\item $\Hoerrdan{m}$ is a closed subspace of the Fréchet space $\Hoerrd{m}$. 
		\item For any $m_1, m_2 \in \R$, $\Hoerrdan{m_1} \cdot \Hoerrdan{m_2} \subseteq \Hoerrdan{m_1 + m_2}$. 
		\item For any $a , \alpha \in \N_0^d$, $\partial^a_x \partial^{\alpha}_{\xi} \Hoerrdan{m} \subseteq \Hoerrdan{m - \rho \abs{\alpha} + \delta \abs{a}}$. 
	\end{enumerate}
\end{prop}
%
%
\begin{proof}
	\begin{enumerate}[(i)]
		\item We have to show that if $f_n \in \Hoerrdan{m}$ converges to $f \in \Hoerrd{m}$ with respect to the  family of seminorms of $\Hoerrd{m}$, $\sfnorm{f_n - f}{m}{a \alpha} \xrightarrow{n \rightarrow \infty} 0$, then $(\partial_x^a \partial_{\xi}^{\alpha} f)(\xi) \in \Alg$ for all $a$, $\alpha$ and $\xi$. But since $\Alg$ is closed, it is enough to show that for any $a , \alpha \in \N_0^d$, the following statement holds: if $g_n \in \Hoerrd{m}$ and $\sfnorm{g_n}{m}{a \alpha}$ tends to $0$ as $n \rightarrow \infty$, then $\snorm{(\partial_x^a \partial_{\xi}^{\alpha} g_n) (\xi)}_{\Alg} \xrightarrow{n \rightarrow \infty} 0$ for all $\xi \in {\R^d}^*$. Hence, $g_n$ and all its derivatives in $x$ and $\xi$ converge uniformly to $0$. 
		\item Statement (ii) follows by applying Lemma \ref{psiDO_reloaded:magB_anisotropic_symbols:lem:equality_symbol_classes}, Leibnitz's rule and the fact that $\Alg$ is an algebra. 
		\item Statement~(iii) is a direct consequence of Lemma~\ref{psiDO_reloaded:magB_anisotropic_symbols:lem:equality_symbol_classes}. 
	\end{enumerate}

\end{proof}
%


\subsection{Magnetic composition} 
\label{psiDO_reloaded:magB_anisotropic_symbols:symbol_composition}
The definition of twisted crossed products implies that anisotropies, \ie particular behaviors in $x$, are respected. On the other hand, in Chapters~\ref{magWQ} and \ref{asymptotics}, we have seen that the product of two Hörmander symbols is again a Hörmander symbol. Combining these two facts leads to 
\begin{thm}\label{psiDO_reloaded:PsiD_theory:magnetic_composition:thm:composition_of_anisotropic_symbols}
	Assume that the components $B_{jk}$ of the magnetic field are smooth elements of $\Alg^{\infty}$. Then, for any $\eps \leq 1$, $\lambda \leq 1$, $m_1 , m_2 \in \R$ and $0 \leq \delta < \rho \leq 1$ or $\rho = 0 = \delta$, one has\index{Weyl product!magnetic}
	\begin{align}\label{psiDO_reloaded:PsiD_theory:magnetic_composition:eqn:composition_of_anisotropic_symbols}
		\Hoerrdan{m_1} \magWel \Hoerrdan{m_2} \subseteq \Hoerrdan{m_1 + m_2}
		. 
	\end{align}
\end{thm}
Before proving this theorem, we need a technical lemma. For consistency, we have included the parameters $\eps$ and $\lambda$ as in Chapter~\ref{asymptotics}, but we only need them for the asymptotic expansions of the product. 
\begin{lem}\label{psiDO_reloaded:PsiD_theory:magnetic_composition:lem:properties_of_flux}
	Assume that the each component $B_{jk}$ of the magnetic field belongs to $\Alg^{\infty}$. Then, for all $a,b,c \in \N_0^d$ and all $x,y,z \in \R^d$, one has:
	\begin{enumerate}[(i)]
		\item $\bigl ( \partial_x^a \partial_y^b \partial_z^c \gBe \bigr )(\cdot , y , z) \in \Alg$, 
		\item $\bigl ( \partial_x^a \partial_y^b \partial_z^c \oBel \bigr )(\cdot , y , z) \in \Alg$, 
		\item $\babs{\bigl ( \partial_x^a \partial_y^b \partial_z^c \oBel \bigr )(\cdot , y , z)} \leq C_{abc} \bigl ( \expval{y} + \expval{z} \bigr )^{\abs{a} + \abs{b} + \abs{c}}$. 
	\end{enumerate}
\end{lem}
\begin{proof}
	The expressions $\bigl ( \partial_x^a \partial_y^b \partial_z^c \gBe \bigr )(\cdot , y , z)$ can be explicitly calculated from equation~\eqref{asymptotics:expansions:eqn:gamma_B_eps}, (i) follows from the completeness of $\Alg$ and (ii) easily follows from (i). Statement (iii) is an immediate generalization of Lemma~\ref{appendix:asymptotics:properties_mag_flux:lem:boundedness_mag_flux} and Corollary~\ref{appendix:asymptotics:cor:properties_flux} in Appendix~\ref{appendix:asymptotics:expansion_twister}. 
\end{proof}
\begin{proof}[Theorem~\ref{psiDO_reloaded:PsiD_theory:magnetic_composition:thm:composition_of_anisotropic_symbols}]
	For simplicity and consistency with Chapter~\ref{asymptotics}, we will only present the proof for $\delta = 0$. We refer to \cite{Lein_Mantoiu_Richard:anisotropic_mag_pseudo:2009} for details in case $0 < \delta < \rho \leq 1$. 
	
	As $\Hoerrdan{m} \subseteq \Hoerrd{m}$, Theorem~2.2 of \cite{Iftimie_Mantiou_Purice:magnetic_psido:2006} ensures that $f \magWel g$ exists as an oscillatory integral and is in Hörmander class $\Hoerrd{m_1 + m_2}$ for $f \in \Hoerrdan{m_1} \subseteq \Hoerrd{m_1}$ and $g \in \Hoerrdan{m_2} \subseteq \Hoerrd{m_2}$. 
	
	The product is also in the proper anisotropic symbol class, \ie for all $\xi \in {\R^d}^*$, we have $\bigl ( \partial_x^a \partial_{\xi}^{\alpha} (f \magWel g) \bigr )(\cdot , \xi) \in \Alg$: 
	if we write out $\partial_x^a \partial_{\xi}^{\alpha} (f \magWel g)$ for arbitrary $a , \alpha \in \N_0^d$ using equation~\eqref{asymptotics:expansions:eqn:Fourier_form_magnetic_Weyl_product}, we get a sum of terms of the type 
	\begin{align*}
		\frac{1}{(2\pi)^{2d}} \int \dd Y \, \int \dd Z \, &\partial_x^b \partial_{\xi}^{\alpha} e^{i \sigma(X,Y+Z)} \,  e^{i \tfrac{\eps}{2} \, \sigma(Y,Z)} 
		\cdot \notag \\ & \cdot 
		\partial_x^c \oBel \bigl ( x - \tfrac{\eps}{2} (y+z),x + \tfrac{\eps}{2} (y-z) , x + \tfrac{\eps}{2} (y+z) \bigr ) 
		\cdot \notag \\ & \cdot 
		\bigl ( \Fs^{-1} f \bigr )(Y) \, \bigl ( \Fs^{-1} g \bigr )(Z)
		, 
	\end{align*}
	where $a = b + c$. As in the proof of Lemma~\ref{appendix:asymptotics:existence_osc_int:lem:remainder}, we can rewrite this expression as an absolutely convergent integral. By assumption on $B$, $B_{jk} \in \Alg^{\infty}$ for all $j,k = 1 , \ldots , d$, and with the help of Propositions~\ref{psiDO_reloaded:PsiD_theory:anisotropic_symbol_classes:prop:properties_anisotropic_symbol_classes} and Lemma~\ref{psiDO_reloaded:PsiD_theory:magnetic_composition:lem:properties_of_flux}, we can use Dominated Convergence to conclude that $\partial_x^a \partial_{\xi}^{\alpha} (f \magBel g)$ is indeed anistropic in $x$. 
\end{proof}
\begin{thm}[Two-parameter expansion of $\magW$]\label{psiDO_reloaded:PsiD_theory:magnetic_composition:thm:2_parameter_expansion_composition_of_anisotropic_symbols}
	Assume the components of $B$ belong to $\Alg^{\infty}$ and $\eps \ll 1$, $\lambda \leq 1$. Then for each precision~$\eprec > 0$, $f \in \Hoerran{m_1}$ and $g \in \Hoerran{m_2}$, $\rho \in [0,1]$, the asymptotic two-parameter expansion of $f \magWel g$ as given by Theorem~\ref{asymptotics:thm:asymptotic_expansion} exists,\index{Weyl product!asymptotic expansion in $\eps$ and $\lambda$} 
	\begin{align}
		f \magWel g = \sum_{n = 0}^N \sum_{k = 0}^n \eps^n \lambda^k \, (f \magWel g)_{(n,k)} + R_N, 
	\end{align}
	where $N \equiv N(\eprec,\eps,\lambda)$ is deduced from Definition~\ref{asymptotics:defn:precision}, 
	\begin{align}
		(f \magWel g)_{(n,k)} \in \Hoerran{m_1 + m_2 - (n+k) \rho} 
	\end{align}
	for all $n \in \N_0$, $0 \leq k \leq n$ and the remainder is also in the proper anisotropic symbol class, $R_N \in \Hoerran{m_1 + m_2 - (N+1) \rho}$. 
\end{thm}
\begin{proof}
	As before, we note that we have already proven the existence of the expansion as a symbol in $\Hoerr{m_1 + m_2}$ (Theorem~\ref{asymptotics:thm:equivalenceProduct}) where 
	\begin{align*}
		(f \magWel g)_{(n,k)} \in \Hoerr{m_1 + m_2 - (n+k) \rho}
	\end{align*}
	and
	\begin{align*}
		R_N \in \Hoerr{m_1 + m_2 - (N+1) \rho} 
		. 
	\end{align*}
	It remains to show that these objects are really in the proper \emph{anisotropic} symbol class. Let us start with the $(n,k)$ term of the expansion: it can be written as a sum of terms of the type 
	\begin{align*}
		B_{n k}^{a \alpha b \beta}(x) \, (\partial_x^a \partial_{\xi}^{\alpha} f)(x,\xi) \, (\partial_x^b \partial_{\xi}^{\beta} g)(x,\xi)
	\end{align*}
	for suitable multiindices $a , \alpha , b , \beta \in \N_0^d$ where $B_{n k}^{a \alpha b \beta} \in \Alg^{\infty}$. To see this, we combine Lemma~\ref{appendix:asymptotics:expansion_twister:lem:expansion_twister} and Lemma~\ref{psiDO_reloaded:PsiD_theory:magnetic_composition:lem:properties_of_flux}. Then $x \mapsto B_{n k}^{a \alpha b \beta}(x) \, (\partial_x^a \partial_{\xi}^{\alpha} f)(x,\xi) \, (\partial_x^b \partial_{\xi}^{\beta} g)(x,\xi) \in \Alg$ follows as claimed. 
	
	The remainder can be written as 
	\begin{align*}
		R_N = f \magWel g - \sum_{n = 0}^N \sum_{k = 0}^n \eps^n \lambda^k \, (f \magWel g)_{(n,k)} \in \Hoerr{m_1 + m_2 - (N+1) \rho}
		. 
	\end{align*}
	Combining Proposition~\ref{psiDO_reloaded:PsiD_theory:anisotropic_symbol_classes:prop:properties_anisotropic_symbol_classes} and  Theorem~\ref{psiDO_reloaded:PsiD_theory:magnetic_composition:thm:composition_of_anisotropic_symbols} with what we have shown in the previous paragraph, we conclude that the remainder must also be in the proper \emph{anisotropic} symbol class, 
	\begin{align*}
		\Hoerr{m_1 + m_2 - (N+1) \rho} \cap \Hoerran{m_1 + m_2} = \Hoerran{m_1 + m_2 - (N+1) \rho} 
		. 
	\end{align*}
	This concludes the proof. 
\end{proof}
In the same way, we can show that the two one-parameter expansions also maps anisotropic symbols onto anisotropic symbols. Since there are no contributions for $k > n$, an immediate corollary is the expansion with respect to $\eps \ll 1$ only: 
\begin{cor}[$\eps$ expansion of $\magW$]
	Under the assumptions of Theorem~\ref{psiDO_reloaded:PsiD_theory:magnetic_composition:thm:2_parameter_expansion_composition_of_anisotropic_symbols}, we can expand the product $f \magWel g$ in powers of $\eps \ll 1$ even if $\lambda = 1$: for each $N \in \N_0$, we can expand the the product\index{Weyl product!asymptotic expansion in $\eps$} of $f \in \Hoerran{m_1}$ and $g \in \Hoerran{m_2}$ as 
	\begin{align*}
		f \magW_{\eps} g = \sum_{n = 0}^N \eps^n (f \magW_{\eps} g)_{(n)} + R_N
	\end{align*}
	where $(f \magW_{\eps} g)_{(n)} \in \Hoerran{m_1 + m_2 - n \rho}$ and $R_N \in \Hoerran{m_1 + m_2 - (N+1) \rho}$. 
\end{cor}
The expansion in $\lambda$ follows from completely analogous arguments. 
\begin{thm}[$\lambda$ expansion of $\magW$]
	Assume the components of the magnetic field $B$ are of class $\Alg^{\infty}$; then for $\lambda \ll 1$ and $\eps \leq 1$, we can expand $\magWel$ of two anisotropic Hörmander symbols $f \in \Hoerran{m_1}$ and $g \in \Hoerran{m_2}$ asymptotically in $\lambda$\index{Weyl product!asymptotic expansion in $\lambda$} such that 
	\begin{align}
		f \magWel g - \sum_{k = 0}^N \lambda^k (f \magWel g)_{(k)} &\in \Hoerran{m_1 + m_2 - 2(N+1) \rho} 
		, 
		\notag \\
		(f \magWel g)_{(k)} &\in \Hoerran{m_1 + m_2 - 2k \rho} 
		. 
	\end{align}
\end{thm}
%



\section{Relevant $C^*$-algebras} 
\label{psiDO_reloaded:relevant_cStar_algebras}
The algebra of bounded operators $\BopL$ often serves as a starting point for analyzing observables of interest. In Chapter~\ref{algebraicPOV:generalized_weyl_calculus}, we have shown how to extend $\OpA : \sFXprodR \longrightarrow \BopL$ and $\Rep^A : \sXprodR \longrightarrow \BopL$, $\sXprodR := \Alg \rtimes^{\omega^B}_{\theta,\nicefrac{1}{2}} \R^d$, to take values in all of the bounded operators. Alternatively, we can argue from the pseudodifferential point of view: if we require the magnetic field $B$ to have components in $\BCont^{\infty}(\R^d)$, we can choose a vector potential $A$ with components in $\Cpol^{\infty}(\R^d)$. In \cite[Proposition~5]{Mantoiu_Purice:magnetic_Weyl_calculus:2004} it has been shown how to extend $\OpA$ to a linear topological isomorphism $\Schwartz'(\R^d) \longrightarrow \mathcal{B} \bigl ( \Schwartz(\R^d) , \Schwartz'(\R^d) \bigr )$. Since $\BopL$ is continuously embedded in $\mathcal{B} \bigl ( \Schwartz(\R^d) , \Schwartz'(\R^d) \bigr )$, we can define 
\begin{defn}
	For any vector potential $A$ with components in $\Cont^{\infty}_{\mathrm{pol}}(\R^d)$ associated to $B$ with components in $\BCont^{\infty}(\R^d)$, we set
	\begin{align*}
		\AB := {\OpA}^{-1} \Bigl ( \BopL \Bigr ) 
		.
	\end{align*}
\end{defn}
It is obviously a vector subspace of $\mathcal{S}'(\Pspace)$ which only depends on the magnetic field by gauge covariance. On convenient subsets, for example on elements of $\AB$ which are also in the magnetic Moyal algebra $\mathcal{M}^B(\Pspace)$ (see~\cite[Section~II.D]{Mantoiu_Purice:magnetic_Weyl_calculus:2004}),  $\AB \cap \mathcal{M}^B(\Pspace)$, the transported product from $\BopL$ coincides with $\magW$, and taking the adjoint in $\BopL$ corresponds to the involution ${}^{\magW}$ (which reduces to complex conjugation in Weyl ordering). Endowed with the transported norm
\begin{align}
	\bnorm{f}_B \equiv \bnorm{f}_{\AB} := \bnorm{\OpA(f)}_{\mathcal{B}(\Hil)}
\end{align}
$\AB$ is a $C^*$-algebra.

With this notation and due to the inclusion $\Hoermrd{m}{\rho}{\delta} \subset \Hoermrd{m}{\delta}{\delta}$ for $\delta<\rho$, Theorem~\ref{magWQ:important_results:continuity_and_selfadjointness:thm:magnetic_Calderon_Vaillancourt} can be rephrased:
\begin{prop}\label{psiDO_reloaded:relevant_cStar_algebras:thm:continuous_embedding}
	For any $0 \leq \delta \leq \rho \leq 1$ with $\delta \neq 1$, $\Hoerrd{0} \hookrightarrow \AB$ can be continuously embedded in $\AB$. 
\end{prop}
We now define two $\Alg$-dependent $C^*$-subalgebras of $\AB$.
\begin{defn}\label{psiDO_reloaded:relevant_cStar_algebras:defn:Cstar_subalgebras}
	We set
	\begin{enumerate}[(i)]
		\item $\fxb$ for the $C^*$-subalgebra of $\AB$ generated by 
		\begin{align*}
			\mathcal{S}(\syspace) \equiv \Hoeran{-\infty} := \bigcap_{m \in \R} \Hoerrdan{m} 
			. 
		\end{align*}
		\item $\MBA$ for the $C^*$-subalgebra of $\,\AB$ generated by $\Hoermrdan{0}{0}{0}$.
	\end{enumerate}
\end{defn}
%
The algebra $\fxb$ is really $\sFXprodR$ in disguise. 
%
\begin{thm}\label{psiDO_reloaded:relevant_cStar_algebras:thm:description_FXprod_via_S}
	The algebra $\fxb$ coincides with $\sFXprodR$. In other words, the partial Fourier transform $\Fourier^{-1} : \Schwartz'(\R^d \times \R^d) \longrightarrow \Schwartz'(\R^d \times {\R^d}^*)$ restricts to a $C^*$-isomorphism $\Fourier^{-1} : \sXprodR \longrightarrow \fxb$.
\end{thm}
\begin{proof}
	The partial Fourier transform $\Fourier^{-1}$ is an isomorphism between $\Schwartz(\R^d;\Alg^{\infty})$ and $\Schwartz(\syspace)$ which intertwines the products and the involutions:
	\begin{align*}
		(\Fourier^{-1} f) \magW (\Fourier^{-1} g) = \Fourier^{-1} \bigl ( f \crossB g \bigr ) 
		, && 
		(\Fourier^{-1} f)^{\magW} = \Fourier^{-1} (f^{\crossB}) 
		. 
	\end{align*}
	The statement follows then from the density of $\Schwartz(\syspace)$ in $\fxb$, and from the density of $\Schwartz(\R^d;\Alg^{\infty})$ in $L^1(\Ralg)$, and hence also in $\sXprodR = \Alg \rtimes^{\omega^B}_{\theta,\nicefrac{1}{2}} \R^d$. This means, $\fxb$ coincides with $\sFXprodR$ as claimed. 
\end{proof}
\begin{remark}
	In Definition \ref{psiDO_reloaded:relevant_cStar_algebras:defn:Cstar_subalgebras}, the algebra $\sFXprodR$ was defined as a closure of a set of smooth elements, but it can easily be guessed that non-smooth elements also belong to $\sFXprodR$, especially in the $\R^d$-variable. The partial Fourier transform of the elements in $L^1(\Ralg)$ belong to $\sFXprodR$, and these elements do not necessarily possess any smoothness except continuity.
\end{remark}
\begin{remark}
	By Lemma A.4 of \cite{Mantoiu_Purice_Richard:Cstar_algebraic_framework:2007}, for any $m < 0$ the set $\Fourier^{-1} \bigl ( \Hoermrdan{m}{1}{0} \bigr )$ is contained in $L^1(\Ralg)$, which implies that $\Hoermrdan{m}{1}{0} \subset \sFXprodR$. A trivial extension of the same lemma to arbitrary $\delta$ implies that $\Fourier^{-1} \bigl ( \Hoermrdan{-0}{1}{\delta} \bigr )$ is also contained in $L^1(\Ralg)$. But by a remark in \cite[p.~17]{Amrein_Boutet_Georgescu:C0_groups_commutator_methods:1996} such an inclusion is no longer true for $\rho \neq 1$. It follows that for $0 \leq \delta \leq \rho < 1$ many elements of $\Fourier^{-1} \bigl ( \Hoerrdan{-0} \bigr )$ only belong to $\CBA \setminus L^1(\Ralg)$.
\end{remark}
We finally state a result about how the algebra $\sFXprodR$ can be generated from a simpler set of its elements. It is an adaptation of \cite[Prop.~2.6]{Mantoiu_Purice_Richard:twisted_X_products:2004}.
\begin{prop}\label{psiDO_reloaded:relevant_cStar_algebras:prop:another_description_FXprod}
	The norm closure in $\AB$ of the subspaces generated by
	\begin{align*}
		\bigl \{ \varphi \magW f \; \vert \; \varphi \in \Alg, \, f \in \Schwartz({\R^d}^*) \bigr \} 
		&& 
		\mbox{and} 
		&&
		\bigl \{ f \magW \varphi \; \vert \; \varphi \in \Alg, \, f \in \Schwartz({\R^d}^*) \bigr \}
	\end{align*}
	are equal and coincide with the $C^*$-algebra $\sFXprodR$.
\end{prop}
It is also easily observed that the $\Hoerrdan{-\infty} \equiv \mathcal{S}(\syspace)$ is really independent of $\rho$ and $\delta$. Part of our interest in the algebra $\sFXprodR$ is due to the following proposition and its corollary. We define
\begin{align*}
	\Hoerrdan{-0} := \bigcup_{m < 0} \Hoerrdan{m}
	. 
\end{align*}
Suitable bounded sequences in $\Hoerrdan{-0}$ in general only converge in the topology of $\Hoerrdan{0}$; this is the crucial point in the proof of the next proposition. 
\begin{prop}\label{psiDO_reloaded:relevant_cStar_algebras:prop:Sminus_in_fxb}
	For every $0 \leq \delta \leq \rho \leq 1$ with $\delta \neq 1$, the space $\Hoerrdan{-0}$ is contained in $\sFXprodR$.
\end{prop}
\begin{proof}
	We adapt the proof of Proposition~1.1.11 in \cite{Hoermander:Fourier_integral_operators_1:1971} to show that any symbol $f \in \Hoerrdan{-0}$ is the limit of a sequence $ \bigl \{ f_{\epsilon} \bigr \}_{0 \leq \epsilon \leq 1} \in \mathcal{S}(\syspace)$ in the topology of $\Hoerrdan{-0}$, see also \cite[Sec.~1]{Grigis_Sjoestrand:microlocal_analysis:1994} for more details. This and Proposition~\ref{magWQ:important_results:continuity_and_selfadjointness:thm:magnetic_Calderon_Vaillancourt} will imply the result.
	
	Let $f \in \Hoerrdan{m}$ for some $m < 0$, $0 \leq \delta \leq \rho \leq 1$, $\delta \neq 1$, and let $\chi \in \mathcal{S}({\R^d}^*)$ with $\chi(0) = 1$. We set $f_{\epsilon}(x,\xi):= \chi(\epsilon \xi) \, f(x,\xi)$ for $0 \leq \epsilon \leq 1$. By using Proposition~\ref{psiDO_reloaded:PsiD_theory:anisotropic_symbol_classes:prop:properties_anisotropic_symbol_classes}~(ii), one has $f_{\epsilon} \in \mathcal{S}(\syspace)$ for all $\epsilon > 0$, and $\bigl \{ f_{\epsilon} \bigr \}_{0 \leq \epsilon \leq 1}$ is a bounded subset of $\Hoerrdan{m}$. Finally, one easily obtains that $f_{\epsilon}$ converges to $f$ as $\epsilon \rightarrow 0$ in the topology of $\Hoerrdan{0}$.
\end{proof}
\begin{remark}
	The same proof shows the density of $\mathcal{S}(\syspace)$ in $\Hoerrdan{m}$ with respect to the topology of $\Hoerrdan{m'}$ for arbitrary $m'>m$.
\end{remark}
\begin{cor}\label{psiDO_reloaded:relevant_cStar_algebras:cor:MBA_contained_mult_fxb}
	The $C^*$-algebra $\MBA$ is contained in the multiplier algebra $\mathcal{M}(\sFXprodR)$ of $\sFXprodR$.
\end{cor}
\begin{proof}
	The statement follows from the fact that $\mathcal{S}(\syspace)$ is a two-sided ideal in $\Hoermrdan{0}{0}{0}$ with respect to $\magWel$, from the definition of $\sFXprodR$ and $\MBA$, and a density argument.
\end{proof}
\begin{remark}
	Let us observe that $\Fourier \Calg^B_{\C} = \Cont_{\infty}({\R^d}^*)$ and $\mathfrak{M}^B_{\C} = \BCont_u ({\R^d}^*)$, while $\mathcal{M}(\Fourier \Calg^B_{\C}) = \BCont({\R^d}^*)$; so the inclusion could be strict. 
\end{remark}
%


\section{Inversion and affiliation} 
\label{psiDO_reloaded:inversion_and_affiliation}
%
Before we proceed, we need to quote some basic facts on so-called $\Psi^*$-algebras which are special $C^*$-subalgebras with Fréchet structure. The results are borrowed from \cite{Lauter:operator_theoretical_approach_melrose_algebras:1998} and some additional material can be found in \cite{Lauter_Monthubert_Nistor:spectral_invariance:2005}. 
\begin{defn}[$\Psi^*$-algebra\index{$\Psi^*$-algebra}]\label{psiDO_reloaded:inversion_affiliation:defn:Psi_star_algebra}
	Let $\Psi$ be a unital $C^*$-subalgebra of a $C^*$-algebra $\Alg$. We say that $\Psi$ is a $\Psi^*$-algebra if it is \emph{spectrally invariant} (or full), \ie 
	\begin{align*}
		\Psi \cap \Alg^{-1} = \Psi^{-1} 
	\end{align*}
	where $\Alg^{-1}$ and $\Psi^{-1}$ are the groups of invertible elements of $\Alg$ and $\Psi$, repsectively, and if $\Psi$ can be endowed with a Fréchet topology $\tau_{\Psi}$ such that $\Psi \hookrightarrow \Alg$ can be continuously embedded in $\Alg$. 
\end{defn}
The reason why $\Psi^*$-algebras are such a useful notion is the following Theorem: 
%
\begin{thm}[Corollary~2.5 in \cite{Lauter:operator_theoretical_approach_melrose_algebras:1998}]
	\label{psiDO_reloaded:inversion_affiliation:thm:Psi_star_sub_also_Psi_star}
	Let $\Psi \subseteq \Alg$ be a $\Psi^*$-algebra and $\Psi' \subseteq \Psi$ be a closed, symmetric subalgebra of $\Psi$ with unit. Then $\Psi' \hookrightarrow \Alg$ endowed with the restricted topology $\tau_{\Psi'} := \tau_{\Psi} \vert_{\Psi'}$ is again a $\Psi^*$-algebra. 
\end{thm}
Furthermore, $\Psi^*$-algebras also have a nice holomorphic functional calculus:\index{holomorphic functional calculus} let $\varphi \in \Psi \subseteq \Alg$ be an element of a $\Psi^*$-algebra and $f : \C \longrightarrow \C$ be a function which is holomorphic in a neighborhood of the spectrum $\sigma(\varphi) := \bigl \{ z \in \C \; \vert \; \varphi - z \, \id \mbox{ is invertible} \bigr \}$. Then 
\begin{align*}
	f(\varphi) := \frac{1}{2 \pi i} \int_{\Gamma} \dd z \, f(z) \, \bigl ( \varphi - z \, \id \bigr )^{-1} \in \Psi 
\end{align*}
is well-defined and again an element of the $\Psi^*$-algebra ($\Gamma$ is a contour surrounding $\sigma(\varphi)$). 
\medskip

\noindent
Recently, Iftimie, Măntoiu and Purice \cite{Iftimie_Mantoiu_Purice:commutator_criteria:2008} have proven that $\Hoermrd{0}{\rho}{0}$ is a $\Psi^*$-algebra in $\AB$ using the magnetic version of the Bony criterion.\index{Bony criterion} Combining this fact with Theorem~\ref{psiDO_reloaded:PsiD_theory:magnetic_composition:thm:composition_of_anisotropic_symbols} immediately yields 
%
\begin{prop}\label{psiDO_reloaded:inversion_affiliation:prop:inversion_of_Hoermander_symbol_Psi_star}
	Let $\Alg$ be a \emph{unital} $\R^d$-algebra; then $\Hoermrdan{0}{\rho}{0}$, $0 \leq \rho \leq 1$, is a $\Psi^*$-algebra, it is stable under the holomorphic functional calculus, $\bigl ( \Hoermrdan{0}{\rho}{0} \bigr )^{(-1)_B}$ is open and the map 
	\begin{align}
		{}^{(-1)_B} : \bigl ( \Hoermrdan{0}{\rho}{0} \bigr )^{(-1)_B} \longrightarrow \Hoermrdan{0}{\rho}{0} , \; 
		 f \mapsto f^{(-1)_B} 
		, 
	\end{align}
	is continuous. Here, we have used the superscript ${}^{(-1)_B}$ for the inverse with respect to the composition law $\magW$.
\end{prop}
%
In order to state the next lemma some notations are needed. For $m > 0$, $\lambda>0$ and $\xi \in {\R^d}^*$, set
\begin{align*}
	\mathfrak{p}_{m,\lambda}(\xi) := \expval{\xi}^m + \lambda 
	. 
\end{align*}
As $\Alg$ is unital, the map $\mathfrak{p}_{m,\lambda}$ is clearly an element of $\Hoermrdan{m}{1}{0}$ and its pointwise inverse an element of $\Hoermrdan{m}{1}{0}$. It has been proven in \cite[Thm.~1.8]{Mantoiu_Purice_Richard:Cstar_algebraic_framework:2007} that for $\lambda$ large enough, $\mathfrak{p}_{m,\lambda}$ is invertible with respect to the composition law $\magW$ and that its inverse $\mathfrak{p}_{m,\lambda}^{(-1)_B}$ belongs to $\sFXprodR$. So for any $m > 0$ we can fix $\lambda = \lambda(m)$ such that $\mathfrak{p}_{m,\lambda(m)}$ is invertible. Then, for arbitrary $m \in \R$ we set
\begin{align*}
	\mathfrak{r}_m :=
	\left \{
	\begin{matrix}
		\mathfrak{p}_{m , \lambda(m)} & \mbox{ for } m > 0 \\
		1 & \mbox{ for } m = 0 \\
		\mathfrak{p}_{\abs{m},\lambda(\abs{m})}^{(-1)_B} & \mbox{ for } m < 0 \\
	\end{matrix}
	\right .
	. 
\end{align*}
The relation $\mathfrak{r}_m^{(-1)_B} = \mathfrak{r}_{-m}$ clearly holds for all $m \in \R$. Let us still show an important property of $\mathfrak{r}_m$.
\begin{lem}
	For any $m \in \R$, one has $\mathfrak{r}_m \in \Hoermrdan{m}{1}{0}$.
\end{lem}
\begin{proof}
	For $m \geq 0$, the statement is trivial from the definition of $\mathfrak{r}_m$. But for $m < 0$, we will use ideas from the proof of \cite[Theorem~1.8]{Mantoiu_Purice_Richard:Cstar_algebraic_framework:2007}: for $\lambda$ large enough, $\mathfrak{p} := \mathfrak{p}_{\abs{m},\lambda(\abs{m})}$ has been shown to be invertible with respect to the composition law $\magW$, and that its inverse is given by the formula
	\begin{align}
		\mathfrak{p}^{(-1)_B}= \mathfrak{p}^{-1} \magW \bigl ( \mathfrak{p} \magW \mathfrak{p}^{-1} \bigr )^{(-1)_B}
		, 
		\label{psiDO_reloaded:inversion_affiliation:eqn:inverse_p_m}
	\end{align}
	where $\mathfrak{p}^{-1}$ is the pointwise inverse of $\mathfrak{p}$, and $\lambda$ has been chosen such that the Moyal inverse $\bigl ( \mathfrak{p} \magW \mathfrak{p}^{-1} \bigr )^{(-1)_B}$ is well defined and belongs to $\AB$. Furthermore, since $\mathfrak{p}^{-1}$ belongs to $\Hoermrdan{-m}{1}{0}$, the product $\mathfrak{p} \magW \mathfrak{p}^{-1}$ belongs to $\Hoermrdan{0}{1}{0}$. It then follows that the inverse of $\mathfrak{p} \magW \mathfrak{p}^{-1}$ also belongs to $\Hoermrdan{0}{1}{0}$ by the $\Psi^*$-property of $\Hoermrdan{0}{1}{0}$. One concludes by observing that the right-hand~side of \eqref{psiDO_reloaded:inversion_affiliation:eqn:inverse_p_m} belongs to $\Hoermrdan{-m}{1}{0}$, and corresponds to $\mathfrak{r}_m$ for $m < 0$.
\end{proof}
\begin{prop}\label{psiDO_reloaded:inversion_affiliation:prop:inversion_Hoermander_m}
	Let $m > 0$, $\rho \in [0,1]$ and $f \in \Hoermrdan{m}{\rho}{0}$. If $f$ is invertible in the magnetic Moyal algebra $\mathcal{M}^B(\Pspace)$ and $\mathfrak{r}_m \magW f^{(-1)_B} \in \AB$, then $f^{(-1)_B}$ belongs to $\Hoermrdan{-m}{\rho}{0}$.
\end{prop}
\begin{proof}
	Let us first observe that
	\begin{align*}
		f \magW \mathfrak{r}_{-m} \in  \Hoermrdan{m}{\rho}{0} \magW \Hoermrdan{-m}{1}{0} \subset \Hoermrdan{0}{\rho}{0}.
	\end{align*}
	Furthermore, this element is invertible in $\AB$ since its inverse $\bigl ( f \magW \mathfrak{r}_{-m} \bigr )^{(-1)_B}$ is equal to $\mathfrak{r}_m \magW f^{(-1)_B}$, which belongs to $\AB$. Then, by the $\Psi^*$-property of $\Hoermrdan{0}{\rho}{0}$, it follows that $\bigl ( f \magW \mathfrak{r}_{-m} \bigr )^{(-1)_B}$ belongs to $\Hoermrdan{0}{\rho}{0}$, and so does $\mathfrak{r}_m \magW f^{(-1)_B}$. Consequently, one has
	\begin{align*}
		f^{(-1)_B} &= \mathfrak{r}_{-m} \magW \bigl ( \mathfrak{r}_m \magW f^{(-1)_B} \bigr ) 
		= \mathfrak{r}_{-m} \magW \bigl ( f \magW \mathfrak{r}_{-m} \bigr )^{(-1)_B} 
		\\
		&
		\in \Hoermrdan{-m}{\rho}{0} \magW \Hoermrdan{0}{\rho}{0} 
		\subset \Hoermrdan{-m}{\rho}{0} 
		.
	\end{align*}
	This concludes the proof. 
\end{proof}
Elliptic, real-valued symbols are one class of functions which satisfy the assumptions of the above proposition. 
\begin{defn}\label{psiDO_reloaded:inversion_affiliation:defn:elliptic_symbol}
	A symbol $f \in \Hoerrdan{m}$ is called \emph{elliptic} if there exist two strictly positive constants $R$ and $C$ such that 
	\begin{align*} 
		\abs{f(x,\xi)} \geq C \expval{\xi}^m 
	\end{align*} 
	for all $x \in \R^d$ and $\abs{\xi} > R$.
\end{defn}
We are now in a position to state and prove our main theorem on inversion:
\begin{thm}[Inversion\index{Moyal resolvent}]\label{psiDO_reloaded:inversion_affiliation:thm:inversion}
	Let $f \in \Hoermrdan{m}{\rho}{0}$, $0 \leq \rho  \leq 1$, be a real-valued elliptic symbol of order $m > 0$. Then for any $z \in \C \setminus \R$, the function $f - z$ is invertible in the magnetic Moyal algebra $\mathcal{M}^B(\Pspace)$ and its inverse $(f - z)^{(-1)_B}$ belongs to $\Hoermrdan{-m}{\rho}{0}$.
\end{thm}
\begin{proof}
	It has been proved in \cite[Thm.~4.1]{Iftimie_Mantiou_Purice:magnetic_psido:2006} that $\OpA(f)$ defines a selfadjoint operator with domain $H^m_A(\R^d)$ (Definition~\ref{magWQ:important_results:continuity_and_selfadjointness:defn:magnetic_Sobolev_space}) for any vector potential $A$ whose components belong to $\Cont^{\infty}_{\mathrm{pol}}(\R^d)$. In particular, $z$ does not belong to the spectrum of $\OpA(f)$, which is independent of $A$ by gauge covariance, and $\OpA(f)-z = \OpA(f-z)$ is invertible. Its inverse belongs to $\BopL$, which means that $(f-z)^{(-1)_B}$ belongs to $\AB$. 
	Hence, the Moyal resolvent $(f-z)^{(-1)_B}$ exists in the magnetic Moyal algebra: by Proposition~\ref{magWQ:extension:magnetic_moyal_algebra:prop:mag_Moyal_algebra_L_S_L_Sprime}, $\Op^A$ is an isomorphism between $\mathcal{M}^B(\Pspace)$ and $\mathcal{L} \bigl ( \Schwartz(\R^d) \bigr ) \cap \mathcal{L} \bigl ( \Schwartz'(\R^d) \bigr )$, and we conclude from $1 , f - z \in \mathcal{M}^B(\Pspace)$ that 
	\begin{align*}
		\Op^A(1) = \Op^A \bigl ( (f-z)^{(-1)_B} \magW (f-z) \bigr ) : \Schwartz^{(\prime)}(\R^d) \longrightarrow \Schwartz^{(\prime)}(\R^d) 
	\end{align*}
	holds. In other words, $\Op^A \bigl ( (f-z)^{(-1)_B} \bigr ) \in \mathcal{L} \bigl ( \Schwartz(\R^d) \bigr ) \cap \mathcal{L} \bigl ( \Schwartz'(\R^d) \bigr )$ and $(f-z)^{(-1)_B}$ exists in $\mathcal{M}^B(\Pspace)$. 
	Furthermore, Theorem~3.1 and Proposition~3.14 of \cite{Iftimie_Mantiou_Purice:magnetic_psido:2006} imply that $\OpA\bigl ( (f-z) \magW \mathfrak{r}_m^{(-1)_B} \bigr )$ is a bijection on $L^2(\R^d)$, and thus $\mathfrak{r}_m \magW (f-z)^{(-1)_B} = \bigl ( (f-z) \magW \mathfrak{r}_m^{(-1)_B} \bigr )^{(-1)_B} \in \AB$. Hence, the assumptions of Proposition~\ref{psiDO_reloaded:inversion_affiliation:prop:inversion_Hoermander_m} are satisfied and we conclude $(f - z)^{(-1)_B} \in \Hoermrdan{-m}{\rho}{0}$ for all $z \in \C \setminus \R$. 
\end{proof}
The preceding theorem implies we can formulate a principle of affiliation: 
\begin{thm}[Principle of affiliation\index{affiliation}]\label{psiDO_reloaded:inversion_affiliation:thm:affiliation}
	For $m > 0$ and $0 \leq \rho \leq 1$, any real-valued elliptic element $f \in \Hoermrdan{m}{\rho}{0}$ defines an observable affiliated to the $C^*$-algebra $\sFXprodR$.
\end{thm}
\begin{proof}
	For $z \in \C\setminus\R$, let us set $r_z:=(\cdot-z)^{-1}$. We also define $\Phi(r_z) := (f-z)^{(-1)_B}$. We first prove that the family $\{\Phi(r_z) \; \vert \; z \in \C \setminus \R \}$ satisfies the resolvent equation. Indeed, for any two complex numbers $z , z' \in \C \setminus \R$, one has
	\begin{align*}
		(f-z)\magW \Phi(r_z)=1 \qquad \mbox{ and }\qquad (f-z')\magW \Phi(r_{z'})=1.
	\end{align*}
	By subtraction, one obtains $(f-z) \magW \bigl ( \Phi(r_z) - \Phi(r_{z'}) \bigr ) + (z' - z) \, \Phi(r_{z'}) = 0$. By multiplying with $\Phi(r_z)$ on the left and using the associativity, one then gets the resolvent equation
	\begin{align*}
		\Phi(r_z) - \Phi(r_{z'}) = (z - z') \, \Phi(r_z) \magW \Phi(r_{z'}) 
		. 
	\end{align*}
	We have thus obtained a map $\C \setminus \R \ni z \mapsto \Phi(r_z) \in \Hoermrdan{-m}{\rho}{0} \subset \sFXprodR$ due to Theorem~\ref{psiDO_reloaded:inversion_affiliation:thm:inversion} and Proposition \ref{psiDO_reloaded:relevant_cStar_algebras:prop:Sminus_in_fxb}. Furthermore, the relation $\Phi(r_z)^{\magW} \equiv \Phi(r_z)^* = \Phi(r_{{z}^*})$ clearly holds. A general argument presented in \cite[p.~364]{Amrein_Boutet_Georgescu:C0_groups_commutator_methods:1996} that once more relies on the Stone-Weierstrass theorem allows to extend the map $\Phi$ to a $C^*$-algebra morphism $\Cont_{\infty}(\R) \longrightarrow \sFXprodR$ in a unique way.
\end{proof}
%

\section{Spectral properties} 
\label{psiDO_reloaded:spectral}
%
The basis for our spectral analysis lies in the $C^*$-algebraic point of view: the starting point is a principle of affiliation for an observable and Chapter~\ref{algebraicPOV:affiliation} gives a brief overview of what is in our toolbox. Some of these tools have no Hilbert space meaning -- which is the main advantage of this point of view. In $C^*$-algebraic terms, the spectrum of an observable $H$ affiliated to a $C^*$-algebra $\Calg$ is defined as 
\begin{align*}
	\sigma(H) := \Bigl \{ \lambda \in \R \; \big \vert \; \forall \varphi \in \Cont_{\infty}(\R) : \varphi(\lambda) \neq 0 \Rightarrow \varphi(H) \neq 0 \Bigr \} 
	. 
\end{align*}
Similarly, if $\Ideal$ is a closed, two-sided idea in $\Calg$, then the $\Ideal$-essential spectrum can be seen to be 
\begin{align*}
	\sigma_{\Ideal}(H) := \sigma \bigl ( \pi_{\Ideal}(H) \bigr ) = \Bigl \{ \lambda \in \R \; \big \vert \; \forall \varphi \in \Cont_{\infty}(\R) : \varphi(\lambda) \neq 0 \Rightarrow \varphi(H) \not\in \Ideal \Bigr \}
	. 
\end{align*}
Of particular interest is $\Cont_{\infty}(\R^d) \rtimes^{\omega^B}_{\theta,\nicefrac{1}{2}} \R^d$ here -- or equivalently, its Fourier transform. In this case, the Schrödinger representation of this algebra is equal to the compact operators $\mathcal{K} \bigl ( L^2(\R^d) \bigr )$ (Theorem~\ref{algebraicPOV:twisted_crossed_products:prop:basic_properties_rep}), and taking the quotient 
\begin{align*}
	\bigl ( {\Alg \rtimes^{\omega^B}_{\theta,\nicefrac{1}{2}} \R^d} \bigr ) / \bigl ( {\Cont_{\infty}(\R^d) \rtimes^{\omega^B}_{\theta,\nicefrac{1}{2}} \R^d} \bigr ) \cong \bigl ( {\Alg} / {\Cont_{\infty}(\R^d)} \bigr ) \rtimes^{\omega^B}_{\theta,\nicefrac{1}{2}} \R^d
\end{align*}
is in some sense equivalent to considering a subalgebra of the so-called Calkin algebra ${\BopL} /$ ${\mathcal{K} \bigl ( L^2(\R^d) \bigr )}$. Intuitively, it is also clear that the essential spectrum is due to the `behavior at infinity' of the operator; in the $C^*$-algebraic framework, we can define these operators `located at infinity' rigorously as quantizations of asymptotic functions which live on quasi orbits. The work is based on the ideas of Măntoiu \cite{Mantoiu:Cstar_algebras_dynamical_systems_at_infinity:2002} as well as Georgescu and Iftimovici \cite{Georgescu_Iftimovici:Cstar_algebras_quantum_hamiltonians:2003}.

\subsection{Families of ideals} 
\label{psiDO_reloaded:spectral:ideals}
The basic idea is to replace the ideal $\Ideal$ with a family of ideals $\{ \Ideal_i \}_{i \in \Index}$ whose intersection is $\Ideal$. The next few bits have been taken from \cite{Mantoiu:Cstar_algebras_dynamical_systems_at_infinity:2002}. 
\begin{defn}[$\Ideal$-sufficient family of ideals\index{$\Ideal$-sufficient family of ideals}]
	Let $\Calg$ be a $C^*$-algebra and $\Ideal$ and ideal of $\Calg$. A family of ideals $\{ \Ideal_i \}_{i \in \Index}$ will be called $\Ideal$-sufficient if 
	\begin{enumerate}[(i)]
		\item $\bigcap_{i \in \Index} \Ideal_i = \Ideal$, 
		\item $\Ideal_i$ is not contained in $\Ideal_j$ if $i \neq j$. 
	\end{enumerate}
	When $\Ideal = \{ 0 \}$, instead of $\{ 0 \}$-sufficient, we simply say sufficient. 
\end{defn}
The next lemma relates an $\Ideal$-sufficient family of ideals $\{ \Ideal_i \}_{i \in \Index}$ with a sufficient family of ideals in ${\Calg} / {\Ideal}$. 
\begin{lem}
	Let $\{ \Ideal_i \}_{i \in \Index}$ be a $\Ideal$-sufficient family of ideals in $\Calg$. Then $\mathfrak{K}_i := {\Ideal_i} / {\Ideal}$ is an ideal in the quotient ${\Calg} / {\Ideal}$ and $\{ \mathfrak{K}_i \}_{i \in \Index}$ is a sufficient family of ideals in ${\Calg} / {\Ideal}$. 
\end{lem}
The next result is the \emph{key} to decomposing the $\Ideal$-essential spectrum of an observable affiliated to $\Calg$ as the closure of the union of essential spectra with respect to larger ideals. 
\begin{prop}\label{psiDO_reloaded:spectral:prop:ideal_essential_spectrum}
	\begin{enumerate}[(i)]
		\item Let $\{ \Ideal_i \}_{i \in \Index}$ be a $\Ideal$-sufficient family of ideals in $\Calg$. Then there is a canonical monomorphism 
		\begin{align*}
			{\Calg} / {\Ideal} \hookrightarrow \prod_{i \in \Index} {\Calg} / {\Ideal_i} 
			. 
		\end{align*}
		\item Let $H$ be an observable affiliated to $\Calg$. For any $i \in \Index$, let us set $H_i := \pi_i(H)$ the image of $H$ by the canonical surjection $\pi_i : \Calg \longrightarrow {\Calg} / {\Ideal_i}$. Then 
		\begin{align*}
			\sigma_{\Ideal}(H) = \overline{\bigcup_{i \in \Index} \sigma(H_i)} 
			. 
		\end{align*}
	\end{enumerate}
\end{prop}
\begin{proof}
	\begin{enumerate}[(i)]
		\item The kernel of $\pi_i$ is $\Ideal_i$. Thus, the kernel of $( \pi_i )_{i \in \Index} : \Calg \longrightarrow \prod_{i \in \Index} {\Calg} / {\Ideal_i}$ is $\Ideal$. 
		\item This follows from (i), the invariance of the spectrum under monomorphisms (Theorem~\ref{algebraicPOV:affiliation:thm:spectrum_of_image_of_H_through_pi}) and from the fact that the spectrum of an observable affiliated to a direct product is the closure of the union of spectra of its components (which follows from an adaption of Proposition~\ref{algebraicPOV:affiliation:prop:fiber_decomposition_of_spectrum}). 
	\end{enumerate}
\end{proof}
%


\subsection{Decomposition of $\Salg$ into quasi orbits} 
\label{psiDO_reloaded:spectral:decomposition_of_salg_into_quasi_orbits}
In our case, the ideal $\Ideal$ will be the twisted crossed product with respect to the algebra $\Cont_{\infty}(\R^d)$. Writing $\Cont_{\infty}(\R^d)$ as the intersection of algebras $\Alg^{\Quasi}$ consisting of functions that do not vanish everywhere at infinity but only in `certain directions,' we can decompose $\Cont_{\infty}(\R^d) \rtimes^{\omega^B}_{\theta,\nicefrac{1}{2}} \R^d$ as an intersection of twisted crossed products of the form $\Alg^{\Quasi} \rtimes^{\omega^B}_{\theta,\nicefrac{1}{2}} \R^d$. Let us define some preliminaries: for any commutative algebra $\Alg$, the Gelfand isomorphism $\Gelf : \Alg \longrightarrow \Cont_{\infty}(\Salg)$ maps any $\varphi \in \Alg$ onto a function $\tilde{\varphi} := \Gelf(\varphi)$ on the Gelfand spectrum $\Salg$ (see Chapter~\ref{algebraicPOV:twisted_crossed_products:gelfand_theory}). If $\Alg$ is an $\R^d$-algebra which contains $\Cont_{\infty}(\R^d)$, the Stone-Weierstrass theorem tells us that we can continuously and densely embed $\R^d$ into $\Salg$. This means, we can think of $\Gelf(\varphi)$ as the unique `extension' of $\varphi$, and for that reason, we will often not distinguish between these two objects. If $\Alg$ is also unital, then the Gelfand spectrum is compact and $\Salg$ is a \emph{compactification} of $\R^d$ and $\Salg \setminus \imath_{\Alg}(\R^d)$ are the points at infinity. 
\begin{example}
	Let $\Alg := \mathrm{span} \; \bigl \{ \varphi_0 + \varphi_1 \; \vert \; \varphi_0 \in \Cont_{\infty}(\R^d) , \; \varphi_1 \in \C \bigr \}$ be the algebra consisting of function that tend to a constant at infinity. Then the Gelfand spectrum $\Salg \cong \R^d \cup \{ \infty \}$ is the one-point compactification of $\R^d$. 
\end{example}
Viewing $\Salg$ as a compactification of $\R^d$, we extend the group law $\tau : \R^d \times \R^d \longrightarrow \R^d$ to a continuous map $\tau : \R^d \times \Salg \longrightarrow \Salg$ (which we denote by the same letter) and the triple $(\Salg , \tau , \R^d)$ is a topological dynamical system. Now $\Salg$ decomposes into orbits and quasi orbits: 
\begin{defn}
	Let $\kappa,\kappa'$ be two elements of $\Salg$. We set $\Orbit_{\kappa} := \bigl \{ \tau_x[\kappa] \mid x \in \R^d \bigr \}$ for \emph{the orbit of} $\kappa$\index{orbit} and $\Quasi_{\kappa}:=\overline{\Orbit_{\kappa}}$ for \emph{the quasi orbit of} $\kappa$,\index{quasi orbit} which is the closure of $\Orbit_{\kappa}$ in $\Salg$. $\Quasialg$ stands for the set of all the quasi orbits. Furthermore, we write
	\begin{enumerate}[(i)]
		\item $\kappa \cong \kappa'$ if $\Orbit_{\kappa} = \Orbit_{\kappa'}$,
		\item $\kappa \prec \kappa'$ if $\Quasi_{\kappa} \subseteq \Quasi_{\kappa'}$,
		\item $\kappa \sim \kappa'$ if $\kappa \prec \kappa'$ and $\kappa' \prec \kappa$, which is equivalent to $\Quasi_{\kappa} = \Quasi_{\kappa'}$.
	\end{enumerate}
\end{defn}
Similarly, we can now show under which circumstances representations are equivalent: 
\begin{defn}
	Let $\Calg$ be a $C^*$-algebra, and for $j \in \{1,2\}$ let $\pi_j :\Calg \longrightarrow \mathcal{B}(\Hil_j)$ be a representation in a Hilbert space $\Hil_j$. We write:
	\begin{enumerate}[(i)]
		\item $\pi_1 \prec \pi_2$ if $\ker \pi_1 \supset \ker \pi_2$,
		\item $\pi_1 \sim \pi_2$ if $\ker \pi_1 = \ker \pi_2$.
	\end{enumerate}
\end{defn}
If $\pi_1 \prec \pi_2$, then the spectrum of the image of an observable $H$ through $\pi_1$ affiliated to $\mathcal{B}(\Hil_1)$ is contained in that of the image through $\pi_2$, \ie if the kernel grows, the spectrum tends to shrink. 
\begin{lem}\label{psiDO_reloaded:spectral:lem:shrinking_spectrum}
	Let $H$ be an observable affiliated to a $C^*$-algebra $\Calg$. If $\pi_1$, $\pi_2$ are two representations of $\Calg$ in some Hilbert spaces $\Hil_j$ satisfying $\pi_1 \prec \pi_2$, then $\sigma \bigl (\pi_1 (H) \bigr ) \subseteq \sigma \bigl (\pi_2 (H) \bigr )$. In particular, if $\pi_1 \sim \pi_2$, then $\sigma \bigl (\pi_1 (H) \bigr ) = \sigma \bigl (\pi_2 (H) \bigr )$.
\end{lem}
\begin{proof}
	For $j \in \{1,2\}$, let $p_j$ denote the morphisms $\Calg \longrightarrow \Calg \setminus \ker \pi_j$, and let $\tilde{\pi}_j$ denote the isomorphisms $\bigl (\Calg \setminus \ker \pi_j \bigr ) \longrightarrow \pi_j(\Calg) \subseteq \mathcal{B}(\Hil_j)$. One clearly has $\tilde{\pi}_j \circ p_j = \pi_j$. Since $\ker \pi_1 \supseteq \ker \pi_2$, there also exists a surjective morphism $m: \Calg \setminus \ker \pi_2 \longrightarrow \Calg \setminus \ker \pi_1$ which satisfies $m\circ p_2 = p_1$. Then, by setting $M:\pi_2(\Calg) \longrightarrow \pi_1(\Calg)$ by $M:=\tilde{\pi}_1 \circ m \circ \tilde{\pi}_2^{-1}$, it follows that $M$ is a surjective morphism and
	\begin{align*}
		M \circ \pi_2= (\tilde{\pi}_1 \circ m \circ \tilde{\pi}^{-1}_2) \circ (\tilde{\pi}_2 \circ p_2)  
		= \tilde{\pi}_1 \circ m \circ p_2 
		= \tilde{\pi}_1\circ p_1 = \pi_1 
		.
	\end{align*}
	At the level of spectra, this equality implies that 
	\begin{align*}
		\sigma \bigl (\pi_1 (H) \bigr ) = \sigma \bigl ( M(\pi_2 (H)) \bigr ) \subseteq \sigma \bigl (\pi_2 (H) \bigr ) 
		. 
	\end{align*}
\end{proof}
Let us consider a specific case: Let $\kappa \in \Salg$, $\varphi \in \Alg \cong \Cont_{\infty}(\Salg)$ and $u \in L^2(\R^d)$. Then 
\begin{align*}
	\bigl ( r_{\kappa}(\varphi) u \bigr )(x) := \bigl ( \theta_x[\varphi] \bigr )(\kappa) \, u(x) 
	= \varphi \bigl ( \tau_x[\kappa] \bigr ) \, u(x) =: \varphi \bigl ( \tau^{\kappa}[x] \bigr ) \, u(x) 
\end{align*}
and 
\begin{align*}
	\bigl ( T_{\kappa}(y) u \bigr )(x) := \omega^B(\kappa ; x , y) \, u(x + y) 
\end{align*}
form a covariant representation of the twisted $C^*$-dynamical system. 
\begin{defn}[Representation by evaluation\index{representation!by evaluation}]
	For $\kappa \in \Salg$, $F \in L^1(\R^d ; \Alg)$ and $f \in \Fourier^{-1} L^1(\R^d ; \Alg)$, we define 
	\begin{align}
		\bigl ( \Rep_{\kappa}(F) u \bigr )(x) :&= \int_{\R^d} \dd y \, F \bigl ( \tau^{\kappa} \bigl ( \tfrac{1}{2} (x + y) \bigr ) ; y - x \bigr ) \, \omega^B(\kappa ; x , y - x) \, u(y) 
	\end{align}
	as well as 
	\begin{align}
		\bigl ( \Op_{\kappa}(f) u \bigr )(x) := \int_{\R^d} \dd y \int_{{\R^d}^*} \dd \eta \, e^{-i (y - x) \cdot \eta} f \bigl ( \tau^{\kappa} \bigl ( \tfrac{1}{2} (x + y) \bigr ) ; \eta \bigr ) \, \omega^B(\kappa ; x , y - x) \, u(y) 
	\end{align}
	for any $u \in L^2(\R^d)$. These definitions extend to the twisted crossed product $\Alg \rtimes^{\omega^B}_{\theta , \nicefrac{1}{2}} \R^d$ and its Fourier transform, respectively. 
\end{defn}
%
Before we proceed, we need a technical lemma: 
\begin{lem}\label{psiDO_reloaded:spectral:lem:strong_continuity_of_Rep_kappa}
	For any $F \in \Alg \rtimes^{\omega^B}_{\theta , \nicefrac{1}{2}} \R^d$, the mapping $\Salg \ni \kappa \mapsto \Rep_{\kappa}(F) \in \BopL$ is strongly continuous.
\end{lem}
In the next proof and again further on, we shall denote by $\Cont_c(\mathcal{Y};\mathcal{Z})$ the set of continuous functions on $\mathcal{Y}$ with compact support and values in $\mathcal{Z}$; the topologies will be evident. We shall also drop the tilde in the notations, identifying elements of $\Alg$ with their isomorphic image in $\Cont_{\infty}(\Salg)$. 
%
\begin{proof}
	We prove the statement for $F \in L^1(\R^d ; \Alg) \subset \Alg \rtimes^{\omega^B}_{\theta , \nicefrac{1}{2}} \R^d$ and $u \in L^2(\R^d)$ and then use a density argument. Let $\kappa$ be an element of $\Salg$ and initially assume that $F \in L^1(\R^d ; \Alg) \subset \Alg \rtimes^{\omega^B}_{\theta , \nicefrac{1}{2}} \R^d$. For all $u \in L^2(\R^d)$ and almost all $x,y \in \R^d$, we can bound 
	\begin{align*}
		&\abs{F \bigl ( \tau^{\kappa} \bigl ( \tfrac{1}{2} (x + y) \bigr ) ; y - x \bigr ) \, \omega^B(\kappa ; x , y - x) \, u(y)} 
		\leq 
		\\
		&\qquad \qquad 
		\leq \babs{F \bigl ( \theta^{\kappa} \bigl ( \tfrac{1}{2} (x + y) \bigr ) ; y - x \bigr )} \, \babs{\omega^B(\kappa ; x , y - x)} \, \babs{u(y)} 
		\\
		&\qquad \qquad 
		\leq \bnorm{F(y - x)}_{\Alg} \, \babs{u(y)} 
	\end{align*}
	point-wise by something which is in $L^2(\R^d)$ and independent of $\kappa$ as it is the convolution of the $L^1(\R^d)$ function $\bnorm{F(- \, \cdot)}_{\Alg}$ and the $L^2(\R^d)$ function $u$. Thus, we can use Dominated Convergence to conclude that at least for $F \in L^1(\R^d ; \Alg)$, the map $\kappa \mapsto \Rep_{\kappa}(F)$ is strongly continuous. 
	
	If $F \in \Alg \rtimes^{\omega^B}_{\theta,\nicefrac{1}{2}} \R^d$ is an arbitrary element of the $C^*$-algebra, we can approximate it by a sequence $\{ F_n \}_{n \in \N}$ in $L^1(\R^d ; \Alg)$. Then the uniform bound 
	\begin{align*}
		\bnorm{\Rep_{\kappa}(F)}_{\mathcal{B}(L^2(\R^d))} \leq \bnorm{F}_{\stwistedXprod} 
	\end{align*}
	can be combined with an $\nicefrac{\eps}{3}$ argument to make the right-hand side of 
	\begin{align*}
		&\bnorm{\bigl ( \Rep_{\kappa}(F) - \Rep_{\kappa'}(F) \bigr ) u}_{\mathcal{B}(L^2(\R^d))} 
		\leq \\
		&\quad \leq \bnorm{\bigl ( \Rep_{\kappa}(F) - \Rep_{\kappa}(F_n) \bigr ) u}_{\mathcal{B}(L^2(\R^d))} 
		+ \bnorm{\bigl ( \Rep_{\kappa}(F_n) - \Rep_{\kappa'}(F_n) \bigr ) u}_{\mathcal{B}(L^2(\R^d))} 
		\\
		&\qquad \medskip
		+ \bnorm{\bigl ( \Rep_{\kappa'}(F_n) - \Rep_{\kappa'}(F) \bigr ) u}_{\mathcal{B}(L^2(\R^d))} 
		\\
		&\quad \leq 2 \bnorm{F - F_n}_{\stwistedXprod} \, \snorm{u}_{L^2(|R^d)} + \bnorm{\bigl ( \Rep_{\kappa}(F_n) - \Rep_{\kappa'}(F_n) \bigr ) u}_{\mathcal{B}(L^2(\R^d))} 
	\end{align*}
	arbitrarily small if we choose $n$ large enough and let $\kappa'$ tend to $\kappa$. 
\end{proof}
The representations associated to different $\kappa$ are related if they belong to the same orbit or quasi orbit. 
\begin{prop}\label{psiDO_reloaded:spectral:prop:relation_reps_by_eval}
	Let $\kappa , \kappa'$ be two elements of $\Salg$.
	\begin{enumerate}[(i)]
		\item If $\kappa \cong \kappa'$, then $\Rep_{\kappa} \cong \Rep_{\kappa'}$
		and $\Op_{\kappa} \cong \Op_{\kappa'}$.
		\item If $\kappa \prec \kappa'$, then $\Rep_{\kappa} \prec \Rep_{\kappa'}$ and $\Op_{\kappa} \prec \Op_{\kappa'}$.
		\item If $\kappa \sim \kappa'$, then $\Rep_{\kappa} \sim \Rep_{\kappa'}$ and $\Op_{\kappa} \sim \Op_{\kappa'}$.
	\end{enumerate}
\end{prop}
\begin{proof}
	We prove the statements for $\Rep_{\kappa}$ only, the corresponding ones for $\Op_{\kappa}$ follow
	from the fact $\sXprodR$ and $\sFXprodR$ are isomorphic.
	\begin{enumerate}[(i)]
		\item Since $\Orbit_{\kappa} =\Orbit_{\kappa'}$, there exists an element $x_0$ of $\R^d$ such that $\tau_{x_0}[\kappa'] = \kappa$. For $u \in L^2(\R^d)$ and $x \in \R^d$ we set
		\begin{align*}
			\bigl ( U_{\kappa \kappa'} u \bigr )(x):= {\omega^B(\kappa; -x_0 , x + x_0)}^* \, u(x + x_0) 
			.
		\end{align*}
		To show unitary equivalence of the two representations, it is enough to show that for all $\varphi \in \Alg$ and $y \in \R^d$ one has
		\begin{align*}
			U_{\kappa \kappa'} \, r_{\kappa'}(\varphi) = r_{\kappa}(\varphi) \, U_{\kappa \kappa'} 
			\quad \hbox{ and }\quad 
			U_{\kappa \kappa'} \, T_{\kappa'}(y) = T_{\kappa}(y) \, U_{\kappa \kappa'} 
			.
		\end{align*}
		The first one is obvious. The second one reduces to
		\begin{align*}
			\omega^B(\kappa; -x_0 , x + x_0) \, \omega^B(\kappa;x,y) = \omega^B \bigl (\tau_{-x_0}[\kappa] ; x + x_0 , y \bigr ) \, \omega^B(\kappa; -x_0 , x + y + x_0),
		\end{align*}
		which is true by the $2$-cocycle property of $\omega^B$. 
		\item If $\kappa \prec \kappa'$, there exists a sequence $(x_m)_{m\in\N}$ in $\R^d$ such that $\theta_{x_m}[\kappa'] \rightarrow \kappa$ in the larger quasi orbit $\Quasi_{\kappa'} \supseteq \Quasi_{\kappa}$. For $F \in \ker \Rep_{\kappa'}$, by point~(i), the continuity of $\theta$ and  Lemma~\ref{psiDO_reloaded:spectral:lem:strong_continuity_of_Rep_kappa}, one obtains
		\begin{align*}
			0 = \Rep_{\theta_{x_m}[\kappa']}(F) \xrightarrow{m \rightarrow \infty} \Rep_{\kappa}(F) = 0 
			. 
		\end{align*}
		Hence, $F$ is contained in the kernel of $\Rep_{\kappa}$. 
		\item is a direct consequence of (ii). 
	\end{enumerate}
\end{proof}
If $\Quasi \in \Quasialg$ is a quasi orbit, then there are two natural spaces associated to it: if $\kappa$ generates $\Quasi$ (\ie the set $\Orbit_{\kappa} := \bigl \{ \tau_x[\kappa] \mid x \in \R^d \bigr \}$ is dense in $\Quasi$), we define 
\begin{align*}
	\Alg_{\kappa} := \bigl \{ \varphi_{\kappa} := x \mapsto \tilde{\varphi} \bigl ( \tau^{\kappa}(x) \bigr ) \mid \tilde{\varphi} \in \Cont(\Salg) \bigr \} 
	. 
\end{align*}
It can be easily seen that if $\kappa'$ generates the same quasi orbit $\Quasi$, then $\Alg_{\kappa}$ and $\Alg_{\kappa'}$ are isomorphic. With a little abuse of notation, we will exploit this fact and write $\Alg_{\Quasi}$ instead. From the definition, we conclude that $\Alg_{\Quasi} \cong \Cont(\Quasi)$ is an $\R^d$-algebra as well: any $\varphi_{\kappa}$ belongs to $\BCont_u(\R^d)$. By restricting $\tilde{\varphi} \in \Cont(\Salg)$ to $\Quasi \subseteq \Salg$, we get a canonical epimorphism 
\begin{align*}
	\mathsf{p}_{\Quasi} : \Cont(\Salg) \longrightarrow \Cont(\Quasi) 
\end{align*}
and this map is needed to define a second algebra associated to each quasi orbit $\Quasi$. If $\kappa$ generates $\Quasi$, then 
\begin{align*}
	\pi_{\kappa} := \tau^{\kappa} \circ \mathsf{p}_{\Quasi} \circ \Gelf 
\end{align*}
maps $\Alg$ onto $\Alg_{\kappa}$. Note that $\Alg_{\kappa}$ has \emph{no reason to be contained in $\Alg$} in any way. The \emph{kernel} of $\pi_{\kappa}$, 
\begin{align*}
	\Alg^{\Quasi} := \bigl \{ \varphi \in \Alg \mid \Gelf(\varphi) \vert_{\Quasi} = 0 \bigr \} 
	, 
\end{align*}
is then the second relevant algebra. As the notation already suggests, $\Alg^{\Quasi}$ is independent of the choice of $\kappa$ for as long as it generates $\Quasi$; hence, by slight abuse of notation we will write $\pi_{\Quasi}$ instead of $\pi_{\kappa}$. All of these objects are related via a family of short exact sequences 
\begin{align}
	0 \longrightarrow \Alg_{\Quasi} \longrightarrow \Alg \overset{\pi_{\Quasi}}{\longrightarrow} \Alg^{\Quasi} \longrightarrow 0 
	\label{psiDO_reloaded:spectral:decomposition_of_salg_into_quasi_orbits:eqn:short_exact_sequence}
\end{align}
indexed by $\Quasi \in \Quasialg$, the set of all quasi orbits of the topological dynamical system $(\Salg , \tau , \R^d)$. This topological dynamical system is called \emph{minimal} if all orbits are dense, \ie if there is only a single quasi orbit. 

The previously proven principle of affiliation of real-valued, elliptic, anisotropic Hörmander symbols combined with these abstract arguments leads to 
\begin{thm}
	Let $f$ be any real and elliptic element of $\Hoermrdan{m}{\rho}{0}$ for $m >0$ and $\rho \in [0,1]$. Assume that each component $B_{jk}$ of the magnetic field belongs to $\Alg^{\infty}$, and let $A \in \Cont^{\infty}_{\mathrm{pol}}(\R^d)$ be a vector potential for $B$.
	\begin{enumerate}[(i)]
		\item If $\kappa \prec \kappa'$ then $\sigma \bigl ( \Op_\kappa(f) \bigr ) \subseteq \sigma \bigl ( \Op_{\kappa'}(f) \bigr ) \subseteq \sigma \bigl ( \Op^A(f) \bigr ) = \sigma(f)$, where $\sigma(f)$ is the spectrum of $f$ in $\sFXprod$.
		\item If the dynamical system $(\Salg,\tau,\R^d)$ is minimal, then all the operators $\Op_\kappa(f)$ have the same spectrum, which coincides with $\sigma(f)$.
	\end{enumerate}
\end{thm}
\begin{proof}
	\begin{enumerate}[(i)]
		\item Use Proposition~\ref{psiDO_reloaded:spectral:prop:relation_reps_by_eval}, Lemma~\ref{psiDO_reloaded:spectral:lem:shrinking_spectrum}, Corollary~\ref{psiDO_reloaded:inversion_affiliation:thm:affiliation} and the faithfulness of the representation $\Op^A$.
		\item follows directly from (i), since a minimal dynamical system has a single quasi-orbit.
	\end{enumerate}
\end{proof}
%


\subsection{The essential spectrum of anisotropic magnetic operators}
\label{psiDO_reloaded:spectral:ess_spec}
Now that all the necessary notions have been introduced, we turn to the spectral analysis of magnetic pseudodifferential operators: our goal is to rewrite the quotient 
\begin{align*}
	\bigl ( {\Alg \rtimes^{\omega^B}_{\theta,\nicefrac{1}{2}} \R^d} \bigr ) / \bigl ( {\Cont_{\infty}(\R^d) \rtimes^{\omega^B}_{\theta,\nicefrac{1}{2}} \R^d} \bigr ) \cong \bigl ( {\Alg} / {\Cont_{\infty}(\R^d)} \bigr ) \rtimes^{\omega^B}_{\theta,\nicefrac{1}{2}} \R^d 
\end{align*}
as the intersection of larger ideals, \ie the twisted crossed products of a $\Cont_{\infty}(\R^d)$-sufficient family of ideals. By throwing away superfluous contributions, we can extract a $\Cont_{\infty}(\R^d)$-sufficient family of ideals from $\{ \Alg^{\Quasi} \}_{\Quasi \in \Quasialg}$. The decomposition of the Gelfand spectrum $\Salg$ into quasi orbits enables us to make the intuitive statement `the essential spectrum\index{spectrum!essential} depends on the behavior of the operator at infinity' into a theorem: for magnetic pseudodifferential operators, the behavior at infinity can be described by the the behavior of the algebra elements on 
\begin{align*}
	\SalgComp := \Salg \setminus \imath_{\Alg}(\R^d) 
\end{align*}
which are \emph{the points at infinity}. This closed set is left invariant by the $\R^d$-action and it inherits the compactness from $\Salg$. Furthermore, we can choose a covering of $\SalgComp$ by a collection $\Quasicover$ of maximal quasi orbits so that $\{ \Alg^{\Quasi} \}_{\Quasi \in \Quasicover}$ is a $\Cont_{\infty}(\R^d)$-sufficient family of ideals\index{$\Ideal$-sufficient family of ideals}, 
\begin{align*}
	\bigcap_{\Quasi \in \Quasicover} \Alg^{\Quasi} = \Cont_{\infty}(\R^d)
	. 
\end{align*}
This gives us a collection of abelian twisted $C^*$-dynamical systems $\bigl ( \Alg^{\Quasi} , \R^d , \theta , \omega^B \bigr )$ and a family of twisted crossed products 
\begin{align*}
	\Alg^{\Quasi} \rtimes^{\omega^B}_{\theta,\nicefrac{1}{2}} \R^d 
\end{align*}
whose intersection is the ideal of interest: 
\begin{align*}
	\bigcap_{\Quasi \in \Quasicover} \Alg^{\Quasi} \rtimes^{\omega^B}_{\theta,\nicefrac{1}{2}} \R^d = \Cont_{\infty}(\R^d) \rtimes^{\omega^B}_{\theta,\nicefrac{1}{2}} \R^d 
\end{align*}
The quotient of the original crossed product $\Alg \rtimes^{\omega^B}_{\theta,\nicefrac{1}{2}} \R^d$ with $\Alg^{\Quasi} \rtimes^{\omega^B}_{\theta,\nicefrac{1}{2}} \R^d$ can be canonically identified with 
\begin{align*}
	\bigl ( \Alg \rtimes^{\omega^B}_{\theta,\nicefrac{1}{2}} \R^d \bigr ) / \bigl ( \Alg^{\Quasi} \rtimes^{\omega^B}_{\theta,\nicefrac{1}{2}} \R^d \bigr ) \cong \bigl ( \Alg / \Alg^{\Quasi} \bigr ) \rtimes^{\omega^{B_{\Quasi}}}_{\theta,\nicefrac{1}{2}} \R^d
	&\cong \Cont(\Quasi) \rtimes^{\omega^{B_{\Quasi}}}_{\theta , \nicefrac{1}{2}} \R^d 
	\\
	&
	\cong \Alg_{\Quasi} \rtimes^{\omega^{B_{\Quasi}}}_{\theta , \nicefrac{1}{2}} \R^d 
\end{align*}
where the $2$-cocycle 
\begin{align*}
	\omega^{B_{\Quasi}} := \pi_{\Quasi} \bigl ( \Gelf(\omega^B) \bigr ) = \omega^{\pi_{\Quasi}(\Gelf(B))} 
\end{align*}
is determined by the asymptotic behavior of $B$ on that particular quasi orbit at infinity. If $\Alg$ is separable, this follows from \cite[Proposition~2.2]{Packer_Raeburn:twisted_X_products_2:1990}; the non-separable case can be proven by retracing the arguments in \cite{Georgescu_Iftimovici:Cstar_algebras_quantum_hamiltonians:2003} for the twisted case. Furthermore, $\Alg / \Alg^{\Quasi}$ can be identified with $\Alg_{\Quasi} \cong \Cont(\Quasi)$. 

We are in a position now to state the main theorems of this section: 
\begin{prop}\label{psiDO_reloaded:spectral:ess_spec:prop:ess_spec}
	Let $\Quasicover \subseteq \Quasialg$ be a covering of $\SalgComp$.
	\begin{enumerate}[(i)]
		\item There exists an injective morphism 
		\begin{align*}
			\bigl ( {\Alg \rtimes^{\omega^B}_{\theta,\nicefrac{1}{2}} \R^d} \bigr ) / \bigl ( {\Cont_{\infty}(\R^d) \rtimes^{\omega^B}_{\theta,\nicefrac{1}{2}} \R^d} \bigr ) \hookrightarrow \prod_{\Quasi \in \Quasicover} \Alg_{\Quasi} \rtimes^{\omega^{B_{\Quasi}}}_{\theta,\nicefrac{1}{2}} \R^d  
			. 
		\end{align*}
		\item If $H$ is an observable affiliated to $\Alg \rtimes^{\omega^B}_{\theta,\nicefrac{1}{2}} \R^d$ and $\Ideal := \Cont_{\infty}(\R^d) \rtimes^{\omega^B}_{\theta,\nicefrac{1}{2}} \R^d$, then we have
		\begin{align}
			\sigma_{\Ideal}(H) = \overline{\bigcup_{\Quasi \in \Quasicover} \sigma \bigl ( \pi^{\rtimes}_{\Quasi}(H) \bigr )} 
			\label{psiDO_reloaded:spectral:ess_spec:eqn:ess_spec}
		\end{align}
		where $\pi^{\rtimes}_{\Quasi} : \Alg \rtimes^{\omega^B}_{\theta , \nicefrac{1}{2}} \R^d \longrightarrow \Alg_{\Quasi} \rtimes^{\omega^{B_{\Quasi}}}_{\theta , \nicefrac{1}{2}} \R^d$ is the projection induced by $\pi_{\Quasi} : \Alg \longrightarrow \Alg_{\Quasi}$ (equation~\eqref{psiDO_reloaded:spectral:decomposition_of_salg_into_quasi_orbits:eqn:short_exact_sequence}) that localizes algebra elements to the quasi orbit $\Quasi$. 
	\end{enumerate}
\end{prop}
\begin{proof}
	The proof follows immediately from Proposition~\ref{psiDO_reloaded:spectral:prop:ideal_essential_spectrum}. 
\end{proof}
If we apply the above proposition to magnetic pseudodifferential operators, we get the main result of this section: 
\begin{thm}\label{psiDO_reloaded:spectral:thm:essential_spectrum}
	Let $m > 0$, $0 \leq \rho \leq 1$ and let $\Quasicover \subset \Quasialg$ be a covering of $\SalgComp$. Then, for any real-valued elliptic element $f$ of $\Hoermrdan{m}{\rho}{0}$ one has
	\begin{align}
		\sigma_{\mathrm{ess}}(f) = \sigma_{\mathrm{ess}} \bigl ( \OpA(f) \bigr ) = \overline{\bigcup_{\Quasi \in \mathbf{Q}} \sigma \bigl ( \Op^{A_{\Quasi}} (f_{\Quasi}) \bigr )}
		, 
		\label{psiDO_reloaded:spectral:eqn:essential_spectrum}
	\end{align}
	%
	where $A$ and $A_{\Quasi}$ are continuous vector potentials for $B$ and $B_{\Quasi} := \pi_{\Quasi}(B)$, and the asymptotic symbol $f_{\Quasi} \in S^m_{\rho,0} \bigl ( {\R^d}^*;\Alg_{\Quasi}^{\infty} \bigr )$ is given as the image of $f$ through $\pi_{\Quasi}$.
\end{thm}
\begin{proof}
	In Proposition~\ref{algebraicPOV:twisted_crossed_products:prop:basic_properties_rep}, it was proven that $\OpA$ and $\Rep^A$ are faithful and irreducible as $\Cont_{\infty}(\R^d) \subseteq \Alg$. This means the essential spectrum of $\OpA(f)$ in the functional analytic sense coincides with $\sigma_{\mathrm{ess}}(f)$ in the algebraic sense. In particular, it does not depend on the choice of vector potential. 
	
	For any $g \in \Fourier^{-1} L^1(\R^d ; \Alg)$, the morphism 
	\begin{align*}
		g \mapsto \Fourier^{-1} \bigl ( \pi_{\Quasi}^{\rtimes} \bigl ( \Fourier g \bigr ) \bigr ) \in \Fourier^{-1} L^1(\R^d ; \Alg)
	\end{align*}
	extends to a morphism $\Fourier \pi_{\Quasi}^{\rtimes} : \sFXprodR \longrightarrow \Fourier \Calg^{B_{\Quasi}}_{\Alg_{\Quasi}}$. Equation~\eqref{psiDO_reloaded:spectral:eqn:essential_spectrum} can now be rewritten as 
	\begin{align*}
		\sigma_{\Ideal}(f) = \overline{\bigcup_{\Quasi \in \Quasicover} \sigma \bigl ( \Fourier \pi_{\Quasi}^{\rtimes}(f) \bigr )} 
	\end{align*}
	for any observable $f$ which is affilated to $\sFXprodR$. The ideal of interest $\Ideal \subset \sFXprodR$ in this case is the image of $\Cont_{\infty}(\R^d) \rtimes^{\omega^B}_{\theta , \nicefrac{1}{2}} \R^d$ through $\Fourier^{-1}$. Real-valued, elliptic, anisotropic symbols of positive order are affiliated to $\sFXprodR$ by Theorem~\ref{psiDO_reloaded:inversion_affiliation:thm:affiliation}. The algebras $\sFXprodR$ and $\Fourier \Calg^{B_{\Quasi}}_{\Alg_{\Quasi}}$ are faithfully represented by $\Op^A$ and $\Op^{A_{\Quasi}}$, respectively, and so we conclude 
	\begin{align*}
		\sigma_{\mathrm{ess}}(f) &\equiv \sigma_{\Ideal}(f) = \overline{\bigcup_{\Quasi \in \Quasicover} \sigma \bigl ( \Fourier \pi_{\Quasi}^{\rtimes}(f) \bigr )} 
		= \sigma_{\mathrm{ess}} \bigl ( \OpA(f) \bigr ) = \overline{\bigcup_{\Quasi \in \mathbf{Q}} \sigma \bigl ( \Op^{A_{\Quasi}} (f_{\Quasi}) \bigr )}
	\end{align*}
	where $f_{\Quasi}$ coincides with $\Fourier \pi_{\Quasi}^{\rtimes}(f)$. 
\end{proof}
This has an immediate consequence: if $\Cont_{\infty}(\R^d)$ is not contained in $\Alg$ (which is equivalent to saying $\Alg \cap \Cont_{\infty}(\R^d) = \{ 0 \}$), then the spectrum is purely essential. 
\begin{cor}
	If $\Cont_{\infty}(\R^d) \not\subseteq \Alg$, then observables $f$ affiliated to $\sFXprodR$ or $\sXprodR$, respectively, have purely essential spectrum, 
	\begin{align*}
		\sigma(f) = \sigma_{\mathrm{ess}}(f) 
		. 
	\end{align*}
\end{cor}
\begin{remark}
	It is now easy to compute the kernel of any of the representations $\Rep_{\kappa}$. Indeed, for $\Quasi = \Quasi_{\kappa}$, the representation $\Rep_{\kappa}$ can be written as the composition between the canonical projection
	\begin{align*}
		\Alg \rtimes^{\omega^B}_{\theta , \nicefrac{1}{2}} \R^d 
		\mapsto \Alg \rtimes^{\omega^B}_{\theta , \nicefrac{1}{2}} \R^d \big / \Alg^{\Quasi} \rtimes^{\omega^B}_{\theta , \nicefrac{1}{2}} \R^d 
		\cong \bigl ( \Alg / {\Alg^{\Quasi}} \bigr ) \rtimes^{\omega^B}_{\theta , \nicefrac{1}{2}} \R^d 
		\cong \Cont(\Quasi) \rtimes^{\omega^{B_{\Quasi}}}_{\theta , \nicefrac{1}{2}} \R^d
	\end{align*}
	and a Schrödinger representation defined by the transversal gauge. It follows that the kernel of $\Rep_{\kappa}$ is equal to $\Alg^{\Quasi} \rtimes^{\omega^B}_{\theta , \nicefrac{1}{2}} \R^d$. Theorem~\ref{psiDO_reloaded:spectral:thm:essential_spectrum} can then be rephrased as: Let $\{ \kappa_j \}_{j \in \Index}$ be a subset of $\SalgComp$ such that $\{ \Quasi_{\kappa_j} \}_{j \in \Index}$ is a covering of $\SalgComp$ by maximal quasi-orbits, then
	\begin{align} 
		\sigma_{\mathrm{ess}} \bigl ( \OpA(f) \bigr ) = \sigma_{\Ideal}(f)
		= \overline{\bigcup_{j \in \Index} \sigma \bigl ( \Op_{\kappa_j}(f) \bigr )} 
		. 
	\end{align}
\end{remark}
\begin{remark}
	Non-propagation results easily follow from this algebraic framework. They have been explicitly exhibited in the non-magnetic case in \cite{Amrein_Mantoiu_Purice:propagation_properties:2002} and in the magnetic case in \cite{Mantoiu_Purice_Richard:Cstar_algebraic_framework:2007}. In these references, the authors were mainly concerned with generalized Schrödinger operators and their results were stated for these operators. But the proof relies only on the $C^*$-algebraic framework, and the results extend \emph{mutatis mutandis} to the classes of symboles introduced in the present paper. For shortness, we do not present these propagation estimates here, but statements and proofs can easily be mimicked from these references.
\end{remark}
%

\chapter{Outlook} 
\label{further_research}
%
There are quite a few directions for further research and we would like to shortly discuss some of them in no particular order.

\section{Magnetic Weyl calculus for other gauge groups} 
\label{further_research:extension_other_gauge_groups}
Electromagnetism is a $U(1)$ gauge theory \cite{Marsden_Ratiu:intro_mechanics_symmetry:1999} and as such particularly simple: the gauge group is abelian and one-dimensional. If we replace $U(1)$ with a non-commu\-tative gauge group, what would the associated Weyl calculus look like? The simplest case is to consider $SU(2)$ or $SO(3)$ which may be helpful in the analysis of particles moving in $\R^d$ with additional internal degrees of freedom. This may offer an alternative point of view on molecular dynamics or serve as a toy model for non-commutative gauge theories such as dynamics on shape space \cite{Littlejohn_Reinsch:dynamics_shape_space:1997,Sattlegger:dynamics_shape_space:2010}.

\paragraph{Multiband molecular dynamics} 
\label{further_research:extension_other_gauge_groups:two_level_hamiltonians}
The semiclassical approximation \cite{Born_Oppenheimer:BO_approximation:1927} works tre\-mendously well in molecular dynamics calculations in theoretical chemistry. Here, the cores are modeled as classical particles which move in effective potentials generated by the electrons. Over the years, this approximation has been subject to a whole phletora of publications (\eg \cite{Spohn_Teufel:time-dep_BO_approx:2001,Martinez_Sordoni:Born-Oppenheimer:2002,Sordoni:Born-Oppenheimer:2003,PST:sapt:2002,PST:Born-Oppenheimer:2007}) and the mechanism, adiabatic decoupling, has been well-understood. One of the remaining problems is to figure out how to mix classical dynamics with quantum mechanics in the event of band crossings. Away from band crossings, classical dynamics that include quantum corrections describe the dynamics very well. Band transitions in the vicinity of conical crossings, however, fall out of the simple classical framework. There are algorithms (\eg hopping algorithms \cite{Lasser_Teufel:hopping:2005}) which mingle classical dynamics away from band crossings with quantum dynamics in the neighborhood of conical crossings (particles may hop from one band to the other). The problem is the lack of control over the phase: if the particle jumps bands, what is the phase factor associated to that? 

We propose an alternative point of view: the starting point is the hamiltonian 
\begin{align}
	H^{\mathcal{A}} = \tfrac{1}{2} (- i \eps \nabla_x - \eps \mathcal{A})^2 + V(\hat{x}) 
	\label{further_research:extension_other_gauge_groups:eqn:Born_Oppenheimer_hamiltonian}
\end{align}
acting on $L^2(\R^d , \C^2)$ where the potential $V$ and the \emph{Berry connection} $\mathcal{A}$ are suitable symmetric $2 \times 2$ matrix-valued functions. In a suitable basis, the potential $V = \bigl ( E_{\alpha} \delta_{\alpha \beta} \bigr )_{1 \leq \alpha , \beta \leq 2}$ consists of the energy band functions associated to the relevant bands $E_1$ and $E_2$, and the Berry connection $\mathcal{A} := i \bigl ( \sscpro{\varphi_{\alpha}}{\nabla_x \varphi_{\beta}}_{L^2(\R^{d N})} \bigr )_{1 \leq \alpha , \beta \leq 2}$ is determined by the electronic eigenfunctions associated to the relevant bands. The latter plays the same role as the magnetic vector potential $A$ and thus, the physics is determined by the \emph{Berry curvature} 
\begin{align*}
	\Omega := \dd \mathcal{A} + \mathcal{A} \wedge \mathcal{A} 
	. 
\end{align*}
The relation between connection and curvature contains the extra term $\mathcal{A} \wedge \mathcal{A}$ since the gauge group $SU(2)$ is non-commutative. 

Our ansatz is that we would like to think of hamiltonian~\eqref{further_research:extension_other_gauge_groups:eqn:Born_Oppenheimer_hamiltonian} as pseudomagnetic quantization of 
\begin{align*}
	h(x, \xi) = \tfrac{1}{2} \xi^2 \, \id_{\C^2} + V(x) 
	, 
\end{align*}
that is we would like to find a quantization procedure which maps $x$ onto multiplication with $x \otimes \id_{\C^2}$ on $L^2(\R^d,\C^2)$ and $\xi$ onto $-i \eps \nabla_x \otimes \id_{\C^2} + \eps \mathcal{A}(\hat{x})$. Once we have such a calculus at hand (Weyl quantization, dequantization/Wigner transform and Weyl product), we can formulate a semiclassical limit via an Egorov-type theorem. We hope to contribute to an improved understanding of the dynamics in the vicinity of band crossings and answer the following question: \emph{Is it possible to find classical equations of motion that include band transitions? } 

We envision the following strategy: 
\begin{enumerate}[(i)]
	\item Reformulate classical dynamics of a particle in a magnetic field as dynamics on the extended phase space $T^* (\R^d \times U(1))$. Physical observables must be gauge-\emph{in}variant functions, \ie those which do not depend on the $U(1)$ variable. 
	\item Make a phase space reduction in the sense of Marsden, Montgommery and Ratiu \cite{Marsden_Montgommery_Ratiu:phase_space_reduction:1990} and Marsden and Ratiu \cite[Chapter~6.7]{Marsden_Ratiu:intro_mechanics_symmetry:1999} by dividing out the gauge group $U(1)$. 
	\item Find a magnetic Weyl quantization on the extended phase space. 
	\item Repeat those three steps for the gauge group $SU(2)$. 
\end{enumerate}
Dealing with the differences means solving conceptual problems, \eg we need to consistently take into account that physical quantities are no longer gauge-\emph{in}variant, but rather gauge-\emph{co}variant. Some other issues which need to be addressed include: 
%
\begin{enumerate}[(i)]
	\item What is the relation between Berry curvature $\Omega = \dd \mathcal{A} + \mathcal{A} \wedge \mathcal{A}$ and possible choices of vector potentials $\mathcal{A}$? The non-trivial topological phase factors associated to band crossings are reminiscent of an Aharonov-Bohm-type effect and we do not expect the cohomology to be trivial. In other words, we not only need to specify $\Omega$, but also the `magnetic fluxes' associated to band crossings have to be part of the input. 
	\item What are the natural representations and what are the building block operators? Are they all equivalent? 
	\item What do the classical equations of motion look like? What is the meaning of offdiagonal elements? 
\end{enumerate}
%

\paragraph{Dynamics on shape space} 
\label{further_research:extension_other_gauge_groups:dynamics_on_shape_space}
Based on a publication by Littlejohn and Reinsch \cite{Littlejohn_Reinsch:dynamics_shape_space:1997}, D.~Sattlegger \cite{Sattlegger:dynamics_shape_space:2010} reiterated how to formulate dynamics on `shape space' as a gauge theory. The idea is simple: if we start with a non-relativistic $N$ particle system where particles interact via rotationally invariant pair potentials, there is an additional symmetry apart from total translations, namely total rotations. If each of the particles moves in $\R^3$, then the space which describes the shape of the molecule or body is given by $\R^{3(N-1)} / SO(3)$ where the group of rotations acts `democratically' on each factor of $\R^3$. Hence $\R^{3(N-1)} \longrightarrow \R^{3(N-1)} / SO(3)$ describes a bundle. If we exclude points of high symmetry, we arrive at a true fiber bundle\footnote{At points of high symmetry, the fiber is not all of $SO(3)$, but only a subgroup. } This fiber bundle is a quite complicated manifold: if $N \geq 4$, it is nontrivial [citation], \ie it cannot be written as a direct product of some space $M$ and $SO(3)$. This complexity is good for otherwise cats would not be able to land on their feet if they are dropped feet up. 

This system is a particularly good model system if one wants to understand non-commutative gauge groups and dynamics on manifolds. A simplified version may be a first step: we replace the shape space bundle by the direct product $R^{3K} \times SO(3)$ for some suitable $K \in \N$. Since $SO(3)$ shares the Lie algebra with $SU(2)$, this case is closely related to the two-level system mentioned before. 

\medskip

\noindent
Once these points have been properly understood, generalizations to more general gauge groups such as $U(N)$ and $SU(N)$ seem straightforward. 


\section{Electromagnetic Weyl calculus} 
\label{further_research:electromagnetic_weyl_calculus}
Magnetic Weyl calculus does not treat electric and magnetic fields symmetrically. If one wants to effectively deal with time-dependent systems, classical mechanics must be studied on \emph{extended phase space} $T^* (\R \times \R^d)$ where the additional coordiate in cofiguration space is time and its conjugate energy is necessary since there is no conservation of energy in time-dependent systems: one can extract energy from and put energy into the system. The magnetic field must be replaced with the electromagnetic field tensor 
\begin{align*}
	\mathcal{F} = \left (
	\begin{matrix}
		0 & -E \\
		E & B \\
	\end{matrix}
	\right ) = \dd \mathcal{A} 
\end{align*}
which is the exterior derivative of $\mathcal{A} = (\phi , A)$. Interesting new structures such as contact manifolds may appear \cite{Marsden_Ratiu:intro_mechanics_symmetry:1999}. This framework would also allow for electric fields with components in $\BCont^{\infty}$, \eg \emph{constant} electric fields. 


\section{Closer inspection of the semiclassical limit} 
\label{further_research:semiclassical_limit}
The assumptions in the proof of Theorem~\ref{asymptotics:semiclassical_limit:thm:semiclassical_limit_observables} are not optimal: if $B \neq \mathrm{constant}$, the standard proof presented in \cite{Robert:tour_semiclassique:1987} no longer ensures that derivatives of the classical flow are bounded. This is due to the crudeness of the method: taking the operator norm (and then applying the Gronwall lemma) destroys the important piece of information that magnetic fields do not change the energy of a system, they only change the particle's direction. It would be interesting to see under what circumstances the derivatives of the classical flow are bounded using more refined techniques. 

Furthermore, it will be interesting to see which results of non-magnetic semiclassical theory carry over after suitable modification. One example for such a result concerns the dynamics of molecules whose interaction may include Coulomb singularities: Castella, Jecko and Knauf \cite{Castella_Jecko_Knauf:semiclassical_estimates_singularities:2008} have elegantly proven the equivalence of the non-trapping condition on orbits of a certain energy to a resolvent estimate for systems with Coulombic singularities. What conditions have to be placed on the magnetic field so that this equivalence carries over? 


\section{Extension to algebra-valued symobls} 
\label{further_research:algebra_valued_symobls}
A lot of machinery introduced in Chapter~\ref{psiDO_reloaded} carries over if $\Alg$ is a general unital abelian algebra with an $\R^d$-action $\theta$. If this action is ergodic, then there exist invariant measures on the Gelfand spectrum $\Salg$ which are concentrated on one quasi orbit. This setting may prove useful in the study of ergodic random Schrödinger operators with \emph{random magnetic fields}. The analysis of operators with random magnetic fields seems to be very challenging as there are only a handful of publications (\eg \cite{Leschke_Warzel_Weichlein:energetic_dynamic_properties_random_B:2006}) compared to the cornucopia of results for the case of random potentials. Here, the natural ergodic measure on $\Salg$ singles out the orbit of `typical configurations' which are assumed by the system with probability $1$. The families of representations correspond to the realizations of the abstract pseudodifferential operator as an operator on, say, $L^2(\R^d)$ or $\ell^2(\Z^d)$. If the probability space (or rather: the topological dynamical system associated to the probability space) has a nontrivial topology, \eg nested quasi orbits $\Quasi' \subsetneq \Quasi$, it would be interesting to see the interplay between `almost sure' properties of the system and the topological structure of the topological dynamical system. First small steps in this direction were taken in \cite{Belmonte_Lein_Mantoiu:mag_twisted_actions:2010} and it stands to reason that the program of \cite{Lein_Mantoiu_Richard:anisotropic_mag_pseudo:2009} can be repeated in the more general framework of \cite{Belmonte_Lein_Mantoiu:mag_twisted_actions:2010}. 


\section{Modulation spaces} 
\label{further_research:modulation_spaces}
%
The extension of magnetic Weyl calculus rested on duality techniques using Schwartz functions and tempered distributions. The objects of interest are typically in the magnetic Moyal algebra $\mathcal{M}^B(\Pspace)$ whose quantization are exactly those operators which are in $\mathcal{L} \bigl ( \Schwartz(\R^d) \bigr ) \cap \mathcal{L} \bigl ( \Schwartz'(\R^d) \bigr )$. This places rather strong conditions on the regularity of the elements of $\mathcal{M}^B(\Pspace)$ since a countably \emph{infinite} number of seminorms needs to be controlled. If we are able to find a \emph{Banach space} $\mathcal{D}(\R^d)$ which sandwiches $L^2(\R^d)$ between itself and its dual, $\mathcal{D}(\R^d) \subset L^2(\R^d) \subset \mathcal{D}'(\R^d)$ and is compatible with the Fourier transform, then this would not only generalize the framework of magnetic Weyl calculus, it would simplify many technical issues because only one norm (\ie at most finitely many seminorms) need to be controlled. Hence, it stands to reason we can extend some results of magnetic Weyl calculus to symbols with much less regularity. 

In case of regular pseudodifferential theory, the answer was given by Feichtinger \cite{Feichtinger:modulation_spaces:1980} (see also \cite{Feichtinger:modulation_spaces:2002} for a more modern account) which was then modified by Măntoiu and Purice \cite{Mantoiu_Purice:modulation_mapping:2010}. It would be interesting to look for applications in the realm of mathematical physics which fall out of the usual pseudodifferential framework because of a lack of regularity. 


\appendix
\chapter{Oscillatory integrals} 
\label{appendix:oscillatory_integrals}
To extend integral expressions such as the Moyal product to broader classes of functions, it is necessary to introduce the concept of oscillatory integrals which give sense to expressions of the form 
\begin{align}
	g(X) := \int_{\R^{d'}} \dd Y \, e^{i \gamma(X,Y)} \, f(X,Y) 
	\label{appendix:oscillatory_integrals:eqn:osc_int} 
\end{align}
where $\gamma : \R^d \times \R^{d'} \longrightarrow \R$ is a phase factor and $f : \R^d \times \R^{d'} \longrightarrow \C$ a function. We will always assume that $\gamma$ is homogeneous of degree $1$, \ie for all $\lambda \geq 0$, we have $\gamma(\lambda X , Y) = \lambda \, \gamma(X,Y) = \gamma(X , \lambda Y)$. If $f(X,\cdot)$ is integrable for all $X \in \R^d$, then $g(X)$ exists for all $X$. If in addition $\babs{f(X,\cdot)}$ is in $L^1(\R^{d'})$ locally uniformly in $X$, then $X \mapsto g(X)$ is continuous by Dominated Convergence. What can one say if $f(X,\cdot) \not\in L^1(\R^{d'})$? 

Let us start with a simple example: from the introductory course on analysis, we know that the series $\sum_{k = 1}^{\infty} (-1)^k \tfrac{1}{k}$ is convergent, but not absolutely convergent. The reason is that the terms alternate signs and cancel each other. Heuristically, we can think of this series as a sum over 
\begin{align*}
	\frac{1}{k} - \frac{1}{k+1} = \frac{k+1 - k}{k(k+1)} = \frac{1}{k(k+1)} \sim \frac{1}{k^2}
\end{align*}
so that the decay is `faster' than naïvely expected.\footnote{Of course, it is in general not permissible to reorder terms in a series that is not absolutely convergent. } Similarly, even if the integrand of \eqref{appendix:oscillatory_integrals:eqn:osc_int} is not absolutely integrable, the oscillations caused by $e^{i \gamma}$ may remedy the situation and allow us to make sense of this expression after all. 

The idea is to interpret it as a tempered distribution.\footnote{Some authors, \eg Hörmander \cite{Hoermander:analysis_PDO1:1983,Hoermander:Fourier_integral_operators_1:1971}, use $\Cont^{\infty}_c(\R^d)$ as space of test functions rather than $\Schwartz(\R^d)$. Since we are interested in the general strategy rather than deep results, we shall ignore this and always work with Schwartz functions and tempered distributions. } For the remainder of the section, let us assume in addition that $\gamma$ and $f$ are $\Cont^{\infty}$ functions of their variables. Hörmander has shown \cite[Lemma~1.2.1]{Hoermander:Fourier_integral_operators_1:1971} that if $\gamma$ has no critical points for $X \neq 0$, then there exists a first-order differential operator $L$ in $X$ and $Y$ such that its adjoint satisfies 
\begin{align*}
	L^* e^{i \gamma} = e^{i \gamma} 
	. 
\end{align*}
In many applications, it is more convenient to use either a first-order pseudodifferential operator or a higher-order operator instead, say $L^2$. We will give some concrete examples later on. Then the distribution $g$ applied to a test function $\varphi \in \Schwartz(\R^d)$ 
\begin{align}
	\bigl ( g , \varphi \bigr ) &= \int_{\R^d} \dd X \int_{\R^{d'}} \dd Y \, e^{i \gamma(X,Y)} \, f(X,Y) \, \varphi(X) 
	\notag \\
	&= \int_{\R^d} \dd X \int_{\R^{d'}} \dd Y \, \bigl ( L^* e^{i \gamma(X,Y)} \bigr ) \, f(X,Y) \, \varphi(X) 
	\notag \\
	&= \int_{\R^d} \dd X \int_{\R^{d'}} \dd Y \, e^{i \gamma(X,Y)} \, L \bigl ( f(X,Y) \, \varphi(X) \bigr ) 
\end{align}
leads to a formal integral expression where we may insert powers of $L^*$ and use partial integration. If $L^N (f \, \varphi)$ is integrable in $X$ and $Y$ for $N \in \N_0$ large enough, then the oscillatory integral on the left-hand side 
\begin{align}
	\int_{\R^d} \dd X &\int_{\R^{d'}} \dd Y \, e^{i \gamma(X,Y)} \, f(X,Y) \, \varphi(X) := 
	\notag \\
	&:= \int_{\R^d} \dd X \int_{\R^{d'}} \dd Y \, e^{i \gamma(X,Y)} \, L^N \bigl ( f(X,Y) \, \varphi(X) \bigr ) 
	\label{appendix:oscillatory_integrals:eqn:definition_osc_int}
\end{align}
%
is \emph{defined} by the right-hand side which is an \emph{ordinary absolutely convergent} integral. To distinguish oscillatory integrals from ordinary integrals, some authors denote them by $\mathrm{O}_s-\int$ or $\mathrm{O}_s$ \cite{Kumanogo:pseudodiff:1981}, but but we shall not do so. Instead, we will either say explicitly in what sense the integration is to be understood or rely on the reader to interpret formulas properly. By extension, we define the distribution $g$ in terms of equation~\eqref{appendix:oscillatory_integrals:eqn:definition_osc_int}. Oscillatory integrals have very convenient properties: 
\begin{enumerate}[(i)]
	\item We may interchange limits and oscillatory integration. If the integrand depends continuously on a parameter, the oscillatory is also continuous in that parameter. 
	\item We may interchange integration with respect to a parameter and oscillatory integration. Hence, if the integrand depends smoothly on a parameter, for instance, and the oscillatory integrals of the derivatives exist, then the oscillatory integral depends smoothly on the parameter. 
	\item A version of Fubini's theorem holds and we may interchange the order of integration. 
\end{enumerate}
%
In other words, we are allowed to do what our analysis professors have taught us not to do! These properties are essential when working with oscillatory integrals. We refer the interested reader to \cite[Chapter~7.8]{Hoermander:analysis_PDO1:1983}, \cite[Chapter~1]{Hoermander:Fourier_integral_operators_1:1971} and \cite[Chapter~1]{Kumanogo:pseudodiff:1981}. 

To give a flavor of the type of manipulations involved, we give two examples: 
\begin{example}
	In the sense of oscillatory integrals, the Fourier transform of the constant function is up to a factor equal to $\delta \in \Schwartz'(\R^d)$, \ie 
	\begin{align*}
		\int_{\R^d} \dd x \, e^{i x \cdot \xi} = (2\pi)^d \, \delta(\xi) 
		. 
	\end{align*}
	This manipulation can be made rigorous either directly by rewriting the equation in terms of the duality bracket. Alternatively, we pick $\varphi , \chi \in \Schwartz(\R^d)$ such that $\chi(0) = 1$ and insert $\chi(\nicefrac{\cdot}{\eps})$ into the integral, 
	\begin{align*}
		\int_{\R^d} \dd \xi &\int_{\R^d} \dd x \, e^{i x \cdot \xi} \, \chi(\nicefrac{x}{\eps}) \, \varphi(\xi) = \int_{\R^d} \dd \xi \, \eps^d \, \hat{\chi}(- \eps \xi) \, \varphi(\xi) 
		\\
		&= (2\pi)^{\nicefrac{d}{2}} \, \int_{\R^d} \dd \xi \, \hat{\chi}(- \xi) \, \varphi(\nicefrac{\xi}{\eps}) 
		\xrightarrow{\eps \searrow 0} (2\pi)^{\nicefrac{d}{2}} \, \varphi(0) \, \int_{\R^d} \dd \xi \, \hat{\chi}(-\xi) 
		\\
		&= (2\pi)^d \, \varphi(0) \, \chi(0) 
		= (2\pi)^d \, \varphi(0) = \bigl ( (2\pi)^d \delta , \varphi \bigr ) 
		. 
	\end{align*}
	Whenever we may integrate out such an exponential factor, we will not detail the above regularization argument, it is implicitly understood. 
\end{example}
\begin{example}
	Let $f \in \Cont^{\infty}_{\mathrm{pol}}(\R^d)$ be polynomially bounded and smooth whose growth is no worse than $\sexpval{x}^m$ for some $m \in \R$, \ie there exists $C > 0$ such that 
	\begin{align*}
		\babs{f(x)} \leq C \, \sexpval{x}^m 
	\end{align*}
	holds for all $x \in \R^d$. Then we choose $L_{\xi} := 1 - \Delta_{\xi}$ as $L$ operator and the relation 
	\begin{align*}
		\sexpval{x}^{-2} L_{\xi} e^{- i x \cdot \xi} = e^{- x \cdot \xi}
	\end{align*}
	allows us to produce factors of $\sexpval{x}^{-2N}$. If we choose $N > \tfrac{1}{2}(m+d)$, then 
	\begin{align*}
		(\Fourier f)(\xi) &= \frac{1}{(2\pi)^{\nicefrac{d}{2}}} \int_{\R^d} \dd x \, e^{- i x \cdot \xi} \, f(x) 
		= \frac{1}{(2\pi)^{\nicefrac{d}{2}}} \int_{\R^d} \dd x \, \bigl ( \sexpval{x}^{-2N} \, L_{\xi}^N e^{- i x \cdot \xi} \bigr ) \, f(x) 
		\\
		&= L_{\xi}^N \frac{1}{(2\pi)^{\nicefrac{d}{2}}} \int_{\R^d} \dd x \,  e^{- i x \cdot \xi} \, \sexpval{x}^{-2N} \, f(x) 
		= L_{\xi}^N \bigl (\Fourier (\sexpval{x}^{-2N} f) \bigr )(\xi)
	\end{align*}
	holds in the sense of oscillatory integrals. The right-hand side is integrable by choice of $N$. Similarly, all derivatives of $\Fourier f$ exist as oscillatory integrals and we conclude $\Fourier f$ is smooth. 
\end{example}
Section~\ref{appendix:asymptotics:existence_osc_int} contains further non-trivial examples of oscillatory integrals. 


\chapter{Miscellaneous estimates} 
\label{appendix:misc}
\begin{lem}\label{appendix:oscillatory_integrals:lem:schwartz_functions_Lp_estimate}
	Let $f \in \Schwartz(\R^d)$. Then for each $1 \leq p < \infty$, the $L^p$ norm of $f$ can be dominated by a finite number of seminorms, 
	\begin{align*}
		\norm{f}_{L^p(\R^d)} \leq C_1(d) \norm{f}_{00} + C_2(d) \max_{\abs{a} = 2n(d)} \norm{f}_{a 0} 
		, 
	\end{align*}
	where $C_1(d) , C_2(d) \in \R^+$ and $n(d) \in \N_0$ only depend on the dimension of $\R^d$. Hence, $f \in L^p(\R^d)$. 
\end{lem}
\begin{proof}
	We split the integral on $\R^d$ into an integral over the unit ball centered at the origin and its complement: let $B_n := \max_{\abs{a} = 2n} \norm{f}_{a0}$, then 
	\begin{align*}
		\norm{f}_{L^p(\R^d)} &= \biggl ( \int_{\R^d} \dd x \, \abs{f(x)}^p \biggr )^{\nicefrac{1}{p}} 
		\\
		&
		\leq \biggl ( \int_{\abs{x} \leq 1} \dd x \, \abs{f(x)}^p \biggr )^{\nicefrac{1}{p}} + \biggl ( \int_{\abs{x} > 1} \dd x \, \abs{f(x)}^p \biggr )^{\nicefrac{1}{p}} 
		\\
		&\leq \norm{f}_{00} \, \biggl ( \int_{\abs{x} \leq 1} \dd x \, 1 \biggr )^{\nicefrac{1}{p}} +  \biggl ( \int_{\abs{x} > 1} \dd x \, \abs{f(x)}^p \frac{\abs{x}^{2np}}{\abs{x}^{2np}} \biggr )^{\nicefrac{1}{p}} 
		\\
		&\leq \mathrm{Vol}(B_1(0))^{\nicefrac{1}{p}} \, \norm{f}_{00} + B_n \, \biggl ( \int_{\abs{x} > 1} \dd x \, \frac{1}{\abs{x}^{2np}} \biggr )^{\nicefrac{1}{p}} 
		. 
	\end{align*}
	If we choose $n$ large enough, $\abs{x}^{-2np}$	is integrable and can be computed explicitly, and we get 
	\begin{align*}
		\norm{f}_{L^p(\R^d)} \leq C_1(d) \, \norm{f}_{00} + C_2(d) \, \max_{\abs{a} = 2n} \norm{f}_{a0} 
		. 
	\end{align*}
	This concludes the proof. 
\end{proof}
\begin{proof}[Lemma~\ref{magWQ:magnetic_weyl_calculus:lem:composition_of_operator_kernels}]
	We need to estimate the seminorms of $K_T \diamond K_S$: let $a , \alpha , b , \beta \in \N_0^d$ be multiindices and introduce the map $\Phi : (x,y,z) \mapsto (2\pi)^{- \nicefrac{d}{2}} \, K_T(x,z) \, K_S(z,y)$. Then $\Phi \in \Schwartz(\R^d \times \R^d \times \R^d)$ is a Schwartz function in all three variables. First, we need to show we can exchange differentiation with respect to $x$ and $y$ and integration with respect to $z$, \ie that for fixed $x$ and $y$ 
	\begin{align}
		x^a y^b &\partial_x^{\alpha} \partial_y^{\beta} (K_T \diamond K_S)(x,y) = 
		x^a y^b \partial_x^{\alpha} \partial_y^{\beta} \frac{1}{(2\pi)^{\nicefrac{d}{2}}} \int_{\R^d_x} \dd z \, K_T(x,z) \, K_S(z,y) 
		\notag \\
		&
		= \frac{1}{(2\pi)^{\nicefrac{d}{2}}} \int_{\R^d_x} \dd z \, x^a y^b \partial_x^{\alpha} \partial_y^{\beta} \bigl ( K_T(x,z) \, K_S(z,y) \bigr ) 
		= \int_{\R^d_x} \dd z \, x^a y^b \partial_x^{\alpha} \partial_y^{\beta} \Phi(x,y,z) 
		\label{weyl_calculus:weyl_product:eqn:composition_of_operator_kernels}
	\end{align}
	holds. We will do this by estimating $\babs{x^a y^b \partial_x^{\alpha} \partial_y^{\beta} \Phi(x,y,z)}$ uniformly in $x$ and $y$ by an integrable function $G(z)$ and then invoking Dominated Convergence. For fixed $x$ and $y$, we estimate the $L^1$ norm of $x^a y^b \partial_x^{\alpha} \partial_y^{\beta} \Phi(x,y,\cdot)$ from above by a finite number of seminorms of $\Phi(x,y,\cdot)$ with the help of Lemma~\ref{appendix:oscillatory_integrals:lem:schwartz_functions_Lp_estimate}, 
	\begin{align*}
		\int_{\R^d} \dd z \, &\babs{x^a y^b \partial_x^{\alpha} \partial_y^{\beta} \Phi(x,y,z)} = \bnorm{x^a y^b \partial_x^{\alpha} \partial_y^{\beta} \Phi(x,y,\cdot)}_{L^1(\R^d)} 
		\\
		&\qquad \leq C_1 \, \sup_{z \in \R^d} \babs{x^a y^b \partial_x^{\alpha} \partial_y^{\beta} \Phi(x,y,z)} + C_2 \, \max_{\abs{c} = 2 n} \sup_{z \in \R^d} \babs{x^a y^b \partial_x^{\alpha} \partial_y^{\beta} \Phi(x,y,z)} 
		\\
		&\qquad = C_1 \, \bnorm{x^a y^b \partial_x^{\alpha} \partial_y^{\beta} \Phi(x,y,\cdot)}_{0 0} + C_2 \, \max_{\abs{c} = 2 n} \bnorm{x^a y^b \partial_x^{\alpha} \partial_y^{\beta} \Phi(x,y,\cdot)}_{c 0} 
		. 
	\end{align*}
	Now we interchange $\sup$ and integration with respect to $z$, 
	\begin{align*}
		\sup_{x,y \in \R^d} \int_{\R^d} \dd z \, \babs{x^a y^b \partial_x^{\alpha} \partial_y^{\beta} \Phi(x,y,z)} &\leq \int_{\R^d} \dd z \, \sup_{x,y \in \R^d} \babs{x^a y^b \partial_x^{\alpha} \partial_y^{\beta} \Phi(x,y,z)} 
		\\
		&= \Bnorm{\sup_{x,y \in \R^d} \babs{x^a y^b \partial_x^{\alpha} \partial_y^{\beta} \Phi(x,y,\cdot)}}_{L^1(\R^d)}
		, 
	\end{align*}
	which can be estimated from above by 
	\begin{align*}
		\Bnorm{\sup_{x,y \in \R^d} \babs{x^a y^b \partial_x^{\alpha} \partial_y^{\beta} \Phi(x,y,\cdot)}}_{L^1(\R^d)} 
		&\leq C_1 \, \sup_{x,y \in \R^d} \sup_{z \in \R^d} \babs{x^a y^b \partial_x^{\alpha} \partial_y^{\beta} \Phi(x,y,z)} 
		+ \\
		&\quad 
		+ C_2 \, \max_{\abs{c} = 2 n} \sup_{x,y \in \R^d} \sup_{z \in \R^d} \babs{x^a y^b \partial_x^{\alpha} \partial_y^{\beta} \Phi(x,y,z)} 
		\\
		&= C_1 \, \bnorm{\Phi}_{a \alpha b \beta 0 0} + C_2 \, \max_{\abs{c} = 2 n} \bnorm{\Phi}_{a \alpha b \beta c 0} < \infty 
		. 
	\end{align*}
	Here $\bigl \{ \norm{\cdot}_{a \alpha b \beta c \gamma} \bigr \}_{a , \alpha , b , \beta , c , \gamma \in \N_0^d}$ is the family of seminorms associated to $\Schwartz(\R^d \times \R^d \times \R^d)$ which are defined by 
	\begin{align*}
		\norm{\Phi}_{a \alpha b \beta c \gamma} := \sup_{x,y,z \in \R^d} \babs{x^a y^b z^c \partial_x^{\alpha} \partial_x^{\beta} \partial_x^{\gamma} \Phi(x,y,z)} 
		. 
	\end{align*}
	This means we have found an \emph{integrable} function 
	\begin{align*}
		G(z) := \sup_{x,y \in \R^d} \babs{x^a y^b \partial_x^{\alpha} \partial_y^{\beta} \Phi(x,y,z)}
	\end{align*} 
	which dominates $x^a y^b \partial_x^{\alpha} \partial_y^{\beta} \Phi(x,y,z)$ for all $x , y \in \R^d$. Hence, exchanging differentiation and integration in equation~\eqref{weyl_calculus:weyl_product:eqn:composition_of_operator_kernels} is possible and we can bound the $\norm{\cdot}_{a \alpha b \beta}$ seminorm on $\Schwartz(\R^d \times \R^d)$ by 
	\begin{align*}
		\bnorm{K_T \diamond K_S}_{a \alpha b \beta} &= \sup_{x,y \in \R^d} \babs{x^a y^b \partial_x^{\alpha} \partial_y^{\beta} (K_T \diamond K_S)(x,y)} 
		\\
		&\leq C_1 \, \bnorm{\Phi}_{a \alpha b \beta 0 0} + C_2 \, \max_{\abs{c} = 2 n} \bnorm{\Phi}_{a \alpha b \beta c 0} < \infty 
		. 
	\end{align*}
	This means $K_T \diamond K_S \in \Schwartz(\R^d \times \R^d)$. 
\end{proof}
%

\chapter{Asymptotic expansion} 
\label{appendix:asymptotics}
\section{Expansion of the twister} 
\label{appendix:asymptotics:expansion_twister}
\begin{lem}
	\label{appendix:asymptotics:expansion_twister:lem:expansion_twister}
	Assume $B$ satisfies Assumption~\ref{asymptotics:assumption:bounded_fields}. Then we can expand $\gBe$ around $x$ to arbitrary order $N$ in powers of $\eps$: 
	\begin{align}
		\gBe(x,y,z) &= - \sum_{n = 1}^{N} \frac{\eps^n}{n!} \, \partial_{x_{j_1}} \cdots \partial_{x_{j_{n-1}}} B_{kl}(x) \, y_k \, z_l \, \left ( - \frac{1}{2} \right )^{n+1} \frac{1}{(n+1)^2} \sum_{c = 1}^n \noverk{n+1}{c} 
		\cdot \notag \\
		&\qquad \qquad \qquad \qquad \cdot 
		\bigl ( (1 - (-1)^{n+1}) c - (1 - (-1)^{c}) (n+1) \bigr ) 
		\cdot \notag \\
		&\qquad \qquad \qquad \qquad \cdot 
		y_{j_1} \cdots y_{j_{c-1}} z_{j_{c}} \cdots z_{j_{n-1}} 
		+ \notag \\
		&\qquad \qquad 
		+ R_N[\gBe](x,y,z) 
		\notag \\
		&=: - \sum_{n = 1}^{N} \eps^n \negmedspace \negmedspace \negmedspace \sum_{\abs{\alpha} + \abs{\beta} = n-1} \negmedspace \negmedspace \negmedspace 
		C_{n,\alpha,\beta} \, \partial_x^{\alpha} \partial_x^{\beta} B_{kl}(x) \, y_k z_l \, y^{\alpha} \, z^{\beta} + R_N[\gBe](x,y,z) 
		\notag \\
		&=: - \sum_{n = 1}^{N} \eps^n \, {\mathcal{L}}_n + R_N[\gBe](x,y,z) 
	\end{align}
	In particular, the flux is of order $\eps$ and the $n$th-order term is a sum of monomials in position of degree $n+1$ and each of the terms is a $\mathcal{BC}^{\infty}(\R^d_x,\mathcal{C}^{\infty}_{\mathrm{pol}}(\R^d_y \times \R^d_z))$ function. The remainder is a $\mathcal{BC}^{\infty}(\R^d,\mathcal{C}^{\infty}_{\mathrm{pol}}(\R^d \times \R^d))$ function that is $\ordere{N+1}$ and can be explicitly written as a bounded function of $x$, $y$ and $z$ as well as $N+2$ factors of $y$ and $z$. 
\end{lem}
\begin{proof}
	We choose the transversal gauge to represent $B$, \ie 
	\begin{align}
		A_l(x + a) &= - \int_0^1 \dd s \, B_{lj} (x + s a) \, s a_j 
		\label{appendix:eqn:transversal_gauge}
	\end{align}
	and rewrite the flux integral into three line integrals over the edges of the triangle. 
	\begin{align*}
		\gBe(x,y,z) &= \frac{1}{\eps} \int_0^1 \dd t \Bigl [ \eps \, (y_l + z_l) \, A_l \bigl (x + \eps (t - \nicefrac{1}{2}) (y+z) \bigr ) + \Bigr . \\
		&\qquad \qquad \Bigl .  
		- \eps \, y_l \, A_l \bigl ( x + \eps (t - \nicefrac{1}{2}) y - \tfrac{\eps}{2} z \bigr ) 
		- \eps \, z_l \, A_l \bigl ( x + \tfrac{\eps}{2} y + \eps (t - \nicefrac{1}{2}) z \bigr ) \Bigr ] 
		\\
		&= \eps \int_{-\nicefrac{1}{2}}^{+\nicefrac{1}{2}} \dd t \int_0^1 \dd s \, s \Bigl [ - B_{lj} \bigl ( x + \eps s t (y + z) \bigr ) \, (y_l + z_l) \, t (y_j + z_j) 
		+ \Bigr. \\
		&\qquad \qquad \qquad \qquad \qquad \qquad \Bigl .  
		+ B_{lj} \bigl ( x + \eps s \bigl ( t y - \tfrac{z}{2} \bigr ) \bigr ) \, y_l \, \bigl ( t y_j - \tfrac{z_j}{2} \bigr ) 
		+ \Bigr. \\
		&\qquad \qquad \qquad \qquad \qquad \qquad \Bigl .  
		+ B_{lj} \bigl ( x + \eps s \bigl ( \tfrac{y}{2} + t z \bigr ) \bigr ) \, z_l \, \bigl ( \tfrac{y_j}{2} + t z_j \bigr ) \Bigr ]
	\end{align*}
	All these terms have a prefactor of $\eps$ which stems from the explicit expression of transversal gauge. We will now Taylor expand each of the three terms up to $N-1$th order around $x$ (so that it is of $N$th order in $\eps$). 
	\begin{align*}
		\int_{-\nicefrac{1}{2}}^{+\nicefrac{1}{2}} \dd t \, \int_0^1 \dd s \, &B_{lj} \bigl ( x + \eps s t (y + z) \bigr ) \, \eps s (y_{j} + z_{j} ) = \\
		&= \int_0^1 \dd s \, \int_{-\nicefrac{1}{2}}^{+\nicefrac{1}{2}} \dd t \, \sum_{n = 0}^{N-1} \frac{\eps^{n+1}}{n!} s^{n+1} s^{-n} \, \partial_{x_{j_1}} \cdots \partial_{x_{j_n}} B_{lj_{n+1}}(x) 
		\cdot \\
		&\qquad \qquad \qquad \qquad \qquad \cdot 
		t^{n+1} \, \prod_{m = 1}^{n+1} (y_{j_m} + z_{j_m}) 
		+ R_{1 \, N \, l}(x,y,z)
	\end{align*}
	The remainder $R_{1 \, N \, l}$ is of order $N+1$ in $\eps$, bounded in $x$ and polynomially bounded in $y$ and $z$. It is a sum of monomials in $y$ and $z$ of degree $N+1$. 
	\begin{align*}
		R_{1 \, N \, l}(x,y,z) &= 
		\int_0^1 \dd \tau \int_0^1 \dd s  \int_{-\nicefrac{1}{2}}^{+\nicefrac{1}{2}} \dd t \, \frac{1}{(N-1)!} (1-\tau)^{N-1} 
		\cdot \\
		&\qquad \qquad \qquad \qquad \cdot 
		\partial_{\tau}^{N} B_{lj}(x + \eps \tau s t (y+z) ) \, \eps s t (y + z) 
		\\
		&= \eps^{N+1} \, \sum_{\sabs{\alpha} = N} \frac{N}{\alpha !} \int_0^1 \dd s  \int_{-\nicefrac{1}{2}}^{+\nicefrac{1}{2}} \dd t \, s \, t^{N+1} \, (y + z)^{\alpha} \, (y_j + z_j) \, 
		\cdot \\
		&\qquad \qquad \qquad \qquad \cdot 
		\int_0^1 \dd \tau \, (1 - \tau)^{N-1} \, \partial_x^{\alpha} B_{lj} (x + \eps \tau st (y+z)) 
	\end{align*}
	The $n$th order term in $\eps$ (the $n-1$th term of the Taylor expansion) reads 
	\begin{align*}
		\frac{\eps^{n}}{(n-1)!} &\int_0^1 \dd s \, s \, \int_{-\nicefrac{1}{2}}^{+\nicefrac{1}{2}} \dd t \, t^{n} \, \partial_{x_{j_1}} \cdots \partial_{x_{j_{n-1}}} B_{lj_{n}}(x) \, \prod_{m = 1}^{n} (y_{j_m} + z_{j_m}) = 
		\\
		&= \frac{1}{2} \frac{\eps^{n}}{n!} \, \left ( \frac{1}{2} \right )^{n+1} \frac{1 + (-1)^{n}}{n+1} \, \partial_{x_{j_1}} \cdots \partial_{x_{j_n}} B_{lj_{n}}(x) \, 
		\cdot \\
		&\qquad \qquad \cdot 
		\sum_{m = 0}^{n} \noverk{n}{m} y_{j_1} \cdots y_{j_m} z_{j_{m+1}} \cdots z_{j_{n}} 
		. 
	\end{align*}
	The other factors can be calculated in the same fashion: 
	\begin{align*}
		\frac{\eps^{n}}{(n-1)!} &\int_{-\nicefrac{1}{2}}^{+\nicefrac{1}{2}} \dd t \, \int_0^1 \dd s \, s^{n} s^{-(n-1)} \, \partial_{x_{j_1}} \cdots \partial_{x_{j_{n-1}}} B_{lj_{n}}(x) \, \prod_{m = 1}^{n} \bigl ( t y_{j_m} - \tfrac{1}{2} z_{j_m} \bigr ) = 
		\\
		&= \frac{\eps^{n}}{n!} \left ( \frac{1}{2} \right )^{n+2} \, \partial_{x_{j_1}} \cdots \partial_{x_{j_{n-1}}} B_{lj_{n}}(x) \, 
		\cdot \\
		&\qquad \qquad \cdot 
		\sum_{m = 0}^{n} \noverk{n}{m} \, \frac{(-1)^{n-m} +(-1)^{n}}{m+1} \, y_{j_1} \cdots y_{j_m} z_{j_{m+1}} \cdots z_{j_{n}} 
	\end{align*}
	The remainder is also of the correct order in $\eps$, namely $n$th order, contains $N+2$ $q$s and a $\mathcal{BC}^{\infty}(\R^d,\mathcal{C}^{\infty}_{\mathrm{pol}}(\R^d \times \R^d))$ function as prefactor: 
	\begin{align*}
		R_{2 \, N \, l}(x,y,z) &= \eps^{N+1} \, \sum_{\sabs{\alpha} = N} \frac{N}{\alpha !} \int_0^1 \dd s  \int_{-\nicefrac{1}{2}}^{+\nicefrac{1}{2}} \dd t \, s \, \bigl ( t y - \tfrac{z}{2} \bigr )^{\alpha} \, \bigl ( t y_j - \tfrac{z_j}{2} \bigr ) 
		\cdot \\
		&\qquad \qquad \qquad \qquad \cdot 
		\int_0^1 \dd \tau \, (1 - \tau)^{N-1} \, \partial_x^{\alpha} B_{lj} \bigl ( x + \eps \tau \bigl ( st y - \tfrac{z}{2} \bigr ) \bigr ) 
	\end{align*}
	The last term satisfies the same properties as $R_{1 \, N \, l}$: 
	\begin{align*}
		\int_{-\nicefrac{1}{2}}^{+\nicefrac{1}{2}} \dd t &\int_0^1 \dd s \, \frac{\eps^{n}}{(n-1)!} s^{n} s^{-(n-1)} \, \partial_{x_{j_1}} \cdots \partial_{x_{j_{n-1}}} B_{lj_{n}}(x) \, \prod_{m = 1}^{n} \bigl ( \tfrac{1}{2} y_{j_m} + t z_{j_m} \bigr ) = 
		\\
		&= \frac{\eps^{n}}{n!}  \, \left ( \frac{1}{2} \right )^{n+2} \, \partial_{x_{j_1}} \cdots \partial_{x_{j_{n-1}}} B_{lj_{n}}(x) \, 
		\cdot \\
		&\qquad \qquad \qquad \cdot 
		\sum_{m = 0}^{n} \noverk{n}{m} \, \frac{1 + (-1)^{n-m}}{n+1-m} \, y_{j_1} \cdots y_{j_m} z_{j_{m+1}} \cdots z_{j_{n}} 
	\end{align*}
	$R_{3 \, N \, l}$ satisfies the same properties as $R_{1 \, N \, l}$ and $R_{2 \, N \, l}$, 
	\begin{align*}
		R_{3 \, N \, l}(x,y,z) &= \eps^{N+1} \, \sum_{\sabs{\alpha} = N} \frac{N}{\alpha !} \int_0^1 \dd s  \int_{-\nicefrac{1}{2}}^{+\nicefrac{1}{2}} \dd t \, s \, \bigl ( \tfrac{y}{2} + t z \bigr )^{\alpha} \, \bigl ( \tfrac{y_j}{2} + t z_j \bigr )  
		\cdot \\
		&\qquad \qquad \qquad \qquad \cdot 
		\int_0^1 \dd \tau \, (1 - \tau)^{N-1} \, \partial_x^{\alpha} B_{lj} \bigl ( x + \eps \tau s \bigl ( \tfrac{y}{2} + t z \bigr ) \bigr ) 
		. 
	\end{align*}
	Put together, we obtain for the $n$th order term: 
	\begin{align*}
		\frac{1}{2} &\frac{\eps^{n}}{n!} \left ( \frac{1}{2} \right )^{n+1} \, \partial_{x_{j_1}} \cdots \partial_{x_{j_{n-1}}} B_{lj_{n}}(x) \, \sum_{m = 0}^{n} \noverk{n}{m} \cdot \\
		&\qquad \cdot \left [ \frac{1 + (-1)^{n}}{n+1} (y_l + z_l) - \frac{(-1)^{n-m} + (-1)^{n}}{m+1} y_l - \frac{1 + (-1)^{n-m}}{n+1-m} z_l \right ] \, 
		\cdot \\
		&\qquad \qquad \qquad \cdot 
		y_{j_1} \cdots y_{j_m} z_{j_{m+1}} \cdots z_{j_{n}} \\
		&= \frac{\eps^{n}}{n!} \left ( - \frac{1}{2} \right )^{n+1} \frac{1}{(n+1)^2} \, \sum_{\sabs{\alpha} + \sabs{\beta} = n-1} \partial_x^{\alpha} \partial_x^{\beta} B_{lk}(x) \, y_l z_k \cdot \\
		&\qquad \cdot \noverk{n+1}{\sabs{\alpha}+1} \bigl ( ( 1 - (-1)^{\sabs{\alpha}+1} ) (n+1) - (1 - (-1)^{n+1}) (\sabs{\alpha} + 1) \bigr ) \, y^{\alpha} z^{\beta}
	\end{align*}
	The total remainder of the expansion reads 
	\begin{align*}
		R_N[\gBe] &= y_l \, \bigl ( R_{1 \, N \, l} - R_{2 \, N \, l} \bigr ) + z_l \, \bigl ( R_{1 \, N \, l} - R_{3 \, N \, l} \bigr ) \in \mathcal{BC}^{\infty}(\R^d,\mathcal{C}^{\infty}_{\mathrm{pol}}(\R^d \times \R^d))
		. 
	\end{align*}
	In total, the remainder is a sum of monomials with bounded coefficients of degree $N+2$ while it is of $\ordere{N+1}$. 
\end{proof}
%


\section{Properties of derivatives of $\gBe$} 
\label{appendix:asymptotics:properties_mag_flux}
For convenience, we give two theorems found in \cite{Iftimie_Mantiou_Purice:magnetic_psido:2006} on the magnetic flux and its expontential which are needed to make the expansion rigorous: 
\begin{lem}\label{appendix:asymptotics:properties_mag_flux:lem:boundedness_mag_flux}
	If the magnetic field $B_{lj}$, $1 \leq l,j \leq n$, satisfies the usual conditions, then 
	\begin{align*}
		\partial_{x_j} \gBe &= D_{jk}(x,y,z) \, y_k + E_{jk}(x,y,z) \, z_k \\
		\partial_{y_j} \gBe &= D'_{jk}(x,y,z) \, y_k + E'_{jk}(x,y,z) \, z_k \\
		\partial_{z_j} \gBe &= D''_{jk}(x,y,z) \, y_k + E''_{jk}(x,y,z) \, z_k 
	\end{align*}
	where the coefficients $D_{jk} , \ldots , E''_{jk} \in \mathcal{BC}^{\infty}(\R^d \times \R^d \times \R^d)$, $1 \leq j,k \leq d$. 
\end{lem}
A direct consequence of this is the following simple corollary: 
\begin{cor}\label{appendix:asymptotics:cor:properties_flux}
	If the magnetic field satisfies the usual conditions, then for any $a,b,c \in {\N_0}^d$, there exists $C_{abc} > 0$ such that 
	\begin{align*}
		\babs{\partial_x^a \partial_y^b \partial_z^c e^{- i \lambda \gBe(x,y,z)}} \leq C_{abc} \bigl ( \sexpval{y} + \sexpval{z} \bigr )^{\sabs{a} + \sabs{b} + \sabs{c}} 
		\leq \tilde{C}_{abc} \sexpval{y}^{\sabs{a} + \sabs{b} + \sabs{c}} \sexpval{z}^{\sabs{a} + \sabs{b} + \sabs{c}} 
	\end{align*}
	holds, \ie derivatives of $e^{- i \lambda \gBe(x,y,z)}$ are $\mathcal{C}^{\infty}_{\mathrm{pol}}$ functions in $y$ and $z$. 
\end{cor}
%


\section{Existence of oscillatory integrals} 
\label{appendix:asymptotics:existence_osc_int}
To derive the adiabatic expansion, we have to ensure the existence of two types of oscillatory integrals, one is relevant for the $(n,k)$ term of the two-parameter expansion, the other is necessary to show existence of remainders and the $k$th term of the $\lambda$ expansion. 
\begin{lem}\label{appendix:asymptotics:existence_osc_int:lem:Lemma1}
	Let $f \in \Hoerrd{m}$, $\rho \in [0,1]$. Then for all multiindices $a, \alpha \in \N_0^d$ 
	\begin{align}
		G(X) := \frac{1}{(2\pi)^d} \int_{\Pspace} \dd Y \, e^{i \sigma(X,Y)} y^a \eta^{\alpha} (\Fs^{-1} f)(Y) = \bigl ( (-i \partial_{\xi})^{a} (+ i \partial_x)^{\alpha} f \bigr )(X) 
	\end{align}
	exists as an oscillatory integral and is in symbol class $\Hoerrd{m - \sabs{a} \rho}$. 
\end{lem}
\begin{proof}
	Since $f$ is a function of tempered growth, we can consider it as an element of $\Schwartz'(\R^{2d})$. Then, we can rewrite $G$ as $G = \Fs \hat{x}^a \hat{\xi}^{\alpha} \Fs$ where $\hat{x}$ and $\hat{\xi}$ are the multiplication operators initially defined on $\Schwartz(\R^{2d})$ which are extended to tempered distributions by duality. Then for any $\varphi \in \Schwartz(\R^{2d})$, we have 
	\begin{align*}
		\bigl ( G , \varphi \bigr ) &= \bigl ( \Fs \hat{x}^a \hat{\xi}^{\alpha} \Fs f , \varphi \bigr ) 
		= \bigl ( f , \Fs \hat{x}^a \hat{\xi}^{\alpha} \Fs \varphi \bigr ) 
		= \bigl ( f , (+i \partial_{\xi})^a (-i \partial_x)^{\alpha} \varphi \bigr ) 
		\\
		&= \bigl ( (-i \partial_{\xi})^a (+i \partial_x)^{\alpha} f , \varphi \bigr ) 
	\end{align*}
	where $( \cdot , \cdot )$ denotes the usual duality bracket. 
	
	Thus, the integral exists as an oscillatory integral. $G = (-i \partial_{\xi})^a (+i \partial_x)^{\alpha} f$ is also in the correct symbol class, namely $\Hoerrd{m - \sabs{a} \rho}$, and the lemma has been proven. 
\end{proof}
The next corollary is an immediate consequence and contains the relevant result for the term-by-term expansion of the magnetic product. 
\begin{cor}\label{appendix:asymptotics:existence_osc_int:Lemma2}
	Let $f \in \Hoerrd{m_1}$, $g \in \Hoerrd{m_2}$, $\rho \in [0,1]$ and $a , \alpha, b, \beta \in \N_0^d$ be arbitrary multiindices. Then for all functions $B \in \mathcal{BC}^{\infty}(\R^d)$ the oscillatory integral 
	\begin{align}
		G(X) := \frac{1}{(2\pi)^{2d}} \int_{\Pspace} \dd Y \int_{\Pspace} \dd Z e^{i \sigma(X,Y+Z)} B(x) \, y^a \eta^{\alpha} (\Fs^{-1} f)(Y) \, z^b \zeta^{\beta} (\Fs^{-1} g)(Z) 
	\end{align}
	exists, is in symbol class $\Hoerrd{m_1 + m_2 - (\sabs{a} + \sabs{b}) \rho}$ and yields 
	\begin{align}
		B(x) \, \bigl ( (-i \partial_{\xi})^a (+i \partial_x)^{\alpha} f \bigr )(X) \, \bigl ( (-i \partial_{\xi})^b (+i \partial_x)^{\beta} g \bigr )(X) 
		. 
	\end{align}
\end{cor}
In the proof of Corollary~\ref{appendix:asymptotics:existence_osc_int:Lemma2} we have used that we could write the integrals as a \emph{product} of two independent integrals. There is, however, a second relevant type of oscillatory integral that cannot be `untangled.' Fortunately, we only need to ensure their existence and not evaluate them explicitly. Again, we will start with a simpler integral over only one phase space variable and then extend the ideas to the full integral in a corollary. 
\begin{lem}\label{appendix:asymptotics:existence_osc_int:lem:remainder}
	Assume $f \in \Hoerrd{m_1}$, $g \in \Hoerrd{m_2}$, $\rho \in [0,1]$, $\eps \in (0,1]$ and $\tau , \tau' \in [0,1]$. Furthermore, let $G_{\tau'} \in \mathcal{BC}^{\infty} \bigl ( \R^d_x,\mathcal{C}^{\infty}_{\mathrm{pol}}(\R^d_y \times \R^d_z) \bigr )$ be such that for all $c , c' , c'' \in \N_0^d$ 
	\begin{align*}
		\babs{\partial_x^c \partial_y^{c'} \partial_z^{c''} G_{\tau'}(x,y,z)} \leq C_{c c' c''} \bigl ( \sexpval{y} + \sexpval{z} \bigr )^{\sabs{c} + \sabs{c'} + \sabs{c''}} 
	\end{align*}
	holds for some finite constant $C_{c c' c''} > 0$. Then for all $a , \alpha , b , \beta \in \N_0^d$ and $\tau , \tau' \in [0,1]$ 
	\begin{align}
		I_{\tau \tau'}(x,\xi) &:= \frac{1}{(2 \pi)^{2d}} \int_{\Pspace} \dd Y \int_{\Pspace} \dd Z \, e^{i \sigma(X,Y+Z)} \, e^{i \tau \frac{\eps}{2} \sigma(Y,Z)} 
		\cdot \notag \\
		&\cdot 
		G_{\tau'}(x,y,z) \, y^a \eta^{\alpha} z^b \zeta^{\beta} \, (\Fs f)(Y) \, (\Fs g)(Z) 
		\label{appendix:asymptotics:existence_osc_int:eqn:remainder}
	\end{align}
	exists as an oscillatory integral in $\Hoerrd{m_1 + m_2 - (\sabs{a} + \sabs{b}) \rho}$. The map $(\tau , \tau') \mapsto I_{\tau \tau'}$ is continuous. 
\end{lem}
\begin{proof}
	Let us rewrite the integral first, the result will serve as a definition for the oscillatory integral $I_{\tau \tau'}$: 
	\begin{align}
		I_{\tau \tau'}(x,\xi) &= \frac{1}{(2\pi)^{4d}} \int_{\Pspace} \dd Y \int_{\Pspace} \dd \tilde{Y} \int_{\Pspace} \dd Z \int_{\Pspace} \dd \tilde{Z} \, \bigl ( (+i \partial_{\tilde{\eta}})^a (-i \partial_{\tilde{y}})^{\alpha} e^{i \sigma(X - \tilde{Y} , Y)} \bigr ) \, 
		\cdot \notag \\
		&\qquad \qquad \cdot 
		\bigl ( (+i \partial_{\tilde{\zeta}})^b (-i \partial_{\tilde{z}})^{\beta} e^{i \sigma(X - \tilde{Z} , Z)} \bigr ) 
		e^{i \tau \frac{\eps}{2} \sigma(Y,Z)} \, G_{\tau'}(x,y,z) \, f(\tilde{Y}) \, g(\tilde{Z}) 
		\notag \\
		&= \frac{1}{(2\pi)^{4d}} \int_{\Pspace} \dd Y \int_{\Pspace} \dd \tilde{Y} \int_{\Pspace} \dd Z \int_{\Pspace} \dd \tilde{Z} \, e^{i \sigma(X - \tilde{Y} , Y)} \, e^{i \sigma(X - \tilde{Z} , Z)} \, e^{i \tau \frac{\eps}{2} \sigma(Y,Z)} 
		\cdot \notag \\
		&\qquad \qquad \cdot 
		G_{\tau'}(x,y,z) \, \bigl ( (-i \partial_{\tilde{\eta}})^a (+i \partial_{\tilde{y}})^{\alpha} f \bigr )(\tilde{Y}) \, \bigl ( (-i \partial_{\tilde{\zeta}})^b (+i \partial_{\tilde{z}})^{\beta} g \bigr )(\tilde{Z}) 
		\label{appendix:asymptotics:existence_osc_int:eqn:remainder_derivatives_converted}
	\end{align}
	By assumption $\partial_x^{\alpha} \partial_{\xi}^a f \in \Hoerrd{m_1 - \abs{a} \rho}$ as well as $\partial_x^{\beta} \partial_{\xi}^b g \in \Hoerrd{m_2 - \abs{b} \rho}$ and we see that it suffices to consider the case $a = b = \alpha = \beta = 0$. In this particular case, we estimate all seminorms: let $n , \nu \in \N_0^d$. Then we have to bound 
	\begin{align*}
		\partial_x^n \partial_{\xi}^{\nu} I_{\tau \tau'}(x,\xi) &= \sum_{\substack{a + b + c = n \\ \alpha + \beta = \nu}} \frac{1}{(2\pi)^{4d}} \int_{\Pspace} \dd Y \int_{\Pspace} \dd \tilde{Y} \int_{\Pspace} \dd Z \int_{\Pspace} \dd \tilde{Z} \, \bigl ( \partial_{x}^a \partial_{\xi}^{\alpha} e^{i \sigma(X - \tilde{Y} , Y)} \bigr ) \, 
		\cdot \\
		&\qquad \qquad \qquad \cdot 
		\bigl ( \partial_{x}^b \partial_{\xi}^{\beta} e^{i \sigma(X - \tilde{Z} , Z)} \bigr ) \, e^{i \tau \frac{\eps}{2} \sigma(Y,Z)} 
		\partial_x^c G_{\tau'}(x,y,z) \, f(\tilde{Y}) \, g(\tilde{Z}) 
		\\
		&= \sum_{\substack{a + b + c = n \\ \alpha + \beta = \nu}} \frac{1}{(2\pi)^{4d}} \int_{\Pspace} \dd Y \int_{\Pspace} \dd \tilde{Y} \int_{\Pspace} \dd Z \int_{\Pspace} \dd \tilde{Z} \, e^{i \sigma(X - \tilde{Y} , Y)} \, e^{i \sigma(X - \tilde{Z} , Z)} \, 
		\cdot \\
		&\qquad \qquad \qquad \cdot 
		e^{i \tau \frac{\eps}{2} \sigma(Y,Z)} \, \partial_x^c G_{\tau'}(x,y,z) \, \partial_{\tilde{y}}^a \partial_{\tilde{\eta}}^{\alpha} f(\tilde{Y}) \, \partial_{\tilde{z}}^b \partial_{\tilde{\zeta}}^{\beta} g (\tilde{Z}) 
		\\
		&= \sum_{\substack{a + b + c = n \\ \alpha + \beta = \nu}} \int_{\R^d} \dd y \int_{{\R^d}^*} \dd \eta \int_{\R^d} \dd z \int_{{\R^d}^*} \dd \zeta \, e^{i \eta \cdot y} \, e^{i \zeta \cdot z} \, \partial_x^c G_{\tau'}(x,y,z) 
		\cdot \\
		&\qquad \qquad \qquad \cdot 
		\partial_x^a \partial_{\xi}^{\alpha} f \bigl ( x - \tfrac{\tau \eps}{2} z , \xi - \eta \bigr ) \, \partial_x^b \partial_{\xi}^{\beta} g \bigl ( x + \tfrac{\tau \eps}{2} y , \xi - \zeta \bigr ) 
		. 
	\end{align*}
	from above by an integrable function. To do that, we insert powers of $\sexpval{y}^{-2}$, $\sexpval{z}^{-2}$, $\sexpval{\eta}^{-2}$ and $\sexpval{\zeta}^{-2}$ via the usual trick, \eg $\sexpval{y}^{-2} (1 - \Delta_{\eta}) e^{i \eta \cdot y} = e^{i \eta \cdot y}$. To simplify notation, we set $L_y := 1 - \Delta_y$; $L_z$, $L_{\eta}$ and $L_{\zeta}$ are defined analogously. Then, we have for any non-negative integers $N_1 , N_2 , K_1 , K_2$ 
	\begin{align}
		\partial_x^n \partial_{\xi}^{\nu} &I_{\tau \tau'}(x,\xi) = 
		\negmedspace \negmedspace 
		\sum_{\substack{a + b + c = n \\ \alpha + \beta = \nu}} 
		\negmedspace \negmedspace 
		(2\pi)^{-2d} \int_{\Pspace} \dd Y \int_{\Pspace} \dd Z \, \bigl ( \sexpval{y}^{-2N_1} L_{\eta}^{N_1} e^{i \eta \cdot y} \bigr ) \, \bigl ( \sexpval{z}^{-2N_2} L_{\zeta}^{N_2} e^{i \zeta \cdot z} \bigr ) \, 
		\cdot \notag \\
		&
		\negmedspace \negmedspace \negmedspace \negmedspace \negmedspace \negmedspace \negmedspace \negmedspace \negmedspace \negmedspace 
		\qquad \qquad \qquad \qquad \cdot 
		\partial_x^c G_{\tau'}(x,y,z) \, \partial_x^a \partial_{\xi}^{\alpha} f \bigl ( x - \tfrac{\tau \eps}{2} z , \xi - \eta \bigr ) \, \partial_x^b \partial_{\xi}^{\beta} g \bigl ( x + \tfrac{\tau \eps}{2} y , \xi - \zeta \bigr ) 
		\displaybreak[2] \notag \\
		&= 
		\negmedspace \negmedspace \negmedspace \negmedspace \negmedspace \negmedspace \negmedspace \negmedspace \negmedspace \negmedspace 
		\sum_{\substack{a + b + c = n \\ \alpha + \beta = \nu \\ \sabs{\alpha'} \leq 2 N_1 , \, \sabs{\beta'} \leq 2 N_2}} \negmedspace \negmedspace \negmedspace \negmedspace \negmedspace \negmedspace \negmedspace \negmedspace 
		C_{\alpha' \beta'} \int_{\Pspace} \dd Y \int_{\Pspace} \dd Z \, \bigl ( \sexpval{\eta}^{-2K_1} L_{y}^{N_1} e^{i \eta \cdot y} \bigr ) \, \bigl ( \sexpval{\zeta}^{-2K_2} L_{z}^{N_2} e^{i \zeta \cdot z} \bigr ) 
		\cdot \notag \\
		&
		\negmedspace \negmedspace \negmedspace \negmedspace \negmedspace \negmedspace \negmedspace \negmedspace \negmedspace \negmedspace 
		\qquad \qquad \qquad \qquad \cdot 
		\sexpval{y}^{-2N_1} \, \sexpval{z}^{-2N_2} \, \partial_x^c G_{\tau'}(x,y,z) 
		\cdot \notag \\
		&
		\negmedspace \negmedspace \negmedspace \negmedspace \negmedspace \negmedspace \negmedspace \negmedspace \negmedspace \negmedspace 
		\qquad \qquad \qquad \qquad \cdot 
		\partial_x^a \partial_{\xi}^{\alpha + \alpha'} f \bigl ( x - \tfrac{\tau \eps}{2} z , \xi - \eta \bigr ) \, \partial_x^b \partial_{\xi}^{\beta + \beta'} g \bigl ( x + \tfrac{\tau \eps}{2} y , \xi - \zeta \bigr ) 
		\displaybreak[2] \notag \\
		&= 
		\negmedspace \negmedspace \negmedspace \negmedspace \negmedspace \negmedspace \negmedspace \negmedspace \negmedspace \negmedspace 
		\sum_{\substack{a + b + c = n \\ \alpha + \beta = \nu \\ \sabs{\alpha'} \leq 2 N_1 , \, \sabs{\beta'} \leq 2 N_2 \\ \sabs{a'} + \sabs{b'} + \sabs{c'} \leq 2 K_1 \\ \sabs{a''} + \sabs{b''} + \sabs{c''} \leq 2 K_2 }} 
		\negmedspace \negmedspace \negmedspace \negmedspace \negmedspace \negmedspace \negmedspace \negmedspace \negmedspace \negmedspace 
		C^{a b c \alpha \beta}_{\alpha' \beta' a' b' c' a'' b'' c''} (\eps \tau)^{\sabs{a''} + \sabs{b'}} \int_{\Pspace} \dd Y \int_{\Pspace} \dd Z \, e^{i \eta \cdot y} \, e^{i \zeta \cdot z} \, \sexpval{y}^{-2N_1} \, \sexpval{z}^{-2N_2} 
		\cdot \notag \\
		&
		\negmedspace \negmedspace \negmedspace \negmedspace \negmedspace \negmedspace \negmedspace \negmedspace \negmedspace \negmedspace 
		\qquad \qquad \qquad \qquad \cdot 
		\sexpval{\eta}^{-2K_1} \, \sexpval{\zeta}^{-2K_2} \, \varphi_{N_1 a'}(y) \, \varphi_{N_2 b''}(z) \, \partial_x^c \partial_y^{x'} \partial_z^{c''} G_{\tau'}(x,y,z) 
		\cdot \notag \\
		&
		\negmedspace \negmedspace \negmedspace \negmedspace \negmedspace \negmedspace \negmedspace \negmedspace \negmedspace \negmedspace 
		\qquad \qquad \qquad \qquad \cdot 
		\partial_x^{a + a''} \partial_{\xi}^{\alpha + \alpha'} f \bigl ( x - \tfrac{\tau \eps}{2} z , \xi - \eta \bigr ) \, \partial_x^{b + b'} \partial_{\xi}^{\beta + \beta'} g \bigl ( x + \tfrac{\tau \eps}{2} y , \xi - \zeta \bigr ) 
		\label{appendix:existence:remainder:eqn:typical_int}
	\end{align}
	Here, the \emph{bounded} functions $\varphi_{N a}$ are defined by $\partial_x^a \sexpval{x}^{-2N} =: \sexpval{x}^{-2N} \, \varphi_{N a}(x)$ for all $N \in \N_0$, $a \in \N_0^d$, and the constants appearing in the sum are defined implicitly. We now estimate the absolute value of each of the terms in the integral in order to find $N_1 , N_2 , K_1 , K_2 \in \N_0$ large enough so that the right-hand side of the above consists of a finite sum of integrable functions. Using the assumptions on $G_{\tau'}$ and the standard estimate $\sexpval{\xi - \eta}^m \leq \sexpval{\xi}^m \sexpval{\eta}^{\sabs{m}}$, we can bound the integrand of the right-hand side of~\eqref{appendix:existence:remainder:eqn:typical_int} in absolute value by 
	\begin{align*}
		&\sum_{\substack{a + b + c = n \\ \alpha + \beta = \nu \\ \sabs{\alpha'} \leq 2 N_1 , \, \sabs{\beta'} \leq 2 N_2 \\ \sabs{a'} + \sabs{b'} + \sabs{c'} \leq 2 K_1 \\ \sabs{a''} + \sabs{b''} + \sabs{c''} \leq 2 K_2 }} \negmedspace \negmedspace \negmedspace \negmedspace \negmedspace \negmedspace \negmedspace \negmedspace \negmedspace 
		\tilde{C}^{a b c \alpha \beta}_{\alpha' \beta' a' b' c' a'' b'' c''} \, \sexpval{y}^{-2N_1 + \sabs{c} + \sabs{c'} + \sabs{c''}} \sexpval{z}^{-2N_2 + \sabs{c} + \sabs{c'} + \sabs{c''}} 
		\cdot \\
		&\qquad \qquad \qquad \qquad \cdot 
		\sexpval{\eta}^{-2K_1} \sexpval{\zeta}^{-2K_2} 
		\sexpval{\xi - \eta}^{m_1 - (\sabs{\alpha} + \sabs{\alpha'}) \rho} \, \sexpval{\xi - \zeta}^{m_2 - (\sabs{\beta} + \sabs{\beta'}) \rho}
		\\
		&\qquad \qquad \leq 
		C \expval{\xi}^{m_1 + m_2 - \abs{\nu} \rho} \, \sexpval{y}^{-2 N_1 + \abs{n} + 2 K_1 + 2 K_2} \, \sexpval{z}^{-2 N_2 + \abs{n} + 2 K_1 + 2 K_2} \, 
		\cdot \\
		&\qquad \qquad \qquad \qquad \cdot 
		\sum_{\alpha + \beta = \nu} \sexpval{\eta}^{-2 K_1 + \sabs{m_1 - \sabs{\alpha} \rho}} \, \sexpval{\zeta}^{-2 K_2 + \sabs{m_2 - \sabs{\beta} \rho}} 
		\\
		&\qquad \qquad \leq 
		\tilde{C} \sexpval{\xi}^{m_1 + m_2 - \sabs{\nu} \rho} \, \sexpval{y}^{-2 N_1 + \sabs{n} + 2 K_1 + 2 K_2} \, \sexpval{z}^{-2 N_2 + \sabs{n} + 2 K_1 + 2 K_2} \, 
		\cdot \\
		&\qquad \qquad \qquad \qquad \cdot 
		\sexpval{\eta}^{-2 K_1 + \sabs{m_1} + \sabs{\nu} \rho} \, \sexpval{\zeta}^{-2 K_2 + \sabs{m_2} + \sabs{\nu} \rho} 
		. 
	\end{align*}
	Choosing $K_1$ and $K_2$ such that $-2 K_j + \sabs{m_j} + \sabs{\nu} \rho < - d$, $j = 1 , 2$, ensures integrability in $\eta$ and $\zeta$. Now that $K_1$ and $K_2$ are fixed, we choose $N_1$ and $N_2$ such that $- 2 N_j + \abs{n} + 2 K_1 + 2 K_2 < - d$ and the right-hand side of the above is an integrable function in $y$, $\eta$, $z$ and $\zeta$ which dominates the absolute value of~\eqref{appendix:existence:remainder:eqn:typical_int}. Thus, we have shown 
	\begin{align*}
		\babs{\partial_x^n \partial_{\xi}^{\nu} I_{\tau \tau'}(x,\xi)} \leq C_{n \nu} \sexpval{\xi}^{m_1 + m_2 - \sabs{\nu} \rho} 
	\end{align*}
	for all $n , \nu \in \N_0^d$ and hence $I_{\tau \tau'}$ exists in $\Hoerrd{m_1 + m_2}$ if the exponents of $y$, $\eta$, $z$ and $\zeta$ in equation~\eqref{appendix:asymptotics:existence_osc_int:eqn:remainder} all vanish, $a = \alpha = b = \beta = 0$. Similarly, for general $a , \alpha , b , \beta \in \N_0^d$, we conclude $I_{\tau \tau'} \in \Hoerrd{m_1 + m_2 - (\sabs{a} + \sabs{b}) \rho}$. Since the above bounds are uniform in $\tau$ and $\tau'$, the continuity of $(\tau , \tau') \mapsto I_{\tau \tau'}$ in the Fréchet topology of $\Hoerrd{m_1 + m_2 - (\sabs{a} + \sabs{b}) \rho}$ follows from dominated convergence. 
\end{proof}
%



\backmatter


\printbibliography
\printindex

\end{document}
